%% file: paper.tex
\documentclass[citeautoscript,aps,pre,reprint,superscriptaddress,longbibliographyaps,nofootinbib]{revtex4-2}

\usepackage{booktabs}
\usepackage{graphicx}
\graphicspath{{figures/}}
\usepackage[dvipsnames]{xcolor}
\usepackage{upgreek}
\usepackage{float}
\usepackage{lipsum}
\usepackage{mathrsfs}
\usepackage{amsmath, amssymb, amsthm}
\usepackage{mathtools}
\usepackage{bbm}
\usepackage{dsfont}
\usepackage{bm}
\usepackage{enumitem}
\usepackage{braket}
\usepackage{url}
\usepackage{framed}
\usepackage{adjustbox}

\usepackage[most]{tcolorbox}
\usepackage{xparse}
\usepackage{lipsum}
\usepackage{tikz}
\usetikzlibrary{positioning,calc,decorations.pathreplacing,shapes.geometric,matrix}

\usepackage[colorlinks=true,linkcolor=teal,citecolor=teal,urlcolor=teal]{hyperref}

\usepackage[dvipsnames]{xcolor}
\definecolor{BlueRome}{HTML}{4287f5}
\definecolor{cyanRed}{RGB}{252, 5, 120}
\hypersetup{colorlinks=true, linkcolor=BlueRome, citecolor=BlueRome, urlcolor=cyanRed}
\usepackage[capitalise]{cleveref}

\usepackage{thm-restate}

\newcommand{\id}{\mathbbm{1}}
\newcommand{\tr}{\operatorname{Tr}}
\newcommand{\rrangle}{\rangle \! \rangle}
\newcommand{\llangle}{\langle \! \langle}

\newcommand{\T}{\operatorname{T}}

\newcommand{\dd}{\operatorname{d}\!}
\newcommand{\Tr}{\operatorname{Tr}}

\newcommand*{\nn}{\nonumber}

\definecolor{labelgray}{gray}{0.45}

\newtheorem{theorem}{Theorem}
\newtheorem{corollary}{Corollary}
\newtheorem{lemma}{Lemma}

\newtheorem{application}{Application}

\crefname{application}{Application}{Applications}
\Crefname{application}{Application}{Applications}

\input{figures/tn_preamble.tex} %colour definitions etc.

\input{figures/gaussian_conj.tex} %UcUdag
\input{figures/orthog_action.tex} %Rc
\input{figures/string_tensor.tex} %\gamma_s
\input{figures/c_product.tex} %\prod_i c_i
\input{figures/string_conjugation.tex} %U\gamma_s Udag
\input{figures/c_product_conj.tex} %U\prod_i c_i Udag
\input{figures/c_product_expanded_conj.tex} %Uc_1 Udag U c_2 Udag U ...
\input{figures/R_prod_conj.tex} %R c_1 R c_2 ...
\input{figures/c_prod_conj_2_replica.tex}

\input{figures/c_product_2_rep_conj_vert.tex}
\input{figures/two_replica_R_conj.tex}
\input{figures/weingarten_p2_part.tex}
\input{figures/weingarten_p1_part.tex}

\input{figures/weingarten_pairing_examples.tex}
\input{figures/weingarten_pairing_input.tex}
\input{figures/weingarten_pairing_output.tex}
\input{figures/contingency_eg_exp.tex}
\input{figures/contingency_eg_exp_long.tex}
\input{figures/contingency_eg_consolid.tex}
\input{figures/loe_operator.tex}
\input{figures/otoc_operator.tex}
\input{figures/full_loe_contract.tex}

\input{figures/loe_pauli_diagram.tex}
\input{figures/loe_pauli_diagram_expanded.tex}
\input{figures/loe_delta_relations.tex}
\input{figures/loe_delta_relations_col.tex}
\input{figures/loe_class_1.tex}

\input{figures/graph_id.tex}
\input{figures/graph_pair.tex}
\input{figures/graph_looped.tex}
\input{figures/graph_combined.tex}
\input{figures/otoc_contraction_k.tex}

\input{figures/otoc_coeffs_4.tex}
\input{figures/ose_class.tex}

\input{figures/pairing_classes.tex}
\input{figures/pairing_e1_w2.tex}
\input{figures/pairing_e2_w2.tex}
\input{figures/pairing_e3_w2.tex}
\input{figures/pairing_e1_w4.tex}
\input{figures/pairing_e2_w4.tex}
\input{figures/pairing_e3_w4.tex}

\input{figures/colouring_orbit.tex}
\input{figures/application_boundary.tex}
\input{figures/identity_two.tex}
\input{figures/swap.tex}
\input{figures/bell_state.tex}
\input{figures/contingency_inputs.tex}

\begin{document}

\title{Exact Operator Complexity Measures from Random Matchgate Unitaries}

\author{Gregory A. L. White}
\email{gregory.white@fu-berlin.de}
\affiliation{Dahlem Center for Complex Quantum Systems, Freie Universit\"at Berlin, 14195 Berlin, Germany}

\author{Jens Eisert}
\email{jense@zedat.fu-berlin.de}
\affiliation{Dahlem Center for Complex Quantum Systems, Freie Universit\"at Berlin, 14195 Berlin, Germany}
\affiliation{Helmholtz-Zentrum Berlin f{\"u}r Materialien und Energie, Berlin, Germany}

\author{Neil Dowling}
\email{ndowling@uni-koeln.de}
\affiliation{Institut f\"ur Theoretische Physik, Universit\"at zu K\"oln, Z\"ulpicher Strasse 77, 50937 K\"oln, Germany}

\begin{abstract}
Matchgate circuits describe classically tractable dynamics in a physically relevant setting while nonetheless exhibiting features characteristic of complex quantum systems, such as high amounts of magic. This invites the broader question: to what extent are typical complexity measures saturated under trivial dynamics? Here, we address this by studying the behaviour of different Heisenberg-picture complexity measures in random matchgate ensembles. Our main technical contribution is a tractable and generalisible free-fermionic unitary Weingarten calculus formulated entirely in the Majorana basis. This method is tailored to computing moments of operators, and simplifies to a combinatorics problem which depends on the boundary conditions prescribed by a given operator functional. Using this framework, we derive closed-form expressions for the typical local-operator entanglement, operator stabiliser entropy, and higher-order out-of-time-ordered correlators. For extensive initial Majorana strings, both operator entanglement and stabiliser entropy obey volume laws, even though the underlying circuits remain classically simulable. Out-of-time-ordered correlators can likewise become exponentially small, showing that their late-time saturation value alone does not witness computational hardness and constituting a classically tractable example of emergent free independence. Finally, we relate these separations to operator resource theories, showing how operator entanglement beyond the Gaussian baseline lower-bounds the number of non-Gaussian gates in doped matchgate circuits. Our results present a practical toolbox for computing exact averages over random matchgates, while clarifying the distinction between quantum resources versus simulation complexity, and their dependence on the representation being exploited.
\end{abstract}

\maketitle

\section{Introduction}\label{sec:introduction}
Operator dynamics plays a central role in the modern understanding of many-body systems {and their non-equilibrium behaviour~\cite{1408.5148,PolkovnikovReview}.} There are two immediate advantages to working in the Heisenberg picture. On the one hand, we gain access to new physical phenomena, inaccessible {directly} to any Schr\"odinger picture-based protocol. 
This includes those predicted by theories of operator hydrodynamics~\cite{Keyserlingk2018,Khemani2018}, higher-order correlations of ergodic systems~\cite{Pappalardi2022freeETH,Foini2019}, and information scrambling~\cite{Hosur2016,Swingle2016,Mi2021}. On the other hand, as quantified, for example, by the corresponding resource theories of entanglement~\cite{Dowlin2024LOE-OSRE} or non-stabiliserness (`magic')~\cite{dowling2024magicheis}, classical simulation in the Heisenberg picture can offer an exponential advantage {over their Schr{\"o}dinger counterpart}. A particularly emblematic example is constituted by that of free-fermionic dynamics. Any initially local operator leads to at-fastest logarithmic growth of operator entanglement~\cite{Prosen2007a,Prosen2009,Dubail_2017}---implying efficient operator tensor-network simulation~\cite{Hartmann2009,dowling2026classicalsim}---whereas state entanglement in a corresponding quench experiment satisfies a volume law~\cite{Calabrese_2005}, meaning an exponential cost of any matrix-product state-based technique~\cite{Schuch_2008}. What is the source of this stark difference? Does the inherent structure of fermionic unitaries lead to the low-operator-resource phenomenon?

\begin{figure*}[t!]
  \centering
  \includegraphics[width=\linewidth]{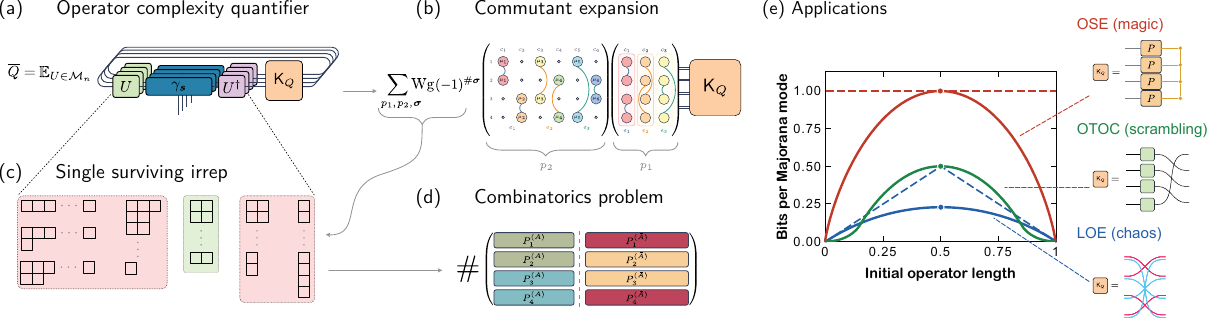}
  \caption{An overview of steps taken and contributions in this work. (a) We consider operational replica quantities of the form $\Tr[(U \gamma_{\bm{\nu}} U^\dagger)^{\otimes k} \mathsf{K}_Q]$ averaged over the ensemble of free-fermionic Gaussian unitaries. $\gamma_{\bm{\nu}}$ is an initial operator whose Heisenberg evolution we probe according to some quantity $Q$ via a boundary operator $\mathsf{K}_Q$. (b) The first panel average is converted into a sum over the free-fermionic commutant, for which we use the Brauer algebra as our basis. (c) We show (\cref{lem:unique-weingarten-sector}) that the resulting expression simplifies substantially, with support on only a single irrep sector, allowing for direct evaluation of the known Weingarten eigenvalue on this sector. (d) The problem thus reduces to a question of counting the basis elements and their prefactors which survive the constraints of the boundary operator $\mathsf{K}_Q$. (e) We compute exact finite and asymptotic expressions for three choices of operator complexity, corresponding to three different $\mathsf{K}_Q$: the \emph{operator stabiliser entropy} (OSE), which quantifies Heisenberg-picture magic; \emph{out-of-time-order correlators}  (OTOCs), which probe scrambling; and \emph{local-operator entanglement}  (LOE), which is conjectured to be a dynamical indicator of many-body chaos. Note that the displayed OTOC data is on a negative-logarithm scale. }
  \label{fig:summary}
\end{figure*}

Gaussian unitaries generalise the notion of free fermions to any Hamiltonian involving only quadratic Majorana terms. Under the Jordan-Wigner transform, such unitaries are one-to-one with matchgate circuits~\cite{Jozsa_2008}. States produced by such circuits/dynamics are characterised by a polynomial-size correlation matrix, making them classically tractable to simulate. A natural avenue for exploring operator complexities for generic Gaussian dynamics is to consider universal, long-time features when averaging over random realizations. Random matchgate circuits of local, two-qubit gates equilibrate to a globally random Gaussian in linear depth~\cite{wan_matchgate_2023,Dias2024classicalsimulation,zhao2023learningopt,braccia2025optimalhaarrandomfermionic}. 
However, averaging even over this global Gaussian quickly becomes difficult. Recent progress in characterising the matchgate commutant has shown that its dimension scales with system size $N$~\cite{sierant2026theorymatchgatecommutant,braccia2026commutantfermionicgaussianunitaries,West_2025}.
This means that even for a fixed number of replicas, an average moment operator becomes analytically intractable through direct calculation, 
in contrast to the Haar unitary case~\cite{collins_integration_2006}. Moreover, the orthonormal commutant 
basis has non-trivial overlap with the Majorana basis~\cite{sierant2026theorymatchgatecommutant,braccia2026commutantfermionicgaussianunitaries,lastres2026geometryfreefermioncommutants}, and so does not obviously help in the task of computing fermionic operator moments. 

Here, we resolve some of these issues and develop methods to analytically compute operator complexities from ensembles of uniformly random Gaussian unitaries. Namely, we compute the matchgate twirl of replicas of arbitrary-length initial Majorana strings; i.e., any Pauli string under the Jordan-Wigner transform. We adopt a first-principles replica approach, leading to our main technical contribution of a closed formula for the replica twirl operator in the Majorana basis, extending the known constructions for $k\leq 3$ replicas~\cite{wan_matchgate_2023, chapman2025fermioni}. Although we only explicitly calculate $k=4$ replica quantities, our \emph{structural} simplification is fully general, and opens the door to computing such quantities are higher replica count (c.f.~\cref{lem:unique-weingarten-sector}).
A key point in our approach is that the coefficients of the replica operator reduce to a counting problem of certain \textit{contingency tables}, {a well-studied object in combinatorics~\cite{Barvinok2008Contingency}.} We show that the complexity of the calculation also reduces significantly upon imposing boundary conditions from the choice of operator, which corresponds to imposing extra symmetry constraints on the tables. 

To showcase the utility of our method, we proceed to compute examples from many-body physics. We calculate the matchgate average of three distinct operator complexities: \emph{(local-)operator entanglement entropy} (LOE), \emph{operator stabiliser entropy} (OSE), and higher-order \emph{out-of-time-order correlators} (OTOCs). Each of these quantities measures a different aspect of how a Heisenberg operator spreads through a system, respectively: how non-local it is (`entanglement'~\cite{Zanardi2001}), how delocalised it is in the Pauli basis (`magic'~\cite{dowling2024magicheis}), and to what degree it does not commute with a reference operator (`free independence'~\cite{Pappalardi2022freeETH}). Each quantity is considered a dynamical indicator of quantum chaos, and each depends non-trivially on $k \geq 4$ copies of the evolved operator, \textit{a priori} inaccessible using existing techniques. We derive closed-form expressions for each quantity at $k=4$ and provide a general recipe for higher-replica versions. In~\cref{fig:summary} we supply a plot summarising our results as well as the methodology of obtaining them.

Our results have implications for classical simulation and quantum algorithms. LOE is the entropy of the Schmidt spectrum of the \emph{matrix-product operator} (MPO) representation of an operator, and so its scaling measures the cost of the MPO  representation~\cite{dowling2026classicalsim}. The OSE is the entropy of the squared coefficients of an operator expanded in the Pauli basis, a measure of the operator's magic resource. It has an analogous operational interpretation compared to LOE, but for Pauli propagation algorithms~\cite{dowling2024magicheis}, whereby operator Pauli components with small coefficients are iteratively discarded~\cite{Rakovszky2022,begusic2024realtime,schuster2024polynomialtime,rudolph2025paulipropag}. In each case, a volume-law scaling of the R\'enyi entropies with index $\alpha\geq 1$, rigorously implies the asymptotically prohibitive cost of the corresponding simulation algorithm up to inverse polynomial errors~\cite{dowling2024magicheis,dowling2026classicalsim}. We find that for extensive initial Majorana strings, both OSE and LOE saturate a volume law on average for deep matchgate circuits, and therefore cannot be represented efficiently using tensor networks or sparse Pauli strings. A notable feature is that while the OSE saturates the upper bound imposed by the fermionic symmetry constraints, the LOE satisfies a suppressed volume-law, with asymptotically less than half the entanglement compared to the (maximal) unitary Haar operator Page curve~\cite{dowling2026page}; cf. the dashed lines in~\cref{fig:summary}(e). This means that while deep, random matchgates lead to a relatively uniform spread of Pauli strings (comparable to the \emph{operator Porter-Thomas distribution}~\cite{dowling2026noise}), the entanglement of these strings is relatively suppressed.

OTOCs quantify how much a (usually) initially local operator fails to commute with a reference observable. While not directly related to the cost of any algorithm, they have recently been measured in experiments and conjectured to be generically hard to simulate classically~\cite{Abanin2025}, surely in the worst case and possibly also in a meaningful sense of average-case. As such, they are an excellent venue for studying the boundary between classically simulable regimes and those where one can expect a quantum advantage in a precise sense. Most relevant to our setting, at long times higher-order OTOCs provide a bridge between random unitaries and chaotic dynamics: in both cases the OTOCs saturate to a universal expression described by free probability. This includes, for instance, both Haar random unitary circuits~\cite{nica2006lectures,fava2023designsfreeprobability,dowling2025freeindep,Claeys2026} and Hamiltonians satisfying (a generalisation of) the \emph{eigenstate thermalisation hypothesis} (ETH)~\cite{Foini2019,Pappalardi2022freeETH,FritzschLattices2025}. Here, we find an exponentially small saturation value of the $8$-point OTOC of extensive Majorana strings over random matchgates, a universal signature of free independence. Our results affirm the intuition that the long-time saturation of the OTOCs is not inherently hard. When viewed in the context of recent OTOC experiments~\cite{Abanin2025}, this naturally implies that its complexity must lie in its finite spatiotemporal profile, if anywhere. In this sense, random matchgates constitute a classically tractable control ensemble for such experiments. These results make evident that even an exponentially small higher-order OTOC, together with non-trivial interference between exponentially many Pauli components, cannot by itself certify quantum computational advantage.

More broadly, these results provide an exactly solvable setting in which distinct notions of operator complexity can be separated from the complexity of simulating the underlying circuit. Matchgate circuits remain efficiently classically simulable, yet an evolved operator may become highly delocalised in the Pauli basis and, for extensive initial strings, extensively entangled in operator space. Thus, neither large operator magic, large operator entanglement, nor strong scrambling alone constitutes a witness of computational hardness: simulability is sensitive to the structure and representation exploited by a classical algorithm. This provides evidence to the suspicion of non-Gaussianity being an orthogonal resource required for quantum advantage, independent of both magic and entanglement, even in the Heisenberg picture.

The remainder of this work is laid out as follows. In~\cref{sec:majorana} we supply some background on Majorana strings, matchgate unitaries and the graphical calculus we employ. In ~\Cref{sec:twirls} we develop our main technical tool: the reduction of the $k$-fold matchgate twirl of a replicated Majorana string to a signed count over binary contingency tables with fixed margins, wherein we derive the exact respective closed-form coefficient in the case of four replicas. In~\cref{sec:apps} we apply this machinery to the computation of LOE, the OSE and higher-order OTOCs in turn, in each case deriving an exact expression and then extracting its corresponding asymptotics. Next, in~\cref{sec:reource_theories} we discuss the consequences of our results for classical simulation, and contextualise them through an analysis of operator resource theories. Lastly, in ~\cref{sec:discussion} we speculate on some fruitful open directions.

\section{Background and Notation}\label{sec:majorana}

We start by detailing some background on the algebra of Majoranas and their transformation via matchgate unitaries. As we proceed, we introduce a graphical, tensor-network notation alongside the equivalent algebraic expressions that is both pedagogical and will help us develop our main technical results. 
Throughout this work, we consider a Hilbert space of $N$ qubits with total dimension $D=2^N$. Consider a basis of $2N$ Majorana operators, or \emph{spinors}, on $N$ qubits, satisfying the anti-commutation relations
\begin{equation}
  \{c_\alpha,c_\beta \} = 2 \delta_{\alpha,  \beta} \id. \label{eq:antiComm}
\end{equation}
These can be mapped to Pauli strings of a one-dimensional spin chain through the Jordan-Wigner transformation~\cite{Jordan1928-ko,Jozsa2008},
\begin{align}\label{eq:JW}
  &c_{2\alpha-1} = \sigma_z^{(1)} \dots \sigma_z^{(\alpha-1)} \sigma_x^{(\alpha)},\\
  &c_{2\alpha} = \sigma_z^{(1)} \dots \sigma_z^{(\alpha-1)} \sigma_y^{(\alpha)}
\end{align}
where $\sigma_z^{(\alpha)}$ is the Pauli $Z$-operator acting on qubit $\alpha \leq N$. Majorana strings of length $L$ are defined as the product of $L$ distinct Majorana operators, i.e., $\gamma_{\bm{\nu}} = c_{\nu_1}c_{\nu_2}\cdots c_{\nu_L}$. We call a Majorana string finite or extensive when $\min(L,2N-L)=\mathcal{O}(1)$ or $\min(L,2N-L)=\Theta(N)$, respectively. 
We depict these strings graphically as follows:
\begin{equation}
 \stringTensor = \cProduct,
\end{equation}
where the vertical wires represent the Majorana spinor label (each of size $2N$), while the horizontal wires represent matrix (physical) indices of the operator (of size $2^N$). Often, we will apply the anti{-}commutation relation~\cref{eq:antiComm} to canonically order the spinor, $\nu_1 < \nu_2 < \dots < \nu_L$, in which case it picks up a $\pm 1$ phase. However, unlike the usual convention, we do not necessarily take the strings to be canonically ordered, and we denote their label by a vector, $\bm \mu$. 

Important examples of Majorana strings are the single-qubit Pauli operators, which can be expressed as the products:
\begin{align}
  &\sigma_x^{(\alpha)} = (-i)^{\alpha-1}c_1 c_2 \dots c_{2\alpha-2} c_{2\alpha-1}{,}\\
  &\sigma_y^{(\alpha)} = (-i)^{\alpha-1}c_1 c_2 \dots c_{2\alpha-2} c_{2\alpha} {,}\\
  &\sigma_z^{(\alpha)} =-i c_{2\alpha-1} c_{2\alpha}.
\end{align}
Through the Jordan-Wigner transform, it is immediate to deduce that, together with the identity operator $\id$, the full set of ordered and normalised Majorana strings of all lengths form an orthonormal operator basis according to the Hilbert-Schmidt inner product, $\tr[\gamma_{\bm{\nu}}^\dagger \gamma_{\bm{\nu}'}] = \delta_{\bm{\nu} \bm{\nu}'} D$. \par 
Moreover, extending the anti-commutation relations of~\cref{eq:antiComm} to arbitrary, canonically-ordered strings means that we need to keep track of a possible phase. We state the useful property that $\gamma_{\bm{\nu}}\gamma_{\bm{\nu}'} = (-1)^{|\bm{\nu}||\bm{\nu}'| - |\bm{\nu}\cap \bm{\nu}'|}\gamma_{\bm{\nu}'}\gamma_{\bm{\nu}}$, where $|\bm{\nu}\cap \bm{\nu}'|$ is the number of shared spinor components. Alternatively, one can express this as $\gamma_{\bm{\nu}}\gamma_{\bm{\nu}'} = (-1)^{I(\bm{\nu},\bm{\nu}')} \gamma_{\bm{\nu}\Delta \bm{\nu}'}$, where $I(\bm{\nu},\bm{\nu}')$ is the \emph{inversion number} $I(S,T)\coloneqq \#\{(a,b)\in S\times T : a > b\}$. Every member of $\bm{\nu}'$ moves left into $\bm{\nu}$, incurring a sign flip for every $\bm{\nu}$ index greater than it, until it is either in numerical order or it collides with the same mode and produces the identity. Hence the number of transpositions is exactly $I(\bm{\nu},\bm{\nu}')$ and the remaining modes are $\bm{\nu}\Delta \bm{\nu}' = (\bm{\nu}\cup\bm{\nu}')\backslash (\bm{\nu}\cap\bm{\nu}')$.

The ensemble we consider in this work is the matchgate unitaries. These are defined as unitary operators generated from a Hamiltonian that involves only terms which have a Majorana length of $L=2$, together with a parity-flipping operator. Key for us, such a unitary may equivalently be defined through its adjoint action on the basis of Majorana spinors,
\begin{equation}
 U c_\alpha U^\dagger = \sum_\beta R_{\alpha,  \beta} c_\beta, \label{eq:matchgate_action}
\end{equation}
where $R \in \mathrm{O}(2N)$, the orthogonal group of $2N\times 2N$ matrices\footnote{Note that there are related groups that $R$ can be chosen to be, representing different symmetries of the system. $R \in \mathrm{SO}(2N)$ represents the fermionic gaussian unitaries, $U = \mathrm{exp}( \sum_{\alpha,\beta =1}^{2N} h_{\alpha,  \beta} c_\alpha c_\beta) $ for a real antisymmetric matrix $h$. Adjoining a reflection operation leads to what we call the matchgate group in this work, $\mathrm{O}(2N)$. The group $ U(N)$ corresponds to fermionic Gaussian unitaries with a particle-preserving symmetry; $\mathrm{U}(N) \subset \mathrm{SO}(2N) \subset \mathrm{O}(2N)$. }. We define the set of all such unitaries, as characterised by $R \in \mathrm{O}(2N)$, to be the matchgate group $\mathcal{M}(2^N)$, and will interchangeably call these (fermionic) Gaussian unitaries. Graphically, this transformation can be depicted as 
\begin{equation}
 \gaussConj = \orthogAct.
\end{equation}
Here, we notice that matchgate unitaries $U$ propagate through a spinor via their adjoint action, leading to its corresponding Majorana representation as a local (on the spinor indices) matrix multiplication by $R$. This allows us to determine the action of a Gaussian unitary on a Majorana string. Namely, inserting the identity $U^\dagger U$ between each unit-length spinor in the string, and then applying~\cref{eq:matchgate_action} yields
\begin{align}
  U\gamma_{\bm{\nu}}U^\dagger &= \stringConj[0.8]\nn \\
  &= \cProdConjExpanded[0.75] \nn\\[0.8em]
  &=\RProdConj[0.8] \label{eq:stringTransf}\\[0.8em]
   &=\sum_{\nu_1'\nu_2'\cdots \nu_L'}c_{\nu_1'}R_{\nu_1\nu_1'}c_{\nu_2'}R_{\nu_2\nu_2'}\cdots c_{\nu_L'}R_{\nu_L\nu_L'} \nn
\end{align}
The above expression is equivalent to a sum over determinants of a submatrix of $R$~\cite{chapman2025fermioni}
\begin{equation}
  U\gamma_{\bm{\nu}} U^\dagger = \sum_{\bm{\nu'} \in \binom{2N}{|\bm{\nu}|}} \det(R_{\bm{\nu},\bm{\nu'}})\gamma_{\bm{\nu'}}.\label{eq:det_formula}
\end{equation}
The advantage of working in this representation is that it translates the ensemble average over a unitary submanifold of $U(2^N)$ into an ensemble average over the orthogonal group $\mathrm{O}(2N)$, a compact group whose Weingarten properties are well-characterised~\cite{collins_integration_2006,Collins_2009}.

\section{Exact twirls over matchgates} \label{sec:twirls}
We look to analytically evaluate averages of replicas of Majorana operators uniformly over the action of Gaussian unitaries. The relevant channel is the $k$-fold twirl,
\begin{equation}
  \Phi^{(k)}_{\mathcal{M}}(X) := \int_{U} U^{\otimes k} X (U^{\dagger})^{\otimes k},
\end{equation}
{of central interest to the study of unitary designs~\cite{Designs}.} Generally, we will take $X=\gamma_{\bm{\nu}}^{\otimes k}$ for an arbitrary string $\bm{\nu}$ whose length is denoted by $L$, and neglect the subscript $\mathcal{M}$ as we deal solely with matchgate twirls. Our strategy is to rewrite the action of matchgates acting on a replica operator using~\cref{eq:stringTransf}, and then to apply the orthogonal Weingarten calculus via the linear action of the orthogonal matrices $R$ on the Majorana indices. It will turn out that considering the replica operator $\gamma_{\bm{\nu}}^{\otimes k}$ will help us simplify this task, before later using the boundary conditions of relevant properties of the replica operator (its average OTOC, LOE, and OSE) to further reduce the task to closed-form expressions. 

\subsection{Initial operator expression}
To start, we employ the considerations from~\cref{sec:majorana} and convert $\Phi^{(k)}[\gamma_{\bm \nu}^{\otimes k}]$ into an algebraic expression summing over the commutant of the $k$-fold orthogonal group. Consider, for example, the Haar second moment over Gaussian unitaries, starting with an initial operator of two copies of the same string $\Phi^{(2)}[\gamma_{\bm \nu}^{\otimes 2}]$.
Expanding each copy in the Majorana-string basis gives
\begin{equation}
\sum_{\substack{\bm a,\bm b : \\ |\bm{a}|=|\bm{b}|=|\bm{\nu}|}}
 \left[\int_{\mathcal G}\!\dd\mu(U) \prod_{\ell=1}^{L}R_{a_\ell , s_\ell}R_{b_\ell , s_\ell}\right] \gamma_{\bm a}\otimes\gamma_{\bm b} .
\end{equation}
Graphically, this is obtained by orienting the two copies oppositely, inserting identities $U^\dagger U$ between neighbouring Majoranas, and then replacing each dressed Majorana by an $R$ tensor~(\cref{eq:stringTransf}), to find
\begin{equation}
 \begin{aligned}
 \Phi^{(2)} [\gamma_{\bm{\nu}_1}\!\otimes\!\gamma_{\bm{\nu}_2}]&={\int_{\mathcal G}\!\dd\mu(U)}\,\cProdTwoRepConj[0.72]\\[0.5em]
 &={\int_{\mathcal G}\!\dd\mu(R)}\,\twoRepRConj[0.72] .
 \end{aligned}
\end{equation}
One can see from this that the 2-fold matchgate moment operator is equivalent to a 2$L$-fold twirl over the orthogonal group $\mathrm{O}(2N)$. Generalising this gives us that $\Phi^{(k)}[\gamma_{\bm{\nu}}^{\otimes k}]$ is equivalently a $kL$-fold orthogonal twirl.
We can therefore apply Schur-Weyl duality to reduce this averaging to a discrete sum over a basis of the $kL$-replica commutant of the orthogonal group. Note that this expression is zero when $kL$ is odd. Henceforth we assume $kL$ is even.\par 

A convenient (though not orthonormal~\cite{sierant2026theorymatchgatecommutant,braccia2026commutantfermionicgaussianunitaries}) basis for the commutant of the $2n$-fold orthogonal group is the Brauer algebra, which is the set of all perfect pairings over $2n$ vertices. In group-theoretic formulation, this is the set of equivalence classes $S_{2n}/H_n$, where $S_{2n}$, $H_n$ are the symmetric and hyperoctahedral groups, respectively. It is well-known that the Haar average of an even number $2n$ matrix elements can be written explicitly as~\cite{collins_integration_2006}
\begin{equation}\label{eq:orthogonal-weingarten-elements}
 \int_{O(d)} \dd\mu(R)\,R_{i_1 , j_1} \cdots R_{i_{2n} , j_{2n}} = \sum_{p_1,p_2\in M_{n}}\delta^{p_1}_{\bm i},\delta^{p_2}_{\bm j}, \operatorname{Wg}_{p_1,p_2}
\end{equation}
Here, $M_n$ denotes the set of perfect matchings of $\{1,\ldots,2n\}$ which, in this instance, may be interpreted as the indices of the different $R$ matrices. The indicator functions follow from the invariant structure of the orthogonal group (i.e., that $RR^{\T}=\id$) and hence not all perfect pairings are valid, only those which pair the same indices together. This restricts us to a set we call \emph{admissible} pairings. That is, for $p_1\in M_n$, we have $\delta^{p_1}_{\bm i}:=\prod_{\{r,s\}\in p_1}\delta_{i_r , i_s}$, so the factor $\delta^{p_1}_{\bm i}$ enforces equality of the indices paired by $p_1$. Lastly, the matrix $\operatorname{Wg}_{p_1,p_2}$ is the \emph{Weingarten} matrix, and is defined as the pseudoinverse of the Gram matrix $G_{p_1,p_2}:=\langle p_1,p_2\rangle$. Since the latter is factorially large in $n$ and the basis is not orthonormal, $\operatorname{Wg}$ quickly becomes unwieldy to evaluate. Nonetheless, as detailed in~\cref{ssec:irreps}, we will exploit the structure of our problem to reduce its influence in the expansion to a simple closed-form expression.  \par 

Turning to our application in the Majorana context, we have $d=2N$ and $2n=kL$, and the perfect matchings are spanned by the SWAPs and Bell projections between pairs of replicas. Graphically, then, for the $\Phi^{(2)}[\gamma_{\bm \nu}^{\otimes 2}]$ case we obtain the expansion
\begin{equation}\label{eq:mainTool}
 \begin{aligned}
 &\sum_{p_1,p_2\in M_{kL/2}}
 \operatorname{Wg}_{p_1,p_2}
 \underbrace{\WeinPTwoTerm[0.69]}_{\delta_{\bm\mu}^{p_2}}~
 \underbrace{\WeinPOneTerm[0.69]}_{\delta_{\bm\nu}^{p_1}} ,
 \end{aligned}
\end{equation}
The multi-index $\bm{\nu} = (\nu^{(1)},\nu^{(2)})$ is the concatenation of the two input strings, and similarly for $\bm\mu=(\bm\mu^{(1)},\bm\mu^{(2)})$. Thus, $p_1$ and $p_2$ respectively pair elements of these multi-indices and, due to orthonormality of the basis we have chosen, only survive if the two pairings connect the same Majorana mode. 
In other words, $\delta^{p_2}_{\bm\mu}$ says that the output Majorana labels are compatible with the matching $p_2$, while $\delta^{p_1}_{\bm\nu}$ imposes the analogous condition on the input labels, which is the analogous notion of admissibility to~\cref{eq:orthogonal-weingarten-elements}. Note that working in the Brauer algebra representation is a common technique for handling averages over matchgates~\cite{wan_matchgate_2023,Cerezo2025,jaeger2026provablescalablequantumgaussian}; we will advance this picture through a combinatorial mapping and irrep simplification.\par

For the general $k$-fold tensor product, acting on identical copies of strings of length $L$, write
\begin{equation}
 \bm\nu := (\underbrace{\nu_1,\ldots,\nu_L}_{\text{copy }1},\ldots,
 \underbrace{\nu_1,\ldots,\nu_L}_{\text{copy }k})
\end{equation}
for the concatenated input labels, and similarly write $\bm\mu=(\bm\mu^{(1)},\ldots,\bm\mu^{(k)})$ for the output labels. Applying equivalent logic gives
\begin{equation}\label{eq:k-moment-weingarten-basic}
 \Phi^{(k)}[\gamma_{\bm{\nu}}^{\otimes k}]
 =\!\sum_{p_1,p_2\in M_{kL/2}}\!\operatorname{Wg}_{p_1,p_2}\!
 \sum_{\bm\mu^{(1)},\ldots,\bm\mu^{(k)}}\!
 \delta^{p_2}_{\bm\mu} \delta^{p_1}_{\bm\nu}
 \bigotimes_{r=1}^{k}\gamma_{\bm\mu^{(r)}}\!.
\end{equation}
Before evaluating this expression, we must first characterise the sets of admissible input and output pairings imposed by the delta tensors. 

The input matching constraint $\delta^{p_1}_{\bm{\nu}}$ is straightforward to characterise because the same length-$L$ string is repeated in every replica. If the string $\gamma_{\bm{\nu}} = c_{\nu_1}c_{\nu_2}\cdots c_{\nu_L}$ were written $k$ times in a vertical column then the only permissible pairings must pair vertices only within the same column, i.e., it lies in $\mathcal{Q}_{k,L}:=M_{k/2}^{\times L}$, where $M_{k/2}$ is the set of perfect matchings of the $k$ replicas in a fixed column. In particular this requires $k$ to be even. For four replicas, we have three possible in-column pairing classes: 
\begin{equation}\label{eq:classes}
 \begin{aligned}
 e_1&=\{\{1,2\},\{3,4\}\},\\
 e_2&=\{\{1,3\},\{2,4\}\},\\
 e_3&=\{\{1,4\},\{2,3\}\}.
 \end{aligned}
\end{equation}
Thus $|\mathcal{Q}_{k,L}|=((k-1)!!)^L$. If $k$ is odd and the input string has no repeated Majoranas, the moment vanishes. Graphically, an example of this for $k=4$ and $L=4$ looks like
\begin{equation}\label{eq:input-pairing-eg}\notag
  p_1 \in \mathcal{Q}_{4,4}\quad : \quad \WeingartenInput[0.8].
\end{equation}
The case of the output $p_2$ pairings is much more complicated since it incorporates all \emph{possible} populations of $kL$ vertices with the different modes and then all possible pairings thereof. In particular, each replica need not contain identical modes and so pairings of the same mode between replicas need not appear in the same position in their respective string. Moreover, a single in-column pairing need not contain the same modes, merely the individual pairs within a column. An example is the following output pairing
\begin{equation}\label{eq:output-pairing-eg}\notag
  \text{admissible}~p_2\quad : \quad \WeingartenOutput[0.8].
\end{equation}
Characterising the full set of admissible $p_2$ is the subject of the next section.

\subsection{Contingency table isomorphism}\label{sec:contingency}
It remains to organise the admissible output labels $\bm\mu$ with respect to the $p_2$ pairings. We shall do this by casting the constraints in a structured way that simplifies both the exposition and the combinatorics. We represent possible output Majorana strings by their mode occupancies. From the setup of the problem, we have (i) no more than one occurrence of each mode per replica, by~\cref{eq:antiComm}; (ii) an even number of occupied modes across the total collection of replicas; and (iii) exactly $L$ occupied modes per replica. 
A convenient way to encode these rules is to use binary matrices whose rows and columns sum to a fixed margin -- this object is called a \emph{binary contingency tables}~\cite{miller2011exactenumerationsamplingmatrices,CANFIELD200832}, and we shall use its elements to designate a given mode occupancy in the output strings. Specifically, let $B\in\mathbb{Z}_2^{k\times 2N}$, where $B_{r,j}=1$ if and only if mode $c_j$ appears in output replica $r$. The allowed tables have row sums $L$ and even column sums,
\begin{equation}
 \sum_{j=1}^{2N}B_{r,j}=L, \qquad m_j:=\sum_{r=1}^{k}B_{r,j}\in 2\mathbb{N}.
\end{equation}
The row-sum condition says that each output string has length $L$, while the even-column condition is exactly the existence of at least one perfect matching $p$ compatible with $\bm\mu$. Equivalently, for a fixed margin vector $\bm m=(m_1,\cdots,m_{2N})$, define the table class
\begin{equation}
 \mathcal{C}_{\bm m}\!:=\! \{B\in\mathbb{Z}_2^{k\times2N} \mid \sum_jB_{r,j}=L, \sum_rB_{r,j}=m_j\}
\end{equation}
where $m_j$ is even, $0\leq m_j\leq k$, and $\sum_jm_j=kL$. These corresponding tables are of the form, e.g.,
\[
 \begin{tikzpicture}[
   baseline=(M.center),
   lab/.style={text=labelgray, font=\normalfont},
   mat/.style={
     matrix of math nodes,
     nodes={inner sep=2pt, minimum width=1.3em},
     column sep=0.8em,
     row sep=0.55em
   }
 ]
 % Main binary matrix
 \matrix (M) [mat, left delimiter={(}, right delimiter={)}] {
   1 & 0 & \cdots & 1 & 0 \\
   0 & 1 & \cdots & 0 & 1 \\
   \vdots & \vdots & \ddots & \vdots & \vdots \\
   1 & 1 & \cdots & 0 & 1 \\
 };

 % Left row labels
 \node[lab, left=1.4em of M-1-1] {$1$};
 \node[lab, left=1.4em of M-2-1] {$2$};
 \node[lab, left=1.4em of M-3-1] {$\vdots$};
 \node[lab, left=1.4em of M-4-1] {$k$};

 % Top column labels
 \node[lab, above=0.7em of M-1-1] {$c_1$};
 \node[lab, above=0.7em of M-1-2] {$c_2$};
 \node[lab, above=0.7em of M-1-3] {$\cdots$};
 \node[lab, above=0.7em of M-1-5] {$c_{2N}$};

 % Right vertical rule
 \draw[labelgray]
  ($(M-1-5.north east)+(1.5em,0.35em)$)
  --
  ($(M-4-5.south east)+(1.5em,-0.35em)$);

 % Right labels
 \node[lab, right=1.6em of M-1-5] {$L$};
 \node[lab, right=1.6em of M-2-5] {$L$};
 \node[lab, right=1.6em of M-3-5] {$~\vdots$};
 \node[lab, right=1.6em of M-4-5] {$L$};

 % Bottom horizontal rule
 \draw[labelgray]
  ($(M-4-1.south west)+(-0.25em,-0.6em)$)
  --
  ($(M-4-5.south east)+(0.25em,-0.6em)$);

 % Bottom labels
 \node[lab, below=1.1em of M-4-1] {$m_1$};
 \node[lab, below=1.1em of M-4-2] {$m_2$};
 \node[lab, below=1.1em of M-4-3] {$\cdots$};
 \node[lab, below=1.1em of M-4-5] {$m_{2N}$};
 \end{tikzpicture}
\]
For each valid margin $\bm{m}$, and each table $B\in \mathcal{C}_{\bm m}$, the compatible output matchings are products of matchings inside the occupied entries of each column, and hence constitute an admissible output pairing. For a given $B$, we denote by $\mathcal{P}(B)$ the set of mode pairings within that table
\begin{equation}
 \mathcal{P}(B)=\prod_{j:m_j>0} M_{m_j/2}.
\end{equation}
For example, graphically,
$$\left(\ContOne[0.7]\right) \longrightarrow \left(\ContTwo[0.7]\right),$$
the left matrix is an example contingency table with the Majorana mode identification made explicit, and the right matrix consolidates these modes into their corresponding strings. It is the coordinates of this consolidated table which ultimately characterises the output pairing.

The above example illustrates the importance of the margin vector. Each column-pairing class can be realised either by two complementary weight-two columns or by assigning one of the three matchings to a weight-four column:
\begin{equation}\label{eq:pairing-classes-diagram}
 \PairingEOneW[0.7]\equiv\PairingEOneWW[0.7]~~;~~\PairingETwoW[0.7]\equiv\PairingETwoWW[0.7]~~;~~\PairingEThreeW[0.7]\equiv\PairingEThreeWW[0.7]
\end{equation}

Finally, the last important consideration in this picture is that each row of $B$ determines the unordered support of one output Majorana string. That is to say, the contingency tables only specify \emph{which} mode appears, but not its order. Since any ordering is valid, then a matrix $B$ represents a sum over all permutations $\sigma_r\in S_L$ of the populated modes across each row $\sum_{\bm{\sigma}\in S_L^k} \bigotimes_{r=1}^k \gamma_{\sigma_r(\bm{\mu}_B^{(r)})}$. For example,
\begin{equation}
  \left(\ContOrbOne[0.7]\right)\longrightarrow\ContOrbTwo[0.7]
\end{equation}
is compatible with the permutations
\begin{equation}
    \ContPairOne[0.6]\quad \ContPairTwo[0.6]\quad \ContPairThree[0.6]\quad\ContPairFour[0.6].
\end{equation}
Alternatively, we can always take the output strings to be canonically ordered; anti-commutativity then gives
\begin{equation}
 \gamma_{\sigma_r(\bm\mu^{(r)})}=\operatorname{sgn}(\sigma_r)\gamma_{\bm\mu^{(r)}}.
\end{equation}
Contingency tables are therefore in exact one-to-one correspondence with a canonically ordered Majorana operator in the moment operator expansion, and the admissible pairings they represent are over all within-column pairings of the same modes along with their signed row permutations. We change our notation slightly to let $p_2$ denote the within-mode pairing, and $\bm{\sigma}\cdot p_2$ to indicate the same pairing after the action of a row permutation.\par 

\begin{equation}
  \left(\ContPairTable[0.6]\right)\ContStrOrdered[0.65] \;=\; (-1)^{\bm\sigma}\left(\ContPairTable[0.6]\right)\ContStrCanonical[0.65]
\end{equation}

To summarise, the moment operator reduces to an enumeration over:
\begin{itemize}
  \item Valid input pairings $p_1$; for an initial operator of the form $\gamma_{\bm{\nu}}^{\otimes k}$ these are all in-column pairings, denoted by $\mathcal{Q}_{k,L}$.
  \item Valid margin vectors $\bm{m}$, with each $m_j$ even, $0\leq m_j\leq k$, and $\sum_jm_j=kL$ (solutions to a Diophantine equation).
  \item Contingency tables $B$ with the columns summing to the margin $\bm{m}$, whose class is denoted by $\mathcal{C}_{\bm{m}}$.
  \item Valid between-replica pairings $p_2$ of each $B$. A pairing between replicas is valid for a given mode if the mode is populated in both rows; these pairings form $\mathcal{P}(B)$.
  \item Permutations $\bm{\sigma}$ which enumerate all the different orders that the populated modes can take in the output Majorana string.
\end{itemize}
Every admissible table $B$ thus labels one operator $\bigotimes_r\gamma_{\bm\mu_B^{(r)}}$ in the moment expansion. Utilising~\cref{eq:k-moment-weingarten-basic}, then, yields
\begin{equation}\label{eq:coefficient-table-final}
  \begin{split}
 \Phi^{(k)}[\gamma_{\bm{\nu}}^{\otimes k}] &={\sum_{\bm m}\sum_{B\in\mathcal C_{\bm m}}} C(B)\bigotimes_{r=1}^{k}\gamma_{\bm\mu_B^{(r)}},~\text{where},\\ 
 C(B)&={\sum_{p_1\in\mathcal Q_{k,L}}\sum_{p_2\in\mathcal P(B)}}\sum_{\bm\sigma\in S_L^k}\operatorname{sgn}(\bm\sigma)\operatorname{Wg}_{(\bm\sigma\cdot p_2,p_1)}.
  \end{split}
\end{equation}
With the admissible input and output pairings completely classified, what remains is to derive a readily computable expression for each $C(B)$.

\subsection{Irrep reduction}~\label{ssec:irreps}
The quantity given in~\cref{eq:coefficient-table-final} is a reduction of the moment operator to a characterisation of admissible pairings in the Brauer algebra which survive the specific Majorana string evaluation, but it nevertheless still contains Weingarten terms which are exponentially costly to evaluate in full generality. These, however, have some group structure which may be exploited in certain contexts to simplify the calculation.
A seminal result in the Weingarten calculus~\cite{collins_integration_2006} over the orthogonal group gives a standard closed form for the orthogonal Weingarten function~\cite{Collins_2009} (Theorem 3.1), which we reproduce here. Since the Weingarten function is \emph{central} -- it depends only on the relative pairing $p_1^{-1}p_2$, and therefore commutes with the simultaneous action of $S_{2n}\equiv S_{kL}$ on the diagram vertices -- it can be diagonalised in the irreps of the symmetric group over the $2n$ vertices. Specifically, for all $p_1,p_2\in M_n$,
\begin{equation}\label{eq:weingarten-general}
\begin{split}
 \operatorname{Wg}_{p_1,p_2}&= \frac{2^n n!}{(2n)!} \sum_{\substack{\lambda\vdash n\\ \ell(\lambda)\leq d}} f^{2\lambda}\, \frac{\omega^\lambda(p_1^{-1}p_2)}{Z_\lambda(1^d)},\\
 \text{where}\quad Z_\lambda(1^d)&=\prod_{(i,j)\in\lambda}(d+2j-i-1),\\
\omega^\lambda(\sigma) &= \frac{1}{2^n n!}\sum_{\zeta\in H_n}\chi^{2\lambda}(\sigma\zeta),
 \end{split}
\end{equation}
and $f^{2\lambda} = \chi^{2\lambda}((1^{2n}))$ is the dimension of the $S_{2n}$ irrep labeled by $2\lambda$. Recall that in the present application $2n=kL$ and $d=2N$.

Although~\cref{eq:weingarten-general} provides an in-principle explicit evaluation of the moment operator by decomposing the Weingarten function into a sum over its \emph{irreducible representations} (irreps), there are still too many of them to evaluate in practice. We show now that this obstacle can be removed in the context of our present problem. Specifically, we will see that the two sums in~\cref{eq:coefficient-table-final} are structured enough to leave only a single surviving irrep in the following lemma:
\begin{restatable}[{Unique non-zero irrep}]{lemma}{irrep} \label{lem:unique-weingarten-sector}
Let $k$ be even. The only orthogonal-Weingarten sector contributing to $C(B)$ in~\cref{eq:coefficient-table-final} is the Young diagram $2\lambda_\star=(k^L)\vdash kL$, where 
\begin{equation}
 \lambda_\star=\left(\left(\frac{k}{2}\right)^L\right).
\end{equation}
Thus, the entire sum in~\cref{eq:weingarten-general} reduces to a single term. 
\end{restatable}
\begin{proof}[Proof sketch]
Let $\Lambda\vdash kL$ label an $S_{kL}$ irrep. The signed output sum over $S_L^k$ transforms as $\operatorname{sgn}_L^{\boxtimes k}$. The corresponding projection therefore requires $\Lambda\unlhd(k^L)$. On the input side, $\Phi_{\rm in}:=\sum_{q\in\mathcal Q_{k,L}}\delta_q$ is invariant under $S_k^L$, which instead requires $\Lambda\unrhd(k^L)$. Thus the sandwich from either side forces that $\Lambda=(k^L)$. Since the pairing rep labels this sector by $2\lambda$, the stated $\lambda_\star$ is the unique survivor. The full argument is given in~Appendix~\ref{app:full-coefficient-calculation}.
\end{proof}

Now that we have that $\Phi^{(k)}[\gamma_{\bm{\nu}}^{\otimes k}]$ lives only on one sector, we may use~\cref{eq:weingarten-general} to compute the Weingarten function explicitly. We substitute the aforementioned Young diagram in, thus turning the remaining table dependence into a signed pairing count. Moreover, the remaining $\langle \Phi_{\rm in},\Phi_{\rm out}\rangle$ inner product enforces that the composition of row permutation $\bm{\sigma}$ with mode pairing $p_2$ must be an in-column pairing, i.e., the pairings which survive the inner product must have $\bm{\sigma}\cdot p_2 \in \mathcal{Q}_{k,L}$, which we denote in the following by an indicator function. 
\begin{corollary}[Reduced moment operator]\label{cor:reduced-moment-operator}
For every admissible table $B$,
\begin{equation}\label{eq:CB-overlap-main}
 \begin{aligned}
 C(B)&=\frac{s(B)}{Z_{\lambda_\star}(1^{2N})},\\
 s(B)&\coloneqq\sum_{p_2\in\mathcal P(B)} \sum_{\bm\sigma\in S_L^k}\operatorname{sgn}(\bm\sigma) [\bm\sigma\!\cdot p_2\in\mathcal Q_{k,L}] .
 \end{aligned}
\end{equation}
Consequently, the moment operator acting on $k$ copies of an Majorana string takes the form
\begin{equation}\label{eq:moment-operator-reduced}
 \begin{split}
 \Phi^{(k)}[\gamma_{\bm \nu}^{\otimes k}] =\sum_{\bm m}\sum_{B\in\mathcal C_{\bm m}} \frac{s(B)}{Z_{\lambda_\star}(1^{2N})} \bigotimes_{r=1}^{k}\gamma_{\bm\mu_B^{(r)}},
 \end{split}
\end{equation}
where
\begin{equation}\label{eq:Zlambda-star}
 Z_{\lambda_\star}(1^{2N}) =\prod_{i=1}^{L}\prod_{j=1}^{k/2}(2N+2j-i-1).
\end{equation}
When $k=4$ this becomes $Z_{\lambda_\star}(1^{2N})=(2N)_L(2N+2)_L$. $(y)_n:=\prod_{i=0}^{n-1} (y-i)$ is the falling factorial.
\end{corollary}
\begin{proof}[Proof sketch]
Let $A$ be the signed row antisymmetriser, $\delta_{\mathcal P(B)}:=\sum_{p\in\mathcal P(B)}\delta_p$, and $\Phi_{\rm in}$ the input vector above. Then $C(B)=\langle\delta_{\mathcal P(B)},A\text{Wg} \Phi_{\rm in}\rangle$. In the central-projector expansion $\operatorname{Wg}=\sum_\lambda\Pi_\lambda/Z_\lambda(1^{2N})$,~\cref{lem:unique-weingarten-sector} gives $A\Pi_\lambda\Phi_{\rm in}=0$ unless $\lambda=\lambda_\star$. Completeness of the projectors then gives $A\Phi_{\rm in}=A\Pi_{\lambda_\star}\Phi_{\rm in}$. Expanding the remaining overlap yields exactly $s(B)$; see~Appendix~\ref{app:full-coefficient-calculation}.
\end{proof}
The remaining calculation pertains to the integer $s(B)$, which is characterised by the indicator function $[\bm{\sigma}\cdot p_2\in\mathcal{Q}_{k,L}]$ and the sign of $\bm{\sigma}$, summed over all admissible pairings in $B$. That is, a pairing survives only if it is the composition of an in-column pairing with a row permutation. 
We first establish a colouring language to talk about positions in the string after numerical ordering. That is, a vertex is characterised by its mode identity $c_j$, and its `colour', which is the position of the mode within the string $\gamma_{\bm{\mu}^{(r)}}$.
Intuitively, the non-zero permutations are those which preserve the pairing class of a given mode. If a node is the $\ell$th occupied entry of row $r$, assign it the colour $\chi_{\bm\sigma}(r,j)=\sigma_r(\ell)$. Every colour occurs once per row, and $\bm\sigma\!\cdot p_2\in\mathcal Q_{k,L}$ precisely when every edge of $p_2$ joins equal colours. Hence
\begin{equation}\label{eq:s-colouring}
 s(B)=\sum_{\bm\sigma\in S_L^k}\operatorname{sgn}(\bm\sigma) \prod_{j=1}^{2N}\omega_j(\bm\sigma),
\end{equation}
where $\omega_j$ counts the monochromatic perfect matchings of column $j$ after applying $\bm{\sigma}$. In particular, in the following, each operator is filled by its mode and outlined by its position in the string (its colour, with the three colours corresponding to the three slots). The matching condition is that an edge of $p_2$ survives only if both ends have the same ring colour. For example, here we have a vanishing pairing (dashed red indicates a mode-colour mismatch), 
\begin{equation*}\label{eq:s-colouring-diagram}
  \ContColourStringsSorted[0.8] \;\longleftrightarrow\; \left(\ContColourTableSorted[0.8]\right),
\end{equation*}
and a fully aligned pairing:
\begin{equation*}
  \ContColourStringsAligned[0.8] \;\longleftrightarrow\; \left(\ContColourTableAligned[0.8]\right).
\end{equation*}

We now specialise to $k=4$ such that the expression can be computed exactly. Every nonempty column has weight two or four. Let $q:=\#\{j:m_j=4\}$ and define the complementary pairing classes $\alpha_1,\alpha_2,$ and $\alpha_3$. If $n_{rr'}$ counts weight-two columns occupying rows $\{r,r'\}$, then pairing complements enforce
\begin{equation}\label{eq:class-balance-main}
 \begin{aligned}
 n_{12}=n_{34}&=: \alpha_1,\\
 n_{13}=n_{24}&=: \alpha_2,\\
 n_{14}=n_{23}&=: \alpha_3.\\
 \end{aligned}
\end{equation}
The $\alpha_i$ hence stratify the columns into pairing classes, with $q$ counting the number of weight four columns that could adopt any pairing identity. It follows then that $\alpha_1 + \alpha_2 + \alpha_3 + q = L$. For the special case where $k=4$, observe that the pairing structure of the contingency matrix column (the mode) is the same as the pairing structure when put in canonical form. That is, the mode pairing type is always the same as the colour pairing type. Consequently, the only way that $\bm{\sigma}\cdot p_2$ can be admissible is if it swaps a \emph{pair} of modes within one class with another pair of modes within the \emph{same} class. 
\begin{equation}
  \ContBlocksRepairBefore[0.8] \;\mapsto\; \ContBlocksRepairAfter[0.8]
\end{equation}
It follows then that $s(B)$ factorises into a sign part and a magnitude part, since the sign is constant across an entire table.
\begin{restatable}[Signed-magnitude decomposition]{lemma}{signMag}\label{lem:signed-magnitude}
\begin{equation}\label{eq:s-split-main}
 s(B)=\varepsilon(B)\mathcal{N}(B),\quad \mathcal N(B):=\sum_{\bm\sigma\in S_L^4} \prod_j\omega_j(\bm\sigma)\geq0 .
\end{equation}
\end{restatable}
\begin{proof}[Proof sketch]
A contributing configuration partitions the nodes into $L$ monochromatic four-node blocks. Exchanging two complete blocks transposes colours in four rows, whereas re-pairing two blocks transposes them in two rows. Both moves preserve $\prod_r\operatorname{sgn}(\sigma_r)$, so the sign is constant. Removing it leaves the unsigned count $\mathcal N(B)$.
\end{proof}
The sign $\varepsilon(B)$ of a given table can be thought of as exactly the sign of the permutation that places it in canonical form from its mode ordering. This is a single property of the contingency table, which can be characterised by the inversion number of the column structure. 
\begin{equation}
  \left(\ContAssemblySorted[0.8]\right) \ContAssemblyArrows[0.8] \left(\ContAssemblyAligned[0.8]\right)
\end{equation}
\begin{restatable}[Class word]{definition}{classword}\label{def:class-word}
Let $W_B$ be the word obtained by reading the modes $j=1,\ldots,2N$ in increasing order and writing down the class $t\in\{1,2,3\}$ of every weight-two column (weight-four and empty columns are skipped). Let $\operatorname{inv}(W_B)$ be its number of inversions.
\end{restatable}
\noindent
For instance, the contingency table
$$\left(~\ContThree[0.65]~\right)$$
has class word $$W_B = (\alpha_1, (\alpha_2)^3,\alpha_3,(\alpha_2)^2, \alpha_1,\alpha_3,(\alpha_2)^2,\alpha_3,\alpha_2,\alpha_3).$$ 
The inversion number of this table is 15, which can be counted by each symbol's contribution: 5 $\alpha_2$ classes before the second $\alpha_1$; $\alpha_3$ at position 5 adds 6, from five later $\alpha_2$ classes and one from the next $\alpha_1$; $\alpha_3$ at position 9 adds 3 from the three later $\alpha_2$ classes; and $\alpha_3$ at position 12 adds one from the later $\alpha_2$.

\begin{restatable}[{Sign of admissible table}]{lemma}{signAd}\label{lem:table-sign}
The sign factor in~\cref{lem:signed-magnitude} is, for a given contingency table $B$,
\begin{equation}\label{eq:EB-main}
 \mathcal E(B)=(-1)^{\operatorname{inv}(W_B)} ,
\end{equation}
where $\operatorname{inv}(W_B)$ is the number of transpositions required to rearrange $W_B$ such that its elements $\{\alpha_i\}$ are in numerical order.
\end{restatable}
\begin{proof}[Proof sketch]
Count $\sum_r\operatorname{inv}(\sigma_r)$ pairwise over matching edges. Two same-class edges meet in zero or two rows and contribute evenly; different classes meet in one row. Modulo two, only cross-class edge pairs remain. The colouring-dependent part occurs four times for each pair of differently classified colours and cancels, leaving $\operatorname{inv}(W_B)$. A weight-four column would add two equal letters, so omitting it does not change the parity.
\end{proof}
This allows the sign of a table to be reduced to an explicit single property of that table, meaning that the sum need not be explicitly performed. It remains only to deal with the magnitude, $\mathcal{N}(B)$, which is a matter of combinatorics. Specifically, the indicator function forces that a surviving row permutation needs to send the configuration to one with in-column pairings preserved. Therefore a row permutation is only valid if it permutes \emph{within} the same column class, leading to a multiplicity of $(\alpha_i + b_i)!$, where $b_i$ is the number of weight-four columns acting as the pairing $e_i$. This argument is summarised in the following lemma.
\begin{restatable}[Magnitude of an admissible table]{lemma}{magAd}\label{lem:table-magnitude}
For an admissible four-replica table $B$, the magnitude coefficient is
\begin{equation}\label{eq:NB}
 \begin{split}
 \mathcal N(B)=L!\!\sum_{\substack{b_1+b_2+b_3=q\\b_t\geq0}} \binom{q}{b_1,b_2,b_3} \prod_{t=1}^{3}(\alpha_t+b_t)! .
 \end{split}
\end{equation}
\end{restatable}
\begin{proof}[Proof sketch]
Regard the pairs of $p_2$ as edges. A weight-two column supplies one forced edge, while a weight-four column supplies two complementary edges and may be assigned a class $e_t$. If $b_t$ weight-four columns are assigned class $t$, their assignments contribute the multinomial factor of all of the different ways the weight four columns can be arranged into the different weight-two classes, and then summed over all of the ways $q$ can be decomposed into $b_1$, $b_2$, and $b_3$. Class $t$ then contains $\alpha_t+b_t$ edges of each complementary type. Summing over the within-class row permutations contributes a factor of $(\alpha_t+b_t)!$. Lastly, there is an overall multiplicity of $L!$ that stems from the relabelling of modes. i.e., when $\bm{\sigma} = \sigma^{\times k}$ for $\sigma \in S_L$.
\end{proof}

Assembling all of the above together, we arrive at our main technical result.
\begin{theorem}[{Four-replica matchgate moment operator}]\label{res:CB-four-replica}
For a Majorana string $\gamma_{\bm{\nu}}$ with length $|\bm{\nu}|=L$, its four-replica twirl is 
\begin{equation}
    \Phi^{(4)}[\gamma_{\bm{\nu}}^{\otimes 4}] = \sum_B C(B) \bigotimes_{r=1}^4 \gamma_{\mu_B^{(r)}},
\end{equation}
where the sum runs over admissible contingency tables $B$, and 
\begin{equation}\label{eq:CB-final}
 C(B)=
\frac{(-1)^{\operatorname{inv}(W_B)}
L!(L+2)!\alpha_1!\alpha_2!\alpha_3!}
{(L-q+2)!(2N)_L(2N+2)_L}.
\end{equation}
Here, we take $B$ to have class multiplicities $\alpha_1,\alpha_2$, and $\alpha_3$, and $q$ weight-four columns, with $\alpha_1+\alpha_2+\alpha_3+q=L$.
\end{theorem}
\begin{proof}[Proof sketch]
By~\cref{cor:reduced-moment-operator}, $C(B)=s(B)/Z_{\lambda_\star}(1^{2N})$. The preceding lemmas give $s(B)=(-1)^{\operatorname{inv}(W_B)}\mathcal N(B)$ and evaluate $\mathcal N(B)$, while~\cref{eq:Zlambda-star} gives $Z_{\lambda_\star}(1^{2N})=(2N)_L(2N+2)_L$. Substitution and a simplification of the resultant sum prove the result.
\end{proof}

With~\cref{res:CB-four-replica} in hand, we may explicitly compute arbitrary matchgate $4$-replica twirls in the Majorana basis for an initial Majorana string (identical on the four replicas). What remains in any given application is to determine and count, with coefficient multiplicity, which tables survive the relevant boundary condition.~\Cref{res:CB-four-replica} is the main technical result of our paper and should be of independent interest outside the present context. Although this compact form is specialised to the $k=4$ case,~\cref{lem:unique-weingarten-sector} holds for any even $k$, provided we compute $s(B)$ as in~\cref{eq:s-colouring}.
Thus our results pertain to the moment operator in an arbitrary number of replicas, but we leave further generalisations of the higher-$k$ expressions to future work. In addition to finding a nicer closed-form expression, the dominant complexity would then be in characterising the contingency classes, a computation which is thought to grow polynomially in $L$ but exponentially in $k$~\cite{miller2011exactenumerationsamplingmatrices}. 

Before moving on to the applications of Theorem~\ref{res:CB-four-replica}, we make a few comparisons with related literature. First, we have extended the results of Refs.~\cite{wan_matchgate_2023,chapman2025fermioni} beyond the matchgate moment operators for $k\leq3$ replicas. This turns out to be a much simpler problem, as replicas of Majorana strings can be assigned a unique column pairing when there are two replicas, and the three-replica case vanishes. Next, our approach is complementary to Refs.~\cite{sierant2026theorymatchgatecommutant,braccia2026commutantfermionicgaussianunitaries}, where a Lie-algebraic approach is adopted to give a general recipe for obtaining an orthogonal basis for the matchgate commutant. Unlike our method, which puts Majorana strings at the forefront, this basis has non-trivial overlap with the Majorana basis. In another distinct approach, in Ref.~\cite{lastres2026geometryfreefermioncommutants} matchgate twirls are handled geometrically through a mapping to the ground state of an effective Heisenberg model. An advantage of both of these commutant approaches is that they lead to explicit expressions for the commutant dimension, hidden in our formulas via degeneracies in our Majorana basis, and are therefore useful for computing some properties of random Gaussian states (which are simply equal to the normalised projection onto the commutant). In contrast, by working directly in the Majorana basis, our approach is particularly tailored to finding twirls of replica Majorana operators, having reduced higher moments to a well-studied enumeration problem. 

\section{Applications to Many-body Physics}\label{sec:apps}
Now we will apply the tools developed in the previous section to a range of relevant many-body quantities. Namely, we derive the matchgate Haar average of the LOE, OSE, and OTOC. This represents the long-time operator properties resultant from both fermionic dynamics and (through the Jordan-Wigner transform) matchgate circuits. We therefore address the question of how the structure of Gaussian unitaries affects the long-time physics of otherwise generic systems.
\par
All three calculations follow the same rough strategy. If a quantity $Q$ has an $r$-replica representation with respect to a `boundary' operator $\mathsf{K}_Q$, then
\begin{equation}\label{eq:boundary-strategy}
 \begin{aligned}
 \overline{Q(O_U)} &=\Tr[\Phi^{(k)}(O^{\otimes k})\mathsf{K}_Q]=\sum_{\bm m}\sum_{B\in\mathcal C_{\bm m}}C(B)\,w_Q(B),\\
 \text{where}&\quad w_Q(B):=\Tr\!\left[\left(\bigotimes_{r=1}^{k}\gamma_{\bm\mu_B^{(r)}}\right)\mathsf{K}_Q\right] .
 \end{aligned}
\end{equation}
The evaluation of $C(B)$ is universal and depends only on~\cref{eq:CB-final}, but in principle there are exponentially many different contingency tables on which $w_Q(B)$ must be evaluated. As we shall see, however, the choice of boundary operator can greatly simplify even further the classes of contingency tables which survive the trace, via delta functions and {cancellations} {arising from $w_Q(B)$}.

\subsection{Operator entanglement}\label{sec:LOE}
The \emph{(local-)operator entanglement} (LOE) entropies of a Hermitian operator $O_U=U^\dagger O U$ are defined via its state representation,
\begin{equation}
  |{O_U}\rrangle := (O_U \otimes \id)\ket{\phi^+},
\end{equation}
where $\ket{\phi^+}:= D^{-1/2} \sum_i \ket{i i}$ is a maximally 
entangled state across two copies of Hilbert space. 
The LOE entropies are then defined as the $k$-R\'enyi entanglement entropies of $|{O_U}\rrangle$ across a spatial bipartition of doubled Hilbert space, $\mathcal{H}_A^{\otimes 2} \otimes \mathcal{H}_{\bar{A}}^{\otimes 2}$ where $\mathcal{H} = \mathcal{H}_A \otimes \mathcal{H}_{\bar{A}}$ {as}~\footnote{As usual for R\'enyi entropies, the definition is extended to $k=0,1,\infty$ via their respective limits.}
\begin{equation}
\begin{split}
  &{\operatorname{LOE}_A^{(k)}} := \frac{1}{1-k}\log\left(E^{(k)}_A(O_U) \right),\quad\text{where}\\
  &E^{(k)}_A(O_U) := \Tr[ \Tr_{\bar{A}}[|{O_U}\rrangle\llangle O_U |]^k ] \\
  &\qquad \qquad \, \,= D^{-k}\Tr[O_U^{\otimes 2 k} T_{\pi_o}^A T_{\pi_e}^{\bar{A}}].
\end{split}
\end{equation}
Here, in anticipation of finding the matchgate-averaged LOE, we define the LOE $k$-purities, $E^{(k)}_A(O_U)$, and rewrite this quantity as a tensor product over replicas of Hilbert space~\cite{dowling2026page}. We have defined the permutations, written in standard cyclic notation, $\pi_e:=(12)(34)\dots([2k-1] [2k] )$ as a pairing permutation which SWAPs the first and second copies of Hilbert space, the third and fourth, and so on, while $\pi_o:=([2k] 1)(23)\dots([2k-2] [2k-1] )$. $T_{\pi_o}^A$ then refers to the permutation unitary representation of $\pi_o$ across the $2k$ replicas of Hilbert space, but acting only on the $A$ subsystem of each replica. Throughout, we take base $2$ logarithms.

Operationally, the value of the LOE governs the memory cost of representing the corresponding operator as a matrix-product operator~\cite{dowling2026classicalsim}. Beyond this direct interpretation, the LOE has also been studied intensely as a dynamical indicator of many-body chaos. It has been observed to grow at-fastest logarithmically in time for locally interacting integrable spin-chain dynamics~\cite{Prosen2007a,Prosen2009,Dubail_2017,Alba2019,Kos2020II,Alba2021,Murciano_2024}, and (typically) linearly for non-integrable dynamics~\cite{Jonay2018,Kos2020,Bertini_2026}. 
Most relevant to this work, for a fermionic hopping model, an initial extensive Majorana operator (i.e., a Pauli $\sigma_x$) displays logarithmic LOE growth at all times~\cite{Dubail_2017}. Its long-time saturation value has been relatively less studied. It is expected to reach a volume in integrable and non-integrable systems alike (with solvable-model exceptions~\cite{Jacoby_2026}), but with a different characteristic shape which should be sensitive to the exact dynamics and initial operator. For instance, (chaotic) random circuit dynamics without conservation laws reach an operator-independent Page-like volume law~\cite{dowling2026page}. What is missing from this picture is exactly what happens at long times for a generic Gaussian dynamics: how does the operator Page curve differ for Fermions compared to arbitrary, chaotic evolution, and can we analytically derive the (anticipated) volume law? 

Before moving on to compute the LOE via Theorem~\ref{res:CB-four-replica}, we recall a general relation between Gaussian operations and the LOE. Namely, any initial Majorana of string length $L$, evolved by any Gaussian unitary, can be exactly encoded into an MPO of bond dimension $2^L$. This leads to the following bound, which is non-trivial but well known~\cite{Prosen2007,Hubig_2017,Dubail_2017}.

\begin{restatable}[{LOE bound under Gaussian evolution}]{proposition}{loeBound}\label{prop:loe-bound}
Let $\gamma_{\bm{\nu}}$ be a Majorana string of length $L$ and let $U$ be any Gaussian unitary. For every bipartition and every $k>1$, we have
\begin{equation}
  E_A^{(k)}(O_U) \geq 2^{-(k-1)\min\{L,\,2N-L\}}.\label{eq:loe_bound}
\end{equation}
Equivalently, $\mathrm{LOE}_A^{(k)}(O_U)\leq\min\{L,2N-L\}$ for all R\'enyi indices $k\geq0$.
\end{restatable}
\begin{proof}[Proof sketch]
By~\cref{eq:matchgate_action}, each Heisenberg-evolved Majorana is a linear combination of single Majoranas, $\sum_\beta R_{\alpha\beta}c_\beta$. Under the Jordan-Wigner representation, this is a sum of Pauli strings of the form $Z^{\otimes a}\id^{\otimes b}$, which admits an exact MPO with bond dimension at most $2$. The product of $L$ such operators is therefore an MPO of bond dimension at most $2^L$. 
See~\cref{app:MPO} for full details.
\end{proof}
In fact, although the argument is mostly straightforward~\cref{eq:loe_bound} is tight. In~\cref{app:max_ent} we give an example of a Gaussian unitary $U$, such that the evolved initial Pauli operator is equal to,
\begin{equation}
    O_U= U \gamma_{\bm{\nu}} U^\dagger = \frac{(-i)^{N/2}}{2^{N/2}}
 \prod_{j=1}^{N}(c_j+c_{N+j}).\label{eq:loe_counter_eg}
\end{equation}
We show that this operator is maximally entangled across the half-chain bipartition. This shows that already Gaussian evolutions can in fact attain a maximal operator entanglement. 
It is an interesting property we shall see, however, that on average over matchgate unitaries, such operators are far from being maximally entangled.

For the rest of this section, we will focus on the matchgate twirls of $k=2$ LOE purity, a four-replica quantity (involving four copies of $U$ and $U^{\dagger}$). $E^{(2)}$ is the simplest non-trivial case but already captures the main features of LOE, and indeed is the most studied quantity in the context of many-body physics. Our results generalise directly: averages over higher LOE $k$-purities can be computed similarly, albeit with more complicated combinatorics. Namely, we look to compute,
\begin{equation}\label{eq:LOE-def}
  \begin{split}
    \overline{E^{(2)}(O_U)} &= D^{-2}\Tr[\Phi^{(4)}[O^{\otimes 4}] \, T_{\pi_o}^A T_{\pi_e}^{\bar{A}}],\quad\text{where,}\\
    &T_{\pi_o}^A T_{\pi_e}^{\bar{A}} = \LOEop.
  \end{split}
\end{equation}
We take the initial operator $O$ to be a Majorana string (a Pauli) of string length $L$, and the overline from here on to indicate the Haar averaging of $U$ over the matchgates. In the final line, we have graphically represented the permutation unitaries $T_{\pi_o}^A T_{\pi_e}^{\bar{A}}$ as a tensor network with subsystems $A$ and $\bar{A}$ represented as blue and red wires, respectively.  
Evaluating the above twirl, from~\cref{eq:mainTool} we have that
\begin{equation}\label{eq:LOE-weingarten}
\begin{aligned}
 &\overline{E^{(2)}(O_U)} ={}D^{-2}\sum_{p_1,p_2\in M_{2L}}\operatorname{Wg}_{p_1,p_2}\\[-0.2em]
 &\times\underbrace{\LOEcontract[0.58]}
 _{\delta_{\bm\mu}^{p_2}\delta_{\bm\mu}^{T'}}
 \underbrace{\WeinPOneTerm[0.75]}_{\delta_{\bm\nu}^{p_1}},
\end{aligned}
\end{equation}
where the boundary reduction $\delta{\bm \mu}^{T'}$ is imposed by
\begin{equation}
\begin{split}
  \LOEPaulisExp[0.45]~=~ \LOEPaulisDelta[0.65] ~=~ \LOEPaulisDeltaSimp[0.65].
  \end{split}
\end{equation}
Thus, the boundary operator $T_{\pi_o}^AT_{\pi_e}^{\bar A}$ pairs the four replicas as $(1,2),(3,4)$ on $A$ and as $(2,3),(4,1)$ on $\bar A$, so a table contributes only if
\begin{equation}\label{eq:LOE-pauli-constraints}
\begin{split}
 &P_1^{(A)}=P_2^{(A)},\quad P_3^{(A)}=P_4^{(A)},\\
 \quad &P_2^{(\bar A)}=P_3^{(\bar A)},\quad P_4^{(\bar A)}=P_1^{(\bar A)} ,
 \end{split}
\end{equation}
where $P_r^{(A)}$, $P_r^{(\bar A)}$ are the restrictions to $A$ and $\bar A$ of the Pauli operator $\gamma_{\bm\mu_B^{(r)}}$ obtained from row $r$ by the Jordan--Wigner map. Equivalently, the structure of the contingency class whose multiplicity and Weingarten coefficients we need to track is given by 
\begin{equation}\label{eq:LOE-class-one}
  \left(~\LOEClassOne[0.8]~\right)
\end{equation}
We show in~Appendix~\ref{app:LOE-calc} that this is indeed the only surviving class of contingency tables.

We first compute an expression for the Weingarten coefficient of each table $B$ in this class via~\cref{eq:CB-final}. First, observe that since all of $A$ precedes all of $\bar{A}$ in the mode ordering, the weight-two columns give the class word $e_1,e_1,\cdots,e_1~e_3,e_3,\cdots, e_3$. Equivalently, this is the statement that the LOE boundary forces all tables to already be in an in-column pairing. Hence $\rm{inv}(W_B) = 0$ for all $B$, $\implies \mathcal{E}(B) = 1$ and the quantity is a sum of only positive terms. \par 

To compute $\mathcal{N}(B)$, one simply needs to {substitute} $\alpha_2=0$ into~\cref{eq:NB} which yields, for a valid $(\alpha_1,\alpha_3,q)$ tuple,
\begin{equation}\label{eq:LOE-mag}
  \frac{\mathcal{N}(B)}{L!} = q!\alpha_1!\alpha_3!\binom{L+2}{q}.
\end{equation}
Lastly, then, to finish our computation we simply need to sum this expression for all valid $(\alpha_1,\alpha_3,q)$, which we do in~Appendix~\ref{app:LOE-calc}. 
\begin{figure*}[t!]
  \centering
  \includegraphics[width=0.95\linewidth]{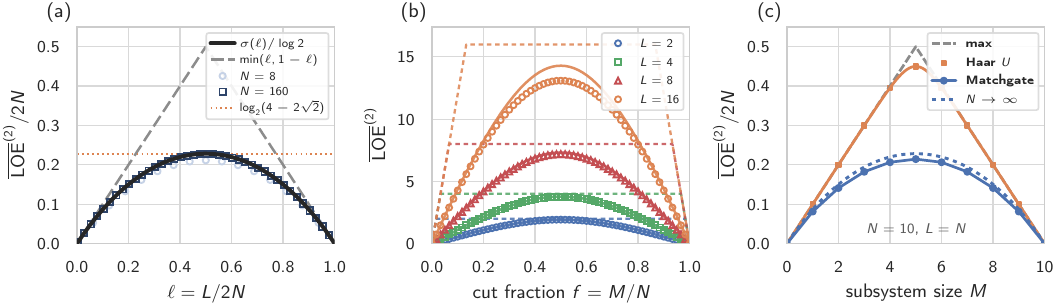}
  \caption{Different parameter slices for the annealed matchgate-average R\'enyi-2 local operator entanglement of a length-$L$ Majorana string across a bipartition with subsystem size $M$, evaluated exactly from Application~\ref{res:LOE-closed-form}. 
    \textbf{(a)} Page-like-curve at the half-chain cut $M=N/2$. We compare exact finite-size results from~\cref{eq:LOE-closed-form}; the asymptotic value~\cref{eq:LOE-page}; and the MPO bond dimension bound,~\cref{eq:loe_bound}.
    \textbf{(b)} Dependence on the cut at fixed string length, for $N=60$ and $L=2,4,8, 16$. Symbols are exact~(\cref{eq:LOE-closed-form}); solid lines are the leading fixed-$L$ expansions~(\cref{eq:LOE-dilute-entropy}) $f=M/N$; dashed lines are the naive upper bound $\min\!\big(L,\,2N-L,\,2M,\,2(N-M)\big)$: the smaller of the Majorana string-length bound (\cref{prop:loe-bound}) and the dimension of the doubled subsystem. 
    \textbf{(c)} A comparison of LOE Page curves for random free-fermionic dynamics ($L=N$ string) with Haar random unitaries.  
    }
  \label{fig:loe-detail}
\end{figure*}

\begin{application}[{Matchgate-averaged LOE purity}]\label{res:LOE-closed-form}
For an initial Majorana string of length $L$ and a bipartition $\mathcal{H}_A\otimes \mathcal{H}_{\bar{A}}$ of the first $M$ versus $\bar{M}:= N-M$ qubits,
\begin{align}\label{eq:LOE-closed-form}
    &\overline{E_A^{(2)}(O_U)}=\frac{L!(L+2)!}{(2N)_{ L}(2N+2)_{L}} \\
\, &\times \sum_{\substack{0\leq w\leq M\\0\leq x\leq\bar M\\w+x\leq\min\{L,2N-L\}}}\frac{(2M)_{{2w}}(2\bar M)_{{2x}}}{w!x!(w+x+2)!}\binom{2N-2w-2x}{L-w-x}. \nn 
\end{align}
\end{application}

Several features of~\cref{eq:LOE-closed-form} are worth remarking on. The right-hand side does not depend on which Majoranas make up the initial string, only on its length. This follows from left invariance of the Haar measure on $\mathrm{O}(2N)$: permuting the selected rows of $R$ is left multiplication by a permutation matrix, an element of $\mathrm{O}(2N)$. Moreover, we see that the expression is invariant under $M\leftrightarrow N-M$ (since LOE is symmetric in the bipartition) and under $L\leftrightarrow 2N-L$ (since $\gamma_{\bm{\nu}}$ and its complement $\gamma_{\bar{\bm{\nu}}}$ differ by the parity operator, also in $\mathrm{O}(2N)$).

While it is exact,~\cref{res:LOE-closed-form} is not particularly transparent. To further analyse~\cref{eq:LOE-closed-form}, we plot it in~\cref{fig:loe-detail} for different parameter regimes. Namely, we plot the annealed average $2$-R\'enyi LOE entropy, 
\begin{equation}\label{eq:LOE-annealed-definition}
    \overline{\mathrm{LOE}}_{A}^{(2)}:= -\log(\overline{E_A^{(2) } }) ,
\end{equation}
against initial string length $L$ and bipartition size, $M=\log_2 (D_A)$. The annealed average has the advantage of only depending on exactly $2k$ replicas, and provides a {lower bound} for the true average by Jensen's inequality, 
\begin{equation}
\int_{U \in \mathcal{M}} {\mathrm{LOE}_{A}^{(k)} (U O U^\dagger)} \geq  \overline{\mathrm{LOE}}_{A}^{(k)}\label{eq:jensens}
\end{equation}
for $k>1$. These differ by the relative fluctuations around the average $k$-purity, and for large $N$ we expect such corrections to be suppressed. Through abuse of notation, as in~\cref{eq:LOE-annealed-definition}, an overline on entropic quantities will always denote the annealed average throughout this work. 

Examining~\cref{fig:loe-detail}, we see some interesting features. First, we saw in~\cref{eq:loe_bound} that the LOE entropies are bounded from above by $\min\{L,2N-L\}$, and the bound is achievable when allowing for arbitrary Gaussian unitaries; see~\cref{eq:loe_counter_eg}. However, from~\cref{fig:loe-detail}, we see that the average case over the matchgate Haar measure is far from saturating the maximal allowable operator entanglement. In part (c), we compare the half-chain LOE, where the fermionic MPO bound of~\cref{eq:loe_bound} is vacuous, to the ``Page curve'' LOE of Haar-evolved operators~\cite{dowling2026page}. We see a gap: while the Page deviation from maximal operator entanglement closes in the thermodynamic limit (at most $\mathcal{O}(1)$), the matchgate curve deviates by $\mathcal{O}(N)$. Such fermionic unitaries therefore typically lead to a volume-law LOE, but with a sub-maximal coefficient. We can make this deviation quantitative by performing an asymptotic expansion of~\cref{eq:LOE-closed-form}. We find that for the half-chain bipartition $M=N/2$ and extensive string length $L=2N\ell$ with $\ell\in(0,1)$ fixed,
\begin{equation}\label{eq:LOE-page}
\overline{\mathrm{LOE}}_{A}^{(2)}=2N\,\sigma(\ell)+\mathcal{O}(1),
\end{equation}
with coefficient,
\begin{equation}
  \sigma(\ell) = \log\!\left[ \frac{ \bigl(\ell-p_\star(\ell)\bigr)^\ell \bigl(1-\ell-p_\star(\ell)\bigr)^{1-\ell}}{ \ell^{2\ell}(1-\ell)^{2(1-\ell)}}
\right],
\end{equation}
where $p_\star(\ell)=({\sqrt{1+4\ell(1-\ell)}-1})/{2}$.
We see that the LOE satisfies a volume law whenever $L$ is extensive, and we see that the maximum is achieved when $L=N$, $\sigma(\frac{1}{2}) = \log\big(4-2\sqrt2\big)\simeq 0.23$. Since the maximal operator entanglement at the half-chain is $N$, random Gaussian dynamics asymptotically yield less than half of the maximal value. 

We should compare the above phenomenon to the Page curve of random Gaussian states~\cite{PhysRevB.103.L241118,Bianchi_Hackl_Kieburg_Rigol_Vidmar_2022}. There, a similar shape is seen, with fermions satisfying a suppressed volume law. There is a key difference: the entanglement density of Fermionic states \textit{decreases} with $N$, before reaching a qualitatively similar shape to our result. Here, the operator entanglement sees the opposite behaviour, where increasing $N$ increases the entanglement density towards a maximum for each $L$; see~\cref{fig:loe-detail}(a).

We may also evaluate~\cref{eq:LOE-closed-form} for few-particle operators, i.e., for $L=\mathcal{O}(1)$. We find that for the half-chain, 
\begin{equation}\label{eq:LOE-dilute-entropy}
\overline{\mathrm{LOE}}_{A}^{(2)} ={}L-\frac{L(L+1)}{2N \ln(2)}+\mathcal O(N^{-2}) .
\end{equation}
with the formula for arbitrary $M$ reported in~\cref{app:LOE-calc}. We see that constant Majorana-length operators are maximally entangled within their number sector, up to polynomially small corrections. This thus characterises the operator entanglement under random matchgate circuits. However, LOE is only one measure of the extent to which a Heisenberg operator can become complex under dynamics. We now apply the same machinery to the OSE, which is sensitive instead to how the operator spreads over the Pauli basis.

\subsection{Operator stabiliser entropies}\label{ssec:OSE}
Consider the set of $N$-qubit Pauli strings $\mathcal{P}_N = \{\id,\sigma_x,\sigma_y,\sigma_z\}^{\otimes N}$, equal to the Pauli group modulo phases. The Pauli strings form an orthonormal basis for the space of operators on $N$ qubits, with respect to the Hilbert-Schmidt inner product. Any operator $O$ can be therefore be expanded in this basis as 
\begin{equation}
  O = \sum_{P \in \mathcal{P}_N} a_P P,
\end{equation}
where we take the operator to be Hilbert-Schmidt normalised, $\Tr[O^\dagger O] = D$, such that $\sum_P |a_P|^2 = 1$ and so the coefficients $\{|a_P|^2 \}_{P \in \mathcal{P}_N}$ form a probability distribution over the Pauli strings for any unitary operator $O$. The \emph{operator stabiliser entropies} (OSE) of a Hermitian operator $O$ are defined as the (classical) R\'enyi entropies of the distribution $\{|a_P|^2 \}_{P \in \mathcal{P}_N}$~\cite{dowling2024magicheis},
\begin{equation}
 \begin{aligned}
 \operatorname{OSE}^{(k)}(O)&:=\frac{1}{1-k}\log M^{(k)},\quad \text{where}\\
 M^{(k)}(O)&:=\sum_{P\in\mathcal P_N}|a_P|^{2k} ,
 \end{aligned}
\end{equation}
where we have defined the operator stabiliser purities $M^{(k)}$, involving exactly $2k$ replicas of $O$. The OSE entropies have been shown to govern the cost of Pauli propagation methods, where a $\Omega(N)$ scaling of $k\geq 1$ entropies indicates an unsuitability of such algorithms. Moreover, the OSE entropies are `magic' monotones, i.e., they measure how many non-Clifford resources are required to implement $U$ to produce the operator $O_U$ for an initial Pauli $O$. It is highly sensitive to noise in a circuit~\cite{dowling2026noise}, and can be shown to scale as $\mathcal{O}(\log(t))$ in certain exactly solvable integrable spin chains~\cite{Alba2019,dowling2024magicheis}.    

As with the LOE purity in the previous section, we can rewrite the OSE $k$-purities $M^{(k)}$ as a $2k$-replica quantity, such that their matchgate twirl is 
\begin{equation}
  \overline{M^{(k)}(O_U)} = D^{-2k} \sum_P \Tr[\Phi^{(2k)}[O^{\otimes 2k}] P^{\otimes 2k}].
\end{equation}
We see that the boundary operator defining the OSE purities for the replica operator is $\sum_P P^{\otimes 2k}$ (an element of the $2k$-replica Clifford commutant~\cite{Bittel_2026}).

\begin{figure}[t!]
  \centering
  \includegraphics[width=0.85\linewidth]{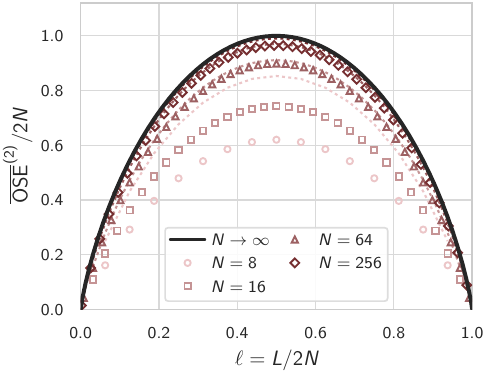}
\caption{Annealed matchgate-averaged R\'enyi-2 OSE from the closed form of~\cref{eq:OSE-closed} in the extensive string-length regime, with the upper bounds of~\cref{eq:ose_bound} shown as dotted lines. }
  \label{fig:ose-numerics}
\end{figure}

We restrict to $k=2$, the simplest non-trivial case, again noting that our methods generalise readily (but tediously) to higher $k$. The strategy is to rewrite the twirl using~\cref{eq:mainTool}, and then to count the contingency tables which have identical Majorana strings across each replica, otherwise the above expression will vanish from the fact that $\Tr[P \gamma_{\bf{\mu}} ] = 0$ if $P\neq\gamma_{\bf{\mu}}$.

Take $O=\gamma_{\bm \nu}$ to be a length-$L$ Majorana string, with its irrelevant phase chosen so that it is a Hermitian Pauli. In the convention of~\cref{sec:majorana}, $c_\alpha^2=\id$, so $O$ is already Hilbert--Schmidt normalised, $\sum_P|a_P|^2=1$. At $k=2$ the twirl is the four-replica moment operator, so~\cref{eq:coefficient-table-final} gives
\begin{equation}
 \overline{M^{(2)}} =D^{-4}\sum_{P\in\mathcal{P}_N}\sum_{B}C(B)\prod_{r=1}^{4}\Tr\!\big[\gamma_{\bm\mu_B^{(r)}}P\big].
\end{equation}
Each Majorana string is a Pauli up to a phase, $\gamma_{\bm\mu}=i^{\theta(\bm\mu)}P_{\bm\mu}$, and $\Tr[\gamma_{\bm\mu}P]=D\,i^{\theta}\delta_{P,P_{\bm\mu}}$. Hence a table survives only if all four rows are identical,
\begin{equation}
 \left(~\OSEClass[0.8]~\right)
\end{equation}
and hence $\overline{M^{(2)}}=\sum_{\bm\mu}C_{B_{\bm\mu}},$
the sum running over all $\binom{2N}{L}$ length-$L$ supports. Such a table has every occupied column of weight four, so in the class notation of~\eqref{eq:classes} $\alpha_1=\alpha_2=\alpha_3=0$, while $q=L$.
This means that the class word $W_B$ is empty and $\mathcal{E}(B)=+1$.~\cref{eq:NB} then gives
\begin{equation}
 \begin{aligned}
 \mathcal N(B_{\bm\mu}) &={L!\!\sum_{b_1+b_2+b_3=L}} \binom{L}{b_1,b_2,b_3}\prod_{t=1}^{3}b_t!\\
 &=(L!)^2\binom{L+2}{2} .
 \end{aligned}
\end{equation}
since every term of the sum equals $L!$ and the number of compositions of $L$ into three non-negative parts is $\binom{L+2}{2}$. Dividing by the `zonal polynomial'~\eqref{eq:Zlambda-star} and multiplying by the number of supports gives us the explicit expression.

\begin{application}[{Matchgate-averaged OSE purity}]\label{res:OSE-closed}
For a Majorana string of length $L$ on $N$ qubits, 
\begin{equation}\label{eq:OSE-closed}
 \overline{M^{(2)}(O_U)} =\frac{(L+2)!}{2\,(2N+2)_L}
\end{equation}
is the value of the operator stabiliser entropy purity averaged over the matchgate ensemble. 
\end{application}
We plot this result in~\cref{fig:ose-numerics} for varying string length $L$ and system size $N$. The random Gaussian unitary spreads $O$ over the subspace of $\binom{2N}{L}$ Majorana strings of the same length. 
The dimension of this sector alone bounds the OSE for every Gaussian unitary, stated in the following proposition and in more detail in~\cref{app:OSE-calc}.

\begin{restatable}[{OSE bound under Gaussian evolution}]{proposition}{oseBound}\label{prop:ose-bound}
Let $\gamma_{\bm{\nu}}$ be a Majorana string of length $L$ and let $U$ be any Gaussian unitary. Then, for every $k>1$,
\begin{equation}
  M^{(k)} (O_U) \geq \binom{2N}{L}^{-k+1},\label{eq:ose_bound}
\end{equation}
with equality if and only if the Pauli distribution $\{|a_P|^2\}$ of $O_U$ is uniform on the $\binom{2N}{L}$ length-$L$ Majorana strings. Equivalently, $\mathrm{OSE}^{(k)}(O_U)\leq\log\binom{2N}{L}$ for all R\'enyi indices $k\geq0$.
\end{restatable}

We can analyse how close a typical Matchgate circuit comes to saturating this bound by first rewriting~\cref{eq:OSE-closed} in terms of the annealed average entropy (\cref{eq:LOE-annealed-definition}), 
\begin{equation}\label{eq:OSE-entropy}
 \begin{aligned}
 \overline{\mathrm{OSE}}^{(2)}:=&-\log\overline{M^{(2)}}\\
=&\log\binom{2N}{L} -\log\frac{(L+1)(L+2)}{2(2N+1)}\\
 &-\log\frac{(2N+2-L)(2N+1-L)}{2N+2} .
 \end{aligned}
\end{equation}
From this, we can study the two limits considered earlier for the LOE. First, if we fix {$L=\mathcal{O}(1)$, and expand for large $N$,} then 
\begin{equation}
\log_2\binom{2N}{L}-\overline{\mathrm{OSE}}^{(2)}=\log_2\frac{(L+1)(L+2)}{2}+\mathcal{O}(N^{-1}).    
\end{equation}
We see that a constant-length Majorana is approximately uniformly spread in Pauli strings according to the second moment, up to a constant correction. If we instead consider an extensive initial string, with density $\ell := L/2N$ for $0<\ell <1$, we find  
 \begin{equation}
     \log\! \binom{2N}{L}-\overline{\mathrm{OSE}}^{(2)}
\approx 2\log(2N)+2\log[\ell(1-\ell)]-1,\label{eq:ose-asymp-extensive}
 \end{equation}
 up to corrections of size $\mathcal{O}(N^{-1})$. We see that there is a logarithmic gap between the maximum allowed by the upper bound of~\cref{eq:ose_bound} and the matchgate average quantity in the largest particle sector. 

For fixed $L$, we saw that the LOE is bounded by $L\log2$, whereas the OSE grows as $L\log N$; for both fixed and extensive $L$, these two quantities therefore exhibit an asymptotic separation. Otherwise, for extensive $L$, LOE deviates from the maximum by a constant relative error, whereas the OSE error in~\cref{eq:ose-asymp-extensive} vanishes in the thermodynamic limit, relative to the leading order. Matchgates therefore exhibit a gap between entanglement and magic. We return to this comparison in~\cref{sec:reource_theories}.

\subsection{Out-of-time-ordered correlators} \label{sssec:OTOC-closed}
Higher-order out-of-time-ordered correlators of a Heisenberg operator $O_U$ are defined as 
\begin{equation}\label{eq:otocdef}
 \operatorname{OTOC}^{(k)}(O_U,X)
 :=D^{-1}\Tr[(O_UX)^k],\quad k\in2\mathbb N .
\end{equation}
Here $X$ is some Hermitian and normalised reference operator. If $O_U$ and $X$ commute, the correlator is one, and so the operators are considered unscrambled. The $k=2$ case is the standard four-point diagnostic of information scrambling~\cite{Shenker_Stanford_2014}, whose long-time behaviour is related to LOE for typical probes~\cite{dowling2023scrambling,Dowling2025thesis}. Higher orders $k\ge4$ probe \textit{free independence}~\cite{Voiculescu1991,nica2006lectures}, a universal long-time operator property from random~\cite{fava2023designsfreeprobability,dowling2025freeindep} and ergodic~\cite{Pappalardi2022freeETH,Vallini2026longtimefreenessin,Claeys2026} dynamics. In particular, the asymptotic vanishing of all OTOCs for two traceless operators indicates that the observables are {freely independent} with respect to each other, a non-commutative analogue of statistical independence~\cite{Voiculescu1991}. OTOCs are moreover central to recent protocols for quantum advantage, where it has been argued that `Pauli interference' effects are hard to reproduce efficiently on a classical computer~\cite{Abanin2025}.

We again apply the strategy for matchgate averages adopted above. We first rewrite~\cref{eq:otocdef} in $k$-replicas of Hilbert space,
\begin{equation}\label{eq:otoc_op}
    \begin{split}
        {\operatorname{OTOC}^{(k)}} &=  D^{-1} \Tr[O_U^{\otimes k}X^{\otimes k}T_{\sigma}],\quad\text{where} \\[1em]
        T_{\sigma} &:=  \OTOCop
    \end{split}
\end{equation}
is the cyclic permutation unitary between replicas, $\sigma=(12\dots k) $. We take the operators $O$ and $X$ to be Majorana strings of length $L$ and $L'$, respectively. Considering the average $2k$-OTOC over the matchgate group, we have that 
\begin{equation}
  \overline{\operatorname{OTOC}}^{(k)} := D^{-1} \Tr[\Phi^{(k)}[O^{\otimes k}]X^{\otimes k}T_{\sigma}].
\end{equation}

The moment operator $\Phi^{(k)}[O^{\otimes k}]$ is given by~\cref{eq:mainTool}. After undoing the cyclic permutation, conjugation by the probe supplies the Pauli phase $XPX=(-1)^{\langle X,P\rangle}P$. The remaining trace is nonzero exactly when the product of the $k$ replica strings is proportional to the identity. Diagrammatically, the contraction is
$$
\sum_{\vec{P}\in \operatorname{Comm}(U_f(N))}\OTOCk[0.8]
$$
Interestingly, if $X$ is chosen to be Pauli and $k$ is (as is necessary) even, then the above expression is non-zero {for every valid binary contingency table in the expansion of the commutant}, each coming with a non-trivial phase. Far from the substantial symmetry reduction imposed by the LOE and OSE calculations (where many terms reduced to once contracted with the boundary operator), we see that the OTOC appears to probe the entire commutant equally, with non-trivial interference effects between components; cf. the Pauli interference arguments of Ref.~\cite{Abanin2025}. Put differently, the pairing structure of the free-fermionic moment operator \emph{guarantees} that each element in the sum comes with an even number of each appearing Majorana spinor species, which means each term in the expansion is proportional to $\id$. The strategy for calculating the matchgate-averaged OTOC is therefore equivalent to enumerating all admissible contingency tables as described in~\cref{sec:contingency} along with their expansion coefficients, the overall sign stemming from the trace of the ordered Paulis, and the sign from (anti-)commuting the probe operator. We show how to account for these signs in the expression and thus make the overall calculation tractable.

\begin{figure}[t!]
  \centering
  \includegraphics[width=\linewidth]{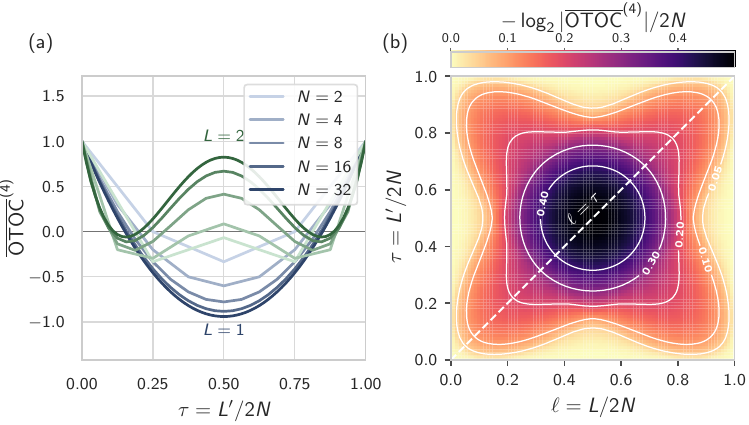}
\caption{Matchgate-averaged four-replica OTOC evaluated from~\cref{eq:otoc-closed}. We write $\ell=L/2N$, $\tau=L'/2N$ for the fractions of the $2N$ Majorana modes occupied by the operator and by the probe, respectively.
    \textbf{(a)} Exact $\overline{\mathrm{OTOC}}^{(4)}$ against $\tau$ for $L=1$ (blue) and $L=2$ (green), with $N=2,4,8,16,32$ shown light to dark. At fixed $L$ the correlator is an $\mathcal{O}(1)$ number for generic probe lengths and does not decay with system size.
    \textbf{(b)} The same extensive decay rate over the whole $(\ell,\tau)$ plane. This is symmetric about the diagonal $\ell=\tau$, and maximal at the most scrambled point $\ell=\tau=\tfrac12$, that is $L=L'=N$, where $\big|\overline{\mathrm{OTOC}}^{(4)}\big|\simeq 2^{-N}$.
    } 
  \label{fig:otoc-numerics}
\end{figure}

We now compute the matchgate-average OTOC for $k=4$. Let $O=\gamma_{\bm{\nu}}$ have length $L$ and let the probe be $X=\gamma_{\bm{\nu}'}$ of length $L'$, with supports $S$ and $S'$ and both observables Hermitian. Inserting~\cref{eq:coefficient-table-final} and commuting the four copies of $X$ to the right using
$X\gamma_{\bm\mu}X=(-1)^{L'L-|S'\cap\bm\mu|}\gamma_{\bm\mu}$ gives
\begin{equation}\label{eq:OTOC-table-sum-main}
 \overline{\mathrm{OTOC}}^{(4)} =\sum_{B\in\mathcal{C}}C(B)\;\varepsilon(B)\;(-1)^{|S'\cap S_2|+|S'\cap S_4|}
\end{equation}
where $S_r$ is row $r$ of the table and
$\varepsilon(B):=D^{-1}\Tr[\gamma_{S_1}\gamma_{S_2}\gamma_{S_3}\gamma_{S_4}]$. The factors of $(-1)^{L'L}$ cancel in pairs. Because the four supports of an admissible table satisfy $S_1\triangle S_2\triangle S_3\triangle S_4=\emptyset$, the product is always proportional to the identity, so $\varepsilon=\pm1$ and every admissible table contributes. Explicitly, sorting the concatenated string, we have $\varepsilon(B) = (-1)^{I(B)}$, where $I(B):=\sum_{r<r'}I(S_r,S_{r'})$ as in prior sections. Once more, we can show that this explicit value depends only on properties of the contingency table alone, and so evaluating the full ensemble average amounts to first computing the coefficient contribution for a given $B$ and a given $X$, and then counting up all the number of ways we can attain those values. First, we have
\begin{restatable}[OTOC table coefficient]{lemma}{OTOCphase}\label{lem:otoc-phase}
The summand in~\cref{eq:OTOC-table-sum-main}, $C(B) \varepsilon(B)(-1)^{|S'\cap S_2|+|S'\cap S_4|}$, evaluates to
  \begin{equation}
\frac{(-1)^{\alpha_2+k_1+k_3}\mathcal{N}(\alpha_1,\alpha_2,\alpha_3,q)}{(2N)_L(2N+2)_L},
\end{equation}
where $k_t$ counts the number of class-$t$ weight-two columns of $B$ whose mode lies in $S'$. 
\end{restatable}
\begin{proof}[Proof sketch]
By~\cref{cor:reduced-moment-operator} and~\cref{lem:signed-magnitude,lem:table-sign}, $C(B)=\mathcal{E}(B)\mathcal{N}(B)/[(2N)_L(2N+2)_L]$ with $\mathcal E(B)=(-1)^{\operatorname{inv}(W_B)}$. It therefore suffices to show (i) that $\mathcal{E}(B)\varepsilon(B)=(-1)^{\alpha_2}$ and (ii) that $(-1)^{|S'\cap S_2|+|S'\cap S_4|}=(-1)^{k_1+k_3}$.

For (i), write $\mathcal{E}(B)\varepsilon(B)=\prod_{p<q}h(\tau_p,\tau_q)$ as a product over pairs of columns, where $\tau_p$ is the column type of mode $p$. Both the class-word inversion and the trace-reordering contribution to $h$ are antisymmetric under exchanging the two columns, so $h$ is symmetric and $\mathcal{E}(B)\varepsilon(B)$ does not depend on the order of the columns. We may therefore sort $B$ into canonical order. In this order $\operatorname{inv}(W_B)=0$, and $\varepsilon(B)=(-1)^{I(B)}$ follows from counting the row inversions $I(S_r,S_{r'})$ over the six row pairs. Every contribution that mixes different classes, or that involves the weight-four columns, appears an even number of times. Of the within-class terms, only $\alpha_2^2\equiv\alpha_2 \pmod 2$ appears an odd number of times.

For (ii), only columns whose mode lies in $S'$ contribute, and such a column changes the parity exactly when it occupies precisely one of rows $2$ and $4$. Weight-four columns and class-$2$ columns (rows $\{1,3\}$ or $\{2,4\}$) occupy both or neither of these rows. Class-$1$ columns (rows $\{1,2\}$ or $\{3,4\}$) and class-$3$ columns (rows $\{1,4\}$ or $\{2,3\}$) always occupy exactly one, which gives $(-1)^{k_1+k_3}$. The full calculation is presented in~Appendix~\ref{app:OTOC-calc}.
\end{proof}
We have thus reduced the summand of~\cref{eq:OTOC-table-sum-main} to class data alone, with every table with the same $(\alpha_1,\alpha_2,\alpha_3,q,k_1,k_3)$ values contributing to the sum identically. What remains is simply to count the multiplicity of each of the possible values that produce this tuple, which we do explicitly in~Appendix~\ref{app:OTOC-calc}. This leads to the following result:
\begin{application}[Matchgate-averaged $8$-point OTOC]\label{res:otoc-closed}
  For all $N$ and operators with string lengths $1 \leq L,L' \leq N$,
  \begin{equation}\label{eq:otoc-closed}
      \overline{\mathrm{OTOC}}^{(4)}\!=\! \sum_{j=0}^{\min\{L,L'\}}\! \frac{(-8)^j(L')_{{j}}(2N-L')_{{j}}L_{{j}}(2N-L)_{{j}}}{j!(2N)_{{2j}}(2N+2)_{{j}}}.
    \end{equation}
\end{application}
We note the symmetry of the expression in $L \leftrightarrow L'$, and that it covers all possible initial Pauli operators through the parity operator symmetry, $L \leftrightarrow 2N-L$. In~\cref{fig:otoc-numerics}, we display~\cref{eq:otoc-closed} for different parameter ranges of string length. We see a sign flip with an odd versus even $L$, and for the magnitude, a non-linear dependence on the lengths $L,L'$ across the extensive regime. 

As in the previous results of Applications~\ref{res:LOE-closed-form} and~\ref{res:OSE-closed}, it is instructive to evaluate our result in the regimes of initial operators with a fixed or extensive Majorana string length. However, for the OTOC we need to consider the length of its two arguments, $L$ and $L'$. We identify three relevant regimes. First, for $L,L'=\mathcal{O}(1)$, we find
\begin{equation}
    \overline{\mathrm{OTOC}}^{(4)} = 1 - \frac{4LL'}{N} + \mathcal{O}(N^{-2}). 
\end{equation}
Interestingly, we see that asymptotically the OTOC becomes equal to $1$, the value it attains when $O_U$ and $X$ commute. Such operators, therefore, do \emph{not} become freely independent after evolving one with deep, random matchgates. This adds to a growing body of evidence that saturation to free independence is a property of genuinely complex quantum systems~\cite{Pappalardi2022freeETH,fava2023designsfreeprobability,dowling2025freeindep}. We see a similar behaviour if the probe is extensive, $L=\mathcal{O}(1)$ and $L'=N$, finding 
\begin{equation}
    \overline{\mathrm{OTOC}}^{(4)}=(-1)^L\left[1-\frac{L(L+1)}{N} +O_L(N^{-2})\right],
\end{equation}
which again is not freely independent as $N \to \infty$. The above two asymptotic expressions can be proved immediately through direct summation and expansion. 

 Choosing extensive $L=L'=N$, we see a qualitative change: 
\begin{equation}
    \overline{\mathrm{OTOC}}^{(4)} \approx \frac{1}{2^{N+1}}\!\left[\sin\left(\!\frac{\pi N}{2}\!\right) \!+ \!\cos\left(\!\frac{\pi N}{2}\!\right)+ \frac{(-1)^N}{\sqrt{2}}\right]
\end{equation}
up to a relative error of $\mathcal{O}(N^{-1}) $. The proof of the above asymptotic formula involves reducing the sum to an equivalent integral and applying endpoint Laplace’s method; see Appendix~\ref{app:OTOC-calc}. In contrast to the finite string length situation, in this dual extensive regime, we see that the OTOC vanishes exponentially with $N$. We therefore conclude, at least at the level of $8$-point OTOCs, such operators become freely independent when one is evolved by a random matchgate unitary.

\section{Operator Resource Theories and Classical Simulation Complexity} \label{sec:reource_theories}
The quantities computed in the previous section are not all independent. The OSE and LOE of any (normalised and Hermitian) operator are related by a bound: for any $U$, for any initial Pauli operator $O$, and for $k\geq1$~\cite{Dowlin2024LOE-OSRE},
\begin{equation}
  M^{(k)}(O_U) \leq E^{(k)}(O_U), \label{eq:bridging}
\end{equation}
where we have written the bound in terms of their purity quantities to connect more readily with our results. This is a strict relationship, with no such simple equivalent for the corresponding magic and entanglement monotones in the Sch\"odinger picture. 

It is natural to ask whether some measure of operator non-Gaussianity fits within, or on either side of, the bound hierarchy in~\cref{eq:bridging}. Irrespective of what measure of non-Gaussianity we choose, a first check is what kind of operator entanglement and magic values a `free' Gaussian unitary generate for an initial Pauli operator. First, it is clear that a Pauli string has zero operator entanglement and magic. On the other hand, we have seen in the previous section that a random Gaussian unitary acting on an extensive Majorana string ($L=\Theta(N)$) leads to a volume law for both quantities. We therefore conclude that no immediate analogue of~\cref{eq:bridging} can hold for a non-Gaussianity measure for any initial $O$. 

However, if we restrict the initial operator to have sub-extensive Majorana length, $\min(L,2N-L)=o(N)$, we have seen that there are non-trivial bounds on both OSE and LOE; cf. Eqs.~\eqref{eq:loe_bound} and \eqref{eq:ose_bound}, and the dashed lines in {Fig.~}\ref{fig:ose-loe-gap}. There is a separation between operator entanglement and magic on the Gaussian orbit: a $L=\mathcal{O}(1)$ initial Majorana string has constant-bounded LOE for any Gaussian evolution, but $\Theta(\log(N))$ OSE. One of the more interesting observations to be gleaned from~\cref{fig:ose-loe-gap} is that {this gap persists even when the bounds of Eqs.~\eqref{eq:loe_bound} and \eqref{eq:ose_bound} become vacuous}. The usual intuition is that a uniformly random unitary typically leads to maximal entanglement (allowable by symmetry constraints)~\cite{Popescu2006}. Here, however, we see a suppressed volume law for LOE, far from maximal. This is in stark contrast to the fact that operator magic saturates to a maximal value for random Gaussian unitaries, and that random unitaries lead to a maximal operator entanglement up to a constant Page correction~\cite{dowling2026page}.

Including some small non-Gaussian resources clearly changes this story. Namely, if one introduces local, non-Gaussian gates into a matchgate circuit (e.g. a SWAP gate), or a defect into an otherwise non-interacting Fermionic Hamiltonian, then the dynamics perturb away from the Gaussian case. Such dynamics may clearly generate operator magic and entanglement beyond the bounds of Eqs.~\eqref{eq:loe_bound} and \eqref{eq:ose_bound}. To explore this intuition, consider a simple example of a non-Gaussian gate, $G:=\exp(ic_{\nu_1} c_{\nu_2} c_{\nu_3} c_{\nu_4})$. Such a quartic gate may change the Majorana string length by $\pm 2$ through its conjugate action. Then, take a standard doped circuit architecture~\cite{trigueros2026unitarydesignsdopedmatchgate},
\begin{equation}
  U_t := U^{(t)} G U^{(t-1)} G \dots U^{(1)} G U^{(0)} 
\end{equation}
where $U^{(i)}$ is an arbitrary matchgate circuit, and $t$ is the number of non-Gaussian gates. For these unitaries, we can refine~\cref{eq:loe_bound} with the following result.
\begin{restatable}[{LOE bound for doped matchgate circuits}]{theorem}{doped} \label{res:doping}
For an initial Majorana string $O=\gamma_{\bm{\nu}}$ of length
$|\bm{\nu}|=L$, let $U_t$ be a matchgate circuit with $t$ quartic gates $G=\exp(i\theta c_{\nu_1} c_{\nu_2} c_{\nu_3} c_{\nu_4})$, for any $\theta\geq 0$.
Across a bipartition of the first $M$ qubits and the next $N-M$, and for $k\geq1$,
\begin{equation}\label{eq:doped-loe-bound}
    \begin{split}
        &\operatorname{LOE}_A^{(k)}(U_t^\dagger O U_t) \leq L+\kappa(\theta)t, \quad \text{where}\\
  &\kappa(\theta):=8h_2\!\left(\frac{2+|\sin(2\theta)|}{8}\right)-4,
    \end{split}
\end{equation}
and $h_2(x):=-x\log_2x-(1-x)\log_2(1-x)$.
\end{restatable}
\begin{proof}[Proof sketch]
    Using an auxiliary fermion representation $\{ f_\alpha\}$ for the vectorised operator~\cite{Calabrese_2005,Prosen2007a}, entropy inequalities lead to the bound: $\operatorname{LOE}^{(1)}_A\leq\tr [ h_2(C/2)]$, where $C_{\mu \nu }=\llangle{O}|f_\mu^{\dagger} f_\nu | O\rrangle $ is the correlation matrix and $h_2$ is the binary entropy. After absorbing the Gaussian layers into the quartic generators, the non-Gaussian gates involve at most $4t$ Majorana directions and increase the mean string length by at most $2t|\sin(2\theta)|$, while occupied modes orthogonal to these directions remain unchanged. Separating these unchanged modes and applying concavity to the remaining correlation eigenvalues, together with R\'enyi entropy monotonicity, yields the final result.
\end{proof}
We see an interplay between the `free' operator entanglement that can be generated within the weight sectors due to Gaussian rotations, and that which is generated from mixing the subsectors due to the non-Gaussian gates. We expect the above result to generalise to other (continuous) measures of non-Gaussianity, such as the correlation entropy of the vectorised operator~\cite{debertolis2025naturalsuper} or Fermionic anti-flatness~\cite{Sierant_2026,turkeshi2026williamsonmajorizationtheoryfermionic}. 

\cref{res:doping} may also be read as a Heisenberg-picture counterpart of the principle that entanglement generation can lower-bound circuit complexity~\cite{Eisert2021EntanglingPower}. Indeed, since operator entanglement is the usual state entanglement after vectorising the operator, the same incremental-entangling logic applies: bounded operator-entangling power per elementary resource limits the LOE that a circuit can generate. To unpack this, it is instructive to reformulate~\cref{res:doping} as a circuit complexity lower-bound.

\begin{figure}[t!]
  \centering
  \includegraphics[width=0.85\linewidth]{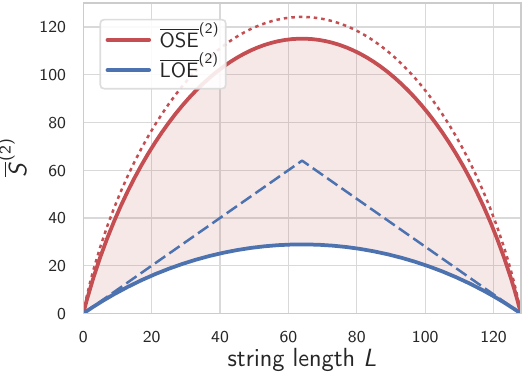}
\caption{
   The magic-versus-entanglement separation for random matchgates on $N=32$ sites and the half-chain cut, with the Fermionic bounds of Eqs.~\eqref{eq:loe_bound} and~\eqref{eq:ose_bound} shown as dashed lines. The shaded region is the window allowed by $\mathrm{LOE}^{(2)}\le \mathrm{OSE}^{(2)}$, the entropic version of~\cref{eq:bridging}, showing a large separation that persists in the thermodynamic limit. Random free-fermionic dynamics therefore rather uniformly spread operators in the Pauli basis, but these strings have suppressed operator entanglement.}
  \label{fig:ose-loe-gap}
\end{figure}

\begin{corollary}[{Non-Gaussian circuit-resource lower bound}]\label{cor:doping-lower-bound}
Under the assumptions of~\cref{res:doping}, the number $t\in\mathbb{N}$ of quartic non-Gaussian gates in $U$ obeys
\begin{equation}
  t\geq \max\!\left\{0,\left\lceil\frac{\mathrm{LOE}^{(k)}(U^\dagger O U)-L}{\kappa(\theta)}\right\rceil\right\}.
  \label{eq:doping-lower-bound}
\end{equation}
\end{corollary}
More generally, whenever examining low Majorana degree observables and when the permitted non-Gaussian primitives have bounded operator-entangling power (i.e., are local gates), LOE growth supplies a lower bound on the corresponding non-Gaussian circuit complexity. On the one hand, this means that operator entanglement growth provides a witness to non-Gaussian resources. On the other hand, together with~\cref{eq:bridging}~\cite{Dowlin2024LOE-OSRE}, we see that the LOE can constitute a unifying measure of distinct quantum resources. Recalling the interpretation of the linear-in-time scaling of LOE as a faithful dynamical signature of quantum chaos~\cite{Prosen2007,Prosen2007a,Kos2020,dowling2023scrambling}, we conclude that any such `chaotic' system generates at least a linear amount of each of the distinct non-classical resources of entanglement, magic, and non-Gaussianity.

The above results can be interpreted through the lens of concrete classical simulation methods. From~\cref{eq:bridging} we know that, at least in 1D where Heisenberg-picture TEBD is easily implementable, tensor networks will always have at least a comparable memory cost to Pauli propagation~\cite{Dowlin2024LOE-OSRE,dowling2026classicalsim}. Beyond the doping scheme described above, the exact extension to near-Gaussian methods remains an interesting open question. Corollary~\ref{cor:doping-lower-bound} suggests a parameterised approach to the simulation of near-Gaussian circuits, in which the complexity is controlled by the number or strength of the non-Gaussian resources. We expect that, for a suitably defined non-Gaussianity monotone which relates directly to the cost of matchgate-based simulation methods (such as those based on Majorana propagation~\cite{miller2025simulation,D_Anna_2026} or few-body revealing~\cite{2025resolvingspacetimestructuresquantum}), an analogue of~\cref{res:doping} should hold.

We may also consider the relation of OTOCs to simulability. While they do not constitute a formal complexity resource (though note that the $k=2$ case can be related to LOE~\cite{dowling2023scrambling}), it has been heuristically argued that their evolution is inherently difficult to reproduce classically, sample-wise for random unitary circuits~\cite{Abanin2025}. We have provided exact expressions for their value in an analogous circuit setting of random matchgates. 
In particular, each individual sample can be evaluated efficiently in place of a trace over a Hilbert space of dimension $2^N$. This is because Majorana strings are themselves matchgate unitaries, and so $(O_UX)^4$ is a Gaussian unitary whose trace can be evaluated at a cost of $\mathcal O(N^3)$ from the $2N\times 2N$ orthogonal matrix using a phase-resolved trace formula~\cite{Klich_2014}.
Then, from Application~\ref{res:otoc-closed}, we know analytically that the long-time behaviour can mirror the unitary case. This makes random matchgates an efficiently computable control ensemble for higher-order OTOC experiments of the kind considered in Ref.~\cite{Abanin2025}. OTOCs built from extensive strings can be exponentially small, and their Pauli expansion can display non-trivial interference effects, despite the classical tractability of the underlying dynamics. Hence neither a small late-time OTOC nor interference among many Pauli components is, by itself, a certificate of quantum computational advantage. Any hardness in the interacting setting must arise from additional structure absent from the free-fermionic contraction, potentially including the finite-time and spatial profile of the correlator or the non-Gaussianity of the dynamics. A detailed
quantitative comparison of these features is left for future work.

\section{Discussion} \label{sec:discussion}
Our contributions in this work are threefold. From a \textit{mathematical perspective}, we have provided a new construction of matchgate Weingarten calculus, entirely within the Majorana basis and so tailored to computing moments of random operators. This is a general, usable toolbox for performing exact averages, which can be directly applied to relevant applications beyond our proof-of-principle examples across many-body physics and quantum information. Our formulae are explicit for $k=4$, and we have supplied a clear recipe for computing with respect to $k\geq 6$. We therefore offer an alternative perspective on the structure of the matchgate group compared to orthonormal constructions~\cite{sierant2026theorymatchgatecommutant,braccia2026commutantfermionicgaussianunitaries}.

In terms of \textit{many-body physics}, we demonstrated the usefulness of our methods by analytically computing three relevant examples of operator complexities under random Gaussian unitaries. We found that the complexity of random operators depends greatly on their initial Majorana string length, with constant-mode operators leading to constant values of the LOE and OTOC, while extensive strings saturate to Haar-unitary-like scaling. Such operators therefore look complex in this representation despite their underlying tractable Gaussian structure. A notable feature here is that the LOE saturates to less than half the Page value, despite no \textit{a priori} structural restrictions. Our results therefore add insight into the expected long-time behaviour of fermions with random interactions. 

From a \textit{quantum-computational perspective}, our results point to a hierarchy of operator complexities. MPO simulation, Pauli propagation, and Majorana-based methods exploit different structures, and their costs need not track one another. A circuit may generate an operator that is highly complex in one representation while remaining efficient in another. Random matchgate circuits provide an exactly solvable example: although they are efficiently classically simulable by exploiting their fermionic Gaussian structure, they can generate extensive operator entanglement and an essentially maximal spreading over the Pauli basis. Thus, neither volume-law operator entanglement nor large operator magic alone constitutes a representation-independent witness of computational hardness. However, for initially low-weight Majorana observables, LOE beyond the Gaussian baseline certifies non-Gaussian circuit resources. These two regimes together illustrate both sides of the same principle: operator complexity becomes a computationally meaningful witness only after the relevant free structure and representation have been specified. Interestingly, there is a distinction between magic and non-Gaussianity in this sense, we have seen that the latter can create operator entanglement in the `free' sector, whereas the former cannot~\cite{Dowlin2024LOE-OSRE}.

There are several fruitful directions to explore based on the Majorana matchgate Weingarten calculus and its applications. First, it should be straightforward to extend~\cref{res:CB-four-replica} to higher replicas $k\geq 5$, based on the procedure outlined in Section~\ref{sec:twirls}. Having access to higher matchgate designs could, for instance, offer insight into the $k$ scaling of higher-order OTOCs and LOE, or allow shadow tomography over exotic representations of the matchgate group; cf. the use of $6$-designs for quotient group shadows~\cite{chang2026classicalshadowssymmetricspaces}. Moreover, it would be interesting to study matchgate twirls under different symmetries; most immediately those with particle-preserving symmetry, as described by the group $\mathrm{U}(N)$. Such a result would allow one to probe the effect of symmetries on the operator entanglement Page curve, for instance. It would be interesting to compare this to the different volume laws of random fermionic states~\cite{Bianchi_Hackl_Kieburg_Rigol_Vidmar_2022}.

In terms of resource theories, one should explore an operator-focused resource theory of non-Gaussianity, with Gaussian conjugations as free operations and Gaussian-evolved Majorana strings as free operators. \cref{res:doping} may offer insight in this direction, suggesting that an analogue of `correlation entropy'~\cite{debertolis2025naturalsuper,Sierant_2026} may control the cost of near-Gaussian simulation methods. This would parallel LOE and OSE, which respectively control the cost of MPO and Pauli-propagation methods~\cite{dowling2026classicalsim,dowling2024magicheis}.

Meanwhile, to better relate the matchgate averages to physical setups, one should study microscopic models where they are achieved and the involved timescales. More broadly, it would be interesting to connect the global matchgate-Haar averages studied here to spatially resolved random Gaussian ensembles such as the Brownian SYK model~\cite{Piroli2019} or all-to-all matchgate circuits~\cite{west2026strongmatchgatedesignsnearly}, and matchgate tensor networks for which continuum field theories retain information about geometry, topology, and disorder~\cite{AltlandMatchgates}. For OTOCs, the $k$ dependence of the (typically exponential) OTOC decay in microscopic models is not yet well understood, and solving a local Gaussian model that eventually reaches our globally averaged results will help uncover the role of free probability in the equilibration of chaotic many-body systems. It is the hope that the present work provides a machinery to further explore these topics at the intersection of quantum many-body physics and quantum information theory.

\section*{Acknowledgements}
We acknowledge useful discussions with Maxwell T. West, Xhek Turkeshi, Antonio Anna Mele, Lorenzo Piroli, Maxime Debertolis, and Lennart Bittel. ND acknowledges funding by the Deutsche Forschungsgemeinschaft (DFG, German Research Foundation) under Germany's Excellence Strategy - Cluster of Excellence Matter and Light for Quantum Computing (ML4Q) EXC 2004/1 - 390534769. GALW acknowledges support from the Alexander von Humboldt foundation. We acknowledge support from the BMFTR (MUNIQC-Atoms, QuSol, Pasquops, Hybrid++, {FermiQP}), the Munich Quantum Valley, the Quantum Flagship (PasQuanS2, MILLENION), QuantERA ({SDPCode}), the Clusters of Excellence MATH+ and ML4Q, the DFG (CRC 183 and SPP 2514), the Einstein Foundation (Einstein Research Unit on Quantum Devices), Berlin Quantum, and the European Research Council (DebuQC). For the CRC 183 and the ML4Q it is the result of a cross-node collaboration involving Berlin and Cologne.

\section*{AI Statement} AI tools were used to assist in simplifying final closed-form expressions, developing proofs of the asymptotic formulas, and generating the plots visualising the analytic results. The main ideas, all other proofs, and the writing were all done without the use of AI.

\section*{Note added}
In the final stages of preparing this manuscript, we became aware of the related work of Ref.~\cite{bayonapena2026operato}, which focuses on the asymptotic von Neumann LOE for random Gaussian unitaries using a Coulomb gas approach. Our results in Section~\ref{sec:LOE} differ in that they are exact for all $D$ and are in terms of the $2$-R\'enyi LOE. 

% \section*{References}

%apsrev4-2.bst 2019-01-14 (MD) hand-edited version of apsrev4-1.bst
%Control: key (0)
%Control: author (8) initials jnrlst
%Control: editor formatted (1) identically to author
%Control: production of article title (0) allowed
%Control: page (0) single
%Control: year (1) truncated
%Control: production of eprint (0) enabled
%

\clearpage
\appendix

\crefalias{section}{appendix}
\crefname{appendix}{Appendix}{Appendices}
\Crefname{appendix}{Appendix}{Appendices}
\onecolumngrid

% \section{Brauer algebra} \label{app:brauer}
% % \GW{Appendix-worthy}
% Here we describe how to compute the coset type of a relative pairing.
% Since $\omega^\lambda$ is a function of the double coset $H_n\sigma H_n$ only, all that is ever needed of a pair $(p_1,p_2)$ is its \emph{coset type} $\Xi[p_1^{-1}p_2]\vdash n$. This is obtained by overlaying the two matchings and reading off the cycle lengths of the resulting union multi-graph, each cycle counting for half its length. For example, take $p_1,p_2\in M_6$ with $p_1=e$ and $p_2=\{\{1,4\},\{2,3\},\{5,11\},\{6,9\},\{7,8\},\{10,12\}\}$ according to
% \begin{equation}
%  p_1 = \GraphId,\qquad p_2 = \GraphPair.
% \end{equation}
% Overlaying,
% \begin{equation}
%  p_1^{-1}p_2= \GraphCombined\implies \Xi[p_1^{-1}p_2] = \GraphLooped,
% \end{equation}
% The corresponding pairing type is $(3^1,2^1,1^1)$.

\section{MPO representation of low-index Majorana strings} \label{app:MPO}
In this section we prove~\cref{prop:loe-bound}, which we restate here for convenience.
\loeBound*
We will prove that a single Heisenberg-evolved Majorana spinor can be written as an MPO of bond dimension {at most} $2$, under any Gaussian evolution. This means that any initial, finite-index operator with Majorana string length $L$, has LOE bounded by 
\begin{equation}
  \mathrm{LOE}^{(k)}(O_U) \leq L,
\end{equation}
for all R\'enyi indices $k\geq 0$, under any Gaussian dynamics $U$ (including the $\alpha=0$ R\'enyi entropy, equal to the log of the bond dimension). Therefore, the long-time limit of LOE for finite-index operators under any Gaussian unitary dynamics is bounded by a constant. This should be contrasted with the OSE, which may have unbounded growth, scaling with $N$. We note that this result is known; it is mentioned in, e.g., Ref.~\cite{Prosen2007,Dubail_2017}, and the presented proof is adapted from Ref.~\cite{Hubig_2017}. We present it here for completeness, to show that {fermions} offer an example of a system with a logarithmic separation between the LOE and OSE, and as a Lemma which is instrumental in deriving~\cref{res:doping}. 

To start, we revisit the definition of a matrix product operator (MPO). Given a set of $N$ local Hilbert spaces $\mathcal{H}_i$ of dimension $d_i$ each and an operator $O$ which acts on the tensor product space $\mathcal{H}=\bigotimes_{j=1}^N\mathcal{H}_i$ we can write the operator $O$ as 
\begin{equation}
  O = \sum_{\vec{\sigma}\vec{\tau}} \alpha_{\vec{\sigma}\vec{\tau}}|\vec{\tau}\rangle\!\langle \vec{\sigma}|,
\end{equation}
where $|\vec{\tau}\rangle$, $\langle\vec{\sigma}|$ enumerate respectively the product basis states of the Hilbert space $\mathcal{H}$ and its dual. That is, $|\vec{\tau}\rangle = \bigotimes_{i=1}^N|\tau_i\rangle = |\tau_1\cdots\tau_N\rangle.$
Explicitly, then, the operator $O$ can be written as
\begin{equation}
  O = \sum_{\sigma_1\tau_1}\cdots\sum_{\sigma_N\tau_N}\alpha_{\tau_1\cdots\tau_N}^{\sigma_1\cdots \sigma_N}|\tau_1\cdots\tau_N\rangle\!\langle\sigma_1\cdots \sigma_N|.
\end{equation}
The coefficient $\alpha\in\mathbb{C}^{\prod_id_i^2}$ may now be decomposed as a set of matrix products. For each site $i$ and for each combination of local states $\{|\tau_i\rangle,\langle\sigma_i|\}$ we introduce a set of matrices $(W_i^{\sigma_i\tau_i})_{w_{i-1},w_i}$ with the property that their matrix-matrix product equals a specific element of the $\alpha$ tensor.

\begin{equation}
\label{eq:MPO-contract}
  \sum_{\vec{w}}(W_{[1]}^{\sigma_1\tau_1})_{w_0,w_1}\cdot (W_{[2]}^{\sigma_2\tau_2})_{w_1,w_2}\cdots(W_{[N]}^{\sigma_N\tau_N})_{w_{N-1},w_N} = \alpha_{\tau_1\cdots\tau_N}^{\sigma_1\cdots \sigma_N}.
\end{equation}

From this expression, we see that there is an alternate way to express $\alpha$ -- this is known as its \emph{matrix product operator} (MPO) representation.
The $W_{[i]}^{\sigma_i\tau_i}$ matrices are not square in general, and have the shape $\chi_{i-1}\times \chi_i$. Define $\chi$ as 
$\max(\{\chi_i\})$ to be the \emph{bond dimension}, which governs the difficulty of computing~\cref{eq:MPO-contract}. We claim that the MPO of the operator $U c_\nu U^\dagger$ has bond dimension equal to either 1 or 2 for all choices of $U\in \mathcal{G}_{N}$ and all choices of $c_\nu$. Let us denote $U c_\nu U^\dagger\equiv c'_\nu$. Noting from~\cref{eq:matchgate_action} that $c'_\nu = \sum_\mu R_{\nu\mu} c_\mu$, we can write this using the Jordan-Wigner transform as 
\begin{equation}
  \sum_{\mu:\mu~\text{mod}~2 = 0}(R_{\nu\mu})Z^{\otimes \frac{\mu}{2}-1}\otimes X\otimes I^{\otimes (2N - \frac{\mu}{2})} + \sum_{\mu:\mu~\text{mod}~2 = 1}(R_{\nu\mu}-1)Z^{\otimes \frac{\mu-1}{2}}\otimes X\otimes I^{\otimes (2N - \frac{\mu-1}{2})}.
\end{equation}
Equivalently, let us re-express this by collecting the local $X$ and $Y$ terms together. Then we have 
\begin{equation}
  c_\nu' = \sum_{j=1}^N Z^{\otimes j-1}\otimes (\alpha_j^+ X_j + \alpha_j^-Y_j)\otimes I^{N-j},
\end{equation}
where $\alpha_j^+ = R_{\nu,2j-1}$ and $\alpha_j^- = R_{\nu,2j}$. We now make the observation that if we define 
\begin{equation}
  V_{[i]}^{\sigma_i\tau_i} = \begin{pmatrix}
    \mathbb{I}_{[i]}^{\sigma_i\tau_i} & 0 \\ 
    \alpha_i^+ X_{[i]}^{\sigma_i\tau_i} + \alpha_i^-Y_{[i]}^{\sigma_i\tau_i} & Z_{[i]}^{\sigma_i\tau_i}
  \end{pmatrix}.
\end{equation}

Then we have 
\begin{equation}
\label{eq:gaussian-MPO}
  c_\nu' = \begin{pmatrix}0 & 1\end{pmatrix}\sum_{\vec{w}}(V_{[1]}^{\sigma_1\tau_1})_{v_0,v_1}\cdot (V_{[2]}^{\sigma_2\tau_2})_{v_1,v_2}\cdots(V_{[N]}^{\sigma_N\tau_N})_{v_{N-1}v_N}\begin{pmatrix} 1 \\ 0\end{pmatrix},
\end{equation}
which is exactly the MPO form given in Equation~\eqref{eq:MPO-contract}. In other words, Equation~\eqref{eq:gaussian-MPO} is the MPO form of $c_\nu'$. And moreover, we see that each $V_{[i]}$ is a $2\times 2$ matrix, meaning that $c_\nu'$ has bond dimension at most 2. Then a Majorana string of length $L$ has a maximal bond dimension of $2^L$, which follows from concatenating $L$ MPOs representing each spinor, each of bond dimension at most $2$. 

Since the operator Schmidt rank across any cut of the chain is bounded by the MPO bond dimension, $\mathrm{LOE}^{(0)}_A(O_U)\leq L$, and hence $\mathrm{LOE}^{(k)}_A(O_U)\leq L$ for all $k\geq0$ by monotonicity of the R\'enyi entropies in $k$. To obtain the bound with $2N-L$, note that the complementary string satisfies $\gamma_{\bar{\bm{\nu}}}=\eta\,P\gamma_{\bm{\nu}}$ for a phase $|\eta|=1$, where $P=\prod_{x=1}^N\sigma_z^{(x)}$ is the global parity operator. Any Gaussian unitary satisfies $UPU^\dagger=\det(R)\,P=\pm P$, so $U^\dagger\gamma_{\bar{\bm{\nu}}}U=\pm\eta\,P\,(U^\dagger\gamma_{\bm{\nu}}U)$. Multiplication by the product operator $P=P_A\otimes P_{\bar A}$ acts as a local unitary on $|O_U\rrangle$ and leaves its operator Schmidt coefficients unchanged. Applying the length-$L$ bound to whichever of $\gamma_{\bm{\nu}}$ and $\gamma_{\bar{\bm{\nu}}}$ is shorter completes the proof of~\cref{prop:loe-bound}.

\section{Proof of the complete contingency table moment operator expression}\label{app:full-coefficient-calculation}
In this section, we detail proofs of technical lemmas that were used in Section~\ref{sec:contingency} to arrive at our main result, Theorem~\ref{res:CB-four-replica}. We start with the pairing expression for $C(B)$. We may express this by writing the two boundary vectors in the abstract diagram basis on $\mathbb{R}^{M_n}$ as
\begin{equation}
 \Phi_{\rm out}(B):=\sum_{p\in\mathcal{P}(B)}\ \sum_{\bm\sigma\in S_L^k}\operatorname{sgn}(\bm\sigma)\,\delta_{\bm\sigma\cdot p}, \qquad \Phi_{\rm in}=\sum_{q\in\mathcal{Q}_{k,L}}\delta_q,
\end{equation}
recalling that $\mathcal{Q}_{k,L}\coloneqq M_{k/2}^{\times L}$. From this, we have that $C(B)=\langle \Phi_{\rm out}(B),\operatorname{Wg}\Phi_{\rm in}\rangle$ with $\operatorname{Wg}$ the Weingarten matrix of~\cref{eq:weingarten-general}. 

\subsection{Proof that $C(B)$ has support on only a single irreducible representation}
We first prove the {lemma} on the unique surviving Weingarten sector, which we reproduce here. 
\irrep*
\begin{proof}
The two essential elements here is that $\Phi_{\rm out}(B)$ transforms according to the sign irrep of $S_L^k$ and $\Phi_{\rm in}$ transforms according to the trivial irrep of $S_k^L$. There are hence two individual projections at hand from structured sums over input and output pairings respectively whose support coincides only on a single irrep. Let us look first at the output pairings which, accounting for the spinor anticommutation relations, is the signed sum over the Young subgroup $K:=S_L^k\subset S_{kL}$. This independently permutes the $L$ Majorana labels in each of the $k$ output rows. Because permuting Majoranas within a row produces the sign representation of $S_L$, the output-side function transforms under $K$ according to
\begin{equation}
 \operatorname{sgn}_L^{\boxtimes k} :=\underbrace{\operatorname{sgn}_L\boxtimes\cdots\boxtimes\operatorname{sgn}_L}_{k\text{ factors}}.
\end{equation}
Viewed as a function on all $kL$ vertices, it therefore lies in the induced representation $\operatorname{Ind}_{S_L^k}^{S_{kL}}\left(\operatorname{sgn}_L^{\boxtimes k}\right).$
Consequently, only those $S_{kL}$ irreducible representations appearing in this induced representation can contribute to the signed row sum. To describe these irreps, let us temporarily denote $\Lambda\vdash kL$ for a partition labelling an irrep of $S_{kL}$. By the Littlewood-Richardson rule, the multiplicity of $V_\Lambda$ in the above induced representation is the coefficient of $s_\Lambda$ in $e_L^k$, where $e_L=s_{(1^L)}$ is the Schur function corresponding to the Frobenius characteristic of the sign representation of $S_L$. We take the $k$th power $e_L^k$ because induction from Young subgroups corresponds to multiplication of symmetric functions. Expressing this decomposition explicitly, we have
\begin{equation}\label{eq:sign-induced-decomposition}
 \operatorname{Ind}_{S_L^k}^{S_{kL}}\left(\operatorname{sgn}_L^{\boxtimes k}\right)\cong \bigoplus_{\Lambda\vdash kL} c_\Lambda V_\Lambda,\qquad e_L^k=\sum_{\Lambda\vdash kL}c_\Lambda s_\Lambda.
\end{equation}
From here, it remains to characterise which partitions $\Lambda\vdash kL$ appear in $e_L^k$. One can do this by using the standard involution $\omega$ of the ring to change from the sign irrep to the trivial irrep. Letting $h_L=s_{(L)}$ be the Frobenius characteristic of the trivial representation, $\omega$ satisfies $\omega(h_r)=e_r$ and $\omega(s_\mu)=s_{\mu'}$, where $\mu'$ denotes the conjugate partition to $\mu$ (a transposition of their Young diagrams). Hence $e_L^k = \omega(h_L^k)$. The product $h_L^k=h_{(L^k)}$ has Schur expansion $h_L^k = \sum_{\mu\vdash kL} K_{\mu,(L^k)}s_\mu$,where $K_{\mu,(L^k)}$ is a Kostka number. Applying $\omega$ gives $e_L^k=\sum_{\mu\vdash kL}K_{\mu,(L^k)}s_{\mu'}$. Equivalently, after writing $\Lambda\equiv\mu'$, we are left with the expansion.
\begin{equation}
  e_L^k=\sum_{\Lambda\vdash kL}K_{\Lambda',(L^k)}s_\Lambda.
\end{equation}
Thus $V_\Lambda$ appears in the induced sign representation if and only if $K_{\Lambda',(L^k)}>0$, where $\Lambda'$ is the conjugate diagram to $\Lambda$. A known property of the Kostka numbers is that $K_{\Lambda',(L^k)}>0 \iff \Lambda' \unrhd (L^k)$, where the symbol $\unrhd$ denotes a dominance ordering of the partitions. Since dominance is reversed by conjugation, $\Lambda'\unrhd (L^k)$ is the same statement as $\Lambda \unlhd (L^k)' = (k^L).$ Lastly, given that both partitions have total size $kL$, this condition is equivalent to saying that $\Lambda$ has at most $k$ nonzero parts, or equivalently that $\Lambda'$ has at most $k$ rows.

The same logic equivalently applies to the other side. The input pairings run over $\mathcal{Q}_{k,L}$, i.e., all ways of matching the $k$ replicas \emph{within} each of the $L$ columns. This set is manifestly invariant under permuting replicas inside a column, so the input vector $\Phi_{\rm in}$ is fixed by the Young subgroup $S_k^{\times L}\equiv S_k^L\subset S_{kL}$. It therefore transforms in the \emph{trivial} representation of $S_k^{L}$ and hence only irreps of $S_{kL}$ that contain an $S_k^L$-invariant vector can contribute. Similar to the above, these are exactly the constituents of $\operatorname{Ind}_{S_k^L}^{S_{kL}}(\operatorname{triv}_k^{\boxtimes L})$, whose Frobenius characteristic is $h_k^L$. Expanding once more in Schur functions, we have $h_k^L=\sum_{\Lambda\vdash kL}K_{\Lambda,(k^L)}\,s_\Lambda,$ and so $V_\Lambda$ occurs if and only if the Kostka number $K_{\Lambda,(k^L)}$ is nonzero, which by the dominance criterion is $\Lambda\unrhd (k^L).$ 

Combining the input and output dominance criteria squeezes $\Lambda$ from both sides,
\begin{equation}
 (k^L)\unlhd\Lambda\unlhd (k^L) \qquad\Longrightarrow\qquad \Lambda=(k^L),
\end{equation}
a rectangular diagram with $L$ rows of length $k$. In the Weingarten expansion~\eqref{eq:weingarten-general} the relevant $S_{2n}$ irrep is $2\lambda$, so
\begin{equation}\label{eq:unique-lambda}
 2\lambda_\star=(k^L), \qquad \lambda_\star=\left(\left(\tfrac{k}{2}\right)^{L}\right), \qquad n=\frac{kL}{2}.
\end{equation}
We thus arrive at the result that the Weingarten function for all input and output pairings may be straightforwardly evaluated using~\cref{eq:weingarten-general} from Ref.~\cite{Collins_2009}.
\end{proof}

\subsection{Evaluating the full moment operator expression}\label{app:ssec:moment-form}
From the previous lemma, we have that the Weingarten function is proportional to a projector onto the $\lambda_\ast$ irrep. We can use this fact to analytically evaluate the rest of $s(B)$. The associated normalising polynomial in~\cref{eq:weingarten-general} is a rectangle of $L$ rows and $k/2$ columns,
\begin{equation}\label{app:eq:Zlambda-star}
 Z_{\lambda_\star}(1^{2N}) =\prod_{i=1}^{L}\prod_{j=1}^{k/2}\left(2N+2j-i-1\right)\;\overset{k=4}{=}\;\prod_{i=1}^{L}(2N-i+1)(2N-i+3)= (2N)_L\,(2N+2)_L ,
\end{equation}
where $(x)_L=x(x-1)\cdots(x-L+1)$ is the falling factorial. Because a single $\lambda$ survives, the table coefficient can be expressed as
\begin{equation}
 C(B) =\frac{1}{Z_{\lambda_\star}(1^{2N})}\big\langle \Phi_{\rm out}(B),\,\Pi_{\lambda_\star}\Phi_{\rm in}\big\rangle, \qquad \text{with} ~ (\Pi_\lambda)_{p_1,p_2}:=\frac{2^n n!}{(2n)!}f^{2\lambda}\,\omega^\lambda(p_1^{-1}p_2),
\end{equation}
Let $A:=\sum_{\bm\sigma\in S_L^k}\operatorname{sgn}(\bm\sigma) P(\bm\sigma)$ be the signed row antisymmetriser. The central projectors are complete and commute with $A$, while the dominance argument above gives $A\Pi_\lambda\Phi_{\rm in}=0$ for $\lambda\neq\lambda_\star$. Therefore $A\Phi_{\rm in} =A\sum_\lambda\Pi_\lambda\Phi_{\rm in} =A\Pi_{\lambda_\star}\Phi_{\rm in}$ and have the coefficient reduction
\begin{equation}\label{eq:CB-overlap}
 C(B)=\frac{s(B)}{Z_{\lambda_\star}(1^{2N})},\qquad s(B):=\big\langle\Phi_{\rm out}(B),\Phi_{\rm in}\big\rangle =\sum_{p\in\mathcal{P}(B)}\sum_{\bm\sigma\in S_L^k}\operatorname{sgn}(\bm\sigma)\, \big[\bm\sigma\cdot p\in \mathcal{Q}_{k,L}\big].
\end{equation}
\noindent
Thus, our initial expression has reduced into a signed count of row permutations that turn a column-compatible \emph{output} pairing into a column-compatible \emph{input} pairing, divided by the single polynomial~\eqref{eq:Zlambda-star}.

Let us now turn to the explicit calculation of $s(B)$. It is convenient to speak of the $kL$ occupied entries of $B$ as \emph{nodes}: a node is a pair $(r,j)$ with $B_{rj}=1$. A row permutation $\bm\sigma\in S_L^k$ is equivalently a \emph{colouring} $\chi$, obtained by listing the nodes of each row in increasing mode order and assigning them the labels $\sigma_r(1),\ldots,\sigma_r(L)$. Thus each colour $c\in[L]$ occurs exactly once in every row, and $\chi$ records the output slot into which each Majorana is permuted; see~\cref{eq:s-colouring-diagram}. The condition $\bm\sigma\cdot p\in\mathcal{Q}_{k,L}$ in~\eqref{eq:CB-overlap} says that $p$ pairs nodes of \emph{equal colour}. Since $p$ pairs only within columns, the constraint factorises over modes,
\begin{equation}\label{eq:s-factorised}
 s(B)=\sum_{\bm\sigma\in S_L^k}\operatorname{sgn}(\bm\sigma)\prod_{j=1}^{2N}\omega_j(\bm\sigma), \qquad \omega_j(\bm\sigma):=\#\Big\{\text{perfect matchings of column }j\text{ into monochromatic pairs}\Big\},
\end{equation}
with the convention $\omega_j=1$ for empty columns. In particular $\omega_j$ depends on $B$ only through the set of rows occupied in column $j$, so $s(B)$ depends on $B$ only through its multiset of column occupancy patterns, together with their left-to-right order. The above expression is valid for general $k$, but to fully reduce it and compute our quantities of interest we will specialise to the case where $k=4$. This has the very useful property that each column has two pairs, and so knowing one pairing type uniquely fixes the other. Thus, the ``mode'' pairing class is always the same as the colour pairing class. This has the consequence that the sign components of a table factorise.
\signMag*
To see this, suppose $\bm\sigma$ contributes, i.e., that $\prod_j\omega_j(\bm\sigma)\neq0$. Any two contributing colourings are related by permuting colours among the ``monochromatic blocks'' of nodes, and every such elementary move exchanges two colours in \emph{an even} number of rows: transposing two colours globally flips $\operatorname{sgn}(\sigma_r)$ in all $k$ rows, giving $(-1)^k=+1$ for even $k$; re-pairing within a colour class flips it in exactly two rows, giving $(-1)^2=+1$. Hence $\operatorname{sgn}(\bm\sigma)$ takes a single value $\mathcal{E}(B)\in\{\pm1\}$ on the whole contributing set, and
\begin{equation}\label{eq:s-split}
 s(B)=\mathcal{E}(B)\,\mathcal{N}(B), \qquad \mathcal{N}(B):=\sum_{\bm\sigma\in S_L^k}\prod_j\omega_j(\bm\sigma)\ \geq 0 .
\end{equation}
Equivalently, $\mathcal{N}(B)$ counts pairs $(p,\bm\sigma)$ with $p\in\mathcal{P}(B)$ monochromatic for $\chi$.

\emph{Computing the magnitude of a $k=4$ table $\mathcal{N}(B)$.---} For the magnitude component of this expression, we have the following lemma
\magAd*
A column with $m_j=2$ occupies a pair of rows; a column with $m_j=4$ occupies all four. Label the three ways of splitting $\{1,2,3,4\}$ into two pairs by
\begin{equation}\label{app:eq:classes}
 e_1=\big\{\{1,2\},\{3,4\}\big\},\qquad
 e_2=\big\{\{1,3\},\{2,4\}\big\},\qquad
 e_3=\big\{\{1,4\},\{2,3\}\big\},
\end{equation}
and call $e_t$ the \emph{class} $t$. A weight-two column carries an unambiguous class, namely the one containing its occupied row pair. Write $n_{rr'}$ for the number of columns occupying exactly rows $\{r,r'\}$, and let $q:=\#\{j: m_j=4\}$.

\noindent
Now unpack let us unpack the magnitude $\mathcal{N}(B)$. Fix $p\in\mathcal{P}(B)$ and think of its pairs as \emph{edges}: a weight-two column contributes one edge, and a weight-four column contributes two edges which necessarily form a complementary pair, i.e.\ a class $t_j\in\{1,2,3\}$; choosing $p$ is exactly choosing $t_j$ for each weight-four column. Let $b_t$ be the number of weight-four columns assigned class $t$, so $b_1+b_2+b_3=q$.

A colour class is a set of four nodes, one per row, which is a union of two edges; those edges cover complementary row pairs and hence belong to the same class. By the row-regularity identities in~\cref{eq:class-balance-main}, class $t$ has $\alpha_t+b_t$ edges on each complementary row pair. Grouping the two sets into colours is a bijection and gives $(\alpha_t+b_t)!$ possibilities. Finally, multiplicity under relabelling the $L$ modes contributes $L!$. Therefore, we find that
\begin{equation}\label{app:eq:NB}
 \mathcal{N}(B) =L!\sum_{\substack{b_1,b_2,b_3\geq0\\ b_1+b_2+b_3=q}}\binom{q}{b_1,b_2,b_3}\prod_{t=1}^{3}(\alpha_t+b_t)!\ .
\end{equation}

\emph{Computing the sign $\mathcal{E}(B)$ for a given table.---} It remains to identify the constant sign of a given contingency table. For this we have the following lemma.
\signAd*
Continue with the edge picture and note that within a row, distinct edges sit in distinct columns (a weight-four column has one node per row, and its two edges partition the rows). Writing $X:=\sum_r\operatorname{inv}(\sigma_r)$, so that $\operatorname{sgn}(\bm\sigma)=(-1)^X$, we may count inversions pairwise over edges,
\begin{equation}
 X=\sum_{\{\epsilon,\epsilon'\}}\big|\pi(\epsilon)\cap\pi(\epsilon')\big|\cdot \big[\text{column order and colour order disagree}\big],
\end{equation}
where $\pi(\epsilon)\subset\{1,2,3,4\}$ is the row pair of the edge $\epsilon$. Two distinct row pairs in a four-element set either coincide, are disjoint, or meet in exactly one element; the first two cases mean $\epsilon,\epsilon'$ are of the \emph{same} class and contribute an even factor. Hence modulo two only cross-class pairs of edges survive, each with weight one:
\begin{equation}
 X\equiv\#\big\{\text{cross-class edge pairs whose column order disagrees with their colour order}\big\}\pmod 2 .
\end{equation}
Now insert the class order as an intermediate reference. For any three total orders, $[\,\text{col}\neq\text{colour}\,]=[\,\text{col}\neq\text{class}\,]\oplus[\,\text{class}\neq\text{colour}\,]$, so
\begin{equation}
 \begin{split}
 X\equiv{}& \underbrace{\#\{\text{cross-class pairs: column order }\neq\text{ class order}\}}_{\text{depends only on }B}\\
 +{}&\underbrace{\#\{\text{cross-class pairs: class order }\neq\text{ colour order}\}}_{\text{colouring dependent}} \pmod 2 .
 \end{split}
\end{equation}
The second term vanishes modulo two: given two colours $a\neq b$ of different classes, all \emph{four} cross pairs of their edges have the same class-versus-colour comparison, so each unordered colour pair contributes $0$ or $4$. The colouring choice no longer features, and we are left with a statistic of $B$ alone, which we term the ``class word'' as defined below. 
\classword*
\noindent
Weight-four columns may be omitted without affecting the parity: assigning one a class contributes two adjacent equal letters, and two equal letters generate an even number of inversions against anything else. For the same reason the parity is independent of which class is assigned to them, and -- since each class occupies an even number $2\alpha_t$ of weight-two columns -- also independent of the arbitrary ordering $e_1<e_2<e_3$ chosen in~\eqref{eq:classes}. Thus
\begin{equation}\label{app:eq:EB}
 \mathcal{E}(B)=(-1)^{\operatorname{inv}(W_B)} .
\end{equation}
For example, a table whose weight-two columns read $(1,2,3,3,2,1)$ from left to right has $6$ inversions and $\mathcal{E}(B)=+1$, whereas $(1,3,2,3,2,1)$ has $7$ and $\mathcal{E}(B)=-1$.

Collecting~\eqref{eq:CB-overlap}, \eqref{eq:s-split}, \eqref{app:eq:NB} and~\eqref{app:eq:EB}, we obtain for the four-replica coefficient of a contingency table:
\begin{equation}\label{app:eq:CB-final}
 C(B) =\frac{(-1)^{\operatorname{inv}(W_B)}\;L! \sum_{b_1+b_2+b_3=q}\binom{q}{b_1,b_2,b_3}\prod_{t=1}^{3}(\alpha_t+b_t)! }{(2N)_L\,(2N+2)_L}.
\end{equation}
Now, we can further simplify this expression to reach~\cref{res:CB-four-replica}, by evaluating the summation in~\cref{app:eq:CB-final}. We have that 
\begin{equation}
    \sum_{b_1+b_2+b_3=q} \binom q{b_1,b_2,b_3} \prod_{t=1}^3(\alpha_t+b_t)! =q!\,\alpha_1!\alpha_2!\alpha_3! \sum_{b_1+b_2+b_3=q} \prod_{t=1}^{3}\binom{\alpha_t+b_t}{b_t}
\end{equation}
where we have expanded the multinomial and used that $\frac{(\alpha_t+b_t)!}{b_t!}=\alpha_t!\binom{\alpha_t+b_t}{b_t}$. The remaining sum can be written as the coefficient of a generating function, giving
\begin{align}
    \sum_{b_1+b_2+b_3=q} \prod_{t=1}^{3}\binom{\alpha_t+b_t}{b_t}&=[z^q]\prod_{t=1}^{3}(1-z)^{-(\alpha_t+1)}\\
&=[z^q](1-z)^{-(\alpha_1+\alpha_2+\alpha_3+3)}\\
&=\binom{\alpha_1+\alpha_2+\alpha_3+q+2}{q}\\
&=\binom{L+2}{q}.
\end{align}
Combining the above, we arrive at the final expression of~\cref{res:CB-four-replica}.

\section{Additional calculations for matchgate-averaged operator complexities}\label{app:operator-calcs}
We now look to wield the previous exact expression for the $k=4$ moment operator to compute operational quantities of the form $\overline{Q} = \Tr[\Phi^{(4)}[\gamma_{\bm \nu}^{\otimes 4}]\mathsf{K}_Q]$, where $\mathsf{K}_Q$ is some boundary operator that enables the exact evaluation of a given quantity $Q$. To this effect, the ensemble average $\overline{Q}$ evaluates to 
\begin{equation}\label{app:eq:boundary-strategy}
 \begin{aligned}
 \overline{Q(O_U)} &=\Tr[\Phi^{(k)}(O^{\otimes k})\mathsf{K}_Q]=\sum_{\bm m}\sum_{B\in\mathcal C_{\bm m}}C(B)\,w_Q(B),\\
 \text{where}&\quad w_Q(B):=\Tr\!\left[\left(\bigotimes_{r=1}^{k}\gamma_{\bm\mu_B^{(r)}}\right)\mathsf{K}_Q\right] .
 \end{aligned}
\end{equation}
The combinatorics of the contingency table structure stems from combining both orthogonal commutant and Majorana spinor properties. We shall see that the boundary operator $\mathsf{K}_Q$ can impose extra conditions on the surviving terms in the summand; ultimately, it is these conditions which need to be counted to compute the ensemble average. In the following, we supply additional details of the proofs of Applications~\ref{res:LOE-closed-form}-\ref{res:otoc-closed} as well as their asymptotic analysis. The proof for the matchgate average OSE result, Application~\ref{res:OSE-closed}, is simple enough to be contained entirely in the Main Text. 

\subsection{Local-operator entanglement}\label{app:LOE-calc}
We start by deriving the local-operator entanglement result, Application~\ref{res:LOE-closed-form}. Rewriting its definition, we have
\begin{equation}\label{app:eq:LOE-def}
  \begin{split}
    \overline{E^{(2)}(O_U)} &= D^{-2}\Tr[\mathcal{E}^{(4)}[O^{\otimes 4}] \, T_{\pi_o}^A T_{\pi_e}^{\bar{A}}],\quad\text{where,}\\
    &T_{\pi_o}^A T_{\pi_e}^{\bar{A}} = \LOEop.
  \end{split}
\end{equation}
That is, we identify $\mathsf{K}_{Q}\equiv T_{\pi_o}^A T_{\pi_e}^{\bar{A}}$ when $Q$ is the LOE purities, $E^{(k)}(O_U)$. First, note that we may write each Majorana string as a Pauli with its bipartition $A\mid \overline{A}$ made explicit, i.e., that $\gamma_{\mu_B^{(r)}} \equiv P_r^{(A)}\otimes P_r^{(\overline{A})}$. We denote by $M$ the largest index of the qubit in the $A$ partition, $X,Z\subseteq[2M]$ the occupied left modes, and $Y,W\subseteq [2N-2M]$ the occupied right modes. 
Evaluating the above expression then enforces, via orthonormality of the Paulis, that 
\begin{equation}\label{app:eq:LOE-pauli-constraints}
\begin{split}
 &P_1^{(A)}=P_2^{(A)},\quad P_3^{(A)}=P_4^{(A)},\\
 \quad &P_2^{(\bar A)}=P_3^{(\bar A)},\quad P_4^{(\bar A)}=P_1^{(\bar A)}.
 \end{split}
\end{equation}
Recalling the definition of the Jordan-Wigner transform (\cref{eq:JW}), we can make the observation that $P_r^{(A)} \propto \gamma_{\mu_r^A}^{(A)}P_A^{|\mu_r^{\overline{A}}|}$
%$P_r^{(A)} = \left(\prod_{j=1}^M c_j^{(r)}\right) \sigma_z^{\otimes (L-M)}$ 
which implies that $\Tr[P_{r}^{(A)}P_{r'}^{(A)}] = \Tr[\prod_{j=1}^{M_r} c_j^{(r)}(\sigma_z^{1:M})^{L_r}\prod_{j=1}^{M_{r'}} c_j^{(r')}(\sigma_z^{1:M})^{L_{r'}}]$. This leads us to several possibilities: First, if $L_r$ and $L_{r'}$ are even, it is equal to $\delta_{\bm{\mu}_r^{(1:M_r)}, \bm{\mu}_{r'}^{(1:M_{r'})}}$. That is, the $Z$ part of the string squares to the identity, and then the substrings on the two replicas are forced to be equal. 
Alternatively, if exactly one of $L_r$ or $L_{r'}$ is even, then we need the two substrings $\mu_{r}^{1:M_r}\cup \mu_{r'}^{1:M_{r'}}$ to contain all modes from 1 to $2M$. This then ensures that $\prod_{j=1}^{M_r} \prod_{j=1}^{M_{r'}} c_j^{(r')}= \sigma_z^{1:2M}$. But this constraint cannot hold, since it requires the first right string equal to the fourth right string, and the second right string equal to the third right string; the first left string equal to the complement of string 2, and similar for 3 and 4. Hence, exactly one of right string 1 and right string 2 is odd, but since $|L_2| + |R_2| = L$, and $|L_1| + |L_2| = 2M$, then $(2M - |L_1|) + |R_2| = L$ which means $|R_1| + |R_2| = (L - |L_1) + (L - 2M + |L_1|) = 2(L-M)$ which, because it's an even number, means you can't have only one of $R_1$ or $R_2$ being odd. We are therefore left with the first case: that the $A$ substring on replicas $1/2$ and $3/4$ must be equal, and that the $\bar{A}$ substring on replicas $1/4$ and $2/3$ must be equal. Graphically, we reiterate that these constraints can be seen in the following way.
\begin{equation}\label{app:eq:LOE-class-one}
  \left(~\LOEClassOne~\right)
\end{equation}

To apply~\cref{res:CB-four-replica} we impose these conditions and read off the column occupancy patterns of the surviving tables. Introduce the two overlaps
\begin{equation}
 u:=|X\cap Z| ,\qquad v:=|Y\cap W| .
\end{equation}
Within $A$ the occupancy of mode $j$ is the indicator vector $(\mathbbm{1}_{j\in X},\mathbbm{1}_{j\in X},\mathbbm{1}_{j\in Z},\mathbbm{1}_{j\in Z})$, and within $\bar A$ it is $(\mathbbm{1}_{j\in W},\mathbbm{1}_{j\in Y},\mathbbm{1}_{j\in Y},\mathbbm{1}_{j\in W})$, so the columns come in only the following types:
\begin{equation}\label{app:eq:LOE-column-types}
 \begin{array}{lcccc}
  \toprule & \multicolumn{2}{c}{\text{weight two}} & \text{weight four} & \text{empty}\\
  \midrule \text{in }A: & \substack{1100\\ L_\ell-u} & \substack{0011\\ L_\ell-u} & \substack{1111\\ u} & \substack{0000\\ 2M-2L_\ell+u}\\[0.6em]
  \text{in }\bar A: & \substack{0110\\ L-L_\ell-v} & \substack{1001\\ L-L_\ell-v} & \substack{1111\\ v} & \substack{0000\\ \cdot}\\
  \bottomrule
 \end{array}.
\end{equation}
Comparing with the classes~\eqref{eq:classes}, every weight-two column inside $A$ occupies rows $\{1,2\}$ or $\{3,4\}$ and hence belongs to class $e_1$, while every weight-two column inside $\bar A$ occupies rows $\{2,3\}$ or $\{1,4\}$ and belongs to class $e_3$, while class $e_2$ never occurs. 
\begin{equation}\label{app:eq:LOE-alphas}
 \begin{aligned}
 \alpha_1&=L_\ell-u, & \alpha_2&=0,\\
 \alpha_3&=L-L_\ell-v, & q&=u+v,
 \end{aligned}
\end{equation}
which indeed satisfies $\alpha_1+\alpha_2+\alpha_3+q=L$. The sign is trivial for every such table. Since all of $A$ precedes all of $\bar A$ in the mode ordering, the class word of~\cref{def:class-word} is
\begin{equation}
 W_B=\underbrace{1\,1\cdots1}_{2\alpha_1} \underbrace{3\,3\cdots3}_{2\alpha_3},
\end{equation}
implying that $\rm{inv}(W_B)=0$ for all $B$. Equivalently, this is the statement that the LOE boundary forces all tables to already be in an in-column pairing, meaning that the quantity is a sum of only positive terms. Setting $\alpha_2=0$ in~\eqref{eq:NB} gives us
\begin{equation}
 \begin{aligned}
 \frac{\mathcal N(B)}{L!} &={\sum_{b_1+b_2+b_3=q}} \frac{q!\,(\alpha_1+b_1)!(\alpha_3+b_3)!}{b_1!\,b_3!} =q!\,\alpha_1!\,\alpha_3! \sum_{b_1+b_2+b_3=q} \binom{\alpha_1+b_1}{\alpha_1} \binom{\alpha_3+b_3}{\alpha_3}=q!\,\alpha_1!\,\alpha_3!\binom{L+2}{q}.
 \end{aligned}
\end{equation}
The third equality follows since $\alpha_1+\alpha_3+q=L$, meaning that the binomial sum is simply $\binom{L+2}{q}$. Hence, for every table,
\begin{equation}\label{app:eq:LOE-sB}
 s(B)=L!\ (u+v)!\ (L_\ell-u)!\ (L-L_\ell-v)!\ \binom{L+2}{u+v} .
\end{equation}

It remains to count how many surviving tables realize a given $(L_\ell,u,v)$. On the $A$ side we must choose $X,Z\subseteq\{1,\ldots,2M\}$ with $|X|=|Z|=L_\ell$ and $|X\cap Z|=u$, which is a four-way multinomial choice into $X\cap Z$, $X\setminus Z$, $Z\setminus X$ and the rest; likewise on the $\bar A$ side. Thus
\begin{equation}\label{app:eq:LOE-counts}
 \begin{aligned}
 \#_{\rm left} &=\frac{(2M)!}{u!\,[(L_\ell-u)!]^2\,(2M-2L_\ell+u)!},\\[0.4em] 
 \#_{\rm right} &=\frac{(2N-2M)!}{v!\,[(L-L_\ell-v)!]^2} \frac{1}{[2N-2M-2(L-L_\ell)+v]!} .
 \end{aligned}
\end{equation}
The squared factorials in the denominators cancel against two of the factorials in~\eqref{app:eq:LOE-sB}, leaving one power of each. Dividing by the content polynomial~\eqref{eq:Zlambda-star} we obtain a form of the main result of this section:
\begin{equation}
    \overline{E^{(2)}_A(O_U)} =\frac{L!\,(2M)!\,\big(2(N-M)\big)!}{\displaystyle(2N)_L\,(2N+2)_L} \sum_{L_\ell=(L-2\bar M)_+}^{\min\{L,2M\}}\sum_{\substack{u=(2L_\ell-2M)_+, \\ v=(2(L-L_\ell)-2\bar M)_+ }}^{L_\ell,L-L_\ell} \!\frac{(u+v)!\,\dbinom{L+2}{u+v}}{\Delta(L_\ell,u,v)},\label{eq:LOE-closed-form-pre}
\end{equation}
where we write $(x)_+:= \max\{0,x\}$ and where 
\begin{equation}
    \Delta(L_\ell,u,v):={}u!\,v!\,(L_\ell-u)! (L-L_\ell-v)! \label{eq:delta}(2M-2L_\ell+u)![2(N-M)-2(L-L_\ell)+v]! .
\end{equation}
The LOE boundary has now produced both required simplifications: it excludes class $e_2$, fixes $\mathcal E(B)=+1$, and reduces the coefficient and table multiplicity to functions of $(L_\ell,u,v)$. Summing those positive contributions gives the exact finite-size result.

It remains to simplify the summations in~\cref{eq:LOE-closed-form-pre} to arrive at Application~\ref{res:LOE-closed-form}. First, introduce the variables
\begin{equation}
w=L_\ell-u,\qquad x=L-L_\ell-v.
\end{equation}
Then we can rewrite: $ L_\ell=u+w$, $L-L_\ell=v+x$, and $u+v+w+x=L$. It follows immediately that the last two factorial arguments of~\cref{eq:delta} become
\begin{align}
2M-2L_\ell+u &=2M-u-2w,\\
2\bar M-2(L-L_\ell)+v &=2\bar M-v-2x.
\end{align}
Thus, the original summation bounds are equivalent to
\begin{equation}
u,v,w,x\geq 0, \qquad u+v+w+x=L, \qquad 2M-u-2w\geq0, \qquad 2\bar M-v-2x\geq0.
\end{equation}
Moreover,
\begin{equation}
(u+v)!\binom{L+2}{u+v} = \frac{(L+2)!}{(w+x+2)!}.
\end{equation}
Therefore,
\begin{equation}
\overline{E_A^{(2)}(O_U)} ={} \frac{L!(2M)!(2\bar M)!(L+2)!}{(2N)_{ L}(2N+2)_{ L}}\sum_{\substack{u,v,w,x\geq0\\u+v+w+x=L\\2M-u-2w\geq0\\2\bar M-v-2x\geq0}}\frac{1}{(w+x+2)!\,u!\,v!\,w!\,x!\,(2M-u-2w)!\,(2\bar M-v-2x)!}.
\label{eq:LOE-uvwx}
\end{equation}
Finally, we can resolve the sum over $u$ and $v$ by applying Vandermonde's identity:
\begin{equation}
\sum_{u+v=L-w-x}\frac{1}{u!\,v!\,(2M-2w-u)!\,(2\bar{M}-2x-v)!}=\frac{1}{(2M-2w)!\,(2\bar{M}-2x)!}\binom{2M-2w+2\bar{M}-2x}{L-w-x}.
\end{equation}
Applying the definition of the falling factorial, $(y)_n:=y!/(y-n)!$, we arrive at the final expression of~\cref{eq:LOE-closed-form}. \par

\emph{Asymptotics.---} The remainder of this Section is devoted to asymptotic analysis of~\cref{eq:LOE-closed-form}, proving the following lemma.
\begin{lemma}[{Asymptotic matchgate-averaged LOE}]\label{lem:LOE-asymptotics}
Let $N\to\infty$. Then, the LOE, evaluated by~\cref{eq:LOE-closed-form} yields:\par
\noindent
\underline{$L=\mathcal{O}(1)$ strings.} 
We first evaluate~\cref{eq:LOE-closed-form} for few-particle operators, i.e., for $L=\mathcal{O}(1)$. For the half-chain, we have that
\begin{equation}\label{app:eq:LOE-dilute-entropy}
\overline{\mathrm{LOE}}_{A}^{(2)} ={}L-\frac{L(L+1)}{2N \ln(2)}+\mathcal O(N^{-2}).
\end{equation}
with the formula for arbitrary $M$ reported in App.~\ref{app:LOE-calc}. \par
\noindent
\underline{Extensive strings.} Take the half-chain cut $M=N/2$ and let $L=2N\ell$ with $\ell\in(0,1)$ fixed. Then
\begin{equation}
\overline{\mathrm{LOE}}_{A}^{(2)}=2N\,\sigma(\ell)+\mathcal{O}(1),
\end{equation}
with
\begin{equation}
\sigma(\ell) = \log\!\left[ \frac{ \bigl(\ell-p_\star\bigr)^\ell \bigl(1-\ell-p_\star\bigr)^{1-\ell}}{ \ell^{2\ell}(1-\ell)^{2(1-\ell)}}\right],
\end{equation}
where $p_\star=p_\star(\ell)=\big(\sqrt{1+4\ell(1-\ell)}-1\big)/2$. In particular, $\sigma$ is maximised at $\ell=\tfrac12$, where $\sigma(\tfrac12)=\log(4-2\sqrt2)$.
\end{lemma}

First, we fix $L=\mathcal{O}(1)$ and consider the half-chain bipartition, $M=N/2$. Rewriting~\cref{eq:LOE-closed-form} under these conditions, we have
\begin{align}
    \overline{E_A^{(2)}(O_U)}&=\frac{L!(L+2)!}{(2N)_{ L}(2N+2)_{L}}  \sum_{\substack{0\leq w,x\leq N/2\\w+x\leq\min\{L,2N-L\}}}\frac{(N)_{{2w}}(N)_{{2x}}}{w!x!(w+x+2)!}\binom{2N-2w-2x}{L-w-x}\\
    &=: \sum_{s=0}^L T_{s} , 
\end{align}
where $T_s$ is the contribution from $x+w=s$. We notice that for a given $L$, each summand scales as 
\begin{equation}
    T_{s}=\mathcal{O}\left(\frac{N^{2w}N^{2x}N^{L-w-x}}{N^{2L}}\right) =\mathcal{O}(N^{w+x-L}).
\end{equation}
So the leading order in $N \to \infty$ comes solely from the terms $T_L$. The subleading terms are when $s=w+x=L-1$, together with the subleading expansion of $T_L$. Using
\begin{equation}
    (N)_r=N^r\left[1-\frac{r(r-1)}{2N}+\mathcal{O}(N^{-2})\right],
\end{equation}
we have that 
\begin{equation}
    T_L=2^{-L}\left[1-\frac{L(L+3)}{2N}+\mathcal{O}(N^{-2})\right],
\end{equation}
and 
\begin{equation}
    T_{L-1}=2^{-L}\left[\frac{L(L+2)}{N}+\mathcal{O}(N^{-2})\right].
\end{equation}
Combining these, we have 
\begin{equation}
    \overline{E_A^{(2)}}=2^{-L}\left[1+\frac{L(L+1)}{2N}+\mathcal{O}(N^{-2})\right],
\end{equation}
which leads directly to the annealed average $2$-R\'enyi LOE in~\cref{eq:LOE-dilute-entropy}. We can directly extend this to arbitrary-size bipartitions $M$. For an arbitrary bipartition, define
$$
f=\frac{M}{N},\qquad g=f(1-f),\qquad \Sigma=f^2+(1-f)^2.
$$
The same expansion applies, with $(N)_{2w}(N)_{2x}$ replaced by $(2M)_{2w}(2N-2M)_{2x}$. Summing over $w+x=s$ using the binomial theorem gives
$$
\begin{aligned}
T_L &=\Sigma^L\left[1+\frac{L}{2N}\left(L-3-\frac{1}{\Sigma}-\frac{2(L-1)(1-3g)}{\Sigma^2}\right)+O(N^{-2})\right],\\
T_{L-1}&=\Sigma^L\left[\frac{L(L+2)}{2N\Sigma}+O(N^{-2})\right].
\end{aligned}
$$
All remaining terms contribute $O(N^{-2})$. Using
$\Sigma=1-2g$, their sum is hence 
$$
\overline{E_A^{(2)}}=\Sigma^L\left[1+\frac{2Lg[1+(L-3)g]}{N\Sigma^2}+O(N^{-2})\right].
$$
Then the annealed entropy in bits is
$$
\overline{\mathrm{LOE}}_A^{(2)}=-L\log_2\Sigma-\frac{2Lg[1+(L-3)g]}{N\Sigma^2\ln2}+O(N^{-2}).
$$
Setting $f=\tfrac12$ recovers the half-chain result.

Now, we again take the half-chain bipartition but choose an extensive string length, $\ell=L/(2N)=\in (0,1)$. At fixed $w+x$, log-convexity of the factorials means that the summand is maximized at $w=x$. Writing $p=(w+x)/(2N)$, Stirling's
formula gives (including the prefactor) a leading scaling of $\Theta(N^{-1})e^{2N\phi(p)}$, where
\begin{equation}
\phi(p)=2\ell\ln\ell+2(1-\ell)\ln(1-\ell) -2p\ln p-p\ln2-(\ell-p)\ln(\ell-p) -(1-\ell-p)\ln(1-\ell-p).
\end{equation}
Since
$$
\phi'(p)=\ln\frac{(\ell-p)(1-\ell-p)}{2p^2}, \qquad \phi''(p)<0,
$$
the unique maximum occurs at
$$
p_\star=\frac{\sqrt{1+4\ell(1-\ell)}-1}{2}.
$$
The peak is nondegenerate in both summation directions, with width $\Theta(\sqrt N)$ in each. Gaussian summation therefore cancels the $N^{-1}$ prefactor, yielding
$$
\overline{E_A^{(2)}}=\exp\!\left[2N\phi(p_\star)+\mathcal{O}(1)\right].
$$
Using $(\ell-p_\star)(1-\ell-p_\star)=2p_\star^2$ gives $-\phi(p_\star)/\ln(2)=:\sigma(\ell)$, and hence (recalling that we use base $2$ logarithms throughout)
$$
\overline{\mathrm{LOE}_{A}}^{(2)} :=
-\log\overline{E_A^{(2)}}=2N\sigma(\ell)+\mathcal{O}(1),
$$
as reported in~\cref{eq:LOE-page}

\subsection{Operator stabilizer entropy}\label{app:OSE-calc}
The closed-form average of~\cref{res:OSE-closed} is derived entirely in the main text. Here, we give the proof of the general bound of~\cref{prop:ose-bound}, which we restate for convenience.
\oseBound*
\begin{proof}
By~\cref{eq:det_formula}, and since $U^\dagger\gamma_{\bm{\nu}}U$ is again of the form~\cref{eq:det_formula} with $R$ replaced by $R^{\T}$,
\begin{equation*}
O_U=\sum_{|\bm\mu|=L}\det\!\big(R^{\T}_{\bm{\nu},\bm\mu}\big)\,\gamma_{\bm\mu},
\end{equation*}
where the sum runs over the $\binom{2N}{L}$ canonically ordered length-$L$ supports. Each $\gamma_{\bm\mu}$ is a Pauli string up to a phase, and distinct supports give distinct Pauli strings. The Pauli distribution $p_P:=|a_P|^2$ of $O_U$ is therefore supported on a set of at most $n:=\binom{2N}{L}$ Pauli strings. For $k>1$, the map $x\mapsto x^k$ is convex, so by Jensen's inequality
\begin{equation*}
M^{(k)}(O_U)=\sum_{P}p_P^k=n\cdot\frac1n\sum_{P\in\operatorname{supp}}p_P^k\;\geq\; n\Big(\frac1n\sum_{P}p_P\Big)^k=n^{1-k},
\end{equation*}
with equality if and only if $p$ is uniform on $n$ strings. Equivalently, $\mathrm{OSE}^{(k)}(O_U)\leq\log n$ for $k>1$. The same bound holds for all $k\geq0$, because the R\'enyi entropy of a distribution with support size at most $n$ never exceeds $\log n$.
\end{proof}

\subsection{Out-of-time-ordered Correlator}\label{app:OTOC-calc}
We again apply the strategy for matchgate averages adopted above. We first rewrite~\cref{eq:otocdef} in $k$-replicas of Hilbert space,
\begin{equation}\label{app:eq:otoc_op}
    \begin{split}
        {\operatorname{OTOC}^{(k)}} &=  D^{-1} \Tr[O_U^{\otimes k}X^{\otimes k}T_{\sigma}],\quad\text{where} \\[1em]
        T_{\sigma} &:=  \OTOCop
    \end{split}
\end{equation}
is the cyclic permutation unitary between replicas, $\sigma=(12\dots k) $. We take the operators $O$ and $X$ to be Majorana strings of length $L$ and $L'$, respectively. Considering the average $2k$-OTOC over the matchgate group, we have that 
\begin{equation}
  \overline{\operatorname{OTOC}^{(k)}} := D^{-1} \Tr[\Phi^{(k)}[O^{\otimes k}]X^{\otimes k}T_{\sigma}].
\end{equation}
The moment operator $\Phi^{(k)}[O^{\otimes k}]$ is given by~\cref{eq:mainTool}. After undoing the cyclic permutation, conjugation by the probe supplies the Pauli phase $XPX=(-1)^{\langle X,P\rangle}P$. The remaining trace is nonzero exactly when the product of the $k$ replica strings is proportional to the identity. Diagrammatically, the contraction is
$$
\sum_{\vec{P}\in \operatorname{Comm}(U_f(N))}\OTOCk[0.8].
$$
We note as an interesting structural point that, in contrast to the LOE and the OSE, every element of the free-fermionic commutant (i.e., every admissible contingency table) contributes to this expression.

Let $O=\gamma_{\bm{\nu}}$ have length $L$ and let the probe be $X=\gamma_{\bm{\nu}'}$ of length $L'$, with supports $S$ and $S'$. Inserting~\cref{eq:coefficient-table-final} and commuting the four copies of $X$ to the right using
$X\gamma_{\bm\mu}X=(-1)^{L'L-|S'\cap\bm\mu|}\gamma_{\bm\mu}$ gives
\begin{equation}\label{eq:OTOC-table-sum}
 \overline{\operatorname{OTOC}}^{(4)} =\sum_{B\in\mathcal{C}}C(B)\;\varepsilon_{\rm tr}(B)\;
 (-1)^{|S'\cap S_2|+|S'\cap S_4|},
\end{equation}
where $S_r$ is row $r$ of the table and $\varepsilon_{\rm tr}(B)=D^{-1}\Tr[\gamma_{S_1}\gamma_{S_2}\gamma_{S_3}\gamma_{S_4}]$. The factors of $(-1)^{L'L}$ cancel in pairs. 
\OTOCphase*
\begin{proof}
  Because the four supports of an admissible table satisfy $S_1\triangle S_2\triangle S_3\triangle S_4=\emptyset$, the product is always proportional to the identity, so $\varepsilon_{\rm tr}=\pm1$ and every admissible table contributes. Explicitly, sorting the concatenated string, we have $\varepsilon_{\rm tr}(B) = (-1)^{I(B)}$, where $I(B):=\sum_{r<r'}I(S_r,S_{r'})$ as in prior sections. Then the signed part of the sum which is independent of the probe is 
\begin{equation}
  \begin{split}
  \mathcal{E}(B)\varepsilon_{\rm tr}(B) &= \prod_{q=2}^{2N}\prod_{p=1}^q(-1)^{[e(\tau_p) > e(\tau_q)] + \sum_{r < r'}[\tau_p^{(r)} = \tau_q^{(r')} = 1]} \\
  & \coloneqq \prod_{p < q} h(\tau_p,\tau_q),
  \end{split}
\end{equation}
with $\tau_p \in \{12,34,13,24,14,23,Q\}$ the column type of mode $p$, and $e(\tau_p) \in \{e_1, e_2, e_3\}$ its corresponding pairing type. Since both the first and the second term are antisymmetric under exchanging of $\tau_p$ and $\tau_q$, we have that $h(\tau_p,\tau_q)$ is overall symmetric. Further, since all permutations can be written as a product of transpositions, and $h(\tau_p,\tau_q)$ is invariant under those transpositions, then $\mathcal{E}(B)\varepsilon_{\rm tr}(B)$ is invariant with respect to the \emph{ordering} of the multi-set of column data in $B$. Hence, without loss of generality, we can place each $B$ into canonical ordering: 
\begin{equation}
  \begin{split}
  &(12)\cdots (12); (34)\cdots (34); \quad(2\alpha_1)\\
  &(13)\cdots (13);(24)\cdots (24); \quad(2\alpha_2)\\
  &(14)\cdots (14); (23)\cdots (23); \quad(2\alpha_3)\\
  &Q\cdots Q;\quad  (q)
  \end{split}
\end{equation}
for which $\rm{inv}(W_B) = 0$. Thus the sign part \emph{only} depends on $I(B)$, which we can exhaustively calculate for this class word. In particular, $I(S_r,S_{r'}),~r<r'$ sum up all elements of $S_r$ which occur \emph{later} than $S_{r'}$. For example, $I(S_1, S_2)$ has a $\binom{\alpha_1}{2} + \binom{q}{2}$ contribution from columns that can be made type $(12)$; plus $\alpha_1\alpha_2 + \alpha_1\alpha_3 + \alpha_2\alpha_3$ from $(12);(13)$; $(12);(14)$; $(23)(14)$ columns respectively; plus $(\alpha_1 + \alpha_2 + \alpha_3)q$ adding up the weight-2, row-2 contributions with respect to the row-1 in weight-4 columns. This calculation may seem opaque at first, but for all six $I(S_r,S_{r'})$ terms we have the $\binom{q}{2}$; $(\alpha_1\alpha_2 + \alpha_1\alpha_3 + \alpha_2\alpha_4)$; and ($\alpha_1+\alpha_2+\alpha_3)q$ terms, hence they are equal to zero modulo two. The remaining exponent is 
\begin{equation}
  \begin{split}
    &\underbrace{\binom{\alpha_1}{2}}_{I(S_1,S_2)} + \underbrace{\binom{\alpha_1}{2}}_{I(S_3,S_4)} + \underbrace{\binom{\alpha_2}{2}}_{I(S_1,S_3)} + \underbrace{\binom{\alpha_2}{2}}_{I(S_2,S_4)} + \underbrace{\binom{\alpha_3}{2}}_{I(S_1,S_4)} + \underbrace{\binom{\alpha_3}{2}}_{I(S_2,S_3)} \\
    &+ \underbrace{\alpha_2^2}_{I(S_2,S_3)} + \underbrace{\alpha_3^2}_{I(S_2,S_4)} + \underbrace{\alpha_3^2}_{I(S_3,S_4)}. 
  \end{split}
\end{equation}
From this we glean that $\alpha_2^2\equiv \alpha_2$ is the only term occurring an odd number of times, and hence $\mathcal{E}(B)\varepsilon_{\rm tr}(B) = (-1)^{\alpha_2}$. For the probe phase term, note that 
\begin{equation}
  (-1)^{|S'\cap S_2| + |S'\cap S_4|} = (-1)^{n_1(S') + n_3(S')},
\end{equation}
where $n_t$ counts the class-$t$ weight two columns of $B$ whose mode lies in $S'$. This yields the stated result. 
\end{proof}

We have thus reduced the summand of~\cref{eq:OTOC-table-sum} to class data alone. Writing $k_1 := n_1(S')$ and $k_3 := n_3(S')$, every table with the same $(\alpha_1,\alpha_2,\alpha_3,q,k_1,k_3)$ values contributes the sum identically. What remains is simply to count the multiplicity of each of the possible values that may be assumed by this tuple, which we do explicitly in~\cref{app:OTOC-calc}. For every fixed tuple, this is as follows. First, consider the $k_i$. Assign the $2N$ modes to $k_1$ class-1 modes from $S'$ and $2\alpha_1 - k_1$ from $S'^c$, for which there are $\binom{L'}{k_1}\binom{2N-L'}{2\alpha_1-k_1}$ many ways; split the $2\alpha_1$ modes into which are type (12) versus (34). Then, from the remainder of $S'$ and $S'^c$, do the same for class 3: $\binom{L'-k_1}{k_3}\binom{2N-L'-2\alpha_1+k_1}{2\alpha_3-k_3}\binom{2\alpha_3}{\alpha_3}$. From the remaining pool, place class 2 $\binom{2N-2p}{2\alpha_2}\binom{2\alpha_2}{\alpha_2}$ ($p=\alpha_1+\alpha_3$), and the $q$ weight-4 columns $\binom{2N-2(p+\alpha_3)}{q}$. It can be shown, using the fact that $\sum_k (-1)^k\binom{A}{k}\binom{B}{m-k} = [u^m](1-u)^A(1+u)^B$ that the resulting alternating sum can be reduced in the following way
\begin{equation}
\sum_{k_1,k_3} (-1)^{k_1+k_3} \binom{L'}{k_1} \binom{2N-L'}{2\alpha_1-k_1} \binom{L'-k_1}{k_3} \binom{2N-L'-2\alpha_1+k_1}{2\alpha_3-k_3}= \binom{2p}{2\alpha_1} K_{2p}(L;2N).
\end{equation}
with $K_m(x;n)$ the standard binary Krawtchouk polynomial\footnote{Defined using the convention $K_m(x;n) := \sum_{j=0}^{m} (-1)^j \binom{x}{j} \binom{n-x}{m-j}$.
}. Grouping then by $p=\alpha_1+\alpha_3$, using that $\sum_p (2p)!/(\alpha_1!\alpha_3!) = (2p)!2^p/p!$; grouping also by $\alpha_2=n-p$ for fixed $n:=L-q$; and lastly summing over $q$ gives the following result:
\begin{equation}
    \overline{\mathrm{OTOC}}^{(4)}=\frac{L!}{(2N)_L(2N+2)_L}\sum_{\substack{p,q\geq 0\\ p+q\leq L}}(-1)^{L-p-q}\binom{L+2}{q} \frac{2^p(2p)!(2N-2p)!}{p!(L-p-q)!(2N-2L+q)!}K_{2p}(L';2N).
\end{equation}
To arrive at the final form of~\cref{res:otoc-closed}, we are left to simplify this expression. In the above double sum, we fix $p$ and sum over $q$. Alternating Vandermonde convolution gives
\begin{equation}
\sum_{q=0}^{L-p}\frac{(-1)^{L-p-q}\binom{L+2}{q}}{(L-p-q)!(2N-2L+q)!}=\frac{(-1)^{L-p}}{(2N-L-p)!}K_{L-p}(L+2;2N-p+2).
\end{equation}
The factorial prefactor also simplifies to
\begin{equation}
\frac{L!(2N)!}{(2N)_L(2N+2)_L\,p!(2N-L-p)!}={\binom{2N-L}{p}}  \binom{2N+2}{L}^{-1}.
\end{equation}
Next, Krawtchouk duality and the generating function $K_r(a;b)=[t^r](1-t)^a(1+t)^{b-a}$ imply
\begin{equation}
\frac{K_{2p}(L';2N)}{\binom{2N}{2p}}=\sum_{j=0}^{\min\{p,L',2N-L'\}}(-4)^j\binom{p}{j}\frac{(L')_j(2N-L')_j}{(2N)_{2j}}.
\end{equation}
Substitute this expansion and interchange the $p$- and $j$-sums. The remaining $p$-sum is evaluated by the binomial identity
\begin{equation}
\sum_p\binom{2N-L}{p}\binom{p}{j}a^p b^{2N-L-p}=\binom{2N-L}{j}a^j(a+b)^{2N-L-j},
\end{equation}
applied with $a=-2t$ and $b=1+t$. Taking the required coefficient of $t^L$ yields $2^j\binom{2N-L}{j}\binom{2N+2-j}{L-j}$. Finally,
\begin{equation}
\frac{\binom{2N+2-j}{L-j}}{\binom{2N+2}{L}}=\frac{(L)_j}{(2N+2)_j}.
\end{equation}
Combining these identities, we arrive at~\cref{eq:otoc-closed}.

\emph{Asymptotics.---} We now derive the $N\to \infty$ expansions of~\cref{eq:otoc-closed} for different string lengths $L,L'$, proving the following lemma.
\begin{lemma}[{Asymptotic matchgate-averaged OTOC}]\label{lem:OTOC-asymptotics}
Let $N\to\infty$. Then,~\cref{eq:otoc-closed} gives for the average-case OTOC:\par
\noindent
\underline{$L=\mathcal{O}(1)$ initial operator and probe:} For fixed $L,L'=\mathcal{O}(1)$,
\begin{equation}
\begin{split}
\overline{\mathrm{OTOC}}^{(4)}={}&1-\frac{4LL'}{N}+\frac{2LL'(L+L'+1)}{N^2} +\frac{8L(L-1)L'(L'-1)}{N^2}+\mathcal{O}_{L,L'}(N^{-3}).
\end{split}
\end{equation}

\noindent
\underline{$L=\mathcal{O}(1)$ initial operator, extensive probe:} For fixed $L=\mathcal{O}(1)$ and $L'=N$,
\begin{equation}
\overline{\mathrm{OTOC}}^{(4)}=(-1)^L\left[1-\frac{L(L+1)}{N}+\mathcal{O}_L(N^{-2})\right].
\end{equation}

\noindent
\underline{Extensive operator and probe:} For $L=L'=N$,
\begin{equation}
\begin{split}
\overline{\mathrm{OTOC}}^{(4)}={}&\frac{1}{2^{N+1}}\left[\sin\!\left(\frac{\pi N}{2}\right)+\cos\!\left(\frac{\pi N}{2}\right)+\frac{(-1)^N}{\sqrt{2}}\right] +\mathcal{O}\!\left(2^{-N}N^{-1}\right).
\end{split}
\end{equation}
\end{lemma}
Let $T_j$ be the $j$-th term of the exact OTOC sum in~\cref{eq:otoc-closed}. First, for fixed $L,L'$, its first terms are
\begin{equation}
T_0=1,\qquad T_1=-\frac{4LL'}{N} +\frac{2LL'(L+L'+1)}{N^2} +O_{L,L'}(N^{-3}),
\end{equation}
\begin{equation}
T_2=\frac{8L(L-1)L'(L'-1)}{N^2} +O_{L,L'}(N^{-3}).
\end{equation}
Since $T_j=O_{L,L'}(N^{-j})$, the remaining terms are $O_{L,L'}(N^{-3})$. Hence
\begin{equation}
\overline{\mathrm{OTOC}}^{(4)} = 1-\frac{4LL'}{N} +\frac{2LL'(L+L'+1)+8L(L-1)L'(L'-1)}{N^2} +O_{L,L'}(N^{-3}).
\end{equation}
For $L'=N$ with fixed $L$, expansion of each of the finitely many terms gives
\begin{equation}
T_j=(-2)^j\binom{L}{j} \left[1-\frac{j(L+1)}{2N}+O_L(N^{-2})\right].
\end{equation}
Summing with the binomial theorem and its derivative yields
\begin{equation}
\overline{\mathrm{OTOC}}^{(4)} = (-1)^L\left[1-\frac{L(L+1)}{N} +O_L(N^{-2})\right].
\end{equation}

For $L=L'=N$, a termwise expansion is invalid because terms with $j=\mathcal{O}(N)$ contribute. Define
\begin{equation}
f_N(z)=\frac{2^N}{\pi}\binom{2N}{N}^{-1} \int_0^\pi\bigl(\sqrt{1-z}-\cos\theta\bigr)^N\,d\theta, \qquad D=z\frac{d}{dz}.
\end{equation}
Integration by parts of $\sin\theta(\sqrt{1-z}-\cos\theta)^{N-1}$ gives
\begin{equation}
z(1-z)f_N'' +\bigl(\tfrac12-N+(N-1)z\bigr)f_N' -\frac{N^2}{4}f_N=0.
\end{equation}
Moreover, $f_N(0)=1$. Differentiating $f_N^2$ twice and substituting into the above formula, we have that 
\begin{equation}
\bigl[D(D-N-\tfrac12)(D-2N-1)-z(D-N)^3\bigr]f_N^2=0.
\end{equation}
Because $f_N(-\sqrt{1-z}) = (-1)^N f_N(\sqrt{1-z})$, $f_N^2$ is a polynomial of degree $N$ in $z$. Writing $f_N^2(z)=\sum_{j=0}^N h_jz^j$ and comparing coefficients gives
\begin{equation}
j(j-N-\tfrac12)(j-2N-1)h_j=(j-N-1)^3h_{j-1}, \qquad h_0=1,
\end{equation}
and therefore
\begin{equation}
f_N^2(z)=\sum_{j=0}^N \frac{(-1)^j(N)_j^3}{j!(N-\tfrac12)_j(2N)_j}\,z^j.
\end{equation}
Using
$(2N)_{2j}=4^j(N)_j(N-\tfrac12)_j$ and
$$
\frac{(2N)_j}{(2N+2)_j}
=\frac{(2N+2-j)(2N+1-j)}{(2N+2)(2N+1)},
$$
the original finite sum becomes exactly
\begin{equation}
\overline{\mathrm{OTOC}}^{(4)}=\left.\frac{(2N+2-D)(2N+1-D)} {(2N+2)(2N+1)}\,f_N^2(z)\right|_{z=2}.
\end{equation}
We must lastly estimate the integral. Put $s=\sqrt{1-z}$, choosing $s=i$ at $z=2$. Gaussian expansions at $\theta=0$ and $\theta=\pi$ give the respective endpoint contributions
\begin{equation}
a_-(z)= \frac{2^N(s-1)^N\sqrt{1-s}}{\binom{2N}{N}\sqrt{2\pi N}},\qquad a_+(z)= \frac{2^N(s+1)^N\sqrt{1+s}}{\binom{2N}{N}\sqrt{2\pi N}}.
\end{equation}
Thus $f_N(z)=a_-(z)[1+\mathcal{O}(N^{-1})]+a_+(z)[1+\mathcal{O}(N^{-1})]$ near $z=2$. Stirling's formula gives
\begin{equation}
a_-(2)=2^{-N/2-1/4}e^{i(3\pi N/4-\pi/8)}[1+\mathcal{O}(N^{-1})],\qquad a_+(2)=2^{-N/2-1/4}e^{i(\pi N/4+\pi/8)}[1+\mathcal{O}(N^{-1})].
\end{equation}
Since $D\log(s\pm1)|_{z=2}=i/(i\pm1)$, the differential operator multiplies $a_-^2$, $a_+^2$, and $a_-a_+$ at leading order by $i/2$, $-i/2$, and $1/4$, respectively. Therefore,
\begin{equation}
\overline{\mathrm{OTOC}}^{(4)} = \frac{i}{2}a_-(2)^2-\frac{i}{2}a_+(2)^2 +\frac12a_-(2)a_+(2) +O(2^{-N}N^{-1}),
\end{equation}
which simplifies to~\cref{eq:otoc-closed}.

\section{Maximal LOE under Gaussian evolution}\label{app:max_ent}

In Application~\ref{res:LOE-closed-form}, we have  shown that the Gaussian-averaged LOE lies below the unitary operator Page curve, and far below the upper bound supplied by~\cref{eq:loe_bound}. In this section, we show that individual (atypical) Gaussian evolutions can in fact attain a maximal operator entanglement.
To demonstrate this, take $N$ even and let $A$ contain the first $N/2$ qubits. Define the local parity operator on the half-chain $A$, $P_A:=\prod_{x=1}^{N/2}\sigma_z^{(x)}$, and choose
\begin{equation}
\begin{aligned}
O&=P_A\otimes\id_{\bar A}=(-i)^{N/2} c_1 \dots c_N ,\\
U&=\exp\!\left(\frac{\pi}{8}\sum_{j=1}^{N}c_jc_{N+j}\right).
\end{aligned}
\end{equation}
We see that $O$ has Majorana length $L=N$ and $U$ is Gaussian. Since the quadratic generators commute,
$U^\dagger c_jU=(c_j+c_{N+j})/\sqrt{2}$, giving
\begin{equation}
O_U=U^\dagger OU=\frac{(-i)^{N/2}}{2^{N/2}}\prod_{j=1}^{N}(c_j+c_{N+j}),
\label{eq:maximal-loe-operator}
\end{equation}
with the product ordered by increasing $j$.

To obtain the spatial Schmidt decomposition, denote by $c_j^{(A)}$ and $c_j^{(\bar A)}$ the Majoranas defined by applying the Jordan--Wigner construction separately within each region, where 
\begin{equation}
c_j=c_j^{(A)}\otimes\id,\qquad c_{N+j}=P_A\otimes c_j^{(\bar A)}.
\end{equation}
Let $\gamma_{\bm{\nu}}^{(A)}$ and $\gamma_{\bm{\nu}}^{(\bar A)}$ denote the corresponding local Majorana strings generated from $c_j^{(A)}$ and $c_j^{(\bar A)}$, respectively. Expanding~\cref{eq:maximal-loe-operator} and collecting the factors on each side yields
\begin{equation}
O_U=2^{-N/2}\sum_{\bm{\nu}\subseteq[N]}\eta_{\bm{\nu}} \left(\gamma_{\bm{\nu}}^{(A)}P_A^{N-|\bm{\nu}|}\right) \otimes\gamma_{\bm{\nu}^c}^{(\bar A)},
\label{eq:maximal-loe-schmidt}
\end{equation}
where $[N]=\{1,\ldots,N\}$, $\bm{\nu}^c=[N]\setminus\bm{\nu}$, and $|\eta_{\bm{\nu}}|=1$. The phases account for Majorana reordering and the overall phase of $O$.

The right factors in~\cref{eq:maximal-loe-schmidt} run through the complete orthonormal local Majorana-string basis. The left factors do likewise, up to phases: for even $|\bm{\nu}|$, the parity factor is the identity, whereas for odd
$|\bm{\nu}|$, multiplication by $P_A$ maps $\gamma_{\bm{\nu}}^{(A)}$ to $\gamma_{\bm{\nu}^c}^{(A)}$ up to a phase. Since $N$ is even, complementation permutes the odd subsets among themselves. Hence both sets of local factors are Hilbert--Schmidt orthonormal, and all $2^N$ operator Schmidt coefficients equal $2^{-N/2}$. We therefore have that the reduced state of the vectorised operator is maximally mixed and so 
\begin{equation}
E_A^{(k)}(O_U)=2^{-(k-1)N},
\end{equation}
saturating the inequality of~\cref{eq:loe_bound}. The suppression of typical Gaussian LOE in the result of Application~\ref{res:LOE-closed-form} therefore does not reflect any inherent obstruction to maximal entanglement on the Gaussian orbit, and~\cref{eq:loe_bound} is tight.

\section{Bound on LOE for doped matchgate circuits}
\doped*

\begin{proof}
We use an adjoint fermionic representation of the vectorised operators~\cite{Prosen2007a}: each canonically ordered Majorana string is a Fock basis state, and its length is the number of occupied modes. The normalised Majorana coefficients of $O_{U}=U_t^\dagger O U_t$ therefore define a pure state under vectorisation and normalisation:
\begin{equation}
    O_U= \sum a_{\bm\mu} \gamma_{\bm \mu} \mapsto \sum a_{\bm\mu} |\gamma_{\bm \mu}\rangle \! \rangle  = \sum a_{\bm\mu} f_{\mu_1}^{\dagger}\cdots f_{\mu_{|\bm \mu |}}^{\dagger}    |0\rangle
\end{equation}
in terms of auxiliary fermion operators $f_\alpha$. Fixed string parity ensures that its entanglement across the stated bond cut equals the LOE (as Gaussian gates preserve the string length, and the non-Gaussian doping is quartic, changing string length by $2$ or $0$).
Its full correlation matrix satisfies
\begin{equation}
 C_{\alpha\beta}=\langle\psi|f_\alpha^\dagger f_\beta|\psi\rangle,\qquad\tr C=\sum_{\bm\mu}|\bm\mu|\,|a_{\bm\mu}|^2.
\end{equation}
Thus $\tr C$ is the mean Majorana-string length; moreover, $0\leq C\leq I$, so its eigenvalues lie in $[0,1]$.

For purely Gaussian evolution of the initial string, we may write the operator entanglement in terms of the binary entropy of the correlation matrix~\cite{Calabrese_2005,Prosen2007a}: $\operatorname{LOE}_A^{(1)}(O_U)=\tr h_2(C_A)$, where $C_A$ is the restriction to the adjoint modes in $A$. After quartic doping, the upper bound $\operatorname{LOE}_A^{(1)}(O_U)\leq\tr h_2(C_A)$ still holds: in a basis diagonalizing $C_A$, each mode with occupation probability $\nu_j$ has entropy $h_2(\nu_j)$, and subadditivity bounds the joint entropy by their sum.

To obtain a bound in terms of the full matrix $C$, let $\Pi_A$ and $\Pi_{\bar A}$ be the projections onto the two sets of modes. The embedded restriction $\Pi_A C\Pi_A$ has the same nonzero eigenvalues as $C^{1/2}\Pi_A C^{1/2}$: these are obtained by reversing the factors $\Pi_A C^{1/2}$ and $C^{1/2}\Pi_A$. The two square-root expressions sum to $C$.
Therefore, purity, R\'enyi monotonicity, and trace concavity give
\begin{equation}
 \operatorname{LOE}^{(k)}_A(O_U) \leq \frac{\tr h_2(C_A)+\tr h_2(C_{\bar A})}{2} \leq \tr h_2(C/2).
 \label{eq:doped-correlation-bound}
\end{equation}
The square roots therefore preserve the entropy sums while making their matrix arguments average to $C/2$.

It remains to bound the eigenvalues of $C$. Moving the Gaussian layers through the quartic gates, leaving one final Gaussian rotation, results in quartic generators involving at most $4t$ independent Majorana directions. We may then choose an orthonormal basis of the initially occupied subspace so that precisely $L-r$ directions are orthogonal to all of these directions, and so $r\leq\min(L,4t)$ (cf.~\cref{app:MPO}). Each such direction commutes with every quartic product, because it anticommutes with each of its four factors. These $L-r$ modes therefore remain occupied and factor out. The remaining state initially has $r$ particles and uses at most $r+4t$ modes.

In a basis aligned with one quartic generator, its conjugation fixes strings with even overlap and mixes odd-overlap strings in pairs of lengths $m,m+2$, with rotation angle $2\theta$. On each pair, the change of the string-length operator has eigenvalues $\pm2|\sin(2\theta)|$. Thus each gate increases the mean particle number by at most $2|\sin(2\theta)|$, including for coherent superpositions. Therefore, $C$ has $L-r$ eigenvalues equal to one, while its remaining nonzero eigenvalues number at most $r+4t$ and sum to at most $r+2t|\sin(2\theta)|$. The final Gaussian rotation preserves this spectrum. Each unit eigenvalue contributes $h_2(1/2)=1$ to~\cref{eq:doped-correlation-bound}. For the remaining eigenvalues, we use concavity and the fact that $h_2$ is increasing on $[0,1/2]$. For $t>0$, this gives
\begin{align}
 \operatorname{LOE}^{(k)}_A(U_t^\dagger O U_t) &\leq L-r+(r+4t) h_2\!\left(\frac{r+2t|\sin(2\theta)|}{2(r+4t)}\right)\notag\\
 &\leq L+\kappa(\theta)t,
\end{align}
The first expression increases with $r$, so the last step uses $r\leq4t$. For $t=0$, $C$ is a rank-$L$ projector and $\tr h_2(C/2)=L$ directly.
\end{proof}
\noindent
Using $|\sin(2\theta)|\leq1$ recovers the maximum value of $\kappa(\pi/4) = (8h_2(3/8)-4\simeq3.64$.

\end{document}

%% file: figures/tn_preamble.tex
\usepackage{tikz}
\usepackage{amsmath}
\usepackage{amssymb}
\usepackage{bm}
\usepackage{graphicx}
\usetikzlibrary{calc,shapes.geometric,arrows.meta,positioning,fit,backgrounds,decorations.pathreplacing,decorations.pathmorphing}
\definecolor{tncol22223b}{RGB}{34,34,59}
\definecolor{tncol2281aa}{RGB}{34,129,170}
\definecolor{tncol26547c}{RGB}{38,84,124}
\definecolor{tncol273043}{RGB}{39,48,67}
\definecolor{tncol2a9d8f}{RGB}{42,157,143}
\definecolor{tncol33415c}{RGB}{51,65,92}
\definecolor{tncol45caff}{RGB}{69,202,255}
\definecolor{tncol4c78a8}{RGB}{76,120,168}
\definecolor{tncol555555}{RGB}{85,85,85}
\definecolor{tncol6b7280}{RGB}{107,114,128}
\definecolor{tncol808080}{RGB}{128,128,128}
\definecolor{tncol81bdc3}{RGB}{129,189,195}
\definecolor{tncol888888}{RGB}{136,136,136}
\definecolor{tncol898989}{RGB}{137,137,137}
\definecolor{tncol8ecae6}{RGB}{142,202,230}
\definecolor{tncol9b2226}{RGB}{155,34,38}
\definecolor{tncol9bf6ff}{RGB}{155,246,255}
\definecolor{tncola0c4ff}{RGB}{160,196,255}
\definecolor{tncolb0ddfd}{RGB}{176,221,253}
\definecolor{tncolb4bd9b}{RGB}{180,189,155}
\definecolor{tncolb5e48c}{RGB}{181,228,140}
\definecolor{tncolbc455a}{RGB}{188,69,90}
\definecolor{tncolbdb2ff}{RGB}{189,178,255}
\definecolor{tncolcaffbf}{RGB}{202,255,191}
\definecolor{tncolccd5c3}{RGB}{204,213,195}
\definecolor{tncold0f0c0}{RGB}{208,240,192}
\definecolor{tncoldbd7d2}{RGB}{219,215,210}
\definecolor{tncoldec4e9}{RGB}{222,196,233}
\definecolor{tncole8edf4}{RGB}{232,237,244}
\definecolor{tncole9c46a}{RGB}{233,196,106}
\definecolor{tncolef476f}{RGB}{239,71,111}
\definecolor{tncolf28e2b}{RGB}{242,142,43}
\definecolor{tncolf3e8b6}{RGB}{243,232,182}
\definecolor{tncolf4a261}{RGB}{244,162,97}
\definecolor{tncolf6cf98}{RGB}{246,207,152}
\definecolor{tncolf9d6d3}{RGB}{249,214,211}
\definecolor{tncolfdba77}{RGB}{253,186,119}
\definecolor{tncolfdf8ec}{RGB}{253,248,236}
\definecolor{tncolfdffb6}{RGB}{253,255,182}
\definecolor{tncolff1b6b}{RGB}{255,27,107}
\definecolor{tncolffa27a}{RGB}{255,162,122}
\definecolor{tncolffadad}{RGB}{255,173,173}
\definecolor{tncolffbf00}{RGB}{255,191,0}
\definecolor{tncolffd166}{RGB}{255,209,102}
\definecolor{tncolffd6a5}{RGB}{255,214,165}
\definecolor{tncolffd6de}{RGB}{255,214,222}

%% file: figures/gaussian_conj.tex
\providecommand{\gaussConj}[1][1]{%
\ensuremath{\begingroup\color{black}%
\vcenter{\hbox{%
\scalebox{#1}{%
\begin{tikzpicture}[baseline=(current bounding box.center)]
\node[rectangle, minimum size=9mm, inner sep=2pt, outer sep=0pt, rounded corners=1.6pt, draw=black, line width=0.6pt, fill=tncold0f0c0, text=black, minimum width=6mm, minimum height=5mm] (t0) at (0,-0) {$U$};
\node[rectangle, minimum size=9mm, inner sep=2pt, outer sep=0pt, rounded corners=1.6pt, draw=black, line width=0.6pt, fill=tncoldec4e9, text=black, minimum width=6mm, minimum height=5mm] (t1) at (1.6,-0) {$U^\dagger$};
\node[rectangle, minimum size=4mm, inner sep=2pt, outer sep=0pt, rounded corners=1.6pt, draw=black, line width=0.6pt, fill=tncolb0ddfd, text=black] (t2) at (0.8,-0) {$c$};
\draw[draw=tncol33415c, line width=0.6pt] (t0.0) to[out=0, in=180] (t2.180);
\draw[draw=tncol33415c, line width=0.6pt] (t2.0) to[out=0, in=180] (t1.180);
\draw[draw=tncol33415c, line width=0.6pt] (t0.west) -- ++(180:0.45);
\draw[draw=tncol33415c, line width=0.6pt] (t1.east) -- ++(0:0.45);
\draw[draw=tncol33415c, line width=0.6pt] (t2.south) -- ++(270:0.45) node[below] {$\alpha$};
\end{tikzpicture}%
}}}%
\endgroup}%
}

%% file: figures/orthog_action.tex
\providecommand{\orthogAct}[1][1]{%
\ensuremath{\begingroup\color{black}%
\vcenter{\hbox{%
\scalebox{#1}{%
\begin{tikzpicture}[baseline=(current bounding box.center)]
\node[rectangle, minimum size=4mm, inner sep=2pt, outer sep=0pt, rounded corners=1.6pt, draw=black, line width=0.6pt, fill=tncolb0ddfd, text=black] (t0) at (0,-0) {$c$};
\node[rectangle, minimum size=9mm, inner sep=2pt, outer sep=0pt, rounded corners=1.6pt, draw=black, line width=0.6pt, fill=tncolffd6de, text=black, minimum width=6mm, minimum height=6mm] (t1) at (0,-1.2) {$R$};
\draw[draw=tncol33415c, line width=0.6pt] (t0.-90) to[out=-90, in=90] node[midway, fill=white, inner sep=1pt] {$\beta$} (t1.90);
\draw[draw=tncol33415c, line width=0.6pt] (t0.west) -- ++(180:0.45);
\draw[draw=tncol33415c, line width=0.6pt] (t0.east) -- ++(0:0.45);
\draw[draw=tncol33415c, line width=0.6pt] (t1.south) -- ++(270:0.45) node[below] {$\alpha$};
\end{tikzpicture}%
}}}%
\endgroup}%
}

%% file: figures/string_tensor.tex
\providecommand{\stringTensor}[1][1]{%
\ensuremath{\begingroup\color{black}%
\vcenter{\hbox{%
\scalebox{#1}{%
\begin{tikzpicture}[baseline=(current bounding box.center)]
\node[rectangle, minimum size=9mm, inner sep=2pt, outer sep=0pt, rounded corners=1.6pt, draw=black, line width=0.6pt, fill=tncol2281aa, text=black, minimum height=4mm, minimum width=15mm] (t0) at (0,-0) {$\gamma_{\bm{\nu}}$};
\draw[draw=tncol33415c, line width=0.6pt] (t0.west) -- ++(180:0.45);
\draw[draw=tncol33415c, line width=0.6pt] (t0.east) -- ++(0:0.45);
\draw[draw=tncol33415c, line width=0.6pt] ($(t0.south) + (0:-0.6)$) -- ++(270:0.45) node[below] {$\nu_1$};
\draw[draw=tncol33415c, line width=0.6pt] ($(t0.south) + (0:-0.2)$) -- ++(270:0.45) node[below] {$\nu_2$};
\node at ($ ($(t0.south) + (0:0.2)$) + (270:0.225) $) {$\cdots$};
\draw[draw=tncol33415c, line width=0.6pt] ($(t0.south) + (0:0.6)$) -- ++(270:0.45) node[below] {$\nu_L$};
\end{tikzpicture}%
}}}%
\endgroup}%
}

%% file: figures/c_product.tex
\providecommand{\cProduct}[1][1]{%
\ensuremath{\begingroup\color{black}%
\vcenter{\hbox{%
\scalebox{#1}{%
\begin{tikzpicture}[baseline=(current bounding box.center)]
\node[rectangle, minimum size=4mm, inner sep=2pt, outer sep=0pt, rounded corners=1.6pt, draw=black, line width=0.6pt, fill=tncolb0ddfd, text=black] (t0) at (0,-0) {$c$};
\node[rectangle, minimum size=4mm, inner sep=2pt, outer sep=0pt, rounded corners=1.6pt, draw=black, line width=0.6pt, fill=tncolb0ddfd, text=black] (t1) at (0.8,-0) {$c$};
\node[rectangle, minimum size=4mm, inner sep=2pt, outer sep=0pt, rounded corners=1.6pt, draw=black, line width=0.6pt, fill=tncolb0ddfd, text=black] (t2) at (2.4,-0) {$c$};
\draw[draw=tncol33415c, line width=0.6pt] (t0.0) to[out=0, in=180] (t1.180);
\draw[draw=tncol33415c, line width=0.6pt] (t1.0) to[out=0, in=180] node[midway, fill=white, inner sep=1pt] {$\cdots$} (t2.180);
\draw[draw=tncol33415c, line width=0.6pt] (t0.south) -- ++(270:0.45) node[below] {$\nu_1$};
\draw[draw=tncol33415c, line width=0.6pt] (t0.west) -- ++(180:0.45);
\draw[draw=tncol33415c, line width=0.6pt] (t1.south) -- ++(270:0.45) node[below] {$\nu_2$};
\draw[draw=tncol33415c, line width=0.6pt] (t2.south) -- ++(270:0.45) node[below] {$\nu_L$};
\draw[draw=tncol33415c, line width=0.6pt] (t2.east) -- ++(0:0.45);
\end{tikzpicture}%
}}}%
\endgroup}%
}

%% file: figures/string_conjugation.tex
\providecommand{\stringConj}[1][1]{%
\ensuremath{\begingroup\color{black}%
\vcenter{\hbox{%
\scalebox{#1}{%
\begin{tikzpicture}[baseline=(current bounding box.center)]
\node[rectangle, minimum size=9mm, inner sep=2pt, outer sep=0pt, rounded corners=1.6pt, draw=black, line width=0.6pt, fill=tncold0f0c0, text=black, minimum width=6mm, minimum height=5mm] (t0) at (0,-0) {$U$};
\node[rectangle, minimum size=9mm, inner sep=2pt, outer sep=0pt, rounded corners=1.6pt, draw=black, line width=0.6pt, fill=tncoldec4e9, text=black, minimum width=6mm, minimum height=5mm] (t1) at (2.4,-0) {$U^\dagger$};
\node[rectangle, minimum size=9mm, inner sep=2pt, outer sep=0pt, rounded corners=1.6pt, draw=black, line width=0.6pt, fill=tncol2281aa, text=black, minimum height=4mm, minimum width=15mm] (t2) at (1.2,-0) {$\gamma_{\bm{\nu}}$};
\draw[draw=tncol33415c, line width=0.6pt] (t0.0) to[out=0, in=180] (t2.180);
\draw[draw=tncol33415c, line width=0.6pt] (t2.0) to[out=0, in=180] (t1.180);
\draw[draw=tncol33415c, line width=0.6pt] (t0.west) -- ++(180:0.45);
\draw[draw=tncol33415c, line width=0.6pt] (t1.east) -- ++(0:0.45);
\draw[draw=tncol33415c, line width=0.6pt] (t2.south) -- ++(270:0.45) node[below] {$\bm{\nu}$};
\end{tikzpicture}%
}}}%
\endgroup}%
}

%% file: figures/c_product_expanded_conj.tex
\providecommand{\cProdConjExpanded}[1][1]{%
\ensuremath{\begingroup\color{black}%
\vcenter{\hbox{%
\scalebox{#1}{%
\begin{tikzpicture}[baseline=(current bounding box.center)]
\node[rectangle, minimum size=9mm, inner sep=2pt, outer sep=0pt, rounded corners=1.6pt, draw=black, line width=0.6pt, fill=tncold0f0c0, text=black, minimum width=6mm, minimum height=5mm] (t0) at (0,-0) {$U$};
\node[rectangle, minimum size=4mm, inner sep=2pt, outer sep=0pt, rounded corners=1.6pt, draw=black, line width=0.6pt, fill=tncolb0ddfd, text=black] (t1) at (0.8,-0) {$c$};
\node[rectangle, minimum size=9mm, inner sep=2pt, outer sep=0pt, rounded corners=1.6pt, draw=black, line width=0.6pt, fill=tncoldec4e9, text=black, minimum width=6mm, minimum height=5mm] (t2) at (1.6,-0) {$U^\dagger$};
\node[rectangle, minimum size=9mm, inner sep=2pt, outer sep=0pt, rounded corners=1.6pt, draw=black, line width=0.6pt, fill=tncold0f0c0, text=black, minimum width=6mm, minimum height=5mm] (t3) at (2.4,-0) {$U$};
\node[rectangle, minimum size=4mm, inner sep=2pt, outer sep=0pt, rounded corners=1.6pt, draw=black, line width=0.6pt, fill=tncolb0ddfd, text=black] (t4) at (3.2,-0) {$c$};
\node[rectangle, minimum size=9mm, inner sep=2pt, outer sep=0pt, rounded corners=1.6pt, draw=black, line width=0.6pt, fill=tncoldec4e9, text=black, minimum width=6mm, minimum height=5mm] (t5) at (4,-0) {$U^\dagger$};
\node[rectangle, minimum size=9mm, inner sep=2pt, outer sep=0pt, rounded corners=1.6pt, draw=black, line width=0.6pt, fill=tncold0f0c0, text=black, minimum width=6mm, minimum height=5mm] (t6) at (5.6,-0) {$U$};
\node[rectangle, minimum size=4mm, inner sep=2pt, outer sep=0pt, rounded corners=1.6pt, draw=black, line width=0.6pt, fill=tncolb0ddfd, text=black] (t7) at (6.4,-0) {$c$};
\node[rectangle, minimum size=9mm, inner sep=2pt, outer sep=0pt, rounded corners=1.6pt, draw=black, line width=0.6pt, fill=tncoldec4e9, text=black, minimum width=6mm, minimum height=5mm] (t8) at (7.2,-0) {$U^\dagger$};
\draw[draw=tncol33415c, line width=0.6pt] (t0.0) to[out=0, in=180] (t1.180);
\draw[draw=tncol33415c, line width=0.6pt] (t1.0) to[out=0, in=180] (t2.180);
\draw[draw=tncol33415c, line width=0.6pt] (t2.0) to[out=0, in=180] (t3.180);
\draw[draw=tncol33415c, line width=0.6pt] (t3.0) to[out=0, in=180] (t4.180);
\draw[draw=tncol33415c, line width=0.6pt] (t4.0) to[out=0, in=180] (t5.180);
\draw[draw=tncol33415c, line width=0.6pt] (t5.0) to[out=0, in=180] node[midway, fill=white, inner sep=1pt] {$\cdots$} (t6.180);
\draw[draw=tncol33415c, line width=0.6pt] (t6.0) to[out=0, in=180] (t7.180);
\draw[draw=tncol33415c, line width=0.6pt] (t7.0) to[out=0, in=180] (t8.180);
\draw[draw=tncol33415c, line width=0.6pt] (t0.west) -- ++(180:0.45);
\draw[draw=tncol33415c, line width=0.6pt] (t1.south) -- ++(270:0.45) node[below] {$\nu_1$};
\draw[draw=tncol33415c, line width=0.6pt] (t4.south) -- ++(270:0.45) node[below] {$\nu_2$};
\draw[draw=tncol33415c, line width=0.6pt] (t7.south) -- ++(270:0.45) node[below] {$\nu_L$};
\draw[draw=tncol33415c, line width=0.6pt] (t8.east) -- ++(0:0.45);
\end{tikzpicture}%
}}}%
\endgroup}%
}

%% file: figures/R_prod_conj.tex
\providecommand{\RProdConj}[1][1]{%
\ensuremath{\begingroup\color{black}%
\vcenter{\hbox{%
\scalebox{#1}{%
\begin{tikzpicture}[baseline=(current bounding box.center)]
\node[rectangle, minimum size=4mm, inner sep=2pt, outer sep=0pt, rounded corners=1.6pt, draw=black, line width=0.6pt, fill=tncolb0ddfd, text=black] (t0) at (0,-0) {$c$};
\node[rectangle, minimum size=4mm, inner sep=2pt, outer sep=0pt, rounded corners=1.6pt, draw=black, line width=0.6pt, fill=tncolb0ddfd, text=black] (t1) at (0.8,-0) {$c$};
\node[rectangle, minimum size=4mm, inner sep=2pt, outer sep=0pt, rounded corners=1.6pt, draw=black, line width=0.6pt, fill=tncolb0ddfd, text=black] (t2) at (2.4,-0) {$c$};
\node[rectangle, minimum size=9mm, inner sep=2pt, outer sep=0pt, rounded corners=1.6pt, draw=black, line width=0.6pt, fill=tncolffd6de, text=black, minimum width=6mm, minimum height=6mm] (t3) at (0,-0.8) {$R$};
\node[rectangle, minimum size=9mm, inner sep=2pt, outer sep=0pt, rounded corners=1.6pt, draw=black, line width=0.6pt, fill=tncolffd6de, text=black, minimum width=6mm, minimum height=6mm] (t4) at (0.8,-0.8) {$R$};
\node[rectangle, minimum size=9mm, inner sep=2pt, outer sep=0pt, rounded corners=1.6pt, draw=black, line width=0.6pt, fill=tncolffd6de, text=black, minimum width=6mm, minimum height=6mm] (t5) at (2.4,-0.8) {$R$};
\draw[draw=tncol33415c, line width=0.6pt] (t0.0) to[out=0, in=180] (t1.180);
\draw[draw=tncol33415c, line width=0.6pt] (t1.0) to[out=0, in=180] node[midway, fill=white, inner sep=1pt] {$\cdots$} (t2.180);
\draw[draw=tncol33415c, line width=0.6pt] (t0.-90) to[out=-90, in=90] (t3.90);
\draw[draw=tncol33415c, line width=0.6pt] (t1.-90) to[out=-90, in=90] (t4.90);
\draw[draw=tncol33415c, line width=0.6pt] (t2.-90) to[out=-90, in=90] (t5.90);
\draw[draw=tncol33415c, line width=0.6pt] (t0.west) -- ++(180:0.45);
\draw[draw=tncol33415c, line width=0.6pt] (t2.east) -- ++(0:0.45);
\draw[draw=tncol33415c, line width=0.6pt] (t3.south) -- ++(270:0.45) node[below] {$\nu_1$};
\draw[draw=tncol33415c, line width=0.6pt] (t4.south) -- ++(270:0.45) node[below] {$\nu_2$};
\draw[draw=tncol33415c, line width=0.6pt] (t5.south) -- ++(270:0.45) node[below] {$\nu_L$};
\end{tikzpicture}%
}}}%
\endgroup}%
}

%% file: figures/c_prod_conj_2_replica.tex
\providecommand{\cProdTwoRepConj}[1][1]{%
\ensuremath{\begingroup\color{black}%
\vcenter{\hbox{%
\scalebox{#1}{%
\begin{tikzpicture}[baseline=(current bounding box.center)]
\node[rectangle, minimum size=9mm, inner sep=2pt, outer sep=0pt, rounded corners=1.6pt, draw=black, line width=0.6pt, fill=tncold0f0c0, text=black, minimum width=6mm, minimum height=5mm] (t0) at (0,-0) {$U$};
\node[rectangle, minimum size=9mm, inner sep=2pt, outer sep=0pt, rounded corners=1.6pt, draw=black, line width=0.6pt, fill=tncoldec4e9, text=black, minimum width=6mm, minimum height=5mm] (t1) at (4,-0) {$U^\dagger$};
\node[rectangle, minimum size=4mm, inner sep=2pt, outer sep=0pt, rounded corners=1.6pt, draw=black, line width=0.6pt, fill=tncolb0ddfd, text=black] (t2) at (0.8,-0) {$c$};
\node[rectangle, minimum size=4mm, inner sep=2pt, outer sep=0pt, rounded corners=1.6pt, draw=black, line width=0.6pt, fill=tncolb0ddfd, text=black] (t3) at (1.6,-0) {$c$};
\node[rectangle, minimum size=4mm, inner sep=2pt, outer sep=0pt, rounded corners=1.6pt, draw=black, line width=0.6pt, fill=tncolb0ddfd, text=black] (t4) at (3.2,-0) {$c$};
\node[rectangle, minimum size=9mm, inner sep=2pt, outer sep=0pt, rounded corners=1.6pt, draw=black, line width=0.6pt, fill=tncold0f0c0, text=black, minimum width=6mm, minimum height=5mm] (t5) at (0,-0.8) {$U$};
\node[rectangle, minimum size=9mm, inner sep=2pt, outer sep=0pt, rounded corners=1.6pt, draw=black, line width=0.6pt, fill=tncoldec4e9, text=black, minimum width=6mm, minimum height=5mm] (t6) at (4,-0.8) {$U^\dagger$};
\node[rectangle, minimum size=4mm, inner sep=2pt, outer sep=0pt, rounded corners=1.6pt, draw=black, line width=0.6pt, fill=tncolb0ddfd, text=black] (t7) at (0.8,-0.8) {$c$};
\node[rectangle, minimum size=4mm, inner sep=2pt, outer sep=0pt, rounded corners=1.6pt, draw=black, line width=0.6pt, fill=tncolb0ddfd, text=black] (t8) at (1.6,-0.8) {$c$};
\node[rectangle, minimum size=4mm, inner sep=2pt, outer sep=0pt, rounded corners=1.6pt, draw=black, line width=0.6pt, fill=tncolb0ddfd, text=black] (t9) at (3.2,-0.8) {$c$};
\draw[draw=tncol33415c, line width=0.6pt] (t2.0) to[out=0, in=180] (t3.180);
\draw[draw=tncol33415c, line width=0.6pt] (t3.0) -- (t4.180);
\begin{scope}[shift={($ (t3.0) !0.5! (t4.180) $)}, rotate=0]
  \fill[white] (-9.1pt,-4.6pt) rectangle (9.1pt,4.6pt);
  \fill (-4.5pt,0pt) circle (0.6pt) (0pt,0pt) circle (0.6pt) (4.5pt,0pt) circle (0.6pt);
\end{scope}
\draw[draw=tncol33415c, line width=0.6pt] (t0.0) to[out=0, in=180] (t2.180);
\draw[draw=tncol33415c, line width=0.6pt] (t4.0) to[out=0, in=180] (t1.180);
\draw[draw=tncol33415c, line width=0.6pt] (t7.0) to[out=0, in=180] (t8.180);
\draw[draw=tncol33415c, line width=0.6pt] (t8.0) -- (t9.180);
\begin{scope}[shift={($ (t8.0) !0.5! (t9.180) $)}, rotate=0]
  \fill[white] (-9.1pt,-4.6pt) rectangle (9.1pt,4.6pt);
  \fill (-4.5pt,0pt) circle (0.6pt) (0pt,0pt) circle (0.6pt) (4.5pt,0pt) circle (0.6pt);
\end{scope}
\draw[draw=tncol33415c, line width=0.6pt] (t5.0) to[out=0, in=180] (t7.180);
\draw[draw=tncol33415c, line width=0.6pt] (t9.0) to[out=0, in=180] (t6.180);
\draw[draw=tncol33415c, line width=0.6pt] (t0.west) -- ++(180:0.35);
\draw[draw=tncol33415c, line width=0.6pt] (t1.east) -- ++(0:0.35);
\draw[draw=tncol33415c, line width=0.6pt] (t2.north) -- ++(90:0.35) node[above] {$\nu_1^{(1)}$};
\draw[draw=tncol33415c, line width=0.6pt] (t3.north) -- ++(90:0.35) node[above] {$\nu_2^{(1)}$};
\draw[draw=tncol33415c, line width=0.6pt] (t4.north) -- ++(90:0.35) node[above] {$\nu_L^{(1)}$};
\draw[draw=tncol33415c, line width=0.6pt] (t5.west) -- ++(180:0.35);
\draw[draw=tncol33415c, line width=0.6pt] (t6.east) -- ++(0:0.35);
\draw[draw=tncol33415c, line width=0.6pt] (t7.south) -- ++(270:0.35) node[below] {$\nu_1^{(2)}$};
\draw[draw=tncol33415c, line width=0.6pt] (t8.south) -- ++(270:0.35) node[below] {$\nu_2^{(2)}$};
\draw[draw=tncol33415c, line width=0.6pt] (t9.south) -- ++(270:0.35) node[below] {$\nu_L^{(2)}$};
\end{tikzpicture}%
}}}%
\endgroup}%
}

%% file: figures/two_replica_R_conj.tex
\providecommand{\twoRepRConj}[1][1]{%
\ensuremath{\begingroup\color{black}%
\vcenter{\hbox{%
\scalebox{#1}{%
\begin{tikzpicture}[baseline=(current bounding box.center)]
\node[rectangle, minimum size=4mm, inner sep=2pt, outer sep=0pt, rounded corners=1.6pt, draw=black, line width=0.6pt, fill=tncolb0ddfd, text=black] (t0) at (0.8,-0) {$c$};
\node[rectangle, minimum size=4mm, inner sep=2pt, outer sep=0pt, rounded corners=1.6pt, draw=black, line width=0.6pt, fill=tncolb0ddfd, text=black] (t1) at (0.8,-0.8) {$c$};
\node[rectangle, minimum size=4mm, inner sep=2pt, outer sep=0pt, rounded corners=1.6pt, draw=black, line width=0.6pt, fill=tncolb0ddfd, text=black] (t2) at (0.8,-2.4) {$c$};
\node[rectangle, minimum size=4mm, inner sep=2pt, outer sep=0pt, rounded corners=1.6pt, draw=black, line width=0.6pt, fill=tncolb0ddfd, text=black] (t3) at (1.6,-0) {$c$};
\node[rectangle, minimum size=4mm, inner sep=2pt, outer sep=0pt, rounded corners=1.6pt, draw=black, line width=0.6pt, fill=tncolb0ddfd, text=black] (t4) at (1.6,-0.8) {$c$};
\node[rectangle, minimum size=4mm, inner sep=2pt, outer sep=0pt, rounded corners=1.6pt, draw=black, line width=0.6pt, fill=tncolb0ddfd, text=black] (t5) at (1.6,-2.4) {$c$};
\node[rectangle, minimum size=9mm, inner sep=2pt, outer sep=0pt, rounded corners=1.6pt, draw=black, line width=0.6pt, fill=tncolffd6de, text=black, minimum width=6mm, minimum height=6mm] (t6) at (0,-0) {$R$};
\node[rectangle, minimum size=9mm, inner sep=2pt, outer sep=0pt, rounded corners=1.6pt, draw=black, line width=0.6pt, fill=tncolffd6de, text=black, minimum width=6mm, minimum height=6mm] (t7) at (0,-0.8) {$R$};
\node[rectangle, minimum size=9mm, inner sep=2pt, outer sep=0pt, rounded corners=1.6pt, draw=black, line width=0.6pt, fill=tncolffd6de, text=black, minimum width=6mm, minimum height=6mm] (t8) at (0,-2.4) {$R$};
\node[rectangle, minimum size=9mm, inner sep=2pt, outer sep=0pt, rounded corners=1.6pt, draw=black, line width=0.6pt, fill=tncolffd6de, text=black, minimum width=6mm, minimum height=6mm] (t9) at (2.4,-0) {$R$};
\node[rectangle, minimum size=9mm, inner sep=2pt, outer sep=0pt, rounded corners=1.6pt, draw=black, line width=0.6pt, fill=tncolffd6de, text=black, minimum width=6mm, minimum height=6mm] (t10) at (2.4,-0.8) {$R$};
\node[rectangle, minimum size=9mm, inner sep=2pt, outer sep=0pt, rounded corners=1.6pt, draw=black, line width=0.6pt, fill=tncolffd6de, text=black, minimum width=6mm, minimum height=6mm] (t11) at (2.4,-2.4) {$R$};
\draw[draw=tncol33415c, line width=0.6pt] (t0.-90) to[out=-90, in=90] (t1.90);
\draw[draw=tncol33415c, line width=0.6pt] (t3.-90) to[out=-90, in=90] (t4.90);
\draw[draw=tncol33415c, line width=0.6pt] (t6.0) to[out=0, in=180] (t0.180);
\draw[draw=tncol33415c, line width=0.6pt] (t9.180) to[out=180, in=0] (t3.0);
\draw[draw=tncol33415c, line width=0.6pt] (t7.0) to[out=0, in=180] (t1.180);
\draw[draw=tncol33415c, line width=0.6pt] (t10.180) to[out=180, in=0] (t4.0);
\draw[draw=tncol33415c, line width=0.6pt] (t8.0) to[out=0, in=180] (t2.180);
\draw[draw=tncol33415c, line width=0.6pt] (t11.180) to[out=180, in=0] (t5.0);
\draw[draw=tncol33415c, line width=0.6pt] (t4.-90) -- (t5.90);
\begin{scope}[shift={($ (t4.-90) !0.5! (t5.90) $)}, rotate=-90]
  \fill[white] (-9.1pt,-4.6pt) rectangle (9.1pt,4.6pt);
  \fill (-4.5pt,0pt) circle (0.6pt) (0pt,0pt) circle (0.6pt) (4.5pt,0pt) circle (0.6pt);
\end{scope}
\draw[draw=tncol33415c, line width=0.6pt] (t1.-90) -- (t2.90);
\begin{scope}[shift={($ (t1.-90) !0.5! (t2.90) $)}, rotate=-90]
  \fill[white] (-9.1pt,-4.6pt) rectangle (9.1pt,4.6pt);
  \fill (-4.5pt,0pt) circle (0.6pt) (0pt,0pt) circle (0.6pt) (4.5pt,0pt) circle (0.6pt);
\end{scope}
\draw[draw=tncol33415c, line width=0.6pt] (t0.north) -- ++(90:0.35);
\draw[draw=tncol33415c, line width=0.6pt] (t2.south) -- ++(270:0.35);
\draw[draw=tncol33415c, line width=0.6pt] (t3.north) -- ++(90:0.35);
\draw[draw=tncol33415c, line width=0.6pt] (t5.south) -- ++(270:0.35);
\draw[draw=tncol33415c, line width=0.6pt] (t6.west) -- ++(180:0.35) node[left] {$\nu_1^{(1)}$};
\draw[draw=tncol33415c, line width=0.6pt] (t7.west) -- ++(180:0.35) node[left] {$\nu_2^{(1)}$};
\draw[draw=tncol33415c, line width=0.6pt] (t8.west) -- ++(180:0.35) node[left] {$\nu_L^{(1)}$};
\draw[draw=tncol33415c, line width=0.6pt] (t9.east) -- ++(0:0.35) node[right] {$\nu_1^{(2)}$};
\draw[draw=tncol33415c, line width=0.6pt] (t10.east) -- ++(0:0.35) node[right] {$\nu_2^{(2)}$};
\draw[draw=tncol33415c, line width=0.6pt] (t11.east) -- ++(0:0.35) node[right] {$\nu_L^{(2)}$};
\end{tikzpicture}%
}}}%
\endgroup}%
}

%% file: figures/weingarten_p2_part.tex
\providecommand{\WeinPTwoTerm}[1][1]{%
\ensuremath{\begingroup\color{black}%
\vcenter{\hbox{%
\scalebox{#1}{%
\begin{tikzpicture}[baseline=(current bounding box.center)]
\node[rectangle, minimum width=0.9cm, minimum height=3.3cm, inner sep=2pt, outer sep=0pt, rounded corners=1.6pt, draw=black, line width=0.6pt, fill=tncoldbd7d2, text=black] (t0) at (0,-1.2) {$p_2$};
\node[rectangle, minimum size=4mm, inner sep=2pt, outer sep=0pt, rounded corners=1.6pt, draw=black, line width=0.6pt, fill=tncolb0ddfd, text=black] (t1) at (1.1,-0) {$c$};
\node[rectangle, minimum size=4mm, inner sep=2pt, outer sep=0pt, rounded corners=1.6pt, draw=black, line width=0.6pt, fill=tncolb0ddfd, text=black] (t2) at (1.1,-0.8) {$c$};
\node[rectangle, minimum size=4mm, inner sep=2pt, outer sep=0pt, rounded corners=1.6pt, draw=black, line width=0.6pt, fill=tncolb0ddfd, text=black] (t3) at (1.1,-2.4) {$c$};
\node[rectangle, minimum size=4mm, inner sep=2pt, outer sep=0pt, rounded corners=1.6pt, draw=black, line width=0.6pt, fill=tncolb0ddfd, text=black] (t4) at (2.2,-0) {$c$};
\node[rectangle, minimum size=4mm, inner sep=2pt, outer sep=0pt, rounded corners=1.6pt, draw=black, line width=0.6pt, fill=tncolb0ddfd, text=black] (t5) at (2.2,-0.8) {$c$};
\node[rectangle, minimum size=4mm, inner sep=2pt, outer sep=0pt, rounded corners=1.6pt, draw=black, line width=0.6pt, fill=tncolb0ddfd, text=black] (t6) at (2.2,-2.4) {$c$};
\draw[draw=tncol33415c, line width=0.6pt] (t1.-90) to[out=-90, in=90] (t2.90);
\draw[draw=tncol33415c, line width=0.6pt] (t2.-90) -- (t3.90);
\begin{scope}[shift={($ (t2.-90) !0.5! (t3.90) $)}, rotate=-90]
  \fill[white] (-9.1pt,-4.6pt) rectangle (9.1pt,4.6pt);
  \fill (-4.5pt,0pt) circle (0.6pt) (0pt,0pt) circle (0.6pt) (4.5pt,0pt) circle (0.6pt);
\end{scope}
\draw[draw=tncol33415c, line width=0.6pt] (t4.-90) to[out=-90, in=90] (t5.90);
\draw[draw=tncol33415c, line width=0.6pt] (t5.-90) -- (t6.90);
\begin{scope}[shift={($ (t5.-90) !0.5! (t6.90) $)}, rotate=-90]
  \fill[white] (-9.1pt,-4.6pt) rectangle (9.1pt,4.6pt);
  \fill (-4.5pt,0pt) circle (0.6pt) (0pt,0pt) circle (0.6pt) (4.5pt,0pt) circle (0.6pt);
\end{scope}
\draw[draw=tncol808080, line width=0.6pt, rounded corners=0.2cm] (t4.east) -- ($ (t4.east) + (0:0.3) $) -- ($ (t4.east) + (0:0.3) + (90:0.57) $) --  node[at start, yshift=-6pt, xshift=-7pt, fill=white, inner sep=-2pt, font=\tiny] {$\mu_1^{(2)}$} ($ ($(t0.west) + (270:-1.2)$) + (180:0.3) + (90:0.57) $) -- ($ ($(t0.west) + (270:-1.2)$) + (180:0.3) $) -- ($(t0.west) + (270:-1.2)$);
\draw[draw=tncol808080, line width=0.6pt, rounded corners=0.2cm] (t5.east) -- ($ (t5.east) + (0:0.3) $) -- ($ (t5.east) + (0:0.3) + (90:0.57) $) --  node[at start, yshift=-6pt, xshift=-7pt, fill=white, inner sep=-2pt, font=\tiny] {$\mu_2^{(2)}$} ($ ($(t0.west) + (270:-0.4)$) + (180:0.3) + (90:0.57) $) -- ($ ($(t0.west) + (270:-0.4)$) + (180:0.3) $) -- ($(t0.west) + (270:-0.4)$);
\draw[draw=tncol808080, line width=0.6pt, rounded corners=0.2cm] (t6.east) -- ($ (t6.east) + (0:0.3) $) -- ($ (t6.east) + (0:0.3) + (90:0.57) $) --  node[at start, yshift=-6pt, xshift=-7pt, fill=white, inner sep=-2pt, font=\tiny] {$\mu_L^{(2)}$} ($ ($(t0.west) + (270:1.2)$) + (180:0.3) + (90:0.57) $) -- ($ ($(t0.west) + (270:1.2)$) + (180:0.3) $) -- ($(t0.west) + (270:1.2)$);
\draw[draw=tncol33415c, line width=0.6pt] (t1.west) to[out=180, in=0] node[midway, yshift=6pt, xshift=1pt, fill=white, inner sep=-2pt, font=\tiny] {$\mu_1^{(1)}$} ($(t0.east) + (270:-1.2)$);
\draw[draw=tncol33415c, line width=0.6pt] (t2.west) to[out=180, in=0] node[midway, yshift=6pt, xshift=1pt, fill=white, inner sep=-2pt, font=\tiny] {$\mu_2^{(1)}$} ($(t0.east) + (270:-0.4)$);
\draw[draw=tncol33415c, line width=0.6pt] (t3.west) to[out=180, in=0] node[midway, yshift=6pt, xshift=1pt, fill=white, inner sep=-2pt, font=\tiny] {$\mu_L^{(1)}$} ($(t0.east) + (270:1.2)$);
\draw[draw=tncol33415c, line width=0.6pt] (t1.north) -- ++(90:0.45);
\draw[draw=tncol33415c, line width=0.6pt] (t3.south) -- ++(270:0.45);
\draw[draw=tncol33415c, line width=0.6pt] (t4.north) -- ++(90:0.45);
\draw[draw=tncol33415c, line width=0.6pt] (t6.south) -- ++(270:0.45);
\end{tikzpicture}%
}}}%
\endgroup}%
}

%% file: figures/weingarten_p1_part.tex
\providecommand{\WeinPOneTerm}[1][1]{%
\ensuremath{\begingroup\color{black}%
\vcenter{\hbox{%
\scalebox{#1}{%
\begin{tikzpicture}[baseline=(current bounding box.center)]
\node[rectangle, minimum width=0.9cm, minimum height=3.3cm, inner sep=2pt, outer sep=0pt, rounded corners=1.6pt, draw=black, line width=0.6pt, fill=tncoldbd7d2, text=black] (t0) at (1.1,-1.2) {$p_1$};
\node[isosceles triangle, minimum width=4mm, minimum height=2mm, inner sep=2pt, outer sep=0pt, isosceles triangle stretches, shape border rotate=180, draw=black, line width=0.6pt, fill=tncolffa27a, text=black] (t1) at (0,-0) {};
\node[isosceles triangle, minimum width=4mm, minimum height=2mm, inner sep=2pt, outer sep=0pt, isosceles triangle stretches, shape border rotate=180, draw=black, line width=0.6pt, fill=tncolffa27a, text=black] (t2) at (0,-0.8) {};
\node[isosceles triangle, minimum width=4mm, minimum height=2mm, inner sep=2pt, outer sep=0pt, isosceles triangle stretches, shape border rotate=180, draw=black, line width=0.6pt, fill=tncolffa27a, text=black] (t3) at (0,-2.4) {};
\node[isosceles triangle, minimum width=4mm, minimum height=2mm, inner sep=2pt, outer sep=0pt, isosceles triangle stretches, shape border rotate=0, draw=black, line width=0.6pt, fill=tncolffa27a, text=black] (t4) at (2.2,-0) {};
\node[isosceles triangle, minimum width=4mm, minimum height=2mm, inner sep=2pt, outer sep=0pt, isosceles triangle stretches, shape border rotate=0, draw=black, line width=0.6pt, fill=tncolffa27a, text=black] (t5) at (2.2,-0.8) {};
\node[isosceles triangle, minimum width=4mm, minimum height=2mm, inner sep=2pt, outer sep=0pt, isosceles triangle stretches, shape border rotate=0, draw=black, line width=0.6pt, fill=tncolffa27a, text=black] (t6) at (2.2,-2.4) {};
\draw[draw=tncol33415c, line width=0.6pt] (t1.0) -- node[above=3pt, fill=white, inner sep=-1pt, font=\tiny] {$\nu_1^{(1)}$} ($(t0.west) + (270:-1.2)$);
\draw[draw=tncol33415c, line width=0.6pt] (t2.0) -- node[above=3pt, fill=white, inner sep=-1pt, font=\tiny] {$\nu_2^{(1)}$} ($(t0.west) + (270:-0.4)$);
\draw[draw=tncol33415c, line width=0.6pt] (t3.0) -- node[above=3pt, fill=white, inner sep=-1pt, font=\tiny] {$\nu_L^{(1)}$} ($(t0.west) + (270:1.2)$);
\draw[draw=tncol33415c, line width=0.6pt] (t4.180) -- node[above=3pt, fill=white, inner sep=-1pt, font=\tiny] {$\nu_1^{(2)}$} ($(t0.east) + (270:-1.2)$);
\draw[draw=tncol33415c, line width=0.6pt] (t5.180) -- node[above=3pt, fill=white, inner sep=-1pt, font=\tiny] {$\nu_2^{(2)}$} ($(t0.east) + (270:-0.4)$);
\draw[draw=tncol33415c, line width=0.6pt] (t6.180) -- node[above=3pt, fill=white, inner sep=-1pt, font=\tiny] {$\nu_L^{(2)}$} ($(t0.east) + (270:1.2)$);
\node at ($ ($(t0.west) + (270:0.4)$) + (180:0.225) $) {$\vdots$};
\node at ($ ($(t0.east) + (270:0.4)$) + (0:0.225) $) {$\vdots$};
\end{tikzpicture}%
}}}%
\endgroup}%
}

%% file: figures/weingarten_pairing_input.tex
\providecommand{\WeingartenInput}[1][1]{%
\ensuremath{\begingroup\color{black}%
\vcenter{\hbox{%
\scalebox{#1}{%
\begin{tikzpicture}[baseline=(current bounding box.center)]
\node[circle, minimum size=4mm, inner sep=0pt, outer sep=0pt, draw=tncol273043, line width=0.45pt, fill=tncolffadad, text=tncol273043] (t0) at (0,-0) {};
\node[circle, minimum size=4mm, inner sep=0pt, outer sep=0pt, draw=tncol273043, line width=0.45pt, fill=tncolffadad, text=tncol273043] (t1) at (0,-0.7) {};
\node[circle, minimum size=4mm, inner sep=0pt, outer sep=0pt, draw=tncol273043, line width=0.45pt, fill=tncolffadad, text=tncol273043] (t2) at (0,-1.4) {};
\node[circle, minimum size=4mm, inner sep=0pt, outer sep=0pt, draw=tncol273043, line width=0.45pt, fill=tncolffadad, text=tncol273043] (t3) at (0,-2.1) {};
\node[circle, minimum size=4mm, inner sep=0pt, outer sep=0pt, draw=tncol273043, line width=0.45pt, fill=tncolffd6a5, text=tncol273043] (t4) at (0.8,-0) {};
\node[circle, minimum size=4mm, inner sep=0pt, outer sep=0pt, draw=tncol273043, line width=0.45pt, fill=tncolffd6a5, text=tncol273043] (t5) at (0.8,-0.7) {};
\node[circle, minimum size=4mm, inner sep=0pt, outer sep=0pt, draw=tncol273043, line width=0.45pt, fill=tncolffd6a5, text=tncol273043] (t6) at (0.8,-1.4) {};
\node[circle, minimum size=4mm, inner sep=0pt, outer sep=0pt, draw=tncol273043, line width=0.45pt, fill=tncolffd6a5, text=tncol273043] (t7) at (0.8,-2.1) {};
\node[circle, minimum size=4mm, inner sep=0pt, outer sep=0pt, draw=tncol273043, line width=0.45pt, fill=tncolfdffb6, text=tncol273043] (t8) at (1.6,-0) {};
\node[circle, minimum size=4mm, inner sep=0pt, outer sep=0pt, draw=tncol273043, line width=0.45pt, fill=tncolfdffb6, text=tncol273043] (t9) at (1.6,-0.7) {};
\node[circle, minimum size=4mm, inner sep=0pt, outer sep=0pt, draw=tncol273043, line width=0.45pt, fill=tncolfdffb6, text=tncol273043] (t10) at (1.6,-1.4) {};
\node[circle, minimum size=4mm, inner sep=0pt, outer sep=0pt, draw=tncol273043, line width=0.45pt, fill=tncolfdffb6, text=tncol273043] (t11) at (1.6,-2.1) {};
\node[circle, minimum size=4mm, inner sep=0pt, outer sep=0pt, draw=tncol273043, line width=0.45pt, fill=tncolcaffbf, text=tncol273043] (t12) at (2.4,-0) {};
\node[circle, minimum size=4mm, inner sep=0pt, outer sep=0pt, draw=tncol273043, line width=0.45pt, fill=tncolcaffbf, text=tncol273043] (t13) at (2.4,-0.7) {};
\node[circle, minimum size=4mm, inner sep=0pt, outer sep=0pt, draw=tncol273043, line width=0.45pt, fill=tncolcaffbf, text=tncol273043] (t14) at (2.4,-1.4) {};
\node[circle, minimum size=4mm, inner sep=0pt, outer sep=0pt, draw=tncol273043, line width=0.45pt, fill=tncolcaffbf, text=tncol273043] (t15) at (2.4,-2.1) {};
\draw[draw=tncol4c78a8, line width=0.85pt] (t0.-90) to[bend left=60] (t1.90);
\draw[draw=tncol4c78a8, line width=0.85pt] (t2.-90) to[bend right=60] (t3.90);
\draw[draw=tncolf28e2b, line width=0.85pt] (t4.-90) to[bend left=60] (t6.90);
\draw[draw=tncolf28e2b, line width=0.85pt] (t5.-90) to[bend right=60] (t7.90);
\draw[draw=tncol2a9d8f, line width=0.85pt] (t8.-90) to[bend left=60] (t11.90);
\draw[draw=tncol2a9d8f, line width=0.85pt] (t9.-90) to[bend right=60] (t10.90);
\draw[draw=tncol4c78a8, line width=0.85pt] (t12.-90) to[bend left=60] (t13.90);
\draw[draw=tncol4c78a8, line width=0.85pt] (t14.-90) to[bend right=60] (t15.90);
\node[anchor=center,font={\small},text=tncol273043] at (0.000,0.440) {$c_{1}$};
\node[anchor=center,font={\small},text=tncol4c78a8] at (0.000,-2.600) {$e_{1}$};
\node[anchor=center,font={\scriptsize},text=tncol6b7280] at (-0.550,0.000) {$1$};
\node[anchor=center,font={\scriptsize},text=tncol6b7280] at (-0.550,-0.700) {$2$};
\node[anchor=center,font={\scriptsize},text=tncol6b7280] at (-0.550,-1.400) {$3$};
\node[anchor=center,font={\scriptsize},text=tncol6b7280] at (-0.550,-2.100) {$4$};
\node[anchor=center,font={\small},text=tncol273043] at (0.800,0.440) {$c_{2}$};
\node[anchor=center,font={\small},text=tncolf28e2b] at (0.800,-2.600) {$e_{2}$};
\node[anchor=center,font={\scriptsize},text=tncol6b7280] at (-0.550,0.000) {$1$};
\node[anchor=center,font={\scriptsize},text=tncol6b7280] at (-0.550,-0.700) {$2$};
\node[anchor=center,font={\scriptsize},text=tncol6b7280] at (-0.550,-1.400) {$3$};
\node[anchor=center,font={\scriptsize},text=tncol6b7280] at (-0.550,-2.100) {$4$};
\node[anchor=center,font={\small},text=tncol273043] at (1.600,0.440) {$c_{3}$};
\node[anchor=center,font={\small},text=tncol2a9d8f] at (1.600,-2.600) {$e_{3}$};
\node[anchor=center,font={\scriptsize},text=tncol6b7280] at (-0.550,0.000) {$1$};
\node[anchor=center,font={\scriptsize},text=tncol6b7280] at (-0.550,-0.700) {$2$};
\node[anchor=center,font={\scriptsize},text=tncol6b7280] at (-0.550,-1.400) {$3$};
\node[anchor=center,font={\scriptsize},text=tncol6b7280] at (-0.550,-2.100) {$4$};
\node[anchor=center,font={\small},text=tncol273043] at (2.400,0.440) {$c_{4}$};
\node[anchor=center,font={\small},text=tncol4c78a8] at (2.400,-2.600) {$e_{1}$};
\node[anchor=center,font={\scriptsize},text=tncol6b7280] at (-0.550,0.000) {$1$};
\node[anchor=center,font={\scriptsize},text=tncol6b7280] at (-0.550,-0.700) {$2$};
\node[anchor=center,font={\scriptsize},text=tncol6b7280] at (-0.550,-1.400) {$3$};
\node[anchor=center,font={\scriptsize},text=tncol6b7280] at (-0.550,-2.100) {$4$};
\begin{pgfonlayer}{background}
\node[draw=tncolffadad, line width=0.6pt, inner sep=3pt, fill=tncolffadad, fill opacity=0.1, rounded corners=3pt, fit=(t0)(t1)(t2)(t3)] {};
\node[draw=tncolffd6a5, line width=0.6pt, inner sep=3pt, fill=tncolffd6a5, fill opacity=0.1, rounded corners=3pt, fit=(t4)(t5)(t6)(t7)] {};
\node[draw=tncolfdffb6, line width=0.6pt, inner sep=3pt, fill=tncolfdffb6, fill opacity=0.1, rounded corners=3pt, fit=(t8)(t9)(t10)(t11)] {};
\node[draw=tncolcaffbf, line width=0.6pt, inner sep=3pt, fill=tncolcaffbf, fill opacity=0.1, rounded corners=3pt, fit=(t12)(t13)(t14)(t15)] {};
\end{pgfonlayer}
\end{tikzpicture}%
}}}%
\endgroup}%
}

%% file: figures/weingarten_pairing_output.tex
\providecommand{\WeingartenOutput}[1][1]{%
\ensuremath{\begingroup\color{black}%
\vcenter{\hbox{%
\scalebox{#1}{%
\begin{tikzpicture}[baseline=(current bounding box.center)]
\node[circle, minimum size=4mm, inner sep=0pt, outer sep=0pt, draw=tncol273043, line width=0.45pt, fill=tncolffadad, text=tncol273043, font=\scriptsize, inner sep=0pt] (t0) at (0,-0) {$\mu_{1}$};
\node[circle, minimum size=4mm, inner sep=0pt, outer sep=0pt, draw=tncol273043, line width=0.45pt, fill=tncolffadad, text=tncol273043, font=\scriptsize, inner sep=0pt] (t1) at (0,-0.7) {$\mu_{1}$};
\node[circle, minimum size=1mm, inner sep=0pt, outer sep=0pt, draw=tncol33415c, line width=0.6pt, fill=tncole8edf4, text=black] (t2) at (0,-1.4) {};
\node[circle, minimum size=1mm, inner sep=0pt, outer sep=0pt, draw=tncol33415c, line width=0.6pt, fill=tncole8edf4, text=black] (t3) at (0,-2.1) {};
\node[circle, minimum size=1mm, inner sep=0pt, outer sep=0pt, draw=tncol33415c, line width=0.6pt, fill=tncole8edf4, text=black] (t4) at (0.7,-0) {};
\node[circle, minimum size=1mm, inner sep=0pt, outer sep=0pt, draw=tncol33415c, line width=0.6pt, fill=tncole8edf4, text=black] (t5) at (0.7,-0.7) {};
\node[circle, minimum size=4mm, inner sep=0pt, outer sep=0pt, draw=tncol273043, line width=0.45pt, fill=tncolffd6a5, text=tncol273043, font=\scriptsize, inner sep=0pt] (t6) at (0.7,-1.4) {$\mu_{2}$};
\node[circle, minimum size=4mm, inner sep=0pt, outer sep=0pt, draw=tncol273043, line width=0.45pt, fill=tncolffd6a5, text=tncol273043, font=\scriptsize, inner sep=0pt] (t7) at (0.7,-2.1) {$\mu_{2}$};
\node[circle, minimum size=4mm, inner sep=0pt, outer sep=0pt, draw=tncol273043, line width=0.45pt, fill=tncolfdffb6, text=tncol273043, font=\scriptsize, inner sep=0pt] (t8) at (1.4,-0) {$\mu_{3}$};
\node[circle, minimum size=1mm, inner sep=0pt, outer sep=0pt, draw=tncol33415c, line width=0.6pt, fill=tncole8edf4, text=black] (t9) at (1.4,-0.7) {};
\node[circle, minimum size=4mm, inner sep=0pt, outer sep=0pt, draw=tncol273043, line width=0.45pt, fill=tncolfdffb6, text=tncol273043, font=\scriptsize, inner sep=0pt] (t10) at (1.4,-1.4) {$\mu_{3}$};
\node[circle, minimum size=1mm, inner sep=0pt, outer sep=0pt, draw=tncol33415c, line width=0.6pt, fill=tncole8edf4, text=black] (t11) at (1.4,-2.1) {};
\node[circle, minimum size=1mm, inner sep=0pt, outer sep=0pt, draw=tncol33415c, line width=0.6pt, fill=tncole8edf4, text=black] (t12) at (2.1,-0) {};
\node[circle, minimum size=4mm, inner sep=0pt, outer sep=0pt, draw=tncol273043, line width=0.45pt, fill=tncolcaffbf, text=tncol273043, font=\scriptsize, inner sep=0pt] (t13) at (2.1,-0.7) {$\mu_{4}$};
\node[circle, minimum size=1mm, inner sep=0pt, outer sep=0pt, draw=tncol33415c, line width=0.6pt, fill=tncole8edf4, text=black] (t14) at (2.1,-1.4) {};
\node[circle, minimum size=4mm, inner sep=0pt, outer sep=0pt, draw=tncol273043, line width=0.45pt, fill=tncolcaffbf, text=tncol273043, font=\scriptsize, inner sep=0pt] (t15) at (2.1,-2.1) {$\mu_{4}$};
\node[circle, minimum size=4mm, inner sep=0pt, outer sep=0pt, draw=tncol273043, line width=0.45pt, fill=tncol9bf6ff, text=tncol273043, font=\scriptsize, inner sep=0pt] (t16) at (2.8,-0) {$\mu_{5}$};
\node[circle, minimum size=1mm, inner sep=0pt, outer sep=0pt, draw=tncol33415c, line width=0.6pt, fill=tncole8edf4, text=black] (t17) at (2.8,-0.7) {};
\node[circle, minimum size=1mm, inner sep=0pt, outer sep=0pt, draw=tncol33415c, line width=0.6pt, fill=tncole8edf4, text=black] (t18) at (2.8,-1.4) {};
\node[circle, minimum size=4mm, inner sep=0pt, outer sep=0pt, draw=tncol273043, line width=0.45pt, fill=tncol9bf6ff, text=tncol273043, font=\scriptsize, inner sep=0pt] (t19) at (2.8,-2.1) {$\mu_{5}$};
\node[circle, minimum size=1mm, inner sep=0pt, outer sep=0pt, draw=tncol33415c, line width=0.6pt, fill=tncole8edf4, text=black] (t20) at (3.5,-0) {};
\node[circle, minimum size=4mm, inner sep=0pt, outer sep=0pt, draw=tncol273043, line width=0.45pt, fill=tncola0c4ff, text=tncol273043, font=\scriptsize, inner sep=0pt] (t21) at (3.5,-0.7) {$\mu_{6}$};
\node[circle, minimum size=4mm, inner sep=0pt, outer sep=0pt, draw=tncol273043, line width=0.45pt, fill=tncola0c4ff, text=tncol273043, font=\scriptsize, inner sep=0pt] (t22) at (3.5,-1.4) {$\mu_{6}$};
\node[circle, minimum size=1mm, inner sep=0pt, outer sep=0pt, draw=tncol33415c, line width=0.6pt, fill=tncole8edf4, text=black] (t23) at (3.5,-2.1) {};
\draw[draw=tncol4c78a8, line width=0.85pt] (t0.-90) to[bend left=60] (t1.90);
\draw[draw=tncol4c78a8, line width=0.85pt] (t6.-90) to[bend right=60] (t7.90);
\draw[draw=tncolf28e2b, line width=0.85pt] (t8.-90) to[bend left=60] (t10.90);
\draw[draw=tncolf28e2b, line width=0.85pt] (t13.-90) to[bend right=60] (t15.90);
\draw[draw=tncol2a9d8f, line width=0.85pt] (t16.-90) to[bend left=60] (t19.90);
\draw[draw=tncol2a9d8f, line width=0.85pt] (t21.-90) to[bend right=60] (t22.90);
\node[anchor=center,font={\small},text=tncol6b7280] at (0.000,0.550) {$c_{1}$};
\node[anchor=center,font={\small},text=tncol6b7280] at (0.700,0.550) {$c_{2}$};
\node[anchor=center,font={\small},text=tncol6b7280] at (1.400,0.550) {$c_{3}$};
\node[anchor=center,font={\small},text=tncol6b7280] at (2.100,0.550) {$c_{4}$};
\node[anchor=center,font={\small},text=tncol6b7280] at (2.800,0.550) {$c_{5}$};
\node[anchor=center,font={\small},text=tncol6b7280] at (3.500,0.550) {$c_{6}$};
\node[anchor=center,font={\scriptsize},text=tncol6b7280] at (-0.550,0.000) {$1$};
\node[anchor=center,font={\scriptsize},text=tncol6b7280] at (-0.550,-0.700) {$2$};
\node[anchor=center,font={\scriptsize},text=tncol6b7280] at (-0.550,-1.400) {$3$};
\node[anchor=center,font={\scriptsize},text=tncol6b7280] at (-0.550,-2.100) {$4$};
\node[anchor=center,font={\small},text=tncol4c78a8] at (0.350,-2.550) {$e_{1}$};
\node[anchor=center,font={\small},text=tncolf28e2b] at (1.750,-2.550) {$e_{2}$};
\node[anchor=center,font={\small},text=tncol2a9d8f] at (3.150,-2.550) {$e_{3}$};
\begin{pgfonlayer}{background}
\node[draw=tncol4c78a8, line width=0.6pt, inner sep=2pt, fill=tncol4c78a8, fill opacity=0.22, rounded corners=3pt, fit=(t0)(t4)(t1)(t5)(t2)(t6)(t3)(t7)] {};
\node[draw=tncolf28e2b, line width=0.6pt, inner sep=2pt, fill=tncolf28e2b, fill opacity=0.22, rounded corners=3pt, fit=(t8)(t12)(t9)(t13)(t10)(t14)(t11)(t15)] {};
\node[draw=tncol2a9d8f, line width=0.6pt, inner sep=2pt, fill=tncol2a9d8f, fill opacity=0.22, rounded corners=3pt, fit=(t16)(t20)(t17)(t21)(t18)(t22)(t19)(t23)] {};
\end{pgfonlayer}
\end{tikzpicture}%
}}}%
\endgroup}%
}

%% file: figures/contingency_eg_exp.tex
\providecommand{\ContOne}[1][1]{%
\ensuremath{\begingroup\color{black}%
\vcenter{\hbox{%
\scalebox{#1}{%
\begin{tikzpicture}[baseline=(current bounding box.center)]
\node[circle, minimum size=4mm, inner sep=0pt, outer sep=0pt, draw=tncol22223b, line width=0.6pt, fill=tncolffadad] (t0) at (0,-0) {};
\node[circle, minimum size=1mm, inner sep=0pt, outer sep=0pt, draw=tncol22223b, line width=0.6pt, fill=tncole8edf4, text=black] (t1) at (0.7,-0) {};
\node[circle, minimum size=4mm, inner sep=0pt, outer sep=0pt, draw=tncol22223b, line width=0.6pt, fill=tncolfdffb6] (t2) at (1.4,-0) {};
\node[circle, minimum size=1mm, inner sep=0pt, outer sep=0pt, draw=tncol22223b, line width=0.6pt, fill=tncole8edf4, text=black] (t3) at (2.1,-0) {};
\node[circle, minimum size=4mm, inner sep=0pt, outer sep=0pt, draw=tncol22223b, line width=0.6pt, fill=tncol9bf6ff] (t4) at (2.8,-0) {};
\node[circle, minimum size=1mm, inner sep=0pt, outer sep=0pt, draw=tncol22223b, line width=0.6pt, fill=tncole8edf4, text=black] (t5) at (3.5,-0) {};
\node[circle, minimum size=4mm, inner sep=0pt, outer sep=0pt, draw=tncol22223b, line width=0.6pt, fill=tncolbdb2ff] (t6) at (4.2,-0) {};
\node[circle, minimum size=1mm, inner sep=0pt, outer sep=0pt, draw=tncol22223b, line width=0.6pt, fill=tncole8edf4, text=black] (t7) at (0,-0.7) {};
\node[circle, minimum size=4mm, inner sep=0pt, outer sep=0pt, draw=tncol22223b, line width=0.6pt, fill=tncolffd6a5] (t8) at (0.7,-0.7) {};
\node[circle, minimum size=4mm, inner sep=0pt, outer sep=0pt, draw=tncol22223b, line width=0.6pt, fill=tncolfdffb6] (t9) at (1.4,-0.7) {};
\node[circle, minimum size=1mm, inner sep=0pt, outer sep=0pt, draw=tncol22223b, line width=0.6pt, fill=tncole8edf4, text=black] (t10) at (2.1,-0.7) {};
\node[circle, minimum size=1mm, inner sep=0pt, outer sep=0pt, draw=tncol22223b, line width=0.6pt, fill=tncole8edf4, text=black] (t11) at (2.8,-0.7) {};
\node[circle, minimum size=4mm, inner sep=0pt, outer sep=0pt, draw=tncol22223b, line width=0.6pt, fill=tncola0c4ff] (t12) at (3.5,-0.7) {};
\node[circle, minimum size=4mm, inner sep=0pt, outer sep=0pt, draw=tncol22223b, line width=0.6pt, fill=tncolbdb2ff] (t13) at (4.2,-0.7) {};
\node[circle, minimum size=1mm, inner sep=0pt, outer sep=0pt, draw=tncol22223b, line width=0.6pt, fill=tncole8edf4, text=black] (t14) at (0,-1.4) {};
\node[circle, minimum size=4mm, inner sep=0pt, outer sep=0pt, draw=tncol22223b, line width=0.6pt, fill=tncolffd6a5] (t15) at (0.7,-1.4) {};
\node[circle, minimum size=4mm, inner sep=0pt, outer sep=0pt, draw=tncol22223b, line width=0.6pt, fill=tncolfdffb6] (t16) at (1.4,-1.4) {};
\node[circle, minimum size=1mm, inner sep=0pt, outer sep=0pt, draw=tncol22223b, line width=0.6pt, fill=tncole8edf4, text=black] (t17) at (2.1,-1.4) {};
\node[circle, minimum size=4mm, inner sep=0pt, outer sep=0pt, draw=tncol22223b, line width=0.6pt, fill=tncol9bf6ff] (t18) at (2.8,-1.4) {};
\node[circle, minimum size=1mm, inner sep=0pt, outer sep=0pt, draw=tncol22223b, line width=0.6pt, fill=tncole8edf4, text=black] (t19) at (3.5,-1.4) {};
\node[circle, minimum size=4mm, inner sep=0pt, outer sep=0pt, draw=tncol22223b, line width=0.6pt, fill=tncolbdb2ff] (t20) at (4.2,-1.4) {};
\node[circle, minimum size=4mm, inner sep=0pt, outer sep=0pt, draw=tncol22223b, line width=0.6pt, fill=tncolffadad] (t21) at (0,-2.1) {};
\node[circle, minimum size=1mm, inner sep=0pt, outer sep=0pt, draw=tncol22223b, line width=0.6pt, fill=tncole8edf4, text=black] (t22) at (0.7,-2.1) {};
\node[circle, minimum size=4mm, inner sep=0pt, outer sep=0pt, draw=tncol22223b, line width=0.6pt, fill=tncolfdffb6] (t23) at (1.4,-2.1) {};
\node[circle, minimum size=1mm, inner sep=0pt, outer sep=0pt, draw=tncol22223b, line width=0.6pt, fill=tncole8edf4, text=black] (t24) at (2.1,-2.1) {};
\node[circle, minimum size=1mm, inner sep=0pt, outer sep=0pt, draw=tncol22223b, line width=0.6pt, fill=tncole8edf4, text=black] (t25) at (2.8,-2.1) {};
\node[circle, minimum size=4mm, inner sep=0pt, outer sep=0pt, draw=tncol22223b, line width=0.6pt, fill=tncola0c4ff] (t26) at (3.5,-2.1) {};
\node[circle, minimum size=4mm, inner sep=0pt, outer sep=0pt, draw=tncol22223b, line width=0.6pt, fill=tncolbdb2ff] (t27) at (4.2,-2.1) {};
\end{tikzpicture}%
}}}%
\endgroup}%
}

%% file: figures/contingency_eg_exp_long.tex
\providecommand{\ContThree}[1][1]{%
\ensuremath{\begingroup\color{black}%
\vcenter{\hbox{%
\scalebox{#1}{%
\begin{tikzpicture}[baseline=(current bounding box.center)]
\node[circle, minimum size=1mm, inner sep=0pt, outer sep=0pt, draw=tncol22223b, line width=0.6pt, fill=tncole8edf4, text=black] (t0) at (0,-0) {};
\node[circle, minimum size=1mm, inner sep=0pt, outer sep=0pt, draw=tncol22223b, line width=0.6pt, fill=tncole8edf4, text=black] (t1) at (0.7,-0) {};
\node[circle, minimum size=4mm, inner sep=0pt, outer sep=0pt, draw=tncol22223b, line width=0.6pt, fill=tncolfdffb6] (t2) at (1.4,-0) {};
\node[circle, minimum size=1mm, inner sep=0pt, outer sep=0pt, draw=tncol22223b, line width=0.6pt, fill=tncole8edf4, text=black] (t3) at (2.1,-0) {};
\node[circle, minimum size=1mm, inner sep=0pt, outer sep=0pt, draw=tncol22223b, line width=0.6pt, fill=tncole8edf4, text=black] (t4) at (2.8,-0) {};
\node[circle, minimum size=4mm, inner sep=0pt, outer sep=0pt, draw=tncol22223b, line width=0.6pt, fill=tncola0c4ff] (t5) at (3.5,-0) {};
\node[circle, minimum size=4mm, inner sep=0pt, outer sep=0pt, draw=tncol22223b, line width=0.6pt, fill=tncolbdb2ff] (t6) at (4.2,-0) {};
\node[regular polygon, minimum size=0.5657cm, inner sep=0pt, outer sep=0pt, draw=tncol22223b, line width=0.6pt, fill=tncolffadad, regular polygon sides=4] (t7) at (4.9,-0) {};
\node[circle, minimum size=1mm, inner sep=0pt, outer sep=0pt, draw=tncol22223b, line width=0.6pt, fill=tncole8edf4, text=black] (t8) at (5.6,-0) {};
\node[circle, minimum size=1mm, inner sep=0pt, outer sep=0pt, draw=tncol22223b, line width=0.6pt, fill=tncole8edf4, text=black] (t9) at (6.3,-0) {};
\node[regular polygon, minimum size=0.5657cm, inner sep=0pt, outer sep=0pt, draw=tncol22223b, line width=0.6pt, fill=tncolcaffbf, regular polygon sides=4] (t10) at (7,-0) {};
\node[regular polygon, minimum size=0.5657cm, inner sep=0pt, outer sep=0pt, draw=tncol22223b, line width=0.6pt, fill=tncol9bf6ff, regular polygon sides=4] (t11) at (7.7,-0) {};
\node[circle, minimum size=1mm, inner sep=0pt, outer sep=0pt, draw=tncol22223b, line width=0.6pt, fill=tncole8edf4, text=black] (t12) at (8.4,-0) {};
\node[regular polygon, minimum size=0.5657cm, inner sep=0pt, outer sep=0pt, draw=tncol22223b, line width=0.6pt, fill=tncolbdb2ff, regular polygon sides=4] (t13) at (9.1,-0) {};
\node[circle, minimum size=1mm, inner sep=0pt, outer sep=0pt, draw=tncol22223b, line width=0.6pt, fill=tncole8edf4, text=black] (t14) at (0,-0.7) {};
\node[circle, minimum size=4mm, inner sep=0pt, outer sep=0pt, draw=tncol22223b, line width=0.6pt, fill=tncolffd6a5] (t15) at (0.7,-0.7) {};
\node[circle, minimum size=1mm, inner sep=0pt, outer sep=0pt, draw=tncol22223b, line width=0.6pt, fill=tncole8edf4, text=black] (t16) at (1.4,-0.7) {};
\node[circle, minimum size=4mm, inner sep=0pt, outer sep=0pt, draw=tncol22223b, line width=0.6pt, fill=tncolcaffbf] (t17) at (2.1,-0.7) {};
\node[circle, minimum size=4mm, inner sep=0pt, outer sep=0pt, draw=tncol22223b, line width=0.6pt, fill=tncol9bf6ff] (t18) at (2.8,-0.7) {};
\node[circle, minimum size=1mm, inner sep=0pt, outer sep=0pt, draw=tncol22223b, line width=0.6pt, fill=tncole8edf4, text=black] (t19) at (3.5,-0.7) {};
\node[circle, minimum size=1mm, inner sep=0pt, outer sep=0pt, draw=tncol22223b, line width=0.6pt, fill=tncole8edf4, text=black] (t20) at (4.2,-0.7) {};
\node[regular polygon, minimum size=0.5657cm, inner sep=0pt, outer sep=0pt, draw=tncol22223b, line width=0.6pt, fill=tncolffadad, regular polygon sides=4] (t21) at (4.9,-0.7) {};
\node[regular polygon, minimum size=0.5657cm, inner sep=0pt, outer sep=0pt, draw=tncol22223b, line width=0.6pt, fill=tncolffd6a5, regular polygon sides=4] (t22) at (5.6,-0.7) {};
\node[regular polygon, minimum size=0.5657cm, inner sep=0pt, outer sep=0pt, draw=tncol22223b, line width=0.6pt, fill=tncolfdffb6, regular polygon sides=4] (t23) at (6.3,-0.7) {};
\node[circle, minimum size=1mm, inner sep=0pt, outer sep=0pt, draw=tncol22223b, line width=0.6pt, fill=tncole8edf4, text=black] (t24) at (7,-0.7) {};
\node[circle, minimum size=1mm, inner sep=0pt, outer sep=0pt, draw=tncol22223b, line width=0.6pt, fill=tncole8edf4, text=black] (t25) at (7.7,-0.7) {};
\node[regular polygon, minimum size=0.5657cm, inner sep=0pt, outer sep=0pt, draw=tncol22223b, line width=0.6pt, fill=tncola0c4ff, regular polygon sides=4] (t26) at (8.4,-0.7) {};
\node[circle, minimum size=1mm, inner sep=0pt, outer sep=0pt, draw=tncol22223b, line width=0.6pt, fill=tncole8edf4, text=black] (t27) at (9.1,-0.7) {};
\node[circle, minimum size=4mm, inner sep=0pt, outer sep=0pt, draw=tncol22223b, line width=0.6pt, fill=tncolffadad] (t28) at (0,-1.4) {};
\node[circle, minimum size=1mm, inner sep=0pt, outer sep=0pt, draw=tncol22223b, line width=0.6pt, fill=tncole8edf4, text=black] (t29) at (0.7,-1.4) {};
\node[circle, minimum size=4mm, inner sep=0pt, outer sep=0pt, draw=tncol22223b, line width=0.6pt, fill=tncolfdffb6] (t30) at (1.4,-1.4) {};
\node[circle, minimum size=1mm, inner sep=0pt, outer sep=0pt, draw=tncol22223b, line width=0.6pt, fill=tncole8edf4, text=black] (t31) at (2.1,-1.4) {};
\node[circle, minimum size=4mm, inner sep=0pt, outer sep=0pt, draw=tncol22223b, line width=0.6pt, fill=tncol9bf6ff] (t32) at (2.8,-1.4) {};
\node[circle, minimum size=4mm, inner sep=0pt, outer sep=0pt, draw=tncol22223b, line width=0.6pt, fill=tncola0c4ff] (t33) at (3.5,-1.4) {};
\node[circle, minimum size=4mm, inner sep=0pt, outer sep=0pt, draw=tncol22223b, line width=0.6pt, fill=tncolbdb2ff] (t34) at (4.2,-1.4) {};
\node[circle, minimum size=1mm, inner sep=0pt, outer sep=0pt, draw=tncol22223b, line width=0.6pt, fill=tncole8edf4, text=black] (t35) at (4.9,-1.4) {};
\node[regular polygon, minimum size=0.5657cm, inner sep=0pt, outer sep=0pt, draw=tncol22223b, line width=0.6pt, fill=tncolffd6a5, regular polygon sides=4] (t36) at (5.6,-1.4) {};
\node[circle, minimum size=1mm, inner sep=0pt, outer sep=0pt, draw=tncol22223b, line width=0.6pt, fill=tncole8edf4, text=black] (t37) at (6.3,-1.4) {};
\node[regular polygon, minimum size=0.5657cm, inner sep=0pt, outer sep=0pt, draw=tncol22223b, line width=0.6pt, fill=tncolcaffbf, regular polygon sides=4] (t38) at (7,-1.4) {};
\node[circle, minimum size=1mm, inner sep=0pt, outer sep=0pt, draw=tncol22223b, line width=0.6pt, fill=tncole8edf4, text=black] (t39) at (7.7,-1.4) {};
\node[circle, minimum size=1mm, inner sep=0pt, outer sep=0pt, draw=tncol22223b, line width=0.6pt, fill=tncole8edf4, text=black] (t40) at (8.4,-1.4) {};
\node[circle, minimum size=1mm, inner sep=0pt, outer sep=0pt, draw=tncol22223b, line width=0.6pt, fill=tncole8edf4, text=black] (t41) at (9.1,-1.4) {};
\node[circle, minimum size=4mm, inner sep=0pt, outer sep=0pt, draw=tncol22223b, line width=0.6pt, fill=tncolffadad] (t42) at (0,-2.1) {};
\node[circle, minimum size=4mm, inner sep=0pt, outer sep=0pt, draw=tncol22223b, line width=0.6pt, fill=tncolffd6a5] (t43) at (0.7,-2.1) {};
\node[circle, minimum size=1mm, inner sep=0pt, outer sep=0pt, draw=tncol22223b, line width=0.6pt, fill=tncole8edf4, text=black] (t44) at (1.4,-2.1) {};
\node[circle, minimum size=4mm, inner sep=0pt, outer sep=0pt, draw=tncol22223b, line width=0.6pt, fill=tncolcaffbf] (t45) at (2.1,-2.1) {};
\node[circle, minimum size=1mm, inner sep=0pt, outer sep=0pt, draw=tncol22223b, line width=0.6pt, fill=tncole8edf4, text=black] (t46) at (2.8,-2.1) {};
\node[circle, minimum size=1mm, inner sep=0pt, outer sep=0pt, draw=tncol22223b, line width=0.6pt, fill=tncole8edf4, text=black] (t47) at (3.5,-2.1) {};
\node[circle, minimum size=1mm, inner sep=0pt, outer sep=0pt, draw=tncol22223b, line width=0.6pt, fill=tncole8edf4, text=black] (t48) at (4.2,-2.1) {};
\node[circle, minimum size=1mm, inner sep=0pt, outer sep=0pt, draw=tncol22223b, line width=0.6pt, fill=tncole8edf4, text=black] (t49) at (4.9,-2.1) {};
\node[circle, minimum size=1mm, inner sep=0pt, outer sep=0pt, draw=tncol22223b, line width=0.6pt, fill=tncole8edf4, text=black] (t50) at (5.6,-2.1) {};
\node[regular polygon, minimum size=0.5657cm, inner sep=0pt, outer sep=0pt, draw=tncol22223b, line width=0.6pt, fill=tncolfdffb6, regular polygon sides=4] (t51) at (6.3,-2.1) {};
\node[circle, minimum size=1mm, inner sep=0pt, outer sep=0pt, draw=tncol22223b, line width=0.6pt, fill=tncole8edf4, text=black] (t52) at (7,-2.1) {};
\node[regular polygon, minimum size=0.5657cm, inner sep=0pt, outer sep=0pt, draw=tncol22223b, line width=0.6pt, fill=tncol9bf6ff, regular polygon sides=4] (t53) at (7.7,-2.1) {};
\node[regular polygon, minimum size=0.5657cm, inner sep=0pt, outer sep=0pt, draw=tncol22223b, line width=0.6pt, fill=tncola0c4ff, regular polygon sides=4] (t54) at (8.4,-2.1) {};
\node[regular polygon, minimum size=0.5657cm, inner sep=0pt, outer sep=0pt, draw=tncol22223b, line width=0.6pt, fill=tncolbdb2ff, regular polygon sides=4] (t55) at (9.1,-2.1) {};
\end{tikzpicture}%
}}}%
\endgroup}%
}

%% file: figures/contingency_eg_consolid.tex
\providecommand{\ContTwo}[1][1]{%
\ensuremath{\begingroup\color{black}%
\vcenter{\hbox{%
\scalebox{#1}{%
\begin{tikzpicture}[baseline=(current bounding box.center)]
\node[circle, minimum size=4mm, inner sep=0pt, outer sep=0pt, draw=tncol22223b, line width=0.6pt, fill=tncolffadad] (t0) at (0,-0) {};
\node[circle, minimum size=4mm, inner sep=0pt, outer sep=0pt, draw=tncol22223b, line width=0.6pt, fill=tncolfdffb6] (t1) at (0.7,-0) {};
\node[circle, minimum size=4mm, inner sep=0pt, outer sep=0pt, draw=tncol22223b, line width=0.6pt, fill=tncol9bf6ff] (t2) at (1.4,-0) {};
\node[circle, minimum size=4mm, inner sep=0pt, outer sep=0pt, draw=tncol22223b, line width=0.6pt, fill=tncolbdb2ff] (t3) at (2.1,-0) {};
\node[circle, minimum size=4mm, inner sep=0pt, outer sep=0pt, draw=tncol22223b, line width=0.6pt, fill=tncolffd6a5] (t4) at (0,-0.7) {};
\node[circle, minimum size=4mm, inner sep=0pt, outer sep=0pt, draw=tncol22223b, line width=0.6pt, fill=tncolfdffb6] (t5) at (0.7,-0.7) {};
\node[circle, minimum size=4mm, inner sep=0pt, outer sep=0pt, draw=tncol22223b, line width=0.6pt, fill=tncola0c4ff] (t6) at (1.4,-0.7) {};
\node[circle, minimum size=4mm, inner sep=0pt, outer sep=0pt, draw=tncol22223b, line width=0.6pt, fill=tncolbdb2ff] (t7) at (2.1,-0.7) {};
\node[circle, minimum size=4mm, inner sep=0pt, outer sep=0pt, draw=tncol22223b, line width=0.6pt, fill=tncolffd6a5] (t8) at (0,-1.4) {};
\node[circle, minimum size=4mm, inner sep=0pt, outer sep=0pt, draw=tncol22223b, line width=0.6pt, fill=tncolfdffb6] (t9) at (0.7,-1.4) {};
\node[circle, minimum size=4mm, inner sep=0pt, outer sep=0pt, draw=tncol22223b, line width=0.6pt, fill=tncol9bf6ff] (t10) at (1.4,-1.4) {};
\node[circle, minimum size=4mm, inner sep=0pt, outer sep=0pt, draw=tncol22223b, line width=0.6pt, fill=tncolbdb2ff] (t11) at (2.1,-1.4) {};
\node[circle, minimum size=4mm, inner sep=0pt, outer sep=0pt, draw=tncol22223b, line width=0.6pt, fill=tncolffadad] (t12) at (0,-2.1) {};
\node[circle, minimum size=4mm, inner sep=0pt, outer sep=0pt, draw=tncol22223b, line width=0.6pt, fill=tncolfdffb6] (t13) at (0.7,-2.1) {};
\node[circle, minimum size=4mm, inner sep=0pt, outer sep=0pt, draw=tncol22223b, line width=0.6pt, fill=tncola0c4ff] (t14) at (1.4,-2.1) {};
\node[circle, minimum size=4mm, inner sep=0pt, outer sep=0pt, draw=tncol22223b, line width=0.6pt, fill=tncolbdb2ff] (t15) at (2.1,-2.1) {};
\draw[draw=tncol888888, line width=0.4pt, dashed] (t0.-90) to[bend left=60] (t12.90);
\draw[draw=tncol888888, line width=0.4pt, dashed] (t4.-90) to[bend left=60] (t8.90);
\draw[draw=tncol888888, line width=0.4pt, dashed] (t5.-90) to[bend right=60] (t13.90);
\draw[draw=tncol888888, line width=0.4pt, dashed] (t1.-90) to[bend left=60] (t9.90);
\draw[draw=tncol888888, line width=0.4pt, dashed] (t2.-90) to[bend left=60] (t10.90);
\draw[draw=tncol888888, line width=0.4pt, dashed] (t6.-90) to[bend right=60] (t14.90);
\draw[draw=tncol888888, line width=0.4pt, dashed] (t7.-90) to[bend left=60] (t11.90);
\draw[draw=tncol888888, line width=0.4pt, dashed] (t3.-90) to[bend right=60] (t15.90);
\end{tikzpicture}%
}}}%
\endgroup}%
}

%% file: figures/loe_operator.tex
\providecommand{\LOEop}[1][1]{%
\ensuremath{\begingroup\color{black}%
\vcenter{\hbox{%
\scalebox{#1}{%
\begin{tikzpicture}[baseline=(current bounding box.center)]
\node[inner sep=0pt, minimum size=0pt, outer sep=0pt] (t0) at (0,-0) {};
\node[inner sep=0pt, minimum size=0pt, outer sep=0pt] (t1) at (0,-0.7) {};
\node[inner sep=0pt, minimum size=0pt, outer sep=0pt] (t2) at (0,-1.4) {};
\node[inner sep=0pt, minimum size=0pt, outer sep=0pt] (t3) at (0,-2.1) {};
\node[inner sep=0pt, minimum size=0pt, outer sep=0pt] (t4) at (1.5,-0) {};
\node[inner sep=0pt, minimum size=0pt, outer sep=0pt] (t5) at (1.5,-0.7) {};
\node[inner sep=0pt, minimum size=0pt, outer sep=0pt] (t6) at (1.5,-1.4) {};
\node[inner sep=0pt, minimum size=0pt, outer sep=0pt] (t7) at (1.5,-2.1) {};
\draw[draw=tncolff1b6b, line width=0.6pt] ($(t0.east) + (270:-0.05)$) -- ($ ($(t0.east) + (270:-0.05)$) + (0:0.15) $) to[out=0, in=180] ($ ($(t5.west) + (270:-0.05)$) + (180:0.15) $) -- ($(t5.west) + (270:-0.05)$);
\draw[draw=tncolff1b6b, line width=0.6pt] ($(t1.east) + (270:-0.05)$) -- ($ ($(t1.east) + (270:-0.05)$) + (0:0.15) $) to[out=0, in=180] ($ ($(t4.west) + (270:-0.05)$) + (180:0.15) $) -- ($(t4.west) + (270:-0.05)$);
\draw[draw=tncolff1b6b, line width=0.6pt] ($(t2.east) + (270:-0.05)$) -- ($ ($(t2.east) + (270:-0.05)$) + (0:0.15) $) to[out=0, in=180] ($ ($(t7.west) + (270:-0.05)$) + (180:0.15) $) -- ($(t7.west) + (270:-0.05)$);
\draw[draw=tncolff1b6b, line width=0.6pt] ($(t3.east) + (270:-0.05)$) -- ($ ($(t3.east) + (270:-0.05)$) + (0:0.15) $) to[out=0, in=180] ($ ($(t6.west) + (270:-0.05)$) + (180:0.15) $) -- ($(t6.west) + (270:-0.05)$);
\draw[draw=tncol45caff, line width=0.6pt] ($(t0.east) + (270:0.05)$) -- ($ ($(t0.east) + (270:0.05)$) + (0:0.15) $) to[out=0, in=180] ($ ($(t7.west) + (270:0.05)$) + (180:0.15) $) -- ($(t7.west) + (270:0.05)$);
\draw[draw=tncol45caff, line width=0.6pt] ($(t3.east) + (270:0.05)$) -- ($ ($(t3.east) + (270:0.05)$) + (0:0.15) $) to[out=0, in=180] ($ ($(t4.west) + (270:0.05)$) + (180:0.15) $) -- ($(t4.west) + (270:0.05)$);
\draw[draw=tncol45caff, line width=0.6pt] ($(t1.east) + (270:0.05)$) -- ($ ($(t1.east) + (270:0.05)$) + (0:0.15) $) to[out=0, in=180] ($ ($(t6.west) + (270:0.05)$) + (180:0.15) $) -- ($(t6.west) + (270:0.05)$);
\draw[draw=tncol45caff, line width=0.6pt] ($(t2.east) + (270:0.05)$) -- ($ ($(t2.east) + (270:0.05)$) + (0:0.15) $) to[out=0, in=180] ($ ($(t5.west) + (270:0.05)$) + (180:0.15) $) -- ($(t5.west) + (270:0.05)$);
\draw[draw=tncolff1b6b, line width=0.6pt] ($(t0.east) + (270:-0.05)$) -- ($ ($(t0.east) + (270:-0.05)$) + (0:0.15) $) to[out=0, in=180] ($ ($(t5.west) + (270:-0.05)$) + (180:0.15) $) -- ($(t5.west) + (270:-0.05)$);
\draw[draw=tncolff1b6b, line width=0.6pt] ($(t1.east) + (270:-0.05)$) -- ($ ($(t1.east) + (270:-0.05)$) + (0:0.15) $) to[out=0, in=180] ($ ($(t4.west) + (270:-0.05)$) + (180:0.15) $) -- ($(t4.west) + (270:-0.05)$);
\draw[draw=tncolff1b6b, line width=0.6pt] ($(t2.east) + (270:-0.05)$) -- ($ ($(t2.east) + (270:-0.05)$) + (0:0.15) $) to[out=0, in=180] ($ ($(t7.west) + (270:-0.05)$) + (180:0.15) $) -- ($(t7.west) + (270:-0.05)$);
\draw[draw=tncolff1b6b, line width=0.6pt] ($(t3.east) + (270:-0.05)$) -- ($ ($(t3.east) + (270:-0.05)$) + (0:0.15) $) to[out=0, in=180] ($ ($(t6.west) + (270:-0.05)$) + (180:0.15) $) -- ($(t6.west) + (270:-0.05)$);
\draw[draw=tncol45caff, line width=0.6pt] ($(t0.east) + (270:0.05)$) -- ($ ($(t0.east) + (270:0.05)$) + (0:0.15) $) to[out=0, in=180] ($ ($(t7.west) + (270:0.05)$) + (180:0.15) $) -- ($(t7.west) + (270:0.05)$);
\draw[draw=tncol45caff, line width=0.6pt] ($(t3.east) + (270:0.05)$) -- ($ ($(t3.east) + (270:0.05)$) + (0:0.15) $) to[out=0, in=180] ($ ($(t4.west) + (270:0.05)$) + (180:0.15) $) -- ($(t4.west) + (270:0.05)$);
\draw[draw=tncol45caff, line width=0.6pt] ($(t1.east) + (270:0.05)$) -- ($ ($(t1.east) + (270:0.05)$) + (0:0.15) $) to[out=0, in=180] ($ ($(t6.west) + (270:0.05)$) + (180:0.15) $) -- ($(t6.west) + (270:0.05)$);
\draw[draw=tncol45caff, line width=0.6pt] ($(t2.east) + (270:0.05)$) -- ($ ($(t2.east) + (270:0.05)$) + (0:0.15) $) to[out=0, in=180] ($ ($(t5.west) + (270:0.05)$) + (180:0.15) $) -- ($(t5.west) + (270:0.05)$);
\draw[draw=tncolff1b6b, line width=0.6pt] ($(t0.east) + (270:-0.05)$) -- ($ ($(t0.east) + (270:-0.05)$) + (0:0.15) $) to[out=0, in=180] ($ ($(t5.west) + (270:-0.05)$) + (180:0.15) $) -- ($(t5.west) + (270:-0.05)$);
\draw[draw=tncolff1b6b, line width=0.6pt] ($(t1.east) + (270:-0.05)$) -- ($ ($(t1.east) + (270:-0.05)$) + (0:0.15) $) to[out=0, in=180] ($ ($(t4.west) + (270:-0.05)$) + (180:0.15) $) -- ($(t4.west) + (270:-0.05)$);
\draw[draw=tncolff1b6b, line width=0.6pt] ($(t2.east) + (270:-0.05)$) -- ($ ($(t2.east) + (270:-0.05)$) + (0:0.15) $) to[out=0, in=180] ($ ($(t7.west) + (270:-0.05)$) + (180:0.15) $) -- ($(t7.west) + (270:-0.05)$);
\draw[draw=tncolff1b6b, line width=0.6pt] ($(t3.east) + (270:-0.05)$) -- ($ ($(t3.east) + (270:-0.05)$) + (0:0.15) $) to[out=0, in=180] ($ ($(t6.west) + (270:-0.05)$) + (180:0.15) $) -- ($(t6.west) + (270:-0.05)$);
\draw[draw=tncol45caff, line width=0.6pt] ($(t0.east) + (270:0.05)$) -- ($ ($(t0.east) + (270:0.05)$) + (0:0.15) $) to[out=0, in=180] ($ ($(t7.west) + (270:0.05)$) + (180:0.15) $) -- ($(t7.west) + (270:0.05)$);
\draw[draw=tncol45caff, line width=0.6pt] ($(t3.east) + (270:0.05)$) -- ($ ($(t3.east) + (270:0.05)$) + (0:0.15) $) to[out=0, in=180] ($ ($(t4.west) + (270:0.05)$) + (180:0.15) $) -- ($(t4.west) + (270:0.05)$);
\draw[draw=tncol45caff, line width=0.6pt] ($(t1.east) + (270:0.05)$) -- ($ ($(t1.east) + (270:0.05)$) + (0:0.15) $) to[out=0, in=180] ($ ($(t6.west) + (270:0.05)$) + (180:0.15) $) -- ($(t6.west) + (270:0.05)$);
\draw[draw=tncol45caff, line width=0.6pt] ($(t2.east) + (270:0.05)$) -- ($ ($(t2.east) + (270:0.05)$) + (0:0.15) $) to[out=0, in=180] ($ ($(t5.west) + (270:0.05)$) + (180:0.15) $) -- ($(t5.west) + (270:0.05)$);
\draw[draw=tncolff1b6b, line width=0.6pt] ($(t0.east) + (270:-0.05)$) -- ($ ($(t0.east) + (270:-0.05)$) + (0:0.15) $) to[out=0, in=180] ($ ($(t5.west) + (270:-0.05)$) + (180:0.15) $) -- ($(t5.west) + (270:-0.05)$);
\draw[draw=tncolff1b6b, line width=0.6pt] ($(t1.east) + (270:-0.05)$) -- ($ ($(t1.east) + (270:-0.05)$) + (0:0.15) $) to[out=0, in=180] ($ ($(t4.west) + (270:-0.05)$) + (180:0.15) $) -- ($(t4.west) + (270:-0.05)$);
\draw[draw=tncolff1b6b, line width=0.6pt] ($(t2.east) + (270:-0.05)$) -- ($ ($(t2.east) + (270:-0.05)$) + (0:0.15) $) to[out=0, in=180] ($ ($(t7.west) + (270:-0.05)$) + (180:0.15) $) -- ($(t7.west) + (270:-0.05)$);
\draw[draw=tncolff1b6b, line width=0.6pt] ($(t3.east) + (270:-0.05)$) -- ($ ($(t3.east) + (270:-0.05)$) + (0:0.15) $) to[out=0, in=180] ($ ($(t6.west) + (270:-0.05)$) + (180:0.15) $) -- ($(t6.west) + (270:-0.05)$);
\draw[draw=tncol45caff, line width=0.6pt] ($(t0.east) + (270:0.05)$) -- ($ ($(t0.east) + (270:0.05)$) + (0:0.15) $) to[out=0, in=180] ($ ($(t7.west) + (270:0.05)$) + (180:0.15) $) -- ($(t7.west) + (270:0.05)$);
\draw[draw=tncol45caff, line width=0.6pt] ($(t3.east) + (270:0.05)$) -- ($ ($(t3.east) + (270:0.05)$) + (0:0.15) $) to[out=0, in=180] ($ ($(t4.west) + (270:0.05)$) + (180:0.15) $) -- ($(t4.west) + (270:0.05)$);
\draw[draw=tncol45caff, line width=0.6pt] ($(t1.east) + (270:0.05)$) -- ($ ($(t1.east) + (270:0.05)$) + (0:0.15) $) to[out=0, in=180] ($ ($(t6.west) + (270:0.05)$) + (180:0.15) $) -- ($(t6.west) + (270:0.05)$);
\draw[draw=tncol45caff, line width=0.6pt] ($(t2.east) + (270:0.05)$) -- ($ ($(t2.east) + (270:0.05)$) + (0:0.15) $) to[out=0, in=180] ($ ($(t5.west) + (270:0.05)$) + (180:0.15) $) -- ($(t5.west) + (270:0.05)$);
\end{tikzpicture}%
}}}%
\endgroup}%
}

%% file: figures/otoc_operator.tex
\providecommand{\OTOCop}[1][1]{%
\ensuremath{\begingroup\color{black}%
\vcenter{\hbox{%
\scalebox{#1}{%
\begin{tikzpicture}[baseline=(current bounding box.center)]
\node[inner sep=0pt, minimum size=0pt, outer sep=0pt] (t0) at (0,-0) {};
\node[inner sep=0pt, minimum size=0pt, outer sep=0pt] (t1) at (0,-0.5) {~~~\vdots};
\node[inner sep=0pt, minimum size=0pt, outer sep=0pt] (t2) at (0,-1) {};
\node[inner sep=0pt, minimum size=0pt, outer sep=0pt] (t3) at (0,-1.5) {};
\node[inner sep=0pt, minimum size=0pt, outer sep=0pt] (t4) at (1.2,-0) {};
\node[inner sep=0pt, minimum size=0pt, outer sep=0pt] (t5) at (1.2,-0.5) {};
\node[inner sep=0pt, minimum size=0pt, outer sep=0pt] (t6) at (1.2,-1) {\hspace{-0.25cm}\vdots};
\node[inner sep=0pt, minimum size=0pt, outer sep=0pt] (t7) at (1.2,-1.5) {};
\draw[draw=tncol33415c, line width=0.6pt] (t0.east) -- ($ (t0.east) + (0:0.15) $) to[out=0, in=180] ($ (t5.west) + (180:0.15) $) -- (t5.west);
\draw[draw=tncol33415c, line width=0.6pt] (t2.east) -- ($ (t2.east) + (0:0.15) $) to[out=0, in=180] ($ (t7.west) + (180:0.15) $) -- (t7.west);
\draw[draw=tncol33415c, line width=0.6pt] (t3.east) -- ($ (t3.east) + (0:0.15) $) to[out=0, in=180] ($ (t4.west) + (180:0.15) $) -- (t4.west);
\end{tikzpicture}%
}}}%
\endgroup}%
}

%% file: figures/full_loe_contract.tex
\providecommand{\LOEcontract}[1][1]{%
\ensuremath{\begingroup\color{black}%
\vcenter{\hbox{%
\scalebox{#1}{%
\begin{tikzpicture}[baseline=(current bounding box.center)]
\node[inner sep=0pt, minimum size=0pt, outer sep=0pt] (t0) at (0,-0) {};
\node[inner sep=0pt, minimum size=0pt, outer sep=0pt] (t1) at (0,-1.2) {};
\node[inner sep=0pt, minimum size=0pt, outer sep=0pt] (t2) at (0,-2.4) {};
\node[inner sep=0pt, minimum size=0pt, outer sep=0pt] (t3) at (0,-3.6) {};
\node[inner sep=0pt, minimum size=0pt, outer sep=0pt] (t4) at (2,-0) {};
\node[inner sep=0pt, minimum size=0pt, outer sep=0pt] (t5) at (2,-1.2) {};
\node[inner sep=0pt, minimum size=0pt, outer sep=0pt] (t6) at (2,-2.4) {};
\node[inner sep=0pt, minimum size=0pt, outer sep=0pt] (t7) at (2,-3.6) {};
\node[rectangle, minimum width=0.9cm, minimum height=4.5cm, inner sep=2pt, outer sep=0pt, rounded corners=1.6pt, draw=black, line width=0.6pt, fill=tncoldbd7d2, text=black] (t8) at (6,-1.8) {$p_2$};
\node[rectangle, minimum size=9mm, inner sep=2pt, outer sep=0pt, rounded corners=1.6pt, draw=black, line width=0.6pt, fill=tncol2281aa, text=black, minimum height=4mm, minimum width=15mm] (t9) at (4,-0) {$\gamma_{\bm{\mu}_B^{(1)}}$};
\node[rectangle, minimum size=9mm, inner sep=2pt, outer sep=0pt, rounded corners=1.6pt, draw=black, line width=0.6pt, fill=tncol2281aa, text=black, minimum height=4mm, minimum width=15mm] (t10) at (4,-1.2) {$\gamma_{\bm{\mu}_B^{(2)}}$};
\node[rectangle, minimum size=9mm, inner sep=2pt, outer sep=0pt, rounded corners=1.6pt, draw=black, line width=0.6pt, fill=tncol2281aa, text=black, minimum height=4mm, minimum width=15mm] (t11) at (4,-2.4) {$\gamma_{\bm{\mu}_B^{(3)}}$};
\node[rectangle, minimum size=9mm, inner sep=2pt, outer sep=0pt, rounded corners=1.6pt, draw=black, line width=0.6pt, fill=tncol2281aa, text=black, minimum height=4mm, minimum width=15mm] (t12) at (4,-3.6) {$\gamma_{\bm{\mu}_B^{(4)}}$};
\draw[draw=tncolff1b6b, line width=0.6pt] ($(t0.east) + (270:-0.05)$) -- ($ ($(t0.east) + (270:-0.05)$) + (0:0.15) $) to[out=0, in=180] ($ ($(t5.west) + (270:-0.05)$) + (180:0.15) $) -- ($(t5.west) + (270:-0.05)$);
\draw[draw=tncolff1b6b, line width=0.6pt] ($(t1.east) + (270:-0.05)$) -- ($ ($(t1.east) + (270:-0.05)$) + (0:0.15) $) to[out=0, in=180] ($ ($(t4.west) + (270:-0.05)$) + (180:0.15) $) -- ($(t4.west) + (270:-0.05)$);
\draw[draw=tncolff1b6b, line width=0.6pt] ($(t2.east) + (270:-0.05)$) -- ($ ($(t2.east) + (270:-0.05)$) + (0:0.15) $) to[out=0, in=180] ($ ($(t7.west) + (270:-0.05)$) + (180:0.15) $) -- ($(t7.west) + (270:-0.05)$);
\draw[draw=tncolff1b6b, line width=0.6pt] ($(t3.east) + (270:-0.05)$) -- ($ ($(t3.east) + (270:-0.05)$) + (0:0.15) $) to[out=0, in=180] ($ ($(t6.west) + (270:-0.05)$) + (180:0.15) $) -- ($(t6.west) + (270:-0.05)$);
\draw[draw=tncol45caff, line width=0.6pt] ($(t0.east) + (270:0.05)$) -- ($ ($(t0.east) + (270:0.05)$) + (0:0.15) $) to[out=0, in=180] ($ ($(t7.west) + (270:0.05)$) + (180:0.15) $) -- ($(t7.west) + (270:0.05)$);
\draw[draw=tncol45caff, line width=0.6pt] ($(t3.east) + (270:0.05)$) -- ($ ($(t3.east) + (270:0.05)$) + (0:0.15) $) to[out=0, in=180] ($ ($(t4.west) + (270:0.05)$) + (180:0.15) $) -- ($(t4.west) + (270:0.05)$);
\draw[draw=tncol45caff, line width=0.6pt] ($(t1.east) + (270:0.05)$) -- ($ ($(t1.east) + (270:0.05)$) + (0:0.15) $) to[out=0, in=180] ($ ($(t6.west) + (270:0.05)$) + (180:0.15) $) -- ($(t6.west) + (270:0.05)$);
\draw[draw=tncol45caff, line width=0.6pt] ($(t2.east) + (270:0.05)$) -- ($ ($(t2.east) + (270:0.05)$) + (0:0.15) $) to[out=0, in=180] ($ ($(t5.west) + (270:0.05)$) + (180:0.15) $) -- ($(t5.west) + (270:0.05)$);
\draw[draw=tncolff1b6b, line width=0.6pt] ($(t0.east) + (270:-0.05)$) -- ($ ($(t0.east) + (270:-0.05)$) + (0:0.15) $) to[out=0, in=180] ($ ($(t5.west) + (270:-0.05)$) + (180:0.15) $) -- ($(t5.west) + (270:-0.05)$);
\draw[draw=tncolff1b6b, line width=0.6pt] ($(t1.east) + (270:-0.05)$) -- ($ ($(t1.east) + (270:-0.05)$) + (0:0.15) $) to[out=0, in=180] ($ ($(t4.west) + (270:-0.05)$) + (180:0.15) $) -- ($(t4.west) + (270:-0.05)$);
\draw[draw=tncolff1b6b, line width=0.6pt] ($(t2.east) + (270:-0.05)$) -- ($ ($(t2.east) + (270:-0.05)$) + (0:0.15) $) to[out=0, in=180] ($ ($(t7.west) + (270:-0.05)$) + (180:0.15) $) -- ($(t7.west) + (270:-0.05)$);
\draw[draw=tncolff1b6b, line width=0.6pt] ($(t3.east) + (270:-0.05)$) -- ($ ($(t3.east) + (270:-0.05)$) + (0:0.15) $) to[out=0, in=180] ($ ($(t6.west) + (270:-0.05)$) + (180:0.15) $) -- ($(t6.west) + (270:-0.05)$);
\draw[draw=tncol45caff, line width=0.6pt] ($(t0.east) + (270:0.05)$) -- ($ ($(t0.east) + (270:0.05)$) + (0:0.15) $) to[out=0, in=180] ($ ($(t7.west) + (270:0.05)$) + (180:0.15) $) -- ($(t7.west) + (270:0.05)$);
\draw[draw=tncol45caff, line width=0.6pt] ($(t3.east) + (270:0.05)$) -- ($ ($(t3.east) + (270:0.05)$) + (0:0.15) $) to[out=0, in=180] ($ ($(t4.west) + (270:0.05)$) + (180:0.15) $) -- ($(t4.west) + (270:0.05)$);
\draw[draw=tncol45caff, line width=0.6pt] ($(t1.east) + (270:0.05)$) -- ($ ($(t1.east) + (270:0.05)$) + (0:0.15) $) to[out=0, in=180] ($ ($(t6.west) + (270:0.05)$) + (180:0.15) $) -- ($(t6.west) + (270:0.05)$);
\draw[draw=tncol45caff, line width=0.6pt] ($(t2.east) + (270:0.05)$) -- ($ ($(t2.east) + (270:0.05)$) + (0:0.15) $) to[out=0, in=180] ($ ($(t5.west) + (270:0.05)$) + (180:0.15) $) -- ($(t5.west) + (270:0.05)$);
\draw[draw=tncolff1b6b, line width=0.6pt] ($(t0.east) + (270:-0.05)$) -- ($ ($(t0.east) + (270:-0.05)$) + (0:0.15) $) to[out=0, in=180] ($ ($(t5.west) + (270:-0.05)$) + (180:0.15) $) -- ($(t5.west) + (270:-0.05)$);
\draw[draw=tncolff1b6b, line width=0.6pt] ($(t1.east) + (270:-0.05)$) -- ($ ($(t1.east) + (270:-0.05)$) + (0:0.15) $) to[out=0, in=180] ($ ($(t4.west) + (270:-0.05)$) + (180:0.15) $) -- ($(t4.west) + (270:-0.05)$);
\draw[draw=tncolff1b6b, line width=0.6pt] ($(t2.east) + (270:-0.05)$) -- ($ ($(t2.east) + (270:-0.05)$) + (0:0.15) $) to[out=0, in=180] ($ ($(t7.west) + (270:-0.05)$) + (180:0.15) $) -- ($(t7.west) + (270:-0.05)$);
\draw[draw=tncolff1b6b, line width=0.6pt] ($(t3.east) + (270:-0.05)$) -- ($ ($(t3.east) + (270:-0.05)$) + (0:0.15) $) to[out=0, in=180] ($ ($(t6.west) + (270:-0.05)$) + (180:0.15) $) -- ($(t6.west) + (270:-0.05)$);
\draw[draw=tncol45caff, line width=0.6pt] ($(t0.east) + (270:0.05)$) -- ($ ($(t0.east) + (270:0.05)$) + (0:0.15) $) to[out=0, in=180] ($ ($(t7.west) + (270:0.05)$) + (180:0.15) $) -- ($(t7.west) + (270:0.05)$);
\draw[draw=tncol45caff, line width=0.6pt] ($(t3.east) + (270:0.05)$) -- ($ ($(t3.east) + (270:0.05)$) + (0:0.15) $) to[out=0, in=180] ($ ($(t4.west) + (270:0.05)$) + (180:0.15) $) -- ($(t4.west) + (270:0.05)$);
\draw[draw=tncol45caff, line width=0.6pt] ($(t1.east) + (270:0.05)$) -- ($ ($(t1.east) + (270:0.05)$) + (0:0.15) $) to[out=0, in=180] ($ ($(t6.west) + (270:0.05)$) + (180:0.15) $) -- ($(t6.west) + (270:0.05)$);
\draw[draw=tncol45caff, line width=0.6pt] ($(t2.east) + (270:0.05)$) -- ($ ($(t2.east) + (270:0.05)$) + (0:0.15) $) to[out=0, in=180] ($ ($(t5.west) + (270:0.05)$) + (180:0.15) $) -- ($(t5.west) + (270:0.05)$);
\draw[draw=tncolff1b6b, line width=0.6pt] ($(t0.east) + (270:-0.05)$) -- ($ ($(t0.east) + (270:-0.05)$) + (0:0.15) $) to[out=0, in=180] ($ ($(t5.west) + (270:-0.05)$) + (180:0.15) $) -- ($(t5.west) + (270:-0.05)$);
\draw[draw=tncolff1b6b, line width=0.6pt] ($(t1.east) + (270:-0.05)$) -- ($ ($(t1.east) + (270:-0.05)$) + (0:0.15) $) to[out=0, in=180] ($ ($(t4.west) + (270:-0.05)$) + (180:0.15) $) -- ($(t4.west) + (270:-0.05)$);
\draw[draw=tncolff1b6b, line width=0.6pt] ($(t2.east) + (270:-0.05)$) -- ($ ($(t2.east) + (270:-0.05)$) + (0:0.15) $) to[out=0, in=180] ($ ($(t7.west) + (270:-0.05)$) + (180:0.15) $) -- ($(t7.west) + (270:-0.05)$);
\draw[draw=tncolff1b6b, line width=0.6pt] ($(t3.east) + (270:-0.05)$) -- ($ ($(t3.east) + (270:-0.05)$) + (0:0.15) $) to[out=0, in=180] ($ ($(t6.west) + (270:-0.05)$) + (180:0.15) $) -- ($(t6.west) + (270:-0.05)$);
\draw[draw=tncol45caff, line width=0.6pt] ($(t0.east) + (270:0.05)$) -- ($ ($(t0.east) + (270:0.05)$) + (0:0.15) $) to[out=0, in=180] ($ ($(t7.west) + (270:0.05)$) + (180:0.15) $) -- ($(t7.west) + (270:0.05)$);
\draw[draw=tncol45caff, line width=0.6pt] ($(t3.east) + (270:0.05)$) -- ($ ($(t3.east) + (270:0.05)$) + (0:0.15) $) to[out=0, in=180] ($ ($(t4.west) + (270:0.05)$) + (180:0.15) $) -- ($(t4.west) + (270:0.05)$);
\draw[draw=tncol45caff, line width=0.6pt] ($(t1.east) + (270:0.05)$) -- ($ ($(t1.east) + (270:0.05)$) + (0:0.15) $) to[out=0, in=180] ($ ($(t6.west) + (270:0.05)$) + (180:0.15) $) -- ($(t6.west) + (270:0.05)$);
\draw[draw=tncol45caff, line width=0.6pt] ($(t2.east) + (270:0.05)$) -- ($ ($(t2.east) + (270:0.05)$) + (0:0.15) $) to[out=0, in=180] ($ ($(t5.west) + (270:0.05)$) + (180:0.15) $) -- ($(t5.west) + (270:0.05)$);
\draw[draw=tncol808080, line width=0.9pt] (t9.south) -- ($ (t9.south) + (270:0.1) $) to[out=270, in=180] ($ ($(t8.west) + (270:-1.8)$) + (180:0.1) $) -- ($(t8.west) + (270:-1.8)$);
\draw[draw=tncol808080, line width=0.9pt] (t10.north) -- ($ (t10.north) + (90:0.1) $) to[out=90, in=180] ($ ($(t8.west) + (270:-0.6)$) + (180:0.1) $) -- ($(t8.west) + (270:-0.6)$);
\draw[draw=tncol808080, line width=0.9pt] (t11.south) -- ($ (t11.south) + (270:0.1) $) to[out=270, in=180] ($ ($(t8.west) + (270:0.6)$) + (180:0.1) $) -- ($(t8.west) + (270:0.6)$);
\draw[draw=tncol808080, line width=0.9pt] (t12.north) -- ($ (t12.north) + (90:0.1) $) to[out=90, in=180] ($ ($(t8.west) + (270:1.8)$) + (180:0.1) $) -- ($(t8.west) + (270:1.8)$);
\draw[draw=tncolff1b6b, line width=0.6pt, rounded corners=0.2cm] ($(t0.180) + (270:-0.05)$) -- ($ ($(t0.180) + (270:-0.05)$) + (180:0.3) $) -- ($ ($(t0.180) + (270:-0.05)$) + (180:0.3) + (90:0.5) $) -- ($ ($(t9.east) + (270:-0.05)$) + (0:0.3) + (90:0.5) $) -- ($ ($(t9.east) + (270:-0.05)$) + (0:0.3) $) -- ($(t9.east) + (270:-0.05)$);
\draw[draw=tncolff1b6b, line width=0.6pt, rounded corners=0.2cm] ($(t1.180) + (270:-0.05)$) -- ($ ($(t1.180) + (270:-0.05)$) + (180:0.3) $) -- ($ ($(t1.180) + (270:-0.05)$) + (180:0.3) + (90:0.5) $) -- ($ ($(t10.east) + (270:-0.05)$) + (0:0.3) + (90:0.5) $) -- ($ ($(t10.east) + (270:-0.05)$) + (0:0.3) $) -- ($(t10.east) + (270:-0.05)$);
\draw[draw=tncolff1b6b, line width=0.6pt, rounded corners=0.2cm] ($(t2.180) + (270:-0.05)$) -- ($ ($(t2.180) + (270:-0.05)$) + (180:0.3) $) -- ($ ($(t2.180) + (270:-0.05)$) + (180:0.3) + (90:0.5) $) -- ($ ($(t11.east) + (270:-0.05)$) + (0:0.3) + (90:0.5) $) -- ($ ($(t11.east) + (270:-0.05)$) + (0:0.3) $) -- ($(t11.east) + (270:-0.05)$);
\draw[draw=tncolff1b6b, line width=0.6pt, rounded corners=0.2cm] ($(t3.180) + (270:-0.05)$) -- ($ ($(t3.180) + (270:-0.05)$) + (180:0.3) $) -- ($ ($(t3.180) + (270:-0.05)$) + (180:0.3) + (90:0.5) $) -- ($ ($(t12.east) + (270:-0.05)$) + (0:0.3) + (90:0.5) $) -- ($ ($(t12.east) + (270:-0.05)$) + (0:0.3) $) -- ($(t12.east) + (270:-0.05)$);
\draw[draw=tncol45caff, line width=0.6pt, rounded corners=0.2cm] ($(t0.180) + (270:0.05)$) -- ($ ($(t0.180) + (270:0.05)$) + (180:0.3) $) -- ($ ($(t0.180) + (270:0.05)$) + (180:0.3) + (90:0.5) $) -- ($ ($(t9.east) + (270:0.05)$) + (0:0.3) + (90:0.5) $) -- ($ ($(t9.east) + (270:0.05)$) + (0:0.3) $) -- ($(t9.east) + (270:0.05)$);
\draw[draw=tncol45caff, line width=0.6pt, rounded corners=0.2cm] ($(t1.180) + (270:0.05)$) -- ($ ($(t1.180) + (270:0.05)$) + (180:0.3) $) -- ($ ($(t1.180) + (270:0.05)$) + (180:0.3) + (90:0.5) $) -- ($ ($(t10.east) + (270:0.05)$) + (0:0.3) + (90:0.5) $) -- ($ ($(t10.east) + (270:0.05)$) + (0:0.3) $) -- ($(t10.east) + (270:0.05)$);
\draw[draw=tncol45caff, line width=0.6pt, rounded corners=0.2cm] ($(t2.180) + (270:0.05)$) -- ($ ($(t2.180) + (270:0.05)$) + (180:0.3) $) -- ($ ($(t2.180) + (270:0.05)$) + (180:0.3) + (90:0.5) $) -- ($ ($(t11.east) + (270:0.05)$) + (0:0.3) + (90:0.5) $) -- ($ ($(t11.east) + (270:0.05)$) + (0:0.3) $) -- ($(t11.east) + (270:0.05)$);
\draw[draw=tncol45caff, line width=0.6pt, rounded corners=0.2cm] ($(t3.180) + (270:0.05)$) -- ($ ($(t3.180) + (270:0.05)$) + (180:0.3) $) -- ($ ($(t3.180) + (270:0.05)$) + (180:0.3) + (90:0.5) $) -- ($ ($(t12.east) + (270:0.05)$) + (0:0.3) + (90:0.5) $) -- ($ ($(t12.east) + (270:0.05)$) + (0:0.3) $) -- ($(t12.east) + (270:0.05)$);
\draw[draw=tncolff1b6b, line width=0.6pt] ($(t4.0) + (270:-0.05)$) -- ($(t9.west) + (270:-0.05)$);
\draw[draw=tncolff1b6b, line width=0.6pt] ($(t5.0) + (270:-0.05)$) -- ($(t10.west) + (270:-0.05)$);
\draw[draw=tncolff1b6b, line width=0.6pt] ($(t6.0) + (270:-0.05)$) -- ($(t11.west) + (270:-0.05)$);
\draw[draw=tncolff1b6b, line width=0.6pt] ($(t7.0) + (270:-0.05)$) -- ($(t12.west) + (270:-0.05)$);
\draw[draw=tncol45caff, line width=0.6pt] ($(t4.0) + (270:0.05)$) -- ($(t9.west) + (270:0.05)$);
\draw[draw=tncol45caff, line width=0.6pt] ($(t5.0) + (270:0.05)$) -- ($(t10.west) + (270:0.05)$);
\draw[draw=tncol45caff, line width=0.6pt] ($(t6.0) + (270:0.05)$) -- ($(t11.west) + (270:0.05)$);
\draw[draw=tncol45caff, line width=0.6pt] ($(t7.0) + (270:0.05)$) -- ($(t12.west) + (270:0.05)$);
\node[draw=tncol888888, line width=0.6pt, inner sep=3pt, dashed, rounded corners=2pt, fit=(t0)(t1)(t2)(t3)(t4)(t5)(t6)(t7)] {};
\end{tikzpicture}%
}}}%
\endgroup}%
}

%% file: figures/loe_pauli_diagram_expanded.tex
\providecommand{\LOEPaulisExp}[1][1]{%
\ensuremath{\begingroup\color{black}%
\vcenter{\hbox{%
\scalebox{#1}{%
\begin{tikzpicture}[baseline=(current bounding box.center)]
\node[inner sep=0pt, minimum size=0pt, outer sep=0pt] (t0) at (0,-0) {};
\node[inner sep=0pt, minimum size=0pt, outer sep=0pt] (t1) at (0,-0.7) {};
\node[inner sep=0pt, minimum size=0pt, outer sep=0pt] (t2) at (0,-1.4) {};
\node[inner sep=0pt, minimum size=0pt, outer sep=0pt] (t3) at (0,-2.1) {};
\node[inner sep=0pt, minimum size=0pt, outer sep=0pt] (t4) at (0,-2.8) {};
\node[inner sep=0pt, minimum size=0pt, outer sep=0pt] (t5) at (0,-3.5) {};
\node[inner sep=0pt, minimum size=0pt, outer sep=0pt] (t6) at (0,-4.2) {};
\node[inner sep=0pt, minimum size=0pt, outer sep=0pt] (t7) at (0,-4.9) {};
\node[inner sep=0pt, minimum size=0pt, outer sep=0pt] (t8) at (2,-0) {};
\node[inner sep=0pt, minimum size=0pt, outer sep=0pt] (t9) at (2,-0.7) {};
\node[inner sep=0pt, minimum size=0pt, outer sep=0pt] (t10) at (2,-1.4) {};
\node[inner sep=0pt, minimum size=0pt, outer sep=0pt] (t11) at (2,-2.1) {};
\node[inner sep=0pt, minimum size=0pt, outer sep=0pt] (t12) at (2,-2.8) {};
\node[inner sep=0pt, minimum size=0pt, outer sep=0pt] (t13) at (2,-3.5) {};
\node[inner sep=0pt, minimum size=0pt, outer sep=0pt] (t14) at (2,-4.2) {};
\node[inner sep=0pt, minimum size=0pt, outer sep=0pt] (t15) at (2,-4.9) {};
\node[rectangle, minimum width=5mm, minimum height=5mm, inner sep=2pt, outer sep=0pt, rounded corners=1.6pt, draw=tncol33415c, line width=0.6pt, fill=tncolb4bd9b, text=black] (t16) at (4,-0) {$P_{\mu_1}^{(A)}$};
\node[rectangle, minimum width=5mm, minimum height=5mm, inner sep=2pt, outer sep=0pt, rounded corners=1.6pt, draw=tncol33415c, line width=0.6pt, fill=tncolfdba77, text=black] (t17) at (4,-1.4) {$P_{\mu_2}^{(A)}$};
\node[rectangle, minimum width=5mm, minimum height=5mm, inner sep=2pt, outer sep=0pt, rounded corners=1.6pt, draw=tncol33415c, line width=0.6pt, fill=tncol81bdc3, text=black] (t18) at (4,-2.8) {$P_{\mu_3}^{(A)}$};
\node[rectangle, minimum width=5mm, minimum height=5mm, inner sep=2pt, outer sep=0pt, rounded corners=1.6pt, draw=tncol33415c, line width=0.6pt, fill=tncolf9d6d3, text=black] (t19) at (4,-4.2) {$P_{\mu_4}^{(A)}$};
\node[rectangle, minimum width=5mm, minimum height=5mm, inner sep=2pt, outer sep=0pt, rounded corners=1.6pt, draw=tncol33415c, line width=0.6pt, fill=tncolbc455a, text=black] (t20) at (4,-0.7) {$P_{\mu_1}^{(\bar{A})}$};
\node[rectangle, minimum width=5mm, minimum height=5mm, inner sep=2pt, outer sep=0pt, rounded corners=1.6pt, draw=tncol33415c, line width=0.6pt, fill=tncolf6cf98, text=black] (t21) at (4,-2.1) {$P_{\mu_2}^{(\bar{A})}$};
\node[rectangle, minimum width=5mm, minimum height=5mm, inner sep=2pt, outer sep=0pt, rounded corners=1.6pt, draw=tncol33415c, line width=0.6pt, fill=tncolfdf8ec, text=black] (t22) at (4,-3.5) {$P_{\mu_3}^{(\bar{A})}$};
\node[rectangle, minimum width=5mm, minimum height=5mm, inner sep=2pt, outer sep=0pt, rounded corners=1.6pt, draw=tncol33415c, line width=0.6pt, fill=tncolccd5c3, text=black] (t23) at (4,-4.9) {$P_{\mu_4}^{(\bar{A})}$};
\draw[draw=tncolff1b6b, line width=0.6pt] (t0.east) -- ($ (t0.east) + (0:0.15) $) to[out=0, in=180] ($ (t10.west) + (180:0.15) $) -- (t10.west);
\draw[draw=tncolff1b6b, line width=0.6pt] (t2.east) -- ($ (t2.east) + (0:0.15) $) to[out=0, in=180] ($ (t8.west) + (180:0.15) $) -- (t8.west);
\draw[draw=tncolff1b6b, line width=0.6pt] (t4.east) -- ($ (t4.east) + (0:0.15) $) to[out=0, in=180] ($ (t14.west) + (180:0.15) $) -- (t14.west);
\draw[draw=tncolff1b6b, line width=0.6pt] (t6.east) -- ($ (t6.east) + (0:0.15) $) to[out=0, in=180] ($ (t12.west) + (180:0.15) $) -- (t12.west);
\draw[draw=tncol45caff, line width=0.6pt] (t1.east) -- ($ (t1.east) + (0:0.15) $) to[out=0, in=180] ($ (t15.west) + (180:0.15) $) -- (t15.west);
\draw[draw=tncol45caff, line width=0.6pt] (t7.east) -- ($ (t7.east) + (0:0.15) $) to[out=0, in=180] ($ (t9.west) + (180:0.15) $) -- (t9.west);
\draw[draw=tncol45caff, line width=0.6pt] (t3.east) -- ($ (t3.east) + (0:0.15) $) to[out=0, in=180] ($ (t13.west) + (180:0.15) $) -- (t13.west);
\draw[draw=tncol45caff, line width=0.6pt] (t5.east) -- ($ (t5.east) + (0:0.15) $) to[out=0, in=180] ($ (t11.west) + (180:0.15) $) -- (t11.west);
\draw[draw=tncolff1b6b, line width=0.6pt] (t8.0) -- (t16.west);
\draw[draw=tncol45caff, line width=0.6pt] (t9.0) -- (t20.west);
\draw[draw=tncolff1b6b, line width=0.6pt, rounded corners=0.2cm] (t0.180) -- ($ (t0.180) + (180:0.3) $) -- ($ (t0.180) + (180:0.3) + (90:0.5) $) -- ($ (t16.east) + (0:0.3) + (90:0.5) $) -- ($ (t16.east) + (0:0.3) $) -- (t16.east);
\draw[draw=tncol45caff, line width=0.6pt, rounded corners=0.2cm] (t1.180) -- ($ (t1.180) + (180:0.3) $) -- ($ (t1.180) + (180:0.3) + (90:0.5) $) -- ($ (t20.east) + (0:0.3) + (90:0.5) $) -- ($ (t20.east) + (0:0.3) $) -- (t20.east);
\draw[draw=tncolff1b6b, line width=0.6pt] (t10.0) -- (t17.west);
\draw[draw=tncol45caff, line width=0.6pt] (t11.0) -- (t21.west);
\draw[draw=tncolff1b6b, line width=0.6pt, rounded corners=0.2cm] (t2.180) -- ($ (t2.180) + (180:0.3) $) -- ($ (t2.180) + (180:0.3) + (90:0.5) $) -- ($ (t17.east) + (0:0.3) + (90:0.5) $) -- ($ (t17.east) + (0:0.3) $) -- (t17.east);
\draw[draw=tncol45caff, line width=0.6pt, rounded corners=0.2cm] (t3.180) -- ($ (t3.180) + (180:0.3) $) -- ($ (t3.180) + (180:0.3) + (90:0.5) $) -- ($ (t21.east) + (0:0.3) + (90:0.5) $) -- ($ (t21.east) + (0:0.3) $) -- (t21.east);
\draw[draw=tncolff1b6b, line width=0.6pt] (t12.0) -- (t18.west);
\draw[draw=tncol45caff, line width=0.6pt] (t13.0) -- (t22.west);
\draw[draw=tncolff1b6b, line width=0.6pt, rounded corners=0.2cm] (t4.180) -- ($ (t4.180) + (180:0.3) $) -- ($ (t4.180) + (180:0.3) + (90:0.5) $) -- ($ (t18.east) + (0:0.3) + (90:0.5) $) -- ($ (t18.east) + (0:0.3) $) -- (t18.east);
\draw[draw=tncol45caff, line width=0.6pt, rounded corners=0.2cm] (t5.180) -- ($ (t5.180) + (180:0.3) $) -- ($ (t5.180) + (180:0.3) + (90:0.5) $) -- ($ (t22.east) + (0:0.3) + (90:0.5) $) -- ($ (t22.east) + (0:0.3) $) -- (t22.east);
\draw[draw=tncolff1b6b, line width=0.6pt] (t14.0) -- (t19.west);
\draw[draw=tncol45caff, line width=0.6pt] (t15.0) -- (t23.west);
\draw[draw=tncolff1b6b, line width=0.6pt, rounded corners=0.2cm] (t6.180) -- ($ (t6.180) + (180:0.3) $) -- ($ (t6.180) + (180:0.3) + (90:0.5) $) -- ($ (t19.east) + (0:0.3) + (90:0.5) $) -- ($ (t19.east) + (0:0.3) $) -- (t19.east);
\draw[draw=tncol45caff, line width=0.6pt, rounded corners=0.2cm] (t7.180) -- ($ (t7.180) + (180:0.3) $) -- ($ (t7.180) + (180:0.3) + (90:0.5) $) -- ($ (t23.east) + (0:0.3) + (90:0.5) $) -- ($ (t23.east) + (0:0.3) $) -- (t23.east);
\end{tikzpicture}%
}}}%
\endgroup}%
}

%% file: figures/loe_delta_relations.tex
\providecommand{\LOEPaulisDelta}[1][1]{%
\ensuremath{\begingroup\color{black}%
\vcenter{\hbox{%
\scalebox{#1}{%
\begin{tikzpicture}[baseline=(current bounding box.center)]
\node[rectangle, minimum width=5mm, minimum height=5mm, inner sep=2pt, outer sep=0pt, rounded corners=1.6pt, draw=tncol33415c, line width=0.6pt, fill=tncolb4bd9b, text=black] (t0) at (0,-0) {$P_{\mu_1}^{(A)}$};
\node[rectangle, minimum width=5mm, minimum height=5mm, inner sep=2pt, outer sep=0pt, rounded corners=1.6pt, draw=tncol33415c, line width=0.6pt, fill=tncolfdba77, text=black] (t1) at (1.25,-0) {$P_{\mu_2}^{(A)}$};
\node[rectangle, minimum width=5mm, minimum height=5mm, inner sep=2pt, outer sep=0pt, rounded corners=1.6pt, draw=tncol33415c, line width=0.6pt, fill=tncol81bdc3, text=black] (t2) at (0,-1) {$P_{\mu_3}^{(A)}$};
\node[rectangle, minimum width=5mm, minimum height=5mm, inner sep=2pt, outer sep=0pt, rounded corners=1.6pt, draw=tncol33415c, line width=0.6pt, fill=tncolf9d6d3, text=black] (t3) at (1.25,-1) {$P_{\mu_4}^{(A)}$};
\node[rectangle, minimum width=5mm, minimum height=5mm, inner sep=2pt, outer sep=0pt, rounded corners=1.6pt, draw=tncol33415c, line width=0.6pt, fill=tncolbc455a, text=black] (t4) at (0,-2) {$P_{\mu_1}^{(\bar{A})}$};
\node[rectangle, minimum width=5mm, minimum height=5mm, inner sep=2pt, outer sep=0pt, rounded corners=1.6pt, draw=tncol33415c, line width=0.6pt, fill=tncolf6cf98, text=black] (t5) at (0,-3) {$P_{\mu_2}^{(\bar{A})}$};
\node[rectangle, minimum width=5mm, minimum height=5mm, inner sep=2pt, outer sep=0pt, rounded corners=1.6pt, draw=tncol33415c, line width=0.6pt, fill=tncolfdf8ec, text=black] (t6) at (1.25,-3) {$P_{\mu_3}^{(\bar{A})}$};
\node[rectangle, minimum width=5mm, minimum height=5mm, inner sep=2pt, outer sep=0pt, rounded corners=1.6pt, draw=tncol33415c, line width=0.6pt, fill=tncolccd5c3, text=black] (t7) at (1.25,-2) {$P_{\mu_4}^{(\bar{A})}$};
\draw[draw=tncolff1b6b, line width=0.6pt, rounded corners=0.2cm] (t0.west) -- ($ (t0.west) + (180:0.3) $) -- ($ (t0.west) + (180:0.3) + (90:0.5) $) -- ($ (t1.east) + (0:0.3) + (90:0.5) $) -- ($ (t1.east) + (0:0.3) $) -- (t1.east);
\draw[draw=tncolff1b6b, line width=0.6pt] (t0.east) -- (t1.west);
\draw[draw=tncolff1b6b, line width=0.6pt, rounded corners=0.2cm] (t2.west) -- ($ (t2.west) + (180:0.3) $) -- ($ (t2.west) + (180:0.3) + (90:0.5) $) -- ($ (t3.east) + (0:0.3) + (90:0.5) $) -- ($ (t3.east) + (0:0.3) $) -- (t3.east);
\draw[draw=tncolff1b6b, line width=0.6pt] (t2.east) -- (t3.west);
\draw[draw=tncol45caff, line width=0.6pt, rounded corners=0.2cm] (t4.west) -- ($ (t4.west) + (180:0.3) $) -- ($ (t4.west) + (180:0.3) + (90:0.5) $) -- ($ (t7.east) + (0:0.3) + (90:0.5) $) -- ($ (t7.east) + (0:0.3) $) -- (t7.east);
\draw[draw=tncol45caff, line width=0.6pt] (t4.east) -- (t7.west);
\draw[draw=tncol45caff, line width=0.6pt, rounded corners=0.2cm] (t5.west) -- ($ (t5.west) + (180:0.3) $) -- ($ (t5.west) + (180:0.3) + (90:0.5) $) -- ($ (t6.east) + (0:0.3) + (90:0.5) $) -- ($ (t6.east) + (0:0.3) $) -- (t6.east);
\draw[draw=tncol45caff, line width=0.6pt] (t5.east) -- (t6.west);
\end{tikzpicture}%
}}}%
\endgroup}%
}

%% file: figures/loe_delta_relations_col.tex
\providecommand{\LOEPaulisDeltaSimp}[1][1]{%
\ensuremath{\begingroup\color{black}%
\vcenter{\hbox{%
\scalebox{#1}{%
\begin{tikzpicture}[baseline=(current bounding box.center)]
\node[rectangle, minimum width=5mm, minimum height=5mm, inner sep=2pt, outer sep=0pt, rounded corners=1.6pt, draw=tncol33415c, line width=0.6pt, fill=tncolb4bd9b, text=black] (t0) at (0,-0) {$P_{\mu_1}^{(A)}$};
\node[rectangle, minimum width=5mm, minimum height=5mm, inner sep=2pt, outer sep=0pt, rounded corners=1.6pt, draw=tncol33415c, line width=0.6pt, fill=tncolb4bd9b, text=black] (t1) at (1.25,-0) {$P_{\mu_2}^{(A)}$};
\node[rectangle, minimum width=5mm, minimum height=5mm, inner sep=2pt, outer sep=0pt, rounded corners=1.6pt, draw=tncol33415c, line width=0.6pt, fill=tncol81bdc3, text=black] (t2) at (0,-1) {$P_{\mu_3}^{(A)}$};
\node[rectangle, minimum width=5mm, minimum height=5mm, inner sep=2pt, outer sep=0pt, rounded corners=1.6pt, draw=tncol33415c, line width=0.6pt, fill=tncol81bdc3, text=black] (t3) at (1.25,-1) {$P_{\mu_4}^{(A)}$};
\node[rectangle, minimum width=5mm, minimum height=5mm, inner sep=2pt, outer sep=0pt, rounded corners=1.6pt, draw=tncol33415c, line width=0.6pt, fill=tncolbc455a, text=black] (t4) at (0,-2) {$P_{\mu_1}^{(\bar{A})}$};
\node[rectangle, minimum width=5mm, minimum height=5mm, inner sep=2pt, outer sep=0pt, rounded corners=1.6pt, draw=tncol33415c, line width=0.6pt, fill=tncolf6cf98, text=black] (t5) at (0,-3) {$P_{\mu_2}^{(\bar{A})}$};
\node[rectangle, minimum width=5mm, minimum height=5mm, inner sep=2pt, outer sep=0pt, rounded corners=1.6pt, draw=tncol33415c, line width=0.6pt, fill=tncolf6cf98, text=black] (t6) at (1.25,-3) {$P_{\mu_3}^{(\bar{A})}$};
\node[rectangle, minimum width=5mm, minimum height=5mm, inner sep=2pt, outer sep=0pt, rounded corners=1.6pt, draw=tncol33415c, line width=0.6pt, fill=tncolbc455a, text=black] (t7) at (1.25,-2) {$P_{\mu_4}^{(\bar{A})}$};
\draw[draw=tncolff1b6b, line width=0.6pt, rounded corners=0.2cm] (t0.west) -- ($ (t0.west) + (180:0.3) $) -- ($ (t0.west) + (180:0.3) + (90:0.5) $) -- ($ (t1.east) + (0:0.3) + (90:0.5) $) -- ($ (t1.east) + (0:0.3) $) -- (t1.east);
\draw[draw=tncolff1b6b, line width=0.6pt] (t0.east) -- (t1.west);
\draw[draw=tncolff1b6b, line width=0.6pt, rounded corners=0.2cm] (t2.west) -- ($ (t2.west) + (180:0.3) $) -- ($ (t2.west) + (180:0.3) + (90:0.5) $) -- ($ (t3.east) + (0:0.3) + (90:0.5) $) -- ($ (t3.east) + (0:0.3) $) -- (t3.east);
\draw[draw=tncolff1b6b, line width=0.6pt] (t2.east) -- (t3.west);
\draw[draw=tncol45caff, line width=0.6pt, rounded corners=0.2cm] (t4.west) -- ($ (t4.west) + (180:0.3) $) -- ($ (t4.west) + (180:0.3) + (90:0.5) $) -- ($ (t7.east) + (0:0.3) + (90:0.5) $) -- ($ (t7.east) + (0:0.3) $) -- (t7.east);
\draw[draw=tncol45caff, line width=0.6pt] (t4.east) -- (t7.west);
\draw[draw=tncol45caff, line width=0.6pt, rounded corners=0.2cm] (t5.west) -- ($ (t5.west) + (180:0.3) $) -- ($ (t5.west) + (180:0.3) + (90:0.5) $) -- ($ (t6.east) + (0:0.3) + (90:0.5) $) -- ($ (t6.east) + (0:0.3) $) -- (t6.east);
\draw[draw=tncol45caff, line width=0.6pt] (t5.east) -- (t6.west);
\end{tikzpicture}%
}}}%
\endgroup}%
}

%% file: figures/loe_class_1.tex
\providecommand{\LOEClassOne}[1][1]{%
\ensuremath{\begingroup\color{black}%
\vcenter{\hbox{%
\scalebox{#1}{%
\begin{tikzpicture}[baseline=(current bounding box.center)]
\node[rectangle, minimum width=2.1cm, minimum height=0.4cm, inner sep=2pt, outer sep=0pt, rounded corners=1.6pt, draw=tncol33415c, line width=0.6pt, fill=tncolb4bd9b, text=black, font=\scriptsize, inner sep=0pt] (t0) at (0.6,-0.25) {$P_{1}^{(A)}$};
\node[rectangle, minimum width=2.1cm, minimum height=0.4cm, inner sep=2pt, outer sep=0pt, rounded corners=1.6pt, draw=tncol33415c, line width=0.6pt, fill=tncolb4bd9b, text=black, font=\scriptsize, inner sep=0pt] (t1) at (0.6,-0.75) {$P_{2}^{(A)}$};
\node[rectangle, minimum width=2.1cm, minimum height=0.4cm, inner sep=2pt, outer sep=0pt, rounded corners=1.6pt, draw=tncol33415c, line width=0.6pt, fill=tncol81bdc3, text=black, font=\scriptsize, inner sep=0pt] (t2) at (0.6,-1.25) {$P_{3}^{(A)}$};
\node[rectangle, minimum width=2.1cm, minimum height=0.4cm, inner sep=2pt, outer sep=0pt, rounded corners=1.6pt, draw=tncol33415c, line width=0.6pt, fill=tncol81bdc3, text=black, font=\scriptsize, inner sep=0pt] (t3) at (0.6,-1.75) {$P_{4}^{(A)}$};
\node[rectangle, minimum width=2.7cm, minimum height=0.4cm, inner sep=2pt, outer sep=0pt, rounded corners=1.6pt, draw=tncol33415c, line width=0.6pt, fill=tncolbc455a, text=black, font=\scriptsize, inner sep=0pt] (t4) at (3.3,-0.25) {$P_{1}^{(\bar{A})}$};
\node[rectangle, minimum width=2.7cm, minimum height=0.4cm, inner sep=2pt, outer sep=0pt, rounded corners=1.6pt, draw=tncol33415c, line width=0.6pt, fill=tncolf6cf98, text=black, font=\scriptsize, inner sep=0pt] (t5) at (3.3,-0.75) {$P_{2}^{(\bar{A})}$};
\node[rectangle, minimum width=2.7cm, minimum height=0.4cm, inner sep=2pt, outer sep=0pt, rounded corners=1.6pt, draw=tncol33415c, line width=0.6pt, fill=tncolf6cf98, text=black, font=\scriptsize, inner sep=0pt] (t6) at (3.3,-1.25) {$P_{3}^{(\bar{A})}$};
\node[rectangle, minimum width=2.7cm, minimum height=0.4cm, inner sep=2pt, outer sep=0pt, rounded corners=1.6pt, draw=tncol33415c, line width=0.6pt, fill=tncolbc455a, text=black, font=\scriptsize, inner sep=0pt] (t7) at (3.3,-1.75) {$P_{4}^{(\bar{A})}$};
\node[inner sep=0pt, minimum size=0pt, outer sep=0pt] (t8) at (1.8,-0) {};
\node[inner sep=0pt, minimum size=0pt, outer sep=0pt] (t9) at (1.8,-0.5) {};
\node[inner sep=0pt, minimum size=0pt, outer sep=0pt] (t10) at (1.8,-1) {};
\node[inner sep=0pt, minimum size=0pt, outer sep=0pt] (t11) at (1.8,-1.5) {};
\node[inner sep=0pt, minimum size=0pt, outer sep=0pt] (t12) at (1.8,-2) {};
\draw[draw=tncol888888, line width=0.6pt, dashed] (t8.-90) -- (t12.90);
\end{tikzpicture}%
}}}%
\endgroup}%
}

%% file: figures/otoc_contraction_k.tex
\providecommand{\OTOCk}[1][1]{%
\ensuremath{\begingroup\color{black}%
\vcenter{\hbox{%
\scalebox{#1}{%
\begin{tikzpicture}[baseline=(current bounding box.center)]
\node[rectangle, minimum width=5mm, minimum height=5mm, inner sep=2pt, outer sep=0pt, rounded corners=1.6pt, draw=tncol33415c, line width=0.6pt, fill=tncolb4bd9b, text=black, font=\scriptsize, inner sep=0pt] (t0) at (0,-0) {$P_{1}$};
\node[rectangle, minimum width=5mm, minimum height=5mm, inner sep=2pt, outer sep=0pt, rounded corners=1.6pt, draw=tncol33415c, line width=0.6pt, fill=white, text=black, font=\scriptsize, inner sep=0pt] (t1) at (0.75,-0) {$X$};
\node[rectangle, minimum width=5mm, minimum height=5mm, inner sep=2pt, outer sep=0pt, rounded corners=1.6pt, draw=tncol33415c, line width=0.6pt, fill=tncolbc455a, text=black, font=\scriptsize, inner sep=0pt] (t2) at (1.5,-0) {$P_{2}$};
\node[rectangle, minimum width=5mm, minimum height=5mm, inner sep=2pt, outer sep=0pt, rounded corners=1.6pt, draw=tncol33415c, line width=0.6pt, fill=white, text=black, font=\scriptsize, inner sep=0pt] (t3) at (2.25,-0) {$X$};
\node[rectangle, minimum width=5mm, minimum height=5mm, inner sep=2pt, outer sep=0pt, rounded corners=1.6pt, draw=tncol33415c, line width=0.6pt, fill=tncolfdba77, text=black, font=\tiny, inner sep=0pt] (t4) at (3.75,-0) {$P_{k-1}$};
\node[rectangle, minimum width=5mm, minimum height=5mm, inner sep=2pt, outer sep=0pt, rounded corners=1.6pt, draw=tncol33415c, line width=0.6pt, fill=white, text=black, font=\scriptsize, inner sep=0pt] (t5) at (4.5,-0) {$X$};
\node[rectangle, minimum width=5mm, minimum height=5mm, inner sep=2pt, outer sep=0pt, rounded corners=1.6pt, draw=tncol33415c, line width=0.6pt, fill=tncol81bdc3, text=black, font=\scriptsize, inner sep=0pt] (t6) at (5.25,-0) {$P_{k}$};
\node[rectangle, minimum width=5mm, minimum height=5mm, inner sep=2pt, outer sep=0pt, rounded corners=1.6pt, draw=tncol33415c, line width=0.6pt, fill=white, text=black, font=\scriptsize, inner sep=0pt] (t7) at (6,-0) {$X$};
\draw[draw=tncol808080, line width=0.6pt] (t0.0) to[out=0, in=180] (t1.180);
\draw[draw=tncol808080, line width=0.6pt] (t1.0) to[out=0, in=180] (t2.180);
\draw[draw=tncol808080, line width=0.6pt] (t2.0) to[out=0, in=180] (t3.180);
\draw[draw=tncol808080, line width=0.6pt] (t3.0) -- (t4.180);
\begin{scope}[shift={($ (t3.0) !0.5! (t4.180) $)}, rotate=0]
  \fill[white] (-8.9pt,-4.4pt) rectangle (8.9pt,4.4pt);
  \fill (-4.5pt,0pt) circle (0.4pt) (0pt,0pt) circle (0.4pt) (4.5pt,0pt) circle (0.4pt);
\end{scope}
\draw[draw=tncol808080, line width=0.6pt] (t4.0) to[out=0, in=180] (t5.180);
\draw[draw=tncol808080, line width=0.6pt] (t5.0) to[out=0, in=180] (t6.180);
\draw[draw=tncol808080, line width=0.6pt] (t6.0) to[out=0, in=180] (t7.180);
\draw[draw=tncol808080, line width=0.6pt, rounded corners=0.2cm] (t7.east) -- ($ (t7.east) + (0:0.3) $) -- ($ (t7.east) + (0:0.3) + (90:0.5) $) -- ($ (t0.west) + (180:0.3) + (90:0.5) $) -- ($ (t0.west) + (180:0.3) $) -- (t0.west);
\end{tikzpicture}%
}}}%
\endgroup}%
}

%% file: figures/ose_class.tex
\providecommand{\OSEClass}[1][1]{%
\ensuremath{\begingroup\color{black}%
\vcenter{\hbox{%
\scalebox{#1}{%
\begin{tikzpicture}[baseline=(current bounding box.center)]
\node[rectangle, minimum width=4.5cm, minimum height=0.4cm, inner sep=2pt, outer sep=0pt, rounded corners=1.6pt, draw=tncol33415c, line width=0.6pt, fill=tncolf6cf98, text=black, font=\scriptsize, inner sep=0pt] (t0) at (1.8,-0.25) {$P_{1}$};
\node[rectangle, minimum width=4.5cm, minimum height=0.4cm, inner sep=2pt, outer sep=0pt, rounded corners=1.6pt, draw=tncol33415c, line width=0.6pt, fill=tncolf6cf98, text=black, font=\scriptsize, inner sep=0pt] (t1) at (1.8,-0.75) {$P_{2}$};
\node[rectangle, minimum width=4.5cm, minimum height=0.4cm, inner sep=2pt, outer sep=0pt, rounded corners=1.6pt, draw=tncol33415c, line width=0.6pt, fill=tncolf6cf98, text=black, font=\scriptsize, inner sep=0pt] (t2) at (1.8,-1.25) {$P_{3}$};
\node[rectangle, minimum width=4.5cm, minimum height=0.4cm, inner sep=2pt, outer sep=0pt, rounded corners=1.6pt, draw=tncol33415c, line width=0.6pt, fill=tncolf6cf98, text=black, font=\scriptsize, inner sep=0pt] (t3) at (1.8,-1.75) {$P_{4}$};
\end{tikzpicture}%
}}}%
\endgroup}%
}

%% file: figures/pairing_e1_w2.tex
\providecommand{\PairingEOneW}[1][1]{%
\ensuremath{\begingroup\color{black}%
\vcenter{\hbox{%
\scalebox{#1}{%
\begin{tikzpicture}[baseline=(current bounding box.center)]
\node[circle, minimum size=4mm, inner sep=0pt, outer sep=0pt, draw=tncol273043, line width=0.45pt, fill=tncolffadad, text=tncol273043, font=\tiny, inner sep=0pt] (t0) at (0,-0.7) {$1$};
\node[circle, minimum size=4mm, inner sep=0pt, outer sep=0pt, draw=tncol273043, line width=0.45pt, fill=tncolffadad, text=tncol273043, font=\tiny, inner sep=0pt] (t1) at (0,-1.4) {$2$};
\node[circle, minimum size=1mm, inner sep=0pt, outer sep=0pt, draw=tncol33415c, line width=0.6pt, fill=tncole8edf4, text=black] (t2) at (0,-2.1) {};
\node[circle, minimum size=1mm, inner sep=0pt, outer sep=0pt, draw=tncol33415c, line width=0.6pt, fill=tncole8edf4, text=black] (t3) at (0,-2.8) {};
\node[circle, minimum size=1mm, inner sep=0pt, outer sep=0pt, draw=tncol33415c, line width=0.6pt, fill=tncole8edf4, text=black] (t4) at (0.7,-0.7) {};
\node[circle, minimum size=1mm, inner sep=0pt, outer sep=0pt, draw=tncol33415c, line width=0.6pt, fill=tncole8edf4, text=black] (t5) at (0.7,-1.4) {};
\node[circle, minimum size=4mm, inner sep=0pt, outer sep=0pt, draw=tncol273043, line width=0.45pt, fill=tncolffd6a5, text=tncol273043, font=\tiny, inner sep=0pt] (t6) at (0.7,-2.1) {$3$};
\node[circle, minimum size=4mm, inner sep=0pt, outer sep=0pt, draw=tncol273043, line width=0.45pt, fill=tncolffd6a5, text=tncol273043, font=\tiny, inner sep=0pt] (t7) at (0.7,-2.8) {$4$};
\draw[draw=tncol4c78a8, line width=0.85pt] (t0.-90) to[bend left=60] (t1.90);
\draw[draw=tncol4c78a8, line width=0.85pt] (t6.-90) to[bend right=60] (t7.90);
\node[anchor=center,font={\small},text=tncol4c78a8] at (0.350,-0.350) {$e_{1}$};
\node[anchor=center,font={\tiny},text=tncol6b7280] at (0.350,-3.255) {$1100\mid 0011$};
\end{tikzpicture}%
}}}%
\endgroup}%
}

%% file: figures/pairing_e2_w2.tex
\providecommand{\PairingETwoW}[1][1]{%
\ensuremath{\begingroup\color{black}%
\vcenter{\hbox{%
\scalebox{#1}{%
\begin{tikzpicture}[baseline=(current bounding box.center)]
\node[circle, minimum size=4mm, inner sep=0pt, outer sep=0pt, draw=tncol273043, line width=0.45pt, fill=tncolfdffb6, text=tncol273043, font=\tiny, inner sep=0pt] (t0) at (0,-0.7) {$1$};
\node[circle, minimum size=1mm, inner sep=0pt, outer sep=0pt, draw=tncol33415c, line width=0.6pt, fill=tncole8edf4, text=black] (t1) at (0,-1.4) {};
\node[circle, minimum size=4mm, inner sep=0pt, outer sep=0pt, draw=tncol273043, line width=0.45pt, fill=tncolfdffb6, text=tncol273043, font=\tiny, inner sep=0pt] (t2) at (0,-2.1) {$3$};
\node[circle, minimum size=1mm, inner sep=0pt, outer sep=0pt, draw=tncol33415c, line width=0.6pt, fill=tncole8edf4, text=black] (t3) at (0,-2.8) {};
\node[circle, minimum size=1mm, inner sep=0pt, outer sep=0pt, draw=tncol33415c, line width=0.6pt, fill=tncole8edf4, text=black] (t4) at (0.7,-0.7) {};
\node[circle, minimum size=4mm, inner sep=0pt, outer sep=0pt, draw=tncol273043, line width=0.45pt, fill=tncolcaffbf, text=tncol273043, font=\tiny, inner sep=0pt] (t5) at (0.7,-1.4) {$2$};
\node[circle, minimum size=1mm, inner sep=0pt, outer sep=0pt, draw=tncol33415c, line width=0.6pt, fill=tncole8edf4, text=black] (t6) at (0.7,-2.1) {};
\node[circle, minimum size=4mm, inner sep=0pt, outer sep=0pt, draw=tncol273043, line width=0.45pt, fill=tncolcaffbf, text=tncol273043, font=\tiny, inner sep=0pt] (t7) at (0.7,-2.8) {$4$};
\draw[draw=tncolf28e2b, line width=0.85pt] (t0.-90) to[bend left=60] (t2.90);
\draw[draw=tncolf28e2b, line width=0.85pt] (t5.-90) to[bend right=60] (t7.90);
\node[anchor=center,font={\small},text=tncolf28e2b] at (0.350,-0.350) {$e_{2}$};
\node[anchor=center,font={\tiny},text=tncol6b7280] at (0.350,-3.255) {$1010\mid 0101$};
\end{tikzpicture}%
}}}%
\endgroup}%
}

%% file: figures/pairing_e3_w2.tex
\providecommand{\PairingEThreeW}[1][1]{%
\ensuremath{\begingroup\color{black}%
\vcenter{\hbox{%
\scalebox{#1}{%
\begin{tikzpicture}[baseline=(current bounding box.center)]
\node[circle, minimum size=4mm, inner sep=0pt, outer sep=0pt, draw=tncol273043, line width=0.45pt, fill=tncol9bf6ff, text=tncol273043, font=\tiny, inner sep=0pt] (t0) at (0,-0.7) {$1$};
\node[circle, minimum size=1mm, inner sep=0pt, outer sep=0pt, draw=tncol33415c, line width=0.6pt, fill=tncole8edf4, text=black] (t1) at (0,-1.4) {};
\node[circle, minimum size=1mm, inner sep=0pt, outer sep=0pt, draw=tncol33415c, line width=0.6pt, fill=tncole8edf4, text=black] (t2) at (0,-2.1) {};
\node[circle, minimum size=4mm, inner sep=0pt, outer sep=0pt, draw=tncol273043, line width=0.45pt, fill=tncol9bf6ff, text=tncol273043, font=\tiny, inner sep=0pt] (t3) at (0,-2.8) {$4$};
\node[circle, minimum size=1mm, inner sep=0pt, outer sep=0pt, draw=tncol33415c, line width=0.6pt, fill=tncole8edf4, text=black] (t4) at (0.7,-0.7) {};
\node[circle, minimum size=4mm, inner sep=0pt, outer sep=0pt, draw=tncol273043, line width=0.45pt, fill=tncola0c4ff, text=tncol273043, font=\tiny, inner sep=0pt] (t5) at (0.7,-1.4) {$2$};
\node[circle, minimum size=4mm, inner sep=0pt, outer sep=0pt, draw=tncol273043, line width=0.45pt, fill=tncola0c4ff, text=tncol273043, font=\tiny, inner sep=0pt] (t6) at (0.7,-2.1) {$3$};
\node[circle, minimum size=1mm, inner sep=0pt, outer sep=0pt, draw=tncol33415c, line width=0.6pt, fill=tncole8edf4, text=black] (t7) at (0.7,-2.8) {};
\draw[draw=tncol2a9d8f, line width=0.85pt] (t0.-90) to[bend left=60] (t3.90);
\draw[draw=tncol2a9d8f, line width=0.85pt] (t5.-90) to[bend right=60] (t6.90);
\node[anchor=center,font={\small},text=tncol2a9d8f] at (0.350,-0.350) {$e_{3}$};
\node[anchor=center,font={\tiny},text=tncol6b7280] at (0.350,-3.255) {$1001\mid 0110$};
\end{tikzpicture}%
}}}%
\endgroup}%
}

%% file: figures/pairing_e1_w4.tex
\providecommand{\PairingEOneWW}[1][1]{%
\ensuremath{\begingroup\color{black}%
\vcenter{\hbox{%
\scalebox{#1}{%
\begin{tikzpicture}[baseline=(current bounding box.center)]
\node[circle, minimum size=4mm, inner sep=0pt, outer sep=0pt, draw=tncol273043, line width=0.45pt, fill=tncol2a9d8f, text=tncol273043, font=\tiny, inner sep=0pt] (t0) at (0,-0.7) {$1$};
\node[circle, minimum size=4mm, inner sep=0pt, outer sep=0pt, draw=tncol273043, line width=0.45pt, fill=tncol2a9d8f, text=tncol273043, font=\tiny, inner sep=0pt] (t1) at (0,-1.4) {$2$};
\node[circle, minimum size=4mm, inner sep=0pt, outer sep=0pt, draw=tncol273043, line width=0.45pt, fill=tncol2a9d8f, text=tncol273043, font=\tiny, inner sep=0pt] (t2) at (0,-2.1) {$3$};
\node[circle, minimum size=4mm, inner sep=0pt, outer sep=0pt, draw=tncol273043, line width=0.45pt, fill=tncol2a9d8f, text=tncol273043, font=\tiny, inner sep=0pt] (t3) at (0,-2.8) {$4$};
\draw[draw=tncol4c78a8, line width=0.9pt] (t0.-90) to[bend left=60] (t1.90);
\draw[draw=tncol4c78a8, line width=0.9pt] (t2.-90) to[bend right=60] (t3.90);
\node[anchor=center,font={\tiny},text=tncol6b7280] at (0.000,-3.255) {$1111$};
\begin{pgfonlayer}{background}
\node[draw=tncol4c78a8, line width=0.6pt, inner sep=3pt, fill=tncolf3e8b6, fill opacity=0.12, rounded corners=3pt, fit=(t0)(t1)(t2)(t3)] {};
\end{pgfonlayer}
\end{tikzpicture}%
}}}%
\endgroup}%
}

%% file: figures/pairing_e2_w4.tex
\providecommand{\PairingETwoWW}[1][1]{%
\ensuremath{\begingroup\color{black}%
\vcenter{\hbox{%
\scalebox{#1}{%
\begin{tikzpicture}[baseline=(current bounding box.center)]
\node[circle, minimum size=4mm, inner sep=0pt, outer sep=0pt, draw=tncol273043, line width=0.45pt, fill=tncol2a9d8f, text=tncol273043, font=\tiny, inner sep=0pt] (t0) at (0,-0.7) {$1$};
\node[circle, minimum size=4mm, inner sep=0pt, outer sep=0pt, draw=tncol273043, line width=0.45pt, fill=tncol2a9d8f, text=tncol273043, font=\tiny, inner sep=0pt] (t1) at (0,-1.4) {$2$};
\node[circle, minimum size=4mm, inner sep=0pt, outer sep=0pt, draw=tncol273043, line width=0.45pt, fill=tncol2a9d8f, text=tncol273043, font=\tiny, inner sep=0pt] (t2) at (0,-2.1) {$3$};
\node[circle, minimum size=4mm, inner sep=0pt, outer sep=0pt, draw=tncol273043, line width=0.45pt, fill=tncol2a9d8f, text=tncol273043, font=\tiny, inner sep=0pt] (t3) at (0,-2.8) {$4$};
\draw[draw=tncolf28e2b, line width=0.9pt] (t0.-90) to[bend left=60] (t2.90);
\draw[draw=tncolf28e2b, line width=0.9pt] (t1.-90) to[bend right=60] (t3.90);
\node[anchor=center,font={\tiny},text=tncol6b7280] at (0.000,-3.255) {$1111$};
\begin{pgfonlayer}{background}
\node[draw=tncolf28e2b, line width=0.6pt, inner sep=3pt, fill=tncolf3e8b6, fill opacity=0.12, rounded corners=3pt, fit=(t0)(t1)(t2)(t3)] {};
\end{pgfonlayer}
\end{tikzpicture}%
}}}%
\endgroup}%
}

%% file: figures/pairing_e3_w4.tex
\providecommand{\PairingEThreeWW}[1][1]{%
\ensuremath{\begingroup\color{black}%
\vcenter{\hbox{%
\scalebox{#1}{%
\begin{tikzpicture}[baseline=(current bounding box.center)]
\node[circle, minimum size=4mm, inner sep=0pt, outer sep=0pt, draw=tncol273043, line width=0.45pt, fill=tncol2a9d8f, text=tncol273043, font=\tiny, inner sep=0pt] (t0) at (0,-0.7) {$1$};
\node[circle, minimum size=4mm, inner sep=0pt, outer sep=0pt, draw=tncol273043, line width=0.45pt, fill=tncol2a9d8f, text=tncol273043, font=\tiny, inner sep=0pt] (t1) at (0,-1.4) {$2$};
\node[circle, minimum size=4mm, inner sep=0pt, outer sep=0pt, draw=tncol273043, line width=0.45pt, fill=tncol2a9d8f, text=tncol273043, font=\tiny, inner sep=0pt] (t2) at (0,-2.1) {$3$};
\node[circle, minimum size=4mm, inner sep=0pt, outer sep=0pt, draw=tncol273043, line width=0.45pt, fill=tncol2a9d8f, text=tncol273043, font=\tiny, inner sep=0pt] (t3) at (0,-2.8) {$4$};
\draw[draw=tncol2a9d8f, line width=0.9pt] (t0.-90) to[bend left=60] (t3.90);
\draw[draw=tncol2a9d8f, line width=0.9pt] (t1.-90) to[bend right=60] (t2.90);
\node[anchor=center,font={\tiny},text=tncol6b7280] at (0.000,-3.255) {$1111$};
\begin{pgfonlayer}{background}
\node[draw=tncol2a9d8f, line width=0.6pt, inner sep=3pt, fill=tncolf3e8b6, fill opacity=0.12, rounded corners=3pt, fit=(t0)(t1)(t2)(t3)] {};
\end{pgfonlayer}
\end{tikzpicture}%
}}}%
\endgroup}%
}

%% file: figures/contingency_inputs.tex
\input{figures/contingency_assembly_aligned.tex}
\input{figures/contingency_assembly_arrows.tex}
\input{figures/contingency_assembly_sorted.tex}

\input{figures/contingency_blocks_aligned.tex}
\input{figures/contingency_blocks_exchange_after.tex}
\input{figures/contingency_blocks_exchange_before.tex}
\input{figures/contingency_blocks_invalid_after.tex}
\input{figures/contingency_blocks_invalid_before.tex}
\input{figures/contingency_blocks_repair_after.tex}
\input{figures/contingency_blocks_repair_before.tex}

\input{figures/contingency_colour_key.tex}
\input{figures/contingency_colour_strings_aligned.tex}
\input{figures/contingency_colour_strings_sorted.tex}
\input{figures/contingency_colour_table_aligned.tex}
\input{figures/contingency_colour_table_sorted.tex}

\input{figures/contingency_eg_consolid.tex}
\input{figures/contingency_eg_exp.tex}
\input{figures/contingency_eg_exp_long.tex}
\input{figures/contingency_eg_exp_orbit.tex}
\input{figures/contingency_eg_orbit_consolid.tex}
\input{figures/contingency_eg_pairing_table.tex}
\input{figures/contingency_eg_pairings_aligned.tex}
\input{figures/contingency_eg_pairings_partial.tex}
\input{figures/contingency_eg_pairings_scrambled.tex}
\input{figures/contingency_eg_pairings_sorted.tex}
\input{figures/contingency_eg_strings_canonical.tex}
\input{figures/contingency_eg_strings_ordered.tex}

%% file: figures/contingency_assembly_aligned.tex
\providecommand{\ContAssemblyAligned}[1][1]{%
\ensuremath{\begingroup\color{black}%
\vcenter{\hbox{%
\scalebox{#1}{%
\begin{tikzpicture}[baseline=(current bounding box.center)]
\node[circle, minimum size=4mm, inner sep=0pt, outer sep=0pt, draw=tncolf4a261, line width=1.6pt, fill=tncolffadad] (t0) at (0,-0) {};
\node[circle, minimum size=4mm, inner sep=0pt, outer sep=0pt, draw=tncol8ecae6, line width=1.6pt, fill=tncolcaffbf] (t1) at (0.7,-0) {};
\node[circle, minimum size=4mm, inner sep=0pt, outer sep=0pt, draw=tncolb5e48c, line width=1.6pt, fill=tncola0c4ff] (t2) at (1.4,-0) {};
\node[circle, minimum size=4mm, inner sep=0pt, outer sep=0pt, draw=tncolf4a261, line width=1.6pt, fill=tncol9bf6ff] (t3) at (0,-0.7) {};
\node[circle, minimum size=4mm, inner sep=0pt, outer sep=0pt, draw=tncol8ecae6, line width=1.6pt, fill=tncolffd6a5] (t4) at (0.7,-0.7) {};
\node[circle, minimum size=4mm, inner sep=0pt, outer sep=0pt, draw=tncolb5e48c, line width=1.6pt, fill=tncolfdffb6] (t5) at (1.4,-0.7) {};
\node[circle, minimum size=4mm, inner sep=0pt, outer sep=0pt, draw=tncolf4a261, line width=1.6pt, fill=tncolffadad] (t6) at (0,-1.4) {};
\node[circle, minimum size=4mm, inner sep=0pt, outer sep=0pt, draw=tncol8ecae6, line width=1.6pt, fill=tncolffd6a5] (t7) at (0.7,-1.4) {};
\node[circle, minimum size=4mm, inner sep=0pt, outer sep=0pt, draw=tncolb5e48c, line width=1.6pt, fill=tncolfdffb6] (t8) at (1.4,-1.4) {};
\node[circle, minimum size=4mm, inner sep=0pt, outer sep=0pt, draw=tncolf4a261, line width=1.6pt, fill=tncol9bf6ff] (t9) at (0,-2.1) {};
\node[circle, minimum size=4mm, inner sep=0pt, outer sep=0pt, draw=tncol8ecae6, line width=1.6pt, fill=tncolcaffbf] (t10) at (0.7,-2.1) {};
\node[circle, minimum size=4mm, inner sep=0pt, outer sep=0pt, draw=tncolb5e48c, line width=1.6pt, fill=tncola0c4ff] (t11) at (1.4,-2.1) {};
\node[inner sep=0pt, minimum size=0pt, outer sep=0pt] (t12) at (0,0.5) {};
\draw[draw=tncol555555, line width=0.6pt] (t0.-90) to[bend right=67.4024] (t6.90);
\draw[draw=tncol555555, line width=0.6pt] (t4.-90) -- (t7.90);
\draw[draw=tncol555555, line width=0.6pt] (t5.-90) -- (t8.90);
\draw[draw=tncol555555, line width=0.6pt] (t1.-90) to[bend right=37.9871] (t10.90);
\draw[draw=tncol555555, line width=0.6pt] (t3.-90) to[bend left=67.4024] (t9.90);
\draw[draw=tncol555555, line width=0.6pt] (t2.-90) to[bend right=37.9871] (t11.90);
\node[anchor=center, inner sep=1pt, font=\small] at (0.000,-2.600) {$e_{2}$};
\node[anchor=center, inner sep=1pt, font=\small] at (0.700,-2.600) {$e_{3}$};
\node[anchor=center, inner sep=1pt, font=\small] at (1.400,-2.600) {$e_{3}$};
\end{tikzpicture}%
}}}%
\endgroup}%
}

%% file: figures/contingency_assembly_arrows.tex
\providecommand{\ContAssemblyArrows}[1][1]{%
\ensuremath{\begingroup\color{black}%
\vcenter{\hbox{%
\scalebox{#1}{%
\begin{tikzpicture}[baseline=(current bounding box.center)]
\node[inner sep=0pt, minimum size=0pt, outer sep=0pt] (t0) at (0,0.35) {};
\node[inner sep=0pt, minimum size=0pt, outer sep=0pt] (t1) at (0,-2.45) {};
\draw[tncol22223b, line width=0.6pt, -{Stealth[length=3.5pt]}] (0,0.000) -- (1.100,0.000);
\draw[tncol22223b, line width=0.6pt, -{Stealth[length=3.5pt]}] (0,-0.700) -- (1.100,-0.700);
\draw[tncol22223b, line width=0.6pt, -{Stealth[length=3.5pt]}] (0,-1.400) -- (1.100,-1.400);
\draw[tncol22223b, line width=0.6pt, -{Stealth[length=3.5pt]}] (0,-2.100) -- (1.100,-2.100);
\node[anchor=west, inner sep=1pt, font=\scriptsize] at (1.220,0.000) {$(+)$};
\node[anchor=south, inner sep=1pt, font=\scriptsize] at (0.550,0.120) {$\sigma_{1}=123$};
\node[anchor=west, inner sep=1pt, font=\scriptsize] at (1.220,-0.700) {$(+)$};
\node[anchor=south, inner sep=1pt, font=\scriptsize] at (0.550,-0.580) {$\sigma_{2}=231$};
\node[anchor=west, inner sep=1pt, font=\scriptsize] at (1.220,-1.400) {$(+)$};
\node[anchor=south, inner sep=1pt, font=\scriptsize] at (0.550,-1.280) {$\sigma_{3}=123$};
\node[anchor=west, inner sep=1pt, font=\scriptsize] at (1.220,-2.100) {$(-)$};
\node[anchor=south, inner sep=1pt, font=\scriptsize] at (0.550,-1.980) {$\sigma_{4}=213$};
\end{tikzpicture}%
}}}%
\endgroup}%
}

%% file: figures/contingency_assembly_sorted.tex
\providecommand{\ContAssemblySorted}[1][1]{%
\ensuremath{\begingroup\color{black}%
\vcenter{\hbox{%
\scalebox{#1}{%
\begin{tikzpicture}[baseline=(current bounding box.center)]
\node[circle, minimum size=4mm, inner sep=0pt, outer sep=0pt, draw=tncolf4a261, line width=1.6pt, fill=tncolffadad] (t0) at (0,-0) {};
\node[circle, minimum size=4mm, inner sep=0pt, outer sep=0pt, draw=tncol8ecae6, line width=1.6pt, fill=tncolcaffbf] (t1) at (0.7,-0) {};
\node[circle, minimum size=4mm, inner sep=0pt, outer sep=0pt, draw=tncolb5e48c, line width=1.6pt, fill=tncola0c4ff] (t2) at (1.4,-0) {};
\node[circle, minimum size=4mm, inner sep=0pt, outer sep=0pt, draw=tncol8ecae6, line width=1.6pt, fill=tncolffd6a5] (t3) at (0,-0.7) {};
\node[circle, minimum size=4mm, inner sep=0pt, outer sep=0pt, draw=tncolb5e48c, line width=1.6pt, fill=tncolfdffb6] (t4) at (0.7,-0.7) {};
\node[circle, minimum size=4mm, inner sep=0pt, outer sep=0pt, draw=tncolf4a261, line width=1.6pt, fill=tncol9bf6ff] (t5) at (1.4,-0.7) {};
\node[circle, minimum size=4mm, inner sep=0pt, outer sep=0pt, draw=tncolf4a261, line width=1.6pt, fill=tncolffadad] (t6) at (0,-1.4) {};
\node[circle, minimum size=4mm, inner sep=0pt, outer sep=0pt, draw=tncol8ecae6, line width=1.6pt, fill=tncolffd6a5] (t7) at (0.7,-1.4) {};
\node[circle, minimum size=4mm, inner sep=0pt, outer sep=0pt, draw=tncolb5e48c, line width=1.6pt, fill=tncolfdffb6] (t8) at (1.4,-1.4) {};
\node[circle, minimum size=4mm, inner sep=0pt, outer sep=0pt, draw=tncol8ecae6, line width=1.6pt, fill=tncolcaffbf] (t9) at (0,-2.1) {};
\node[circle, minimum size=4mm, inner sep=0pt, outer sep=0pt, draw=tncolf4a261, line width=1.6pt, fill=tncol9bf6ff] (t10) at (0.7,-2.1) {};
\node[circle, minimum size=4mm, inner sep=0pt, outer sep=0pt, draw=tncolb5e48c, line width=1.6pt, fill=tncola0c4ff] (t11) at (1.4,-2.1) {};
\draw[draw=tncol555555, line width=0.6pt] (t0.-90) to[bend right=67.4024] (t6.90);
\draw[draw=tncol555555, line width=0.6pt] (t3.-45) -- (t7.135);
\draw[draw=tncol555555, line width=0.6pt] (t4.-45) -- (t8.135);
\draw[draw=tncol555555, line width=0.6pt] (t1.-108) -- (t9.72);
\draw[draw=tncol555555, line width=0.6pt] (t5.-117) -- (t10.63);
\draw[draw=tncol555555, line width=0.6pt] (t2.-90) to[bend right=37.9871] (t11.90);
\end{tikzpicture}%
}}}%
\endgroup}%
}

%% file: figures/contingency_blocks_repair_after.tex
\providecommand{\ContBlocksRepairAfter}[1][1]{%
\ensuremath{\begingroup\color{black}%
\vcenter{\hbox{%
\scalebox{#1}{%
\begin{tikzpicture}[baseline=(current bounding box.center)]
\node[circle, minimum size=4mm, inner sep=0pt, outer sep=0pt, draw=tncol22223b, line width=0.6pt, fill=tncolffadad] (t0) at (0,-0) {};
\node[circle, minimum size=4mm, inner sep=0pt, outer sep=0pt, draw=tncol22223b, line width=0.6pt, fill=tncolcaffbf] (t1) at (0.7,-0) {};
\node[circle, minimum size=4mm, inner sep=0pt, outer sep=0pt, draw=tncol22223b, line width=0.6pt, fill=tncola0c4ff] (t2) at (1.4,-0) {};
\node[circle, minimum size=4mm, inner sep=0pt, outer sep=0pt, draw=tncol22223b, line width=0.6pt, fill=tncol9bf6ff] (t3) at (0,-0.7) {};
\node[circle, minimum size=4mm, inner sep=0pt, outer sep=0pt, draw=tncol22223b, line width=0.6pt, fill=tncolfdffb6] (t4) at (0.7,-0.7) {};
\node[circle, minimum size=4mm, inner sep=0pt, outer sep=0pt, draw=tncol22223b, line width=0.6pt, fill=tncolffd6a5] (t5) at (1.4,-0.7) {};
\node[circle, minimum size=4mm, inner sep=0pt, outer sep=0pt, draw=tncol22223b, line width=0.6pt, fill=tncolffadad] (t6) at (0,-1.4) {};
\node[circle, minimum size=4mm, inner sep=0pt, outer sep=0pt, draw=tncol22223b, line width=0.6pt, fill=tncolfdffb6] (t7) at (0.7,-1.4) {};
\node[circle, minimum size=4mm, inner sep=0pt, outer sep=0pt, draw=tncol22223b, line width=0.6pt, fill=tncolffd6a5] (t8) at (1.4,-1.4) {};
\node[circle, minimum size=4mm, inner sep=0pt, outer sep=0pt, draw=tncol22223b, line width=0.6pt, fill=tncol9bf6ff] (t9) at (0,-2.1) {};
\node[circle, minimum size=4mm, inner sep=0pt, outer sep=0pt, draw=tncol22223b, line width=0.6pt, fill=tncolcaffbf] (t10) at (0.7,-2.1) {};
\node[circle, minimum size=4mm, inner sep=0pt, outer sep=0pt, draw=tncol22223b, line width=0.6pt, fill=tncola0c4ff] (t11) at (1.4,-2.1) {};
\node[inner sep=0pt, minimum size=0pt, outer sep=0pt] (t12) at (0,0.5) {};
\draw[draw=tncol555555, line width=0.6pt] (t0.-90) to[bend right=67.4024] (t6.90);
\draw[draw=tncol555555, line width=0.6pt] (t5.-90) -- (t8.90);
\draw[draw=tncol555555, line width=0.6pt] (t4.-90) -- (t7.90);
\draw[draw=tncol555555, line width=0.6pt] (t1.-90) to[bend right=37.9871] (t10.90);
\draw[draw=tncol555555, line width=0.6pt] (t3.-90) to[bend left=67.4024] (t9.90);
\draw[draw=tncol555555, line width=0.6pt] (t2.-90) to[bend right=37.9871] (t11.90);
\node[anchor=center, inner sep=1pt, font=\small] at (0.000,-2.600) {$e_{2}$};
\node[anchor=center, inner sep=1pt, font=\small] at (0.700,-2.600) {$e_{3}$};
\node[anchor=center, inner sep=1pt, font=\small] at (1.400,-2.600) {$e_{3}$};
\end{tikzpicture}%
}}}%
\endgroup}%
}

%% file: figures/contingency_blocks_repair_before.tex
\providecommand{\ContBlocksRepairBefore}[1][1]{%
\ensuremath{\begingroup\color{black}%
\vcenter{\hbox{%
\scalebox{#1}{%
\begin{tikzpicture}[baseline=(current bounding box.center)]
\node[circle, minimum size=4mm, inner sep=0pt, outer sep=0pt, draw=tncol22223b, line width=0.6pt, fill=tncolffadad] (t0) at (0,-0) {};
\node[circle, minimum size=4mm, inner sep=0pt, outer sep=0pt, draw=tncol22223b, line width=0.6pt, fill=tncolcaffbf] (t1) at (0.7,-0) {};
\node[circle, minimum size=4mm, inner sep=0pt, outer sep=0pt, draw=tncol22223b, line width=0.6pt, fill=tncola0c4ff] (t2) at (1.4,-0) {};
\node[circle, minimum size=4mm, inner sep=0pt, outer sep=0pt, draw=tncol22223b, line width=0.6pt, fill=tncol9bf6ff] (t3) at (0,-0.7) {};
\node[circle, minimum size=4mm, inner sep=0pt, outer sep=0pt, draw=tncol22223b, line width=0.6pt, fill=tncolffd6a5] (t4) at (0.7,-0.7) {};
\node[circle, minimum size=4mm, inner sep=0pt, outer sep=0pt, draw=tncol22223b, line width=0.6pt, fill=tncolfdffb6] (t5) at (1.4,-0.7) {};
\node[circle, minimum size=4mm, inner sep=0pt, outer sep=0pt, draw=tncol22223b, line width=0.6pt, fill=tncolffadad] (t6) at (0,-1.4) {};
\node[circle, minimum size=4mm, inner sep=0pt, outer sep=0pt, draw=tncol22223b, line width=0.6pt, fill=tncolffd6a5] (t7) at (0.7,-1.4) {};
\node[circle, minimum size=4mm, inner sep=0pt, outer sep=0pt, draw=tncol22223b, line width=0.6pt, fill=tncolfdffb6] (t8) at (1.4,-1.4) {};
\node[circle, minimum size=4mm, inner sep=0pt, outer sep=0pt, draw=tncol22223b, line width=0.6pt, fill=tncol9bf6ff] (t9) at (0,-2.1) {};
\node[circle, minimum size=4mm, inner sep=0pt, outer sep=0pt, draw=tncol22223b, line width=0.6pt, fill=tncolcaffbf] (t10) at (0.7,-2.1) {};
\node[circle, minimum size=4mm, inner sep=0pt, outer sep=0pt, draw=tncol22223b, line width=0.6pt, fill=tncola0c4ff] (t11) at (1.4,-2.1) {};
\node[inner sep=0pt, minimum size=0pt, outer sep=0pt] (t12) at (0,0.5) {};
\draw[draw=tncol555555, line width=0.6pt] (t0.-90) to[bend right=67.4024] (t6.90);
\draw[draw=tncol555555, line width=0.6pt] (t4.-90) -- (t7.90);
\draw[draw=tncol555555, line width=0.6pt] (t5.-90) -- (t8.90);
\draw[draw=tncol555555, line width=0.6pt] (t1.-90) to[bend right=37.9871] (t10.90);
\draw[draw=tncol555555, line width=0.6pt] (t3.-90) to[bend left=67.4024] (t9.90);
\draw[draw=tncol555555, line width=0.6pt] (t2.-90) to[bend right=37.9871] (t11.90);
\node[anchor=center, inner sep=1pt, font=\small] at (0.000,-2.600) {$e_{2}$};
\node[anchor=center, inner sep=1pt, font=\small] at (0.700,-2.600) {$e_{3}$};
\node[anchor=center, inner sep=1pt, font=\small] at (1.400,-2.600) {$e_{3}$};
\begin{pgfonlayer}{background}
\node[draw=tncol9b2226, line width=0.6pt, inner sep=2.5pt, fill=tncol9b2226, fill opacity=0.18, rounded corners=4pt, fit=(t4)(t5)] {};
\node[draw=tncol9b2226, line width=0.6pt, inner sep=2.5pt, fill=tncol9b2226, fill opacity=0.18, rounded corners=4pt, fit=(t7)(t8)] {};
\end{pgfonlayer}
\end{tikzpicture}%
}}}%
\endgroup}%
}

%% file: figures/contingency_colour_strings_aligned.tex
\providecommand{\ContColourStringsAligned}[1][1]{%
\ensuremath{\begingroup\color{black}%
\vcenter{\hbox{%
\scalebox{#1}{%
\begin{tikzpicture}[baseline=(current bounding box.center)]
\node[rectangle, minimum size=9mm, inner sep=2pt, outer sep=0pt, rounded corners=1.6pt, draw=black, line width=0.6pt, fill=tncol2281aa, text=black, minimum height=2mm, minimum width=15mm, font=\small] (t0) at (0,0) {$\gamma_{\mu_B^{(1)}}$};
\node[isosceles triangle, minimum width=4.2mm, minimum height=2.5mm, inner sep=0pt, outer sep=0pt, isosceles triangle stretches, shape border rotate=90, draw=tncolf4a261, line width=1.0, fill=tncolffadad, text=black] (t1) at (-0.5,0.45) {};
\node[isosceles triangle, minimum width=4.2mm, minimum height=2.5mm, inner sep=0pt, outer sep=0pt, isosceles triangle stretches, shape border rotate=90, draw=tncol8ecae6, line width=1.0, fill=tncolcaffbf, text=black] (t2) at (0,0.45) {};
\node[isosceles triangle, minimum width=4.2mm, minimum height=2.5mm, inner sep=0pt, outer sep=0pt, isosceles triangle stretches, shape border rotate=90, draw=tncolb5e48c, line width=1.0, fill=tncola0c4ff, text=black] (t3) at (0.5,0.45) {};
\node[rectangle, minimum size=9mm, inner sep=2pt, outer sep=0pt, rounded corners=1.6pt, draw=black, line width=0.6pt, fill=tncol2281aa, text=black, minimum height=2mm, minimum width=15mm, font=\small] (t4) at (0,-1.25) {$\gamma_{\mu_B^{(2)}}$};
\node[isosceles triangle, minimum width=4.2mm, minimum height=2.5mm, inner sep=0pt, outer sep=0pt, isosceles triangle stretches, shape border rotate=90, draw=tncolf4a261, line width=1.0, fill=tncol9bf6ff, text=black] (t5) at (-0.5,-0.8) {};
\node[isosceles triangle, minimum width=4.2mm, minimum height=2.5mm, inner sep=0pt, outer sep=0pt, isosceles triangle stretches, shape border rotate=90, draw=tncol8ecae6, line width=1.0, fill=tncolffd6a5, text=black] (t6) at (0,-0.8) {};
\node[isosceles triangle, minimum width=4.2mm, minimum height=2.5mm, inner sep=0pt, outer sep=0pt, isosceles triangle stretches, shape border rotate=90, draw=tncolb5e48c, line width=1.0, fill=tncolfdffb6, text=black] (t7) at (0.5,-0.8) {};
\node[rectangle, minimum size=9mm, inner sep=2pt, outer sep=0pt, rounded corners=1.6pt, draw=black, line width=0.6pt, fill=tncol2281aa, text=black, minimum height=2mm, minimum width=15mm, font=\small] (t8) at (0,-2.5) {$\gamma_{\mu_B^{(3)}}$};
\node[isosceles triangle, minimum width=4.2mm, minimum height=2.5mm, inner sep=0pt, outer sep=0pt, isosceles triangle stretches, shape border rotate=90, draw=tncolf4a261, line width=1.0, fill=tncolffadad, text=black] (t9) at (-0.5,-2.05) {};
\node[isosceles triangle, minimum width=4.2mm, minimum height=2.5mm, inner sep=0pt, outer sep=0pt, isosceles triangle stretches, shape border rotate=90, draw=tncol8ecae6, line width=1.0, fill=tncolffd6a5, text=black] (t10) at (0,-2.05) {};
\node[isosceles triangle, minimum width=4.2mm, minimum height=2.5mm, inner sep=0pt, outer sep=0pt, isosceles triangle stretches, shape border rotate=90, draw=tncolb5e48c, line width=1.0, fill=tncolfdffb6, text=black] (t11) at (0.5,-2.05) {};
\node[rectangle, minimum size=9mm, inner sep=2pt, outer sep=0pt, rounded corners=1.6pt, draw=black, line width=0.6pt, fill=tncol2281aa, text=black, minimum height=2mm, minimum width=15mm, font=\small] (t12) at (0,-3.75) {$\gamma_{\mu_B^{(4)}}$};
\node[isosceles triangle, minimum width=4.2mm, minimum height=2.5mm, inner sep=0pt, outer sep=0pt, isosceles triangle stretches, shape border rotate=90, draw=tncolf4a261, line width=1.0, fill=tncol9bf6ff, text=black] (t13) at (-0.5,-3.3) {};
\node[isosceles triangle, minimum width=4.2mm, minimum height=2.5mm, inner sep=0pt, outer sep=0pt, isosceles triangle stretches, shape border rotate=90, draw=tncol8ecae6, line width=1.0, fill=tncolcaffbf, text=black] (t14) at (0,-3.3) {};
\node[isosceles triangle, minimum width=4.2mm, minimum height=2.5mm, inner sep=0pt, outer sep=0pt, isosceles triangle stretches, shape border rotate=90, draw=tncolb5e48c, line width=1.0, fill=tncola0c4ff, text=black] (t15) at (0.5,-3.3) {};
\node[inner sep=0pt, minimum size=0pt, outer sep=0pt] (t16) at (0,-4.9) {};
\draw[draw=tncol33415c, line width=0.6pt] ($(t0.north) + (0:-0.5)$) -- (t1.270);
\draw[draw=tncol33415c, line width=0.6pt] (t0.north) -- (t2.270);
\draw[draw=tncol33415c, line width=0.6pt] ($(t0.north) + (0:0.5)$) -- (t3.270);
\draw[draw=tncol33415c, line width=0.6pt] ($(t4.north) + (0:-0.5)$) -- (t5.270);
\draw[draw=tncol33415c, line width=0.6pt] (t4.north) -- (t6.270);
\draw[draw=tncol33415c, line width=0.6pt] ($(t4.north) + (0:0.5)$) -- (t7.270);
\draw[draw=tncol33415c, line width=0.6pt] ($(t8.north) + (0:-0.5)$) -- (t9.270);
\draw[draw=tncol33415c, line width=0.6pt] (t8.north) -- (t10.270);
\draw[draw=tncol33415c, line width=0.6pt] ($(t8.north) + (0:0.5)$) -- (t11.270);
\draw[draw=tncol33415c, line width=0.6pt] ($(t12.north) + (0:-0.5)$) -- (t13.270);
\draw[draw=tncol33415c, line width=0.6pt] (t12.north) -- (t14.270);
\draw[draw=tncol33415c, line width=0.6pt] ($(t12.north) + (0:0.5)$) -- (t15.270);
\draw[draw=tncol33415c, line width=0.6pt] (t0.west) -- ++(180:0.25);
\draw[draw=tncol33415c, line width=0.6pt] (t0.east) -- ++(0:0.25);
\draw[draw=tncol33415c, line width=0.6pt] (t4.west) -- ++(180:0.25);
\draw[draw=tncol33415c, line width=0.6pt] (t4.east) -- ++(0:0.25);
\draw[draw=tncol33415c, line width=0.6pt] (t8.west) -- ++(180:0.25);
\draw[draw=tncol33415c, line width=0.6pt] (t8.east) -- ++(0:0.25);
\draw[draw=tncol33415c, line width=0.6pt] (t12.west) -- ++(180:0.25);
\draw[draw=tncol33415c, line width=0.6pt] (t12.east) -- ++(0:0.25);
\node[anchor=south, inner sep=1pt, font=\scriptsize] at (-0.500,0.700) {$c_{1}$};
\node[anchor=south, inner sep=1pt, font=\scriptsize] at (0.000,0.700) {$c_{4}$};
\node[anchor=south, inner sep=1pt, font=\scriptsize] at (0.500,0.700) {$c_{6}$};
\node[anchor=south, inner sep=1pt, font=\scriptsize] at (-0.500,-0.550) {$c_{5}$};
\node[anchor=south, inner sep=1pt, font=\scriptsize] at (0.000,-0.550) {$c_{2}$};
\node[anchor=south, inner sep=1pt, font=\scriptsize] at (0.500,-0.550) {$c_{3}$};
\node[anchor=south, inner sep=1pt, font=\scriptsize] at (-0.500,-1.800) {$c_{1}$};
\node[anchor=south, inner sep=1pt, font=\scriptsize] at (0.000,-1.800) {$c_{2}$};
\node[anchor=south, inner sep=1pt, font=\scriptsize] at (0.500,-1.800) {$c_{3}$};
\node[anchor=south, inner sep=1pt, font=\scriptsize] at (-0.500,-3.050) {$c_{5}$};
\node[anchor=south, inner sep=1pt, font=\scriptsize] at (0.000,-3.050) {$c_{4}$};
\node[anchor=south, inner sep=1pt, font=\scriptsize] at (0.500,-3.050) {$c_{6}$};
\end{tikzpicture}%
}}}%
\endgroup}%
}

%% file: figures/contingency_colour_strings_sorted.tex
\providecommand{\ContColourStringsSorted}[1][1]{%
\ensuremath{\begingroup\color{black}%
\vcenter{\hbox{%
\scalebox{#1}{%
\begin{tikzpicture}[baseline=(current bounding box.center)]
\node[rectangle, minimum size=9mm, inner sep=2pt, outer sep=0pt, rounded corners=1.6pt, draw=black, line width=0.6pt, fill=tncol2281aa, text=black, minimum height=2mm, minimum width=15mm, font=\small] (t0) at (0,0) {$\gamma_{\mu_B^{(1)}}$};
\node[isosceles triangle, minimum width=4.2mm, minimum height=2.5mm, inner sep=0pt, outer sep=0pt, isosceles triangle stretches, shape border rotate=90, draw=tncolf4a261, line width=1.0, fill=tncolffadad, text=black] (t1) at (-0.5,0.45) {};
\node[isosceles triangle, minimum width=4.2mm, minimum height=2.5mm, inner sep=0pt, outer sep=0pt, isosceles triangle stretches, shape border rotate=90, draw=tncol8ecae6, line width=1.0, fill=tncolcaffbf, text=black] (t2) at (0,0.45) {};
\node[isosceles triangle, minimum width=4.2mm, minimum height=2.5mm, inner sep=0pt, outer sep=0pt, isosceles triangle stretches, shape border rotate=90, draw=tncolb5e48c, line width=1.0, fill=tncola0c4ff, text=black] (t3) at (0.5,0.45) {};
\node[rectangle, minimum size=9mm, inner sep=2pt, outer sep=0pt, rounded corners=1.6pt, draw=black, line width=0.6pt, fill=tncol2281aa, text=black, minimum height=2mm, minimum width=15mm, font=\small] (t4) at (0,-1.25) {$\gamma_{\mu_B^{(2)}}$};
\node[isosceles triangle, minimum width=4.2mm, minimum height=2.5mm, inner sep=0pt, outer sep=0pt, isosceles triangle stretches, shape border rotate=90, draw=tncolf4a261, line width=1.0, fill=tncolffd6a5, text=black] (t5) at (-0.5,-0.8) {};
\node[isosceles triangle, minimum width=4.2mm, minimum height=2.5mm, inner sep=0pt, outer sep=0pt, isosceles triangle stretches, shape border rotate=90, draw=tncol8ecae6, line width=1.0, fill=tncolfdffb6, text=black] (t6) at (0,-0.8) {};
\node[isosceles triangle, minimum width=4.2mm, minimum height=2.5mm, inner sep=0pt, outer sep=0pt, isosceles triangle stretches, shape border rotate=90, draw=tncolb5e48c, line width=1.0, fill=tncol9bf6ff, text=black] (t7) at (0.5,-0.8) {};
\node[rectangle, minimum size=9mm, inner sep=2pt, outer sep=0pt, rounded corners=1.6pt, draw=black, line width=0.6pt, fill=tncol2281aa, text=black, minimum height=2mm, minimum width=15mm, font=\small] (t8) at (0,-2.5) {$\gamma_{\mu_B^{(3)}}$};
\node[isosceles triangle, minimum width=4.2mm, minimum height=2.5mm, inner sep=0pt, outer sep=0pt, isosceles triangle stretches, shape border rotate=90, draw=tncolf4a261, line width=1.0, fill=tncolffadad, text=black] (t9) at (-0.5,-2.05) {};
\node[isosceles triangle, minimum width=4.2mm, minimum height=2.5mm, inner sep=0pt, outer sep=0pt, isosceles triangle stretches, shape border rotate=90, draw=tncol8ecae6, line width=1.0, fill=tncolffd6a5, text=black] (t10) at (0,-2.05) {};
\node[isosceles triangle, minimum width=4.2mm, minimum height=2.5mm, inner sep=0pt, outer sep=0pt, isosceles triangle stretches, shape border rotate=90, draw=tncolb5e48c, line width=1.0, fill=tncolfdffb6, text=black] (t11) at (0.5,-2.05) {};
\node[rectangle, minimum size=9mm, inner sep=2pt, outer sep=0pt, rounded corners=1.6pt, draw=black, line width=0.6pt, fill=tncol2281aa, text=black, minimum height=2mm, minimum width=15mm, font=\small] (t12) at (0,-3.75) {$\gamma_{\mu_B^{(4)}}$};
\node[isosceles triangle, minimum width=4.2mm, minimum height=2.5mm, inner sep=0pt, outer sep=0pt, isosceles triangle stretches, shape border rotate=90, draw=tncolf4a261, line width=1.0, fill=tncolcaffbf, text=black] (t13) at (-0.5,-3.3) {};
\node[isosceles triangle, minimum width=4.2mm, minimum height=2.5mm, inner sep=0pt, outer sep=0pt, isosceles triangle stretches, shape border rotate=90, draw=tncol8ecae6, line width=1.0, fill=tncol9bf6ff, text=black] (t14) at (0,-3.3) {};
\node[isosceles triangle, minimum width=4.2mm, minimum height=2.5mm, inner sep=0pt, outer sep=0pt, isosceles triangle stretches, shape border rotate=90, draw=tncolb5e48c, line width=1.0, fill=tncola0c4ff, text=black] (t15) at (0.5,-3.3) {};
\node[inner sep=0pt, minimum size=0pt, outer sep=0pt] (t16) at (0,-4.9) {};
\draw[draw=tncol33415c, line width=0.6pt] ($(t0.north) + (0:-0.5)$) -- (t1.270);
\draw[draw=tncol33415c, line width=0.6pt] (t0.north) -- (t2.270);
\draw[draw=tncol33415c, line width=0.6pt] ($(t0.north) + (0:0.5)$) -- (t3.270);
\draw[draw=tncol33415c, line width=0.6pt] ($(t4.north) + (0:-0.5)$) -- (t5.270);
\draw[draw=tncol33415c, line width=0.6pt] (t4.north) -- (t6.270);
\draw[draw=tncol33415c, line width=0.6pt] ($(t4.north) + (0:0.5)$) -- (t7.270);
\draw[draw=tncol33415c, line width=0.6pt] ($(t8.north) + (0:-0.5)$) -- (t9.270);
\draw[draw=tncol33415c, line width=0.6pt] (t8.north) -- (t10.270);
\draw[draw=tncol33415c, line width=0.6pt] ($(t8.north) + (0:0.5)$) -- (t11.270);
\draw[draw=tncol33415c, line width=0.6pt] ($(t12.north) + (0:-0.5)$) -- (t13.270);
\draw[draw=tncol33415c, line width=0.6pt] (t12.north) -- (t14.270);
\draw[draw=tncol33415c, line width=0.6pt] ($(t12.north) + (0:0.5)$) -- (t15.270);
\draw[draw=tncol33415c, line width=0.6pt] (t0.west) -- ++(180:0.25);
\draw[draw=tncol33415c, line width=0.6pt] (t0.east) -- ++(0:0.25);
\draw[draw=tncol33415c, line width=0.6pt] (t4.west) -- ++(180:0.25);
\draw[draw=tncol33415c, line width=0.6pt] (t4.east) -- ++(0:0.25);
\draw[draw=tncol33415c, line width=0.6pt] (t8.west) -- ++(180:0.25);
\draw[draw=tncol33415c, line width=0.6pt] (t8.east) -- ++(0:0.25);
\draw[draw=tncol33415c, line width=0.6pt] (t12.west) -- ++(180:0.25);
\draw[draw=tncol33415c, line width=0.6pt] (t12.east) -- ++(0:0.25);
\node[anchor=south, inner sep=1pt, font=\scriptsize] at (-0.500,0.700) {$c_{1}$};
\node[anchor=south, inner sep=1pt, font=\scriptsize] at (0.000,0.700) {$c_{4}$};
\node[anchor=south, inner sep=1pt, font=\scriptsize] at (0.500,0.700) {$c_{6}$};
\node[anchor=south, inner sep=1pt, font=\scriptsize] at (-0.500,-0.550) {$c_{2}$};
\node[anchor=south, inner sep=1pt, font=\scriptsize] at (0.000,-0.550) {$c_{3}$};
\node[anchor=south, inner sep=1pt, font=\scriptsize] at (0.500,-0.550) {$c_{5}$};
\node[anchor=south, inner sep=1pt, font=\scriptsize] at (-0.500,-1.800) {$c_{1}$};
\node[anchor=south, inner sep=1pt, font=\scriptsize] at (0.000,-1.800) {$c_{2}$};
\node[anchor=south, inner sep=1pt, font=\scriptsize] at (0.500,-1.800) {$c_{3}$};
\node[anchor=south, inner sep=1pt, font=\scriptsize] at (-0.500,-3.050) {$c_{4}$};
\node[anchor=south, inner sep=1pt, font=\scriptsize] at (0.000,-3.050) {$c_{5}$};
\node[anchor=south, inner sep=1pt, font=\scriptsize] at (0.500,-3.050) {$c_{6}$};
\end{tikzpicture}%
}}}%
\endgroup}%
}

%% file: figures/contingency_colour_table_aligned.tex
\providecommand{\ContColourTableAligned}[1][1]{%
\ensuremath{\begingroup\color{black}%
\vcenter{\hbox{%
\scalebox{#1}{%
\begin{tikzpicture}[baseline=(current bounding box.center)]
\node[circle, minimum size=4mm, inner sep=0pt, outer sep=0pt, draw=tncolf4a261, line width=1.6pt, fill=tncolffadad] (t0) at (0,-0) {};
\node[circle, minimum size=4mm, inner sep=0pt, outer sep=0pt, draw=tncol8ecae6, line width=1.6pt, fill=tncolcaffbf] (t1) at (2.1,-0) {};
\node[circle, minimum size=4mm, inner sep=0pt, outer sep=0pt, draw=tncolb5e48c, line width=1.6pt, fill=tncola0c4ff] (t2) at (3.5,-0) {};
\node[circle, minimum size=4mm, inner sep=0pt, outer sep=0pt, draw=tncolf4a261, line width=1.6pt, fill=tncol9bf6ff] (t3) at (2.8,-0.8) {};
\node[circle, minimum size=4mm, inner sep=0pt, outer sep=0pt, draw=tncol8ecae6, line width=1.6pt, fill=tncolffd6a5] (t4) at (0.7,-0.8) {};
\node[circle, minimum size=4mm, inner sep=0pt, outer sep=0pt, draw=tncolb5e48c, line width=1.6pt, fill=tncolfdffb6] (t5) at (1.4,-0.8) {};
\node[circle, minimum size=4mm, inner sep=0pt, outer sep=0pt, draw=tncolf4a261, line width=1.6pt, fill=tncolffadad] (t6) at (0,-1.6) {};
\node[circle, minimum size=4mm, inner sep=0pt, outer sep=0pt, draw=tncol8ecae6, line width=1.6pt, fill=tncolffd6a5] (t7) at (0.7,-1.6) {};
\node[circle, minimum size=4mm, inner sep=0pt, outer sep=0pt, draw=tncolb5e48c, line width=1.6pt, fill=tncolfdffb6] (t8) at (1.4,-1.6) {};
\node[circle, minimum size=4mm, inner sep=0pt, outer sep=0pt, draw=tncolf4a261, line width=1.6pt, fill=tncol9bf6ff] (t9) at (2.8,-2.4) {};
\node[circle, minimum size=4mm, inner sep=0pt, outer sep=0pt, draw=tncol8ecae6, line width=1.6pt, fill=tncolcaffbf] (t10) at (2.1,-2.4) {};
\node[circle, minimum size=4mm, inner sep=0pt, outer sep=0pt, draw=tncolb5e48c, line width=1.6pt, fill=tncola0c4ff] (t11) at (3.5,-2.4) {};
\node[circle, minimum size=1mm, inner sep=0pt, outer sep=0pt, draw=tncol22223b, line width=0.6pt, fill=tncole8edf4, text=black] (t12) at (0.7,-0) {};
\node[circle, minimum size=1mm, inner sep=0pt, outer sep=0pt, draw=tncol22223b, line width=0.6pt, fill=tncole8edf4, text=black] (t13) at (1.4,-0) {};
\node[circle, minimum size=1mm, inner sep=0pt, outer sep=0pt, draw=tncol22223b, line width=0.6pt, fill=tncole8edf4, text=black] (t14) at (2.8,-0) {};
\node[circle, minimum size=1mm, inner sep=0pt, outer sep=0pt, draw=tncol22223b, line width=0.6pt, fill=tncole8edf4, text=black] (t15) at (0,-0.8) {};
\node[circle, minimum size=1mm, inner sep=0pt, outer sep=0pt, draw=tncol22223b, line width=0.6pt, fill=tncole8edf4, text=black] (t16) at (2.1,-0.8) {};
\node[circle, minimum size=1mm, inner sep=0pt, outer sep=0pt, draw=tncol22223b, line width=0.6pt, fill=tncole8edf4, text=black] (t17) at (3.5,-0.8) {};
\node[circle, minimum size=1mm, inner sep=0pt, outer sep=0pt, draw=tncol22223b, line width=0.6pt, fill=tncole8edf4, text=black] (t18) at (2.1,-1.6) {};
\node[circle, minimum size=1mm, inner sep=0pt, outer sep=0pt, draw=tncol22223b, line width=0.6pt, fill=tncole8edf4, text=black] (t19) at (2.8,-1.6) {};
\node[circle, minimum size=1mm, inner sep=0pt, outer sep=0pt, draw=tncol22223b, line width=0.6pt, fill=tncole8edf4, text=black] (t20) at (3.5,-1.6) {};
\node[circle, minimum size=1mm, inner sep=0pt, outer sep=0pt, draw=tncol22223b, line width=0.6pt, fill=tncole8edf4, text=black] (t21) at (0,-2.4) {};
\node[circle, minimum size=1mm, inner sep=0pt, outer sep=0pt, draw=tncol22223b, line width=0.6pt, fill=tncole8edf4, text=black] (t22) at (0.7,-2.4) {};
\node[circle, minimum size=1mm, inner sep=0pt, outer sep=0pt, draw=tncol22223b, line width=0.6pt, fill=tncole8edf4, text=black] (t23) at (1.4,-2.4) {};
\node[inner sep=0pt, minimum size=0pt, outer sep=0pt] (t24) at (0,-2.9) {};
\draw[draw=tncol555555, line width=0.6pt] (t0.-90) to[bend right=53.8838] (t6.90);
\draw[draw=tncol555555, line width=0.6pt] (t4.-90) -- (t7.90);
\draw[draw=tncol555555, line width=0.6pt] (t5.-90) -- (t8.90);
\draw[draw=tncol555555, line width=0.6pt] (t1.-90) to[bend right=32.5849] (t10.90);
\draw[draw=tncol555555, line width=0.6pt] (t3.-90) to[bend right=53.8838] (t9.90);
\draw[draw=tncol555555, line width=0.6pt] (t2.-90) to[bend right=32.5849] (t11.90);
\node[anchor=center, inner sep=1pt, font=\small] at (0.000,0.500) {$c_{1}$};
\node[anchor=center, inner sep=1pt, font=\small] at (0.700,0.500) {$c_{2}$};
\node[anchor=center, inner sep=1pt, font=\small] at (1.400,0.500) {$c_{3}$};
\node[anchor=center, inner sep=1pt, font=\small] at (2.100,0.500) {$c_{4}$};
\node[anchor=center, inner sep=1pt, font=\small] at (2.800,0.500) {$c_{5}$};
\node[anchor=center, inner sep=1pt, font=\small] at (3.500,0.500) {$c_{6}$};
\end{tikzpicture}%
}}}%
\endgroup}%
}

%% file: figures/contingency_colour_table_sorted.tex
\providecommand{\ContColourTableSorted}[1][1]{%
\ensuremath{\begingroup\color{black}%
\vcenter{\hbox{%
\scalebox{#1}{%
\begin{tikzpicture}[baseline=(current bounding box.center)]
\node[circle, minimum size=4mm, inner sep=0pt, outer sep=0pt, draw=tncolf4a261, line width=1.6pt, fill=tncolffadad] (t0) at (0,-0) {};
\node[circle, minimum size=4mm, inner sep=0pt, outer sep=0pt, draw=tncol8ecae6, line width=1.6pt, fill=tncolcaffbf] (t1) at (2.1,-0) {};
\node[circle, minimum size=4mm, inner sep=0pt, outer sep=0pt, draw=tncolb5e48c, line width=1.6pt, fill=tncola0c4ff] (t2) at (3.5,-0) {};
\node[circle, minimum size=4mm, inner sep=0pt, outer sep=0pt, draw=tncolf4a261, line width=1.6pt, fill=tncolffd6a5] (t3) at (0.7,-0.8) {};
\node[circle, minimum size=4mm, inner sep=0pt, outer sep=0pt, draw=tncol8ecae6, line width=1.6pt, fill=tncolfdffb6] (t4) at (1.4,-0.8) {};
\node[circle, minimum size=4mm, inner sep=0pt, outer sep=0pt, draw=tncolb5e48c, line width=1.6pt, fill=tncol9bf6ff] (t5) at (2.8,-0.8) {};
\node[circle, minimum size=4mm, inner sep=0pt, outer sep=0pt, draw=tncolf4a261, line width=1.6pt, fill=tncolffadad] (t6) at (0,-1.6) {};
\node[circle, minimum size=4mm, inner sep=0pt, outer sep=0pt, draw=tncol8ecae6, line width=1.6pt, fill=tncolffd6a5] (t7) at (0.7,-1.6) {};
\node[circle, minimum size=4mm, inner sep=0pt, outer sep=0pt, draw=tncolb5e48c, line width=1.6pt, fill=tncolfdffb6] (t8) at (1.4,-1.6) {};
\node[circle, minimum size=4mm, inner sep=0pt, outer sep=0pt, draw=tncolf4a261, line width=1.6pt, fill=tncolcaffbf] (t9) at (2.1,-2.4) {};
\node[circle, minimum size=4mm, inner sep=0pt, outer sep=0pt, draw=tncol8ecae6, line width=1.6pt, fill=tncol9bf6ff] (t10) at (2.8,-2.4) {};
\node[circle, minimum size=4mm, inner sep=0pt, outer sep=0pt, draw=tncolb5e48c, line width=1.6pt, fill=tncola0c4ff] (t11) at (3.5,-2.4) {};
\node[circle, minimum size=1mm, inner sep=0pt, outer sep=0pt, draw=tncol22223b, line width=0.6pt, fill=tncole8edf4, text=black] (t12) at (0.7,-0) {};
\node[circle, minimum size=1mm, inner sep=0pt, outer sep=0pt, draw=tncol22223b, line width=0.6pt, fill=tncole8edf4, text=black] (t13) at (1.4,-0) {};
\node[circle, minimum size=1mm, inner sep=0pt, outer sep=0pt, draw=tncol22223b, line width=0.6pt, fill=tncole8edf4, text=black] (t14) at (2.8,-0) {};
\node[circle, minimum size=1mm, inner sep=0pt, outer sep=0pt, draw=tncol22223b, line width=0.6pt, fill=tncole8edf4, text=black] (t15) at (0,-0.8) {};
\node[circle, minimum size=1mm, inner sep=0pt, outer sep=0pt, draw=tncol22223b, line width=0.6pt, fill=tncole8edf4, text=black] (t16) at (2.1,-0.8) {};
\node[circle, minimum size=1mm, inner sep=0pt, outer sep=0pt, draw=tncol22223b, line width=0.6pt, fill=tncole8edf4, text=black] (t17) at (3.5,-0.8) {};
\node[circle, minimum size=1mm, inner sep=0pt, outer sep=0pt, draw=tncol22223b, line width=0.6pt, fill=tncole8edf4, text=black] (t18) at (2.1,-1.6) {};
\node[circle, minimum size=1mm, inner sep=0pt, outer sep=0pt, draw=tncol22223b, line width=0.6pt, fill=tncole8edf4, text=black] (t19) at (2.8,-1.6) {};
\node[circle, minimum size=1mm, inner sep=0pt, outer sep=0pt, draw=tncol22223b, line width=0.6pt, fill=tncole8edf4, text=black] (t20) at (3.5,-1.6) {};
\node[circle, minimum size=1mm, inner sep=0pt, outer sep=0pt, draw=tncol22223b, line width=0.6pt, fill=tncole8edf4, text=black] (t21) at (0,-2.4) {};
\node[circle, minimum size=1mm, inner sep=0pt, outer sep=0pt, draw=tncol22223b, line width=0.6pt, fill=tncole8edf4, text=black] (t22) at (0.7,-2.4) {};
\node[circle, minimum size=1mm, inner sep=0pt, outer sep=0pt, draw=tncol22223b, line width=0.6pt, fill=tncole8edf4, text=black] (t23) at (1.4,-2.4) {};
\node[inner sep=0pt, minimum size=0pt, outer sep=0pt] (t24) at (0,-2.9) {};
\draw[draw=tncol555555, line width=0.6pt] (t0.-90) to[bend right=53.8838] (t6.90);
\draw[draw=tncol9b2226, line width=0.6pt, densely dashed] (t3.-90) -- (t7.90);
\draw[draw=tncol9b2226, line width=0.6pt, densely dashed] (t4.-90) -- (t8.90);
\draw[draw=tncol9b2226, line width=0.6pt, densely dashed] (t1.-90) to[bend right=32.5849] (t9.90);
\draw[draw=tncol9b2226, line width=0.6pt, densely dashed] (t5.-90) to[bend right=53.8838] (t10.90);
\draw[draw=tncol555555, line width=0.6pt] (t2.-90) to[bend right=32.5849] (t11.90);
\node[anchor=center, inner sep=1pt, font=\small] at (0.000,0.500) {$c_{1}$};
\node[anchor=center, inner sep=1pt, font=\small] at (0.700,0.500) {$c_{2}$};
\node[anchor=center, inner sep=1pt, font=\small] at (1.400,0.500) {$c_{3}$};
\node[anchor=center, inner sep=1pt, font=\small] at (2.100,0.500) {$c_{4}$};
\node[anchor=center, inner sep=1pt, font=\small] at (2.800,0.500) {$c_{5}$};
\node[anchor=center, inner sep=1pt, font=\small] at (3.500,0.500) {$c_{6}$};
\end{tikzpicture}%
}}}%
\endgroup}%
}

%% file: figures/contingency_eg_exp_orbit.tex
\providecommand{\ContOrbOne}[1][1]{%
\ensuremath{\begingroup\color{black}%
\vcenter{\hbox{%
\scalebox{#1}{%
\begin{tikzpicture}[baseline=(current bounding box.center)]
\node[circle, minimum size=4mm, inner sep=0pt, outer sep=0pt, draw=tncol22223b, line width=0.6pt, fill=tncolffadad] (t0) at (0,-0) {};
\node[circle, minimum size=1mm, inner sep=0pt, outer sep=0pt, draw=tncol22223b, line width=0.6pt, fill=tncole8edf4, text=black] (t1) at (0.7,-0) {};
\node[circle, minimum size=1mm, inner sep=0pt, outer sep=0pt, draw=tncol22223b, line width=0.6pt, fill=tncole8edf4, text=black] (t2) at (1.4,-0) {};
\node[circle, minimum size=4mm, inner sep=0pt, outer sep=0pt, draw=tncol22223b, line width=0.6pt, fill=tncolcaffbf] (t3) at (2.1,-0) {};
\node[circle, minimum size=1mm, inner sep=0pt, outer sep=0pt, draw=tncol22223b, line width=0.6pt, fill=tncole8edf4, text=black] (t4) at (2.8,-0) {};
\node[circle, minimum size=4mm, inner sep=0pt, outer sep=0pt, draw=tncol22223b, line width=0.6pt, fill=tncola0c4ff] (t5) at (3.5,-0) {};
\node[circle, minimum size=1mm, inner sep=0pt, outer sep=0pt, draw=tncol22223b, line width=0.6pt, fill=tncole8edf4, text=black] (t6) at (0,-0.7) {};
\node[circle, minimum size=4mm, inner sep=0pt, outer sep=0pt, draw=tncol22223b, line width=0.6pt, fill=tncolffd6a5] (t7) at (0.7,-0.7) {};
\node[circle, minimum size=4mm, inner sep=0pt, outer sep=0pt, draw=tncol22223b, line width=0.6pt, fill=tncolfdffb6] (t8) at (1.4,-0.7) {};
\node[circle, minimum size=1mm, inner sep=0pt, outer sep=0pt, draw=tncol22223b, line width=0.6pt, fill=tncole8edf4, text=black] (t9) at (2.1,-0.7) {};
\node[circle, minimum size=4mm, inner sep=0pt, outer sep=0pt, draw=tncol22223b, line width=0.6pt, fill=tncol9bf6ff] (t10) at (2.8,-0.7) {};
\node[circle, minimum size=1mm, inner sep=0pt, outer sep=0pt, draw=tncol22223b, line width=0.6pt, fill=tncole8edf4, text=black] (t11) at (3.5,-0.7) {};
\node[circle, minimum size=4mm, inner sep=0pt, outer sep=0pt, draw=tncol22223b, line width=0.6pt, fill=tncolffadad] (t12) at (0,-1.4) {};
\node[circle, minimum size=4mm, inner sep=0pt, outer sep=0pt, draw=tncol22223b, line width=0.6pt, fill=tncolffd6a5] (t13) at (0.7,-1.4) {};
\node[circle, minimum size=4mm, inner sep=0pt, outer sep=0pt, draw=tncol22223b, line width=0.6pt, fill=tncolfdffb6] (t14) at (1.4,-1.4) {};
\node[circle, minimum size=1mm, inner sep=0pt, outer sep=0pt, draw=tncol22223b, line width=0.6pt, fill=tncole8edf4, text=black] (t15) at (2.1,-1.4) {};
\node[circle, minimum size=1mm, inner sep=0pt, outer sep=0pt, draw=tncol22223b, line width=0.6pt, fill=tncole8edf4, text=black] (t16) at (2.8,-1.4) {};
\node[circle, minimum size=1mm, inner sep=0pt, outer sep=0pt, draw=tncol22223b, line width=0.6pt, fill=tncole8edf4, text=black] (t17) at (3.5,-1.4) {};
\node[circle, minimum size=1mm, inner sep=0pt, outer sep=0pt, draw=tncol22223b, line width=0.6pt, fill=tncole8edf4, text=black] (t18) at (0,-2.1) {};
\node[circle, minimum size=1mm, inner sep=0pt, outer sep=0pt, draw=tncol22223b, line width=0.6pt, fill=tncole8edf4, text=black] (t19) at (0.7,-2.1) {};
\node[circle, minimum size=1mm, inner sep=0pt, outer sep=0pt, draw=tncol22223b, line width=0.6pt, fill=tncole8edf4, text=black] (t20) at (1.4,-2.1) {};
\node[circle, minimum size=4mm, inner sep=0pt, outer sep=0pt, draw=tncol22223b, line width=0.6pt, fill=tncolcaffbf] (t21) at (2.1,-2.1) {};
\node[circle, minimum size=4mm, inner sep=0pt, outer sep=0pt, draw=tncol22223b, line width=0.6pt, fill=tncol9bf6ff] (t22) at (2.8,-2.1) {};
\node[circle, minimum size=4mm, inner sep=0pt, outer sep=0pt, draw=tncol22223b, line width=0.6pt, fill=tncola0c4ff] (t23) at (3.5,-2.1) {};
\end{tikzpicture}%
}}}%
\endgroup}%
}

%% file: figures/contingency_eg_orbit_consolid.tex
\providecommand{\ContOrbTwo}[1][1]{%
\ensuremath{\begingroup\color{black}%
\vcenter{\hbox{%
\scalebox{#1}{%
\begin{tikzpicture}[baseline=(current bounding box.center)]
\node[circle, minimum size=4mm, inner sep=0pt, outer sep=0pt, draw=tncol22223b, line width=0.6pt, fill=tncolffadad] (t0) at (0,-0) {};
\node[circle, minimum size=4mm, inner sep=0pt, outer sep=0pt, draw=tncol22223b, line width=0.6pt, fill=tncolcaffbf] (t1) at (0.7,-0) {};
\node[circle, minimum size=4mm, inner sep=0pt, outer sep=0pt, draw=tncol22223b, line width=0.6pt, fill=tncola0c4ff] (t2) at (1.4,-0) {};
\node[circle, minimum size=4mm, inner sep=0pt, outer sep=0pt, draw=tncol22223b, line width=0.6pt, fill=tncolffd6a5] (t3) at (0,-0.7) {};
\node[circle, minimum size=4mm, inner sep=0pt, outer sep=0pt, draw=tncol22223b, line width=0.6pt, fill=tncolfdffb6] (t4) at (0.7,-0.7) {};
\node[circle, minimum size=4mm, inner sep=0pt, outer sep=0pt, draw=tncol22223b, line width=0.6pt, fill=tncol9bf6ff] (t5) at (1.4,-0.7) {};
\node[circle, minimum size=4mm, inner sep=0pt, outer sep=0pt, draw=tncol22223b, line width=0.6pt, fill=tncolffadad] (t6) at (0,-1.4) {};
\node[circle, minimum size=4mm, inner sep=0pt, outer sep=0pt, draw=tncol22223b, line width=0.6pt, fill=tncolffd6a5] (t7) at (0.7,-1.4) {};
\node[circle, minimum size=4mm, inner sep=0pt, outer sep=0pt, draw=tncol22223b, line width=0.6pt, fill=tncolfdffb6] (t8) at (1.4,-1.4) {};
\node[circle, minimum size=4mm, inner sep=0pt, outer sep=0pt, draw=tncol22223b, line width=0.6pt, fill=tncolcaffbf] (t9) at (0,-2.1) {};
\node[circle, minimum size=4mm, inner sep=0pt, outer sep=0pt, draw=tncol22223b, line width=0.6pt, fill=tncol9bf6ff] (t10) at (0.7,-2.1) {};
\node[circle, minimum size=4mm, inner sep=0pt, outer sep=0pt, draw=tncol22223b, line width=0.6pt, fill=tncola0c4ff] (t11) at (1.4,-2.1) {};
\node[anchor=center,font={\small},text=tncol273043] at (-0.500,0.000) {$\sigma_{1}$};
\node[anchor=center,font={\small},text=tncol273043] at (-0.500,-0.700) {$\sigma_{2}$};
\node[anchor=center,font={\small},text=tncol273043] at (-0.500,-1.400) {$\sigma_{3}$};
\node[anchor=center,font={\small},text=tncol273043] at (-0.500,-2.100) {$\sigma_{4}$};
\begin{pgfonlayer}{background}
\node[draw=tncolffd6a5, line width=0.6pt, inner sep=3pt, fill=tncolffd6a5, fill opacity=0.1, rounded corners=3pt, fit=(t0)(t1)(t2)] {};
\node[draw=tncolffd6a5, line width=0.6pt, inner sep=3pt, fill=tncolffd6a5, fill opacity=0.1, rounded corners=3pt, fit=(t3)(t4)(t5)] {};
\node[draw=tncolffd6a5, line width=0.6pt, inner sep=3pt, fill=tncolffd6a5, fill opacity=0.1, rounded corners=3pt, fit=(t6)(t7)(t8)] {};
\node[draw=tncolffd6a5, line width=0.6pt, inner sep=3pt, fill=tncolffd6a5, fill opacity=0.1, rounded corners=3pt, fit=(t9)(t10)(t11)] {};
\end{pgfonlayer}
\end{tikzpicture}%
}}}%
\endgroup}%
}

%% file: figures/contingency_eg_pairing_table.tex
\providecommand{\ContPairTable}[1][1]{%
\ensuremath{\begingroup\color{black}%
\vcenter{\hbox{%
\scalebox{#1}{%
\begin{tikzpicture}[baseline=(current bounding box.center)]
\node[circle, minimum size=4mm, inner sep=0pt, outer sep=0pt, draw=tncol22223b, line width=0.6pt, fill=tncola0c4ff] (t0) at (0,-0) {};
\node[circle, minimum size=4mm, inner sep=0pt, outer sep=0pt, draw=tncol22223b, line width=0.6pt, fill=tncolffadad] (t1) at (0.7,-0) {};
\node[circle, minimum size=4mm, inner sep=0pt, outer sep=0pt, draw=tncol22223b, line width=0.6pt, fill=tncolcaffbf] (t2) at (1.4,-0) {};
\node[circle, minimum size=4mm, inner sep=0pt, outer sep=0pt, draw=tncol22223b, line width=0.6pt, fill=tncolfdffb6] (t3) at (0,-0.7) {};
\node[circle, minimum size=4mm, inner sep=0pt, outer sep=0pt, draw=tncol22223b, line width=0.6pt, fill=tncol9bf6ff] (t4) at (0.7,-0.7) {};
\node[circle, minimum size=4mm, inner sep=0pt, outer sep=0pt, draw=tncol22223b, line width=0.6pt, fill=tncolffd6a5] (t5) at (1.4,-0.7) {};
\node[circle, minimum size=4mm, inner sep=0pt, outer sep=0pt, draw=tncol22223b, line width=0.6pt, fill=tncolffd6a5] (t6) at (0,-1.4) {};
\node[circle, minimum size=4mm, inner sep=0pt, outer sep=0pt, draw=tncol22223b, line width=0.6pt, fill=tncolfdffb6] (t7) at (0.7,-1.4) {};
\node[circle, minimum size=4mm, inner sep=0pt, outer sep=0pt, draw=tncol22223b, line width=0.6pt, fill=tncolffadad] (t8) at (1.4,-1.4) {};
\node[circle, minimum size=4mm, inner sep=0pt, outer sep=0pt, draw=tncol22223b, line width=0.6pt, fill=tncolcaffbf] (t9) at (0,-2.1) {};
\node[circle, minimum size=4mm, inner sep=0pt, outer sep=0pt, draw=tncol22223b, line width=0.6pt, fill=tncola0c4ff] (t10) at (0.7,-2.1) {};
\node[circle, minimum size=4mm, inner sep=0pt, outer sep=0pt, draw=tncol22223b, line width=0.6pt, fill=tncol9bf6ff] (t11) at (1.4,-2.1) {};
\draw[draw=tncol555555, line width=0.6pt] (t1.-63) -- (t8.117);
\draw[draw=tncol555555, line width=0.6pt] (t5.-153) -- (t6.27);
\draw[draw=tncol555555, line width=0.6pt] (t3.-45) -- (t7.135);
\draw[draw=tncol555555, line width=0.6pt] (t2.-124) -- (t9.56);
\draw[draw=tncol555555, line width=0.6pt] (t4.-63) -- (t11.117);
\draw[draw=tncol555555, line width=0.6pt] (t0.-72) -- (t10.108);
\end{tikzpicture}%
}}}%
\endgroup}%
}

%% file: figures/contingency_eg_pairings_aligned.tex
\providecommand{\ContPairFour}[1][1]{%
\ensuremath{\begingroup\color{black}%
\vcenter{\hbox{%
\scalebox{#1}{%
\begin{tikzpicture}[baseline=(current bounding box.center)]
\node[circle, minimum size=4mm, inner sep=0pt, outer sep=0pt, draw=tncol22223b, line width=0.6pt, fill=tncolffadad] (t0) at (0,-0) {};
\node[circle, minimum size=4mm, inner sep=0pt, outer sep=0pt, draw=tncol22223b, line width=0.6pt, fill=tncolcaffbf] (t1) at (0.7,-0) {};
\node[circle, minimum size=4mm, inner sep=0pt, outer sep=0pt, draw=tncol22223b, line width=0.6pt, fill=tncola0c4ff] (t2) at (1.4,-0) {};
\node[circle, minimum size=4mm, inner sep=0pt, outer sep=0pt, draw=tncol22223b, line width=0.6pt, fill=tncol9bf6ff] (t3) at (0,-0.7) {};
\node[circle, minimum size=4mm, inner sep=0pt, outer sep=0pt, draw=tncol22223b, line width=0.6pt, fill=tncolffd6a5] (t4) at (0.7,-0.7) {};
\node[circle, minimum size=4mm, inner sep=0pt, outer sep=0pt, draw=tncol22223b, line width=0.6pt, fill=tncolfdffb6] (t5) at (1.4,-0.7) {};
\node[circle, minimum size=4mm, inner sep=0pt, outer sep=0pt, draw=tncol22223b, line width=0.6pt, fill=tncolffadad] (t6) at (0,-1.4) {};
\node[circle, minimum size=4mm, inner sep=0pt, outer sep=0pt, draw=tncol22223b, line width=0.6pt, fill=tncolffd6a5] (t7) at (0.7,-1.4) {};
\node[circle, minimum size=4mm, inner sep=0pt, outer sep=0pt, draw=tncol22223b, line width=0.6pt, fill=tncolfdffb6] (t8) at (1.4,-1.4) {};
\node[circle, minimum size=4mm, inner sep=0pt, outer sep=0pt, draw=tncol22223b, line width=0.6pt, fill=tncol9bf6ff] (t9) at (0,-2.1) {};
\node[circle, minimum size=4mm, inner sep=0pt, outer sep=0pt, draw=tncol22223b, line width=0.6pt, fill=tncolcaffbf] (t10) at (0.7,-2.1) {};
\node[circle, minimum size=4mm, inner sep=0pt, outer sep=0pt, draw=tncol22223b, line width=0.6pt, fill=tncola0c4ff] (t11) at (1.4,-2.1) {};
\draw[draw=tncol555555, line width=0.6pt] (t0.-90) to[bend right=67.4024] (t6.90);
\draw[draw=tncol555555, line width=0.6pt] (t4.-90) -- (t7.90);
\draw[draw=tncol555555, line width=0.6pt] (t5.-90) -- (t8.90);
\draw[draw=tncol555555, line width=0.6pt] (t1.-90) to[bend right=37.9871] (t10.90);
\draw[draw=tncol555555, line width=0.6pt] (t3.-90) to[bend left=67.4024] (t9.90);
\draw[draw=tncol555555, line width=0.6pt] (t2.-90) to[bend right=37.9871] (t11.90);
\node[anchor=center,font={\small},text=tncol273043] at (-0.500,0.000) {$\sigma_{1}$};
\node[anchor=center,font={\small},text=tncol273043] at (-0.500,-0.700) {$\sigma_{2}$};
\node[anchor=center,font={\small},text=tncol273043] at (-0.500,-1.400) {$\sigma_{3}$};
\node[anchor=center,font={\small},text=tncol273043] at (-0.500,-2.100) {$\sigma_{4}$};
\begin{pgfonlayer}{background}
\node[draw=tncolffd6a5, line width=0.6pt, inner sep=3pt, fill=tncolffd6a5, fill opacity=0.1, rounded corners=3pt, fit=(t0)(t1)(t2)] {};
\node[draw=tncolffd6a5, line width=0.6pt, inner sep=3pt, fill=tncolffd6a5, fill opacity=0.1, rounded corners=3pt, fit=(t3)(t4)(t5)] {};
\node[draw=tncolffd6a5, line width=0.6pt, inner sep=3pt, fill=tncolffd6a5, fill opacity=0.1, rounded corners=3pt, fit=(t6)(t7)(t8)] {};
\node[draw=tncolffd6a5, line width=0.6pt, inner sep=3pt, fill=tncolffd6a5, fill opacity=0.1, rounded corners=3pt, fit=(t9)(t10)(t11)] {};
\end{pgfonlayer}
\end{tikzpicture}%
}}}%
\endgroup}%
}

%% file: figures/contingency_eg_pairings_partial.tex
\providecommand{\ContPairThree}[1][1]{%
\ensuremath{\begingroup\color{black}%
\vcenter{\hbox{%
\scalebox{#1}{%
\begin{tikzpicture}[baseline=(current bounding box.center)]
\node[circle, minimum size=4mm, inner sep=0pt, outer sep=0pt, draw=tncol22223b, line width=0.6pt, fill=tncolffadad] (t0) at (0,-0) {};
\node[circle, minimum size=4mm, inner sep=0pt, outer sep=0pt, draw=tncol22223b, line width=0.6pt, fill=tncolcaffbf] (t1) at (0.7,-0) {};
\node[circle, minimum size=4mm, inner sep=0pt, outer sep=0pt, draw=tncol22223b, line width=0.6pt, fill=tncola0c4ff] (t2) at (1.4,-0) {};
\node[circle, minimum size=4mm, inner sep=0pt, outer sep=0pt, draw=tncol22223b, line width=0.6pt, fill=tncol9bf6ff] (t3) at (0,-0.7) {};
\node[circle, minimum size=4mm, inner sep=0pt, outer sep=0pt, draw=tncol22223b, line width=0.6pt, fill=tncolffd6a5] (t4) at (0.7,-0.7) {};
\node[circle, minimum size=4mm, inner sep=0pt, outer sep=0pt, draw=tncol22223b, line width=0.6pt, fill=tncolfdffb6] (t5) at (1.4,-0.7) {};
\node[circle, minimum size=4mm, inner sep=0pt, outer sep=0pt, draw=tncol22223b, line width=0.6pt, fill=tncolffadad] (t6) at (0,-1.4) {};
\node[circle, minimum size=4mm, inner sep=0pt, outer sep=0pt, draw=tncol22223b, line width=0.6pt, fill=tncolffd6a5] (t7) at (0.7,-1.4) {};
\node[circle, minimum size=4mm, inner sep=0pt, outer sep=0pt, draw=tncol22223b, line width=0.6pt, fill=tncolfdffb6] (t8) at (1.4,-1.4) {};
\node[circle, minimum size=4mm, inner sep=0pt, outer sep=0pt, draw=tncol22223b, line width=0.6pt, fill=tncolcaffbf] (t9) at (0,-2.1) {};
\node[circle, minimum size=4mm, inner sep=0pt, outer sep=0pt, draw=tncol22223b, line width=0.6pt, fill=tncol9bf6ff] (t10) at (0.7,-2.1) {};
\node[circle, minimum size=4mm, inner sep=0pt, outer sep=0pt, draw=tncol22223b, line width=0.6pt, fill=tncola0c4ff] (t11) at (1.4,-2.1) {};
\draw[draw=tncol555555, line width=0.6pt] (t0.-90) to[bend right=67.4024] (t6.90);
\draw[draw=tncol555555, line width=0.6pt] (t4.-90) -- (t7.90);
\draw[draw=tncol555555, line width=0.6pt] (t5.-90) -- (t8.90);
\draw[draw=tncol555555, line width=0.6pt] (t1.-108) -- (t9.72);
\draw[draw=tncol555555, line width=0.6pt] (t3.-63) -- (t10.117);
\draw[draw=tncol555555, line width=0.6pt] (t2.-90) to[bend right=37.9871] (t11.90);
\node[anchor=center,font={\small},text=tncol273043] at (-0.500,0.000) {$\sigma_{1}$};
\node[anchor=center,font={\small},text=tncol273043] at (-0.500,-0.700) {$\sigma_{2}$};
\node[anchor=center,font={\small},text=tncol273043] at (-0.500,-1.400) {$\sigma_{3}$};
\node[anchor=center,font={\small},text=tncol273043] at (-0.500,-2.100) {$\sigma_{4}$};
\begin{pgfonlayer}{background}
\node[draw=tncolffd6a5, line width=0.6pt, inner sep=3pt, fill=tncolffd6a5, fill opacity=0.1, rounded corners=3pt, fit=(t0)(t1)(t2)] {};
\node[draw=tncolffd6a5, line width=0.6pt, inner sep=3pt, fill=tncolffd6a5, fill opacity=0.1, rounded corners=3pt, fit=(t3)(t4)(t5)] {};
\node[draw=tncolffd6a5, line width=0.6pt, inner sep=3pt, fill=tncolffd6a5, fill opacity=0.1, rounded corners=3pt, fit=(t6)(t7)(t8)] {};
\node[draw=tncolffd6a5, line width=0.6pt, inner sep=3pt, fill=tncolffd6a5, fill opacity=0.1, rounded corners=3pt, fit=(t9)(t10)(t11)] {};
\end{pgfonlayer}
\end{tikzpicture}%
}}}%
\endgroup}%
}

%% file: figures/contingency_eg_pairings_scrambled.tex
\providecommand{\ContPairTwo}[1][1]{%
\ensuremath{\begingroup\color{black}%
\vcenter{\hbox{%
\scalebox{#1}{%
\begin{tikzpicture}[baseline=(current bounding box.center)]
\node[circle, minimum size=4mm, inner sep=0pt, outer sep=0pt, draw=tncol22223b, line width=0.6pt, fill=tncola0c4ff] (t0) at (0,-0) {};
\node[circle, minimum size=4mm, inner sep=0pt, outer sep=0pt, draw=tncol22223b, line width=0.6pt, fill=tncolffadad] (t1) at (0.7,-0) {};
\node[circle, minimum size=4mm, inner sep=0pt, outer sep=0pt, draw=tncol22223b, line width=0.6pt, fill=tncolcaffbf] (t2) at (1.4,-0) {};
\node[circle, minimum size=4mm, inner sep=0pt, outer sep=0pt, draw=tncol22223b, line width=0.6pt, fill=tncolfdffb6] (t3) at (0,-0.7) {};
\node[circle, minimum size=4mm, inner sep=0pt, outer sep=0pt, draw=tncol22223b, line width=0.6pt, fill=tncol9bf6ff] (t4) at (0.7,-0.7) {};
\node[circle, minimum size=4mm, inner sep=0pt, outer sep=0pt, draw=tncol22223b, line width=0.6pt, fill=tncolffd6a5] (t5) at (1.4,-0.7) {};
\node[circle, minimum size=4mm, inner sep=0pt, outer sep=0pt, draw=tncol22223b, line width=0.6pt, fill=tncolffd6a5] (t6) at (0,-1.4) {};
\node[circle, minimum size=4mm, inner sep=0pt, outer sep=0pt, draw=tncol22223b, line width=0.6pt, fill=tncolfdffb6] (t7) at (0.7,-1.4) {};
\node[circle, minimum size=4mm, inner sep=0pt, outer sep=0pt, draw=tncol22223b, line width=0.6pt, fill=tncolffadad] (t8) at (1.4,-1.4) {};
\node[circle, minimum size=4mm, inner sep=0pt, outer sep=0pt, draw=tncol22223b, line width=0.6pt, fill=tncolcaffbf] (t9) at (0,-2.1) {};
\node[circle, minimum size=4mm, inner sep=0pt, outer sep=0pt, draw=tncol22223b, line width=0.6pt, fill=tncola0c4ff] (t10) at (0.7,-2.1) {};
\node[circle, minimum size=4mm, inner sep=0pt, outer sep=0pt, draw=tncol22223b, line width=0.6pt, fill=tncol9bf6ff] (t11) at (1.4,-2.1) {};
\draw[draw=tncol555555, line width=0.6pt] (t1.-63) -- (t8.117);
\draw[draw=tncol555555, line width=0.6pt] (t5.-153) -- (t6.27);
\draw[draw=tncol555555, line width=0.6pt] (t3.-45) -- (t7.135);
\draw[draw=tncol555555, line width=0.6pt] (t2.-124) -- (t9.56);
\draw[draw=tncol555555, line width=0.6pt] (t4.-63) -- (t11.117);
\draw[draw=tncol555555, line width=0.6pt] (t0.-72) -- (t10.108);
\node[anchor=center,font={\small},text=tncol273043] at (-0.500,0.000) {$\sigma_{1}$};
\node[anchor=center,font={\small},text=tncol273043] at (-0.500,-0.700) {$\sigma_{2}$};
\node[anchor=center,font={\small},text=tncol273043] at (-0.500,-1.400) {$\sigma_{3}$};
\node[anchor=center,font={\small},text=tncol273043] at (-0.500,-2.100) {$\sigma_{4}$};
\begin{pgfonlayer}{background}
\node[draw=tncolffd6a5, line width=0.6pt, inner sep=3pt, fill=tncolffd6a5, fill opacity=0.1, rounded corners=3pt, fit=(t0)(t1)(t2)] {};
\node[draw=tncolffd6a5, line width=0.6pt, inner sep=3pt, fill=tncolffd6a5, fill opacity=0.1, rounded corners=3pt, fit=(t3)(t4)(t5)] {};
\node[draw=tncolffd6a5, line width=0.6pt, inner sep=3pt, fill=tncolffd6a5, fill opacity=0.1, rounded corners=3pt, fit=(t6)(t7)(t8)] {};
\node[draw=tncolffd6a5, line width=0.6pt, inner sep=3pt, fill=tncolffd6a5, fill opacity=0.1, rounded corners=3pt, fit=(t9)(t10)(t11)] {};
\end{pgfonlayer}
\end{tikzpicture}%
}}}%
\endgroup}%
}

%% file: figures/contingency_eg_pairings_sorted.tex
\providecommand{\ContPairOne}[1][1]{%
\ensuremath{\begingroup\color{black}%
\vcenter{\hbox{%
\scalebox{#1}{%
\begin{tikzpicture}[baseline=(current bounding box.center)]
\node[circle, minimum size=4mm, inner sep=0pt, outer sep=0pt, draw=tncol22223b, line width=0.6pt, fill=tncolffadad] (t0) at (0,-0) {};
\node[circle, minimum size=4mm, inner sep=0pt, outer sep=0pt, draw=tncol22223b, line width=0.6pt, fill=tncolcaffbf] (t1) at (0.7,-0) {};
\node[circle, minimum size=4mm, inner sep=0pt, outer sep=0pt, draw=tncol22223b, line width=0.6pt, fill=tncola0c4ff] (t2) at (1.4,-0) {};
\node[circle, minimum size=4mm, inner sep=0pt, outer sep=0pt, draw=tncol22223b, line width=0.6pt, fill=tncolffd6a5] (t3) at (0,-0.7) {};
\node[circle, minimum size=4mm, inner sep=0pt, outer sep=0pt, draw=tncol22223b, line width=0.6pt, fill=tncolfdffb6] (t4) at (0.7,-0.7) {};
\node[circle, minimum size=4mm, inner sep=0pt, outer sep=0pt, draw=tncol22223b, line width=0.6pt, fill=tncol9bf6ff] (t5) at (1.4,-0.7) {};
\node[circle, minimum size=4mm, inner sep=0pt, outer sep=0pt, draw=tncol22223b, line width=0.6pt, fill=tncolffadad] (t6) at (0,-1.4) {};
\node[circle, minimum size=4mm, inner sep=0pt, outer sep=0pt, draw=tncol22223b, line width=0.6pt, fill=tncolffd6a5] (t7) at (0.7,-1.4) {};
\node[circle, minimum size=4mm, inner sep=0pt, outer sep=0pt, draw=tncol22223b, line width=0.6pt, fill=tncolfdffb6] (t8) at (1.4,-1.4) {};
\node[circle, minimum size=4mm, inner sep=0pt, outer sep=0pt, draw=tncol22223b, line width=0.6pt, fill=tncolcaffbf] (t9) at (0,-2.1) {};
\node[circle, minimum size=4mm, inner sep=0pt, outer sep=0pt, draw=tncol22223b, line width=0.6pt, fill=tncol9bf6ff] (t10) at (0.7,-2.1) {};
\node[circle, minimum size=4mm, inner sep=0pt, outer sep=0pt, draw=tncol22223b, line width=0.6pt, fill=tncola0c4ff] (t11) at (1.4,-2.1) {};
\draw[draw=tncol555555, line width=0.6pt] (t0.-90) to[bend right=67.4024] (t6.90);
\draw[draw=tncol555555, line width=0.6pt] (t3.-45) -- (t7.135);
\draw[draw=tncol555555, line width=0.6pt] (t4.-45) -- (t8.135);
\draw[draw=tncol555555, line width=0.6pt] (t1.-108) -- (t9.72);
\draw[draw=tncol555555, line width=0.6pt] (t5.-117) -- (t10.63);
\draw[draw=tncol555555, line width=0.6pt] (t2.-90) to[bend right=37.9871] (t11.90);
\node[anchor=center,font={\small},text=tncol273043] at (-0.500,0.000) {$\sigma_{1}$};
\node[anchor=center,font={\small},text=tncol273043] at (-0.500,-0.700) {$\sigma_{2}$};
\node[anchor=center,font={\small},text=tncol273043] at (-0.500,-1.400) {$\sigma_{3}$};
\node[anchor=center,font={\small},text=tncol273043] at (-0.500,-2.100) {$\sigma_{4}$};
\begin{pgfonlayer}{background}
\node[draw=tncolffd6a5, line width=0.6pt, inner sep=3pt, fill=tncolffd6a5, fill opacity=0.1, rounded corners=3pt, fit=(t0)(t1)(t2)] {};
\node[draw=tncolffd6a5, line width=0.6pt, inner sep=3pt, fill=tncolffd6a5, fill opacity=0.1, rounded corners=3pt, fit=(t3)(t4)(t5)] {};
\node[draw=tncolffd6a5, line width=0.6pt, inner sep=3pt, fill=tncolffd6a5, fill opacity=0.1, rounded corners=3pt, fit=(t6)(t7)(t8)] {};
\node[draw=tncolffd6a5, line width=0.6pt, inner sep=3pt, fill=tncolffd6a5, fill opacity=0.1, rounded corners=3pt, fit=(t9)(t10)(t11)] {};
\end{pgfonlayer}
\end{tikzpicture}%
}}}%
\endgroup}%
}

%% file: figures/contingency_eg_strings_canonical.tex
\providecommand{\ContStrCanonical}[1][1]{%
\ensuremath{\begingroup\color{black}%
\vcenter{\hbox{%
\scalebox{#1}{%
\begin{tikzpicture}[baseline=(current bounding box.center)]
\node[rectangle, minimum size=9mm, inner sep=2pt, outer sep=0pt, rounded corners=1.6pt, draw=black, line width=0.6pt, fill=tncol2281aa, text=black, minimum height=2mm, minimum width=15mm, font=\small] (t0) at (0,0) {$\gamma_{\mu_B^{(1)}}$};
\node[isosceles triangle, minimum width=4.2mm, minimum height=2.5mm, inner sep=0pt, outer sep=0pt, isosceles triangle stretches, shape border rotate=90, draw=tncol22223b, line width=0.6pt, fill=tncolffadad, text=black] (t1) at (-0.5,0.45) {};
\node[isosceles triangle, minimum width=4.2mm, minimum height=2.5mm, inner sep=0pt, outer sep=0pt, isosceles triangle stretches, shape border rotate=90, draw=tncol22223b, line width=0.6pt, fill=tncolcaffbf, text=black] (t2) at (0,0.45) {};
\node[isosceles triangle, minimum width=4.2mm, minimum height=2.5mm, inner sep=0pt, outer sep=0pt, isosceles triangle stretches, shape border rotate=90, draw=tncol22223b, line width=0.6pt, fill=tncola0c4ff, text=black] (t3) at (0.5,0.45) {};
\node[rectangle, minimum size=9mm, inner sep=2pt, outer sep=0pt, rounded corners=1.6pt, draw=black, line width=0.6pt, fill=tncol2281aa, text=black, minimum height=2mm, minimum width=15mm, font=\small] (t4) at (0,-1) {$\gamma_{\mu_B^{(2)}}$};
\node[isosceles triangle, minimum width=4.2mm, minimum height=2.5mm, inner sep=0pt, outer sep=0pt, isosceles triangle stretches, shape border rotate=90, draw=tncol22223b, line width=0.6pt, fill=tncolffd6a5, text=black] (t5) at (-0.5,-0.55) {};
\node[isosceles triangle, minimum width=4.2mm, minimum height=2.5mm, inner sep=0pt, outer sep=0pt, isosceles triangle stretches, shape border rotate=90, draw=tncol22223b, line width=0.6pt, fill=tncolfdffb6, text=black] (t6) at (0,-0.55) {};
\node[isosceles triangle, minimum width=4.2mm, minimum height=2.5mm, inner sep=0pt, outer sep=0pt, isosceles triangle stretches, shape border rotate=90, draw=tncol22223b, line width=0.6pt, fill=tncol9bf6ff, text=black] (t7) at (0.5,-0.55) {};
\node[rectangle, minimum size=9mm, inner sep=2pt, outer sep=0pt, rounded corners=1.6pt, draw=black, line width=0.6pt, fill=tncol2281aa, text=black, minimum height=2mm, minimum width=15mm, font=\small] (t8) at (0,-2) {$\gamma_{\mu_B^{(3)}}$};
\node[isosceles triangle, minimum width=4.2mm, minimum height=2.5mm, inner sep=0pt, outer sep=0pt, isosceles triangle stretches, shape border rotate=90, draw=tncol22223b, line width=0.6pt, fill=tncolffadad, text=black] (t9) at (-0.5,-1.55) {};
\node[isosceles triangle, minimum width=4.2mm, minimum height=2.5mm, inner sep=0pt, outer sep=0pt, isosceles triangle stretches, shape border rotate=90, draw=tncol22223b, line width=0.6pt, fill=tncolffd6a5, text=black] (t10) at (0,-1.55) {};
\node[isosceles triangle, minimum width=4.2mm, minimum height=2.5mm, inner sep=0pt, outer sep=0pt, isosceles triangle stretches, shape border rotate=90, draw=tncol22223b, line width=0.6pt, fill=tncolfdffb6, text=black] (t11) at (0.5,-1.55) {};
\node[rectangle, minimum size=9mm, inner sep=2pt, outer sep=0pt, rounded corners=1.6pt, draw=black, line width=0.6pt, fill=tncol2281aa, text=black, minimum height=2mm, minimum width=15mm, font=\small] (t12) at (0,-3) {$\gamma_{\mu_B^{(4)}}$};
\node[isosceles triangle, minimum width=4.2mm, minimum height=2.5mm, inner sep=0pt, outer sep=0pt, isosceles triangle stretches, shape border rotate=90, draw=tncol22223b, line width=0.6pt, fill=tncolcaffbf, text=black] (t13) at (-0.5,-2.55) {};
\node[isosceles triangle, minimum width=4.2mm, minimum height=2.5mm, inner sep=0pt, outer sep=0pt, isosceles triangle stretches, shape border rotate=90, draw=tncol22223b, line width=0.6pt, fill=tncol9bf6ff, text=black] (t14) at (0,-2.55) {};
\node[isosceles triangle, minimum width=4.2mm, minimum height=2.5mm, inner sep=0pt, outer sep=0pt, isosceles triangle stretches, shape border rotate=90, draw=tncol22223b, line width=0.6pt, fill=tncola0c4ff, text=black] (t15) at (0.5,-2.55) {};
\node[inner sep=0pt, minimum size=0pt, outer sep=0pt] (t16) at (0,-3.6) {};
\draw[draw=tncol33415c, line width=0.6pt] ($(t0.north) + (0:-0.5)$) -- (t1.270);
\draw[draw=tncol33415c, line width=0.6pt] (t0.north) -- (t2.270);
\draw[draw=tncol33415c, line width=0.6pt] ($(t0.north) + (0:0.5)$) -- (t3.270);
\draw[draw=tncol33415c, line width=0.6pt] ($(t4.north) + (0:-0.5)$) -- (t5.270);
\draw[draw=tncol33415c, line width=0.6pt] (t4.north) -- (t6.270);
\draw[draw=tncol33415c, line width=0.6pt] ($(t4.north) + (0:0.5)$) -- (t7.270);
\draw[draw=tncol33415c, line width=0.6pt] ($(t8.north) + (0:-0.5)$) -- (t9.270);
\draw[draw=tncol33415c, line width=0.6pt] (t8.north) -- (t10.270);
\draw[draw=tncol33415c, line width=0.6pt] ($(t8.north) + (0:0.5)$) -- (t11.270);
\draw[draw=tncol33415c, line width=0.6pt] ($(t12.north) + (0:-0.5)$) -- (t13.270);
\draw[draw=tncol33415c, line width=0.6pt] (t12.north) -- (t14.270);
\draw[draw=tncol33415c, line width=0.6pt] ($(t12.north) + (0:0.5)$) -- (t15.270);
\draw[draw=tncol33415c, line width=0.6pt] (t0.west) -- ++(180:0.25);
\draw[draw=tncol33415c, line width=0.6pt] (t0.east) -- ++(0:0.25);
\draw[draw=tncol33415c, line width=0.6pt] (t4.west) -- ++(180:0.25);
\draw[draw=tncol33415c, line width=0.6pt] (t4.east) -- ++(0:0.25);
\draw[draw=tncol33415c, line width=0.6pt] (t8.west) -- ++(180:0.25);
\draw[draw=tncol33415c, line width=0.6pt] (t8.east) -- ++(0:0.25);
\draw[draw=tncol33415c, line width=0.6pt] (t12.west) -- ++(180:0.25);
\draw[draw=tncol33415c, line width=0.6pt] (t12.east) -- ++(0:0.25);
\end{tikzpicture}%
}}}%
\endgroup}%
}

%% file: figures/contingency_eg_strings_ordered.tex
\providecommand{\ContStrOrdered}[1][1]{%
\ensuremath{\begingroup\color{black}%
\vcenter{\hbox{%
\scalebox{#1}{%
\begin{tikzpicture}[baseline=(current bounding box.center)]
\node[rectangle, minimum size=9mm, inner sep=2pt, outer sep=0pt, rounded corners=1.6pt, draw=black, line width=0.6pt, fill=tncol2281aa, text=black, minimum height=2mm, minimum width=15mm, font=\small] (t0) at (0,0) {$\gamma_{\mu_B^{(1)}}$};
\node[isosceles triangle, minimum width=4.2mm, minimum height=2.5mm, inner sep=0pt, outer sep=0pt, isosceles triangle stretches, shape border rotate=90, draw=tncol22223b, line width=0.6pt, fill=tncola0c4ff, text=black] (t1) at (-0.5,0.45) {};
\node[isosceles triangle, minimum width=4.2mm, minimum height=2.5mm, inner sep=0pt, outer sep=0pt, isosceles triangle stretches, shape border rotate=90, draw=tncol22223b, line width=0.6pt, fill=tncolffadad, text=black] (t2) at (0,0.45) {};
\node[isosceles triangle, minimum width=4.2mm, minimum height=2.5mm, inner sep=0pt, outer sep=0pt, isosceles triangle stretches, shape border rotate=90, draw=tncol22223b, line width=0.6pt, fill=tncolcaffbf, text=black] (t3) at (0.5,0.45) {};
\node[rectangle, minimum size=9mm, inner sep=2pt, outer sep=0pt, rounded corners=1.6pt, draw=black, line width=0.6pt, fill=tncol2281aa, text=black, minimum height=2mm, minimum width=15mm, font=\small] (t4) at (0,-1) {$\gamma_{\mu_B^{(2)}}$};
\node[isosceles triangle, minimum width=4.2mm, minimum height=2.5mm, inner sep=0pt, outer sep=0pt, isosceles triangle stretches, shape border rotate=90, draw=tncol22223b, line width=0.6pt, fill=tncolfdffb6, text=black] (t5) at (-0.5,-0.55) {};
\node[isosceles triangle, minimum width=4.2mm, minimum height=2.5mm, inner sep=0pt, outer sep=0pt, isosceles triangle stretches, shape border rotate=90, draw=tncol22223b, line width=0.6pt, fill=tncol9bf6ff, text=black] (t6) at (0,-0.55) {};
\node[isosceles triangle, minimum width=4.2mm, minimum height=2.5mm, inner sep=0pt, outer sep=0pt, isosceles triangle stretches, shape border rotate=90, draw=tncol22223b, line width=0.6pt, fill=tncolffd6a5, text=black] (t7) at (0.5,-0.55) {};
\node[rectangle, minimum size=9mm, inner sep=2pt, outer sep=0pt, rounded corners=1.6pt, draw=black, line width=0.6pt, fill=tncol2281aa, text=black, minimum height=2mm, minimum width=15mm, font=\small] (t8) at (0,-2) {$\gamma_{\mu_B^{(3)}}$};
\node[isosceles triangle, minimum width=4.2mm, minimum height=2.5mm, inner sep=0pt, outer sep=0pt, isosceles triangle stretches, shape border rotate=90, draw=tncol22223b, line width=0.6pt, fill=tncolffd6a5, text=black] (t9) at (-0.5,-1.55) {};
\node[isosceles triangle, minimum width=4.2mm, minimum height=2.5mm, inner sep=0pt, outer sep=0pt, isosceles triangle stretches, shape border rotate=90, draw=tncol22223b, line width=0.6pt, fill=tncolfdffb6, text=black] (t10) at (0,-1.55) {};
\node[isosceles triangle, minimum width=4.2mm, minimum height=2.5mm, inner sep=0pt, outer sep=0pt, isosceles triangle stretches, shape border rotate=90, draw=tncol22223b, line width=0.6pt, fill=tncolffadad, text=black] (t11) at (0.5,-1.55) {};
\node[rectangle, minimum size=9mm, inner sep=2pt, outer sep=0pt, rounded corners=1.6pt, draw=black, line width=0.6pt, fill=tncol2281aa, text=black, minimum height=2mm, minimum width=15mm, font=\small] (t12) at (0,-3) {$\gamma_{\mu_B^{(4)}}$};
\node[isosceles triangle, minimum width=4.2mm, minimum height=2.5mm, inner sep=0pt, outer sep=0pt, isosceles triangle stretches, shape border rotate=90, draw=tncol22223b, line width=0.6pt, fill=tncolcaffbf, text=black] (t13) at (-0.5,-2.55) {};
\node[isosceles triangle, minimum width=4.2mm, minimum height=2.5mm, inner sep=0pt, outer sep=0pt, isosceles triangle stretches, shape border rotate=90, draw=tncol22223b, line width=0.6pt, fill=tncola0c4ff, text=black] (t14) at (0,-2.55) {};
\node[isosceles triangle, minimum width=4.2mm, minimum height=2.5mm, inner sep=0pt, outer sep=0pt, isosceles triangle stretches, shape border rotate=90, draw=tncol22223b, line width=0.6pt, fill=tncol9bf6ff, text=black] (t15) at (0.5,-2.55) {};
\node[inner sep=0pt, minimum size=0pt, outer sep=0pt] (t16) at (0,-3.6) {};
\draw[draw=tncol33415c, line width=0.6pt] ($(t0.north) + (0:-0.5)$) -- (t1.270);
\draw[draw=tncol33415c, line width=0.6pt] (t0.north) -- (t2.270);
\draw[draw=tncol33415c, line width=0.6pt] ($(t0.north) + (0:0.5)$) -- (t3.270);
\draw[draw=tncol33415c, line width=0.6pt] ($(t4.north) + (0:-0.5)$) -- (t5.270);
\draw[draw=tncol33415c, line width=0.6pt] (t4.north) -- (t6.270);
\draw[draw=tncol33415c, line width=0.6pt] ($(t4.north) + (0:0.5)$) -- (t7.270);
\draw[draw=tncol33415c, line width=0.6pt] ($(t8.north) + (0:-0.5)$) -- (t9.270);
\draw[draw=tncol33415c, line width=0.6pt] (t8.north) -- (t10.270);
\draw[draw=tncol33415c, line width=0.6pt] ($(t8.north) + (0:0.5)$) -- (t11.270);
\draw[draw=tncol33415c, line width=0.6pt] ($(t12.north) + (0:-0.5)$) -- (t13.270);
\draw[draw=tncol33415c, line width=0.6pt] (t12.north) -- (t14.270);
\draw[draw=tncol33415c, line width=0.6pt] ($(t12.north) + (0:0.5)$) -- (t15.270);
\draw[draw=tncol33415c, line width=0.6pt] (t0.west) -- ++(180:0.25);
\draw[draw=tncol33415c, line width=0.6pt] (t0.east) -- ++(0:0.25);
\draw[draw=tncol33415c, line width=0.6pt] (t4.west) -- ++(180:0.25);
\draw[draw=tncol33415c, line width=0.6pt] (t4.east) -- ++(0:0.25);
\draw[draw=tncol33415c, line width=0.6pt] (t8.west) -- ++(180:0.25);
\draw[draw=tncol33415c, line width=0.6pt] (t8.east) -- ++(0:0.25);
\draw[draw=tncol33415c, line width=0.6pt] (t12.west) -- ++(180:0.25);
\draw[draw=tncol33415c, line width=0.6pt] (t12.east) -- ++(0:0.25);
\end{tikzpicture}%
}}}%
\endgroup}%
}

%% file: paper.bbl
\begin{thebibliography}{84}%
\makeatletter
\providecommand \@ifxundefined [1]{%
 \@ifx{#1\undefined}
}%
\providecommand \@ifnum [1]{%
 \ifnum #1\expandafter \@firstoftwo
 \else \expandafter \@secondoftwo
 \fi
}%
\providecommand \@ifx [1]{%
 \ifx #1\expandafter \@firstoftwo
 \else \expandafter \@secondoftwo
 \fi
}%
\providecommand \natexlab [1]{#1}%
\providecommand \enquote  [1]{``#1''}%
\providecommand \bibnamefont  [1]{#1}%
\providecommand \bibfnamefont [1]{#1}%
\providecommand \citenamefont [1]{#1}%
\providecommand \href@noop [0]{\@secondoftwo}%
\providecommand \href [0]{\begingroup \@sanitize@url \@href}%
\providecommand \@href[1]{\@@startlink{#1}\@@href}%
\providecommand \@@href[1]{\endgroup#1\@@endlink}%
\providecommand \@sanitize@url [0]{\catcode `\\12\catcode `\$12\catcode `\&12\catcode `\#12\catcode `\^12\catcode `\_12\catcode `\%12\relax}%
\providecommand \@@startlink[1]{}%
\providecommand \@@endlink[0]{}%
\providecommand \url  [0]{\begingroup\@sanitize@url \@url }%
\providecommand \@url [1]{\endgroup\@href {#1}{\urlprefix }}%
\providecommand \urlprefix  [0]{URL }%
\providecommand \Eprint [0]{\href }%
\providecommand \doibase [0]{https://doi.org/}%
\providecommand \selectlanguage [0]{\@gobble}%
\providecommand \bibinfo  [0]{\@secondoftwo}%
\providecommand \bibfield  [0]{\@secondoftwo}%
\providecommand \translation [1]{[#1]}%
\providecommand \BibitemOpen [0]{}%
\providecommand \bibitemStop [0]{}%
\providecommand \bibitemNoStop [0]{.\EOS\space}%
\providecommand \EOS [0]{\spacefactor3000\relax}%
\providecommand \BibitemShut  [1]{\csname bibitem#1\endcsname}%
\let\auto@bib@innerbib\@empty
%</preamble>
\bibitem [{\citenamefont {Eisert}\ \emph {et~al.}(2015)\citenamefont {Eisert}, \citenamefont {Friesdorf},\ and\ \citenamefont {Gogolin}}]{1408.5148}%
  \BibitemOpen
  \bibfield  {author} {\bibinfo {author} {\bibfnamefont {J.}~\bibnamefont {Eisert}}, \bibinfo {author} {\bibfnamefont {M.}~\bibnamefont {Friesdorf}},\ and\ \bibinfo {author} {\bibfnamefont {C.}~\bibnamefont {Gogolin}},\ }\bibfield  {title} {\bibinfo {title} {Quantum many-body systems out of equilibrium},\ }\href {https://doi.org/doi:10.1038/nphys3215} {\bibfield  {journal} {\bibinfo  {journal} {Nature Phys.}\ }\textbf {\bibinfo {volume} {11}},\ \bibinfo {pages} {124} (\bibinfo {year} {2015})}\BibitemShut {NoStop}%
\bibitem [{\citenamefont {Polkovnikov}\ \emph {et~al.}(2011)\citenamefont {Polkovnikov}, \citenamefont {Sengupta}, \citenamefont {Silva},\ and\ \citenamefont {Vengalattore}}]{PolkovnikovReview}%
  \BibitemOpen
  \bibfield  {author} {\bibinfo {author} {\bibfnamefont {A.}~\bibnamefont {Polkovnikov}}, \bibinfo {author} {\bibfnamefont {K.}~\bibnamefont {Sengupta}}, \bibinfo {author} {\bibfnamefont {A.}~\bibnamefont {Silva}},\ and\ \bibinfo {author} {\bibfnamefont {M.}~\bibnamefont {Vengalattore}},\ }\bibfield  {title} {\bibinfo {title} {Non-equilibrium dynamics of closed interacting quantum systems},\ }\href {https://doi.org/10.1103/RevModPhys.83.863} {\bibfield  {journal} {\bibinfo  {journal} {Rev. Mod. Phys.}\ }\textbf {\bibinfo {volume} {83}},\ \bibinfo {pages} {863} (\bibinfo {year} {2011})}\BibitemShut {NoStop}%
\bibitem [{\citenamefont {von Keyserlingk}\ \emph {et~al.}(2018)\citenamefont {von Keyserlingk}, \citenamefont {Rakovszky}, \citenamefont {Pollmann},\ and\ \citenamefont {Sondhi}}]{Keyserlingk2018}%
  \BibitemOpen
  \bibfield  {author} {\bibinfo {author} {\bibfnamefont {C.~W.}\ \bibnamefont {von Keyserlingk}}, \bibinfo {author} {\bibfnamefont {T.}~\bibnamefont {Rakovszky}}, \bibinfo {author} {\bibfnamefont {F.}~\bibnamefont {Pollmann}},\ and\ \bibinfo {author} {\bibfnamefont {S.~L.}\ \bibnamefont {Sondhi}},\ }\bibfield  {title} {\bibinfo {title} {{Operator hydrodynamics, OTOCs, and entanglement growth in systems without conservation laws}},\ }\href {https://doi.org/10.1103/PhysRevX.8.021013} {\bibfield  {journal} {\bibinfo  {journal} {Phys. Rev. X}\ }\textbf {\bibinfo {volume} {8}},\ \bibinfo {pages} {021013} (\bibinfo {year} {2018})}\BibitemShut {NoStop}%
\bibitem [{\citenamefont {Khemani}\ \emph {et~al.}(2018)\citenamefont {Khemani}, \citenamefont {Vishwanath},\ and\ \citenamefont {Huse}}]{Khemani2018}%
  \BibitemOpen
  \bibfield  {author} {\bibinfo {author} {\bibfnamefont {V.}~\bibnamefont {Khemani}}, \bibinfo {author} {\bibfnamefont {A.}~\bibnamefont {Vishwanath}},\ and\ \bibinfo {author} {\bibfnamefont {D.~A.}\ \bibnamefont {Huse}},\ }\bibfield  {title} {\bibinfo {title} {Operator spreading and the emergence of dissipative hydrodynamics under unitary evolution with conservation laws},\ }\href {https://doi.org/10.1103/PhysRevX.8.031057} {\bibfield  {journal} {\bibinfo  {journal} {Phys. Rev. X}\ }\textbf {\bibinfo {volume} {8}},\ \bibinfo {pages} {031057} (\bibinfo {year} {2018})}\BibitemShut {NoStop}%
\bibitem [{\citenamefont {Pappalardi}\ \emph {et~al.}(2022)\citenamefont {Pappalardi}, \citenamefont {Foini},\ and\ \citenamefont {Kurchan}}]{Pappalardi2022freeETH}%
  \BibitemOpen
  \bibfield  {author} {\bibinfo {author} {\bibfnamefont {S.}~\bibnamefont {Pappalardi}}, \bibinfo {author} {\bibfnamefont {L.}~\bibnamefont {Foini}},\ and\ \bibinfo {author} {\bibfnamefont {J.}~\bibnamefont {Kurchan}},\ }\bibfield  {title} {\bibinfo {title} {Eigenstate thermalization hypothesis and free probability},\ }\href {https://doi.org/10.1103/PhysRevLett.129.170603} {\bibfield  {journal} {\bibinfo  {journal} {Phys. Rev. Lett.}\ }\textbf {\bibinfo {volume} {129}},\ \bibinfo {pages} {170603} (\bibinfo {year} {2022})}\BibitemShut {NoStop}%
\bibitem [{\citenamefont {Foini}\ and\ \citenamefont {Kurchan}(2019)}]{Foini2019}%
  \BibitemOpen
  \bibfield  {author} {\bibinfo {author} {\bibfnamefont {L.}~\bibnamefont {Foini}}\ and\ \bibinfo {author} {\bibfnamefont {J.}~\bibnamefont {Kurchan}},\ }\bibfield  {title} {\bibinfo {title} {Eigenstate thermalization hypothesis and out of time order correlators},\ }\href {https://doi.org/10.1103/PhysRevE.99.042139} {\bibfield  {journal} {\bibinfo  {journal} {Phys. Rev. E}\ }\textbf {\bibinfo {volume} {99}},\ \bibinfo {pages} {042139} (\bibinfo {year} {2019})}\BibitemShut {NoStop}%
\bibitem [{\citenamefont {Hosur}\ \emph {et~al.}(2016)\citenamefont {Hosur}, \citenamefont {Qi}, \citenamefont {Roberts},\ and\ \citenamefont {Yoshida}}]{Hosur2016}%
  \BibitemOpen
  \bibfield  {author} {\bibinfo {author} {\bibfnamefont {P.}~\bibnamefont {Hosur}}, \bibinfo {author} {\bibfnamefont {X.-L.}\ \bibnamefont {Qi}}, \bibinfo {author} {\bibfnamefont {D.~A.}\ \bibnamefont {Roberts}},\ and\ \bibinfo {author} {\bibfnamefont {B.}~\bibnamefont {Yoshida}},\ }\bibfield  {title} {\bibinfo {title} {Chaos in quantum channels},\ }\href {https://doi.org/10.1007/JHEP02(2016)004} {\bibfield  {journal} {\bibinfo  {journal} {J. High Energy Phys.}\ }\textbf {\bibinfo {volume} {2016}}\bibinfo  {number} { (2)},\ \bibinfo {pages} {4}}\BibitemShut {NoStop}%
\bibitem [{\citenamefont {Swingle}\ \emph {et~al.}(2016)\citenamefont {Swingle}, \citenamefont {Bentsen}, \citenamefont {Schleier-Smith},\ and\ \citenamefont {Hayden}}]{Swingle2016}%
  \BibitemOpen
\bibfield  {number} {  }\bibfield  {author} {\bibinfo {author} {\bibfnamefont {B.}~\bibnamefont {Swingle}}, \bibinfo {author} {\bibfnamefont {G.}~\bibnamefont {Bentsen}}, \bibinfo {author} {\bibfnamefont {M.}~\bibnamefont {Schleier-Smith}},\ and\ \bibinfo {author} {\bibfnamefont {P.}~\bibnamefont {Hayden}},\ }\bibfield  {title} {\bibinfo {title} {Measuring the scrambling of quantum information},\ }\href {https://doi.org/10.1103/PhysRevA.94.040302} {\bibfield  {journal} {\bibinfo  {journal} {Phys. Rev. A}\ }\textbf {\bibinfo {volume} {94}},\ \bibinfo {pages} {040302} (\bibinfo {year} {2016})}\BibitemShut {NoStop}%
\bibitem [{\citenamefont {{Google Quantum AI}}(2021)}]{Mi2021}%
  \BibitemOpen
  \bibfield  {author} {\bibinfo {author} {\bibnamefont {{Google Quantum AI}}},\ }\bibfield  {title} {\bibinfo {title} {Information scrambling in quantum circuits},\ }\href {https://doi.org/10.1126/science.abg5029} {\bibfield  {journal} {\bibinfo  {journal} {Science}\ }\textbf {\bibinfo {volume} {374}},\ \bibinfo {pages} {1479} (\bibinfo {year} {2021})}\BibitemShut {NoStop}%
\bibitem [{\citenamefont {Dowling}\ \emph {et~al.}(2025{\natexlab{a}})\citenamefont {Dowling}, \citenamefont {Modi},\ and\ \citenamefont {White}}]{Dowlin2024LOE-OSRE}%
  \BibitemOpen
  \bibfield  {author} {\bibinfo {author} {\bibfnamefont {N.}~\bibnamefont {Dowling}}, \bibinfo {author} {\bibfnamefont {K.}~\bibnamefont {Modi}},\ and\ \bibinfo {author} {\bibfnamefont {G.~A.~L.}\ \bibnamefont {White}},\ }\bibfield  {title} {\bibinfo {title} {Bridging entanglement and magic resources within operator space},\ }\href {https://doi.org/10.1103/c7k1-xcwy} {\bibfield  {journal} {\bibinfo  {journal} {Phys. Rev. Lett.}\ }\textbf {\bibinfo {volume} {135}},\ \bibinfo {pages} {160201} (\bibinfo {year} {2025}{\natexlab{a}})}\BibitemShut {NoStop}%
\bibitem [{\citenamefont {Dowling}\ \emph {et~al.}(2025{\natexlab{b}})\citenamefont {Dowling}, \citenamefont {Kos},\ and\ \citenamefont {Turkeshi}}]{dowling2024magicheis}%
  \BibitemOpen
  \bibfield  {author} {\bibinfo {author} {\bibfnamefont {N.}~\bibnamefont {Dowling}}, \bibinfo {author} {\bibfnamefont {P.}~\bibnamefont {Kos}},\ and\ \bibinfo {author} {\bibfnamefont {X.}~\bibnamefont {Turkeshi}},\ }\bibfield  {title} {\bibinfo {title} {Magic resources of the heisenberg picture},\ }\href {https://doi.org/10.1103/p7xt-s9nz} {\bibfield  {journal} {\bibinfo  {journal} {Phys. Rev. Lett.}\ }\textbf {\bibinfo {volume} {135}},\ \bibinfo {pages} {050401} (\bibinfo {year} {2025}{\natexlab{b}})}\BibitemShut {NoStop}%
\bibitem [{\citenamefont {Prosen}\ and\ \citenamefont {Pi\ifmmode~\check{z}\else \v{z}\fi{}orn}(2007)}]{Prosen2007a}%
  \BibitemOpen
  \bibfield  {author} {\bibinfo {author} {\bibfnamefont {T.}~\bibnamefont {Prosen}}\ and\ \bibinfo {author} {\bibfnamefont {I.}~\bibnamefont {Pi\ifmmode~\check{z}\else \v{z}\fi{}orn}},\ }\bibfield  {title} {\bibinfo {title} {Operator space entanglement entropy in a transverse {Ising} chain},\ }\href {https://doi.org/10.1103/PhysRevA.76.032316} {\bibfield  {journal} {\bibinfo  {journal} {Phys. Rev. A}\ }\textbf {\bibinfo {volume} {76}},\ \bibinfo {pages} {032316} (\bibinfo {year} {2007})}\BibitemShut {NoStop}%
\bibitem [{\citenamefont {Pi\ifmmode~\check{z}\else \v{z}\fi{}orn}\ and\ \citenamefont {Prosen}(2009)}]{Prosen2009}%
  \BibitemOpen
  \bibfield  {author} {\bibinfo {author} {\bibfnamefont {I.}~\bibnamefont {Pi\ifmmode~\check{z}\else \v{z}\fi{}orn}}\ and\ \bibinfo {author} {\bibfnamefont {T.}~\bibnamefont {Prosen}},\ }\bibfield  {title} {\bibinfo {title} {Operator space entanglement entropy in $xy$ spin chains},\ }\href {https://doi.org/10.1103/PhysRevB.79.184416} {\bibfield  {journal} {\bibinfo  {journal} {Phys. Rev. B}\ }\textbf {\bibinfo {volume} {79}},\ \bibinfo {pages} {184416} (\bibinfo {year} {2009})}\BibitemShut {NoStop}%
\bibitem [{\citenamefont {Dubail}(2017)}]{Dubail_2017}%
  \BibitemOpen
  \bibfield  {author} {\bibinfo {author} {\bibfnamefont {J.}~\bibnamefont {Dubail}},\ }\bibfield  {title} {\bibinfo {title} {Entanglement scaling of operators: a conformal field theory approach, with a glimpse of simulability of long-time dynamics in 1+1d},\ }\href {https://doi.org/10.1088/1751-8121/aa6f38} {\bibfield  {journal} {\bibinfo  {journal} {J. Phys. A}\ }\textbf {\bibinfo {volume} {50}},\ \bibinfo {pages} {234001} (\bibinfo {year} {2017})}\BibitemShut {NoStop}%
\bibitem [{\citenamefont {Hartmann}\ \emph {et~al.}(2009)\citenamefont {Hartmann}, \citenamefont {Prior}, \citenamefont {Clark},\ and\ \citenamefont {Plenio}}]{Hartmann2009}%
  \BibitemOpen
  \bibfield  {author} {\bibinfo {author} {\bibfnamefont {M.~J.}\ \bibnamefont {Hartmann}}, \bibinfo {author} {\bibfnamefont {J.}~\bibnamefont {Prior}}, \bibinfo {author} {\bibfnamefont {S.~R.}\ \bibnamefont {Clark}},\ and\ \bibinfo {author} {\bibfnamefont {M.~B.}\ \bibnamefont {Plenio}},\ }\bibfield  {title} {\bibinfo {title} {Density matrix renormalization group in the {Heisenberg} picture},\ }\href {https://doi.org/10.1103/PhysRevLett.102.057202} {\bibfield  {journal} {\bibinfo  {journal} {Phys. Rev. Lett.}\ }\textbf {\bibinfo {volume} {102}},\ \bibinfo {pages} {057202} (\bibinfo {year} {2009})}\BibitemShut {NoStop}%
\bibitem [{\citenamefont {Dowling}(2026)}]{dowling2026classicalsim}%
  \BibitemOpen
  \bibfield  {author} {\bibinfo {author} {\bibfnamefont {N.}~\bibnamefont {Dowling}},\ }\href@noop {} {\bibinfo {title} {Classical simulability from operator entanglement scaling}} (\bibinfo {year} {2026}),\ \Eprint {https://arxiv.org/abs/2603.05656} {arXiv:2603.05656} \BibitemShut {NoStop}%
\bibitem [{\citenamefont {Calabrese}\ and\ \citenamefont {Cardy}(2005)}]{Calabrese_2005}%
  \BibitemOpen
  \bibfield  {author} {\bibinfo {author} {\bibfnamefont {P.}~\bibnamefont {Calabrese}}\ and\ \bibinfo {author} {\bibfnamefont {J.}~\bibnamefont {Cardy}},\ }\bibfield  {title} {\bibinfo {title} {Evolution of entanglement entropy in one-dimensional systems},\ }\href {https://doi.org/10.1088/1742-5468/2005/04/p04010} {\bibfield  {journal} {\bibinfo  {journal} {J. Stat. Mech.: Theory Exp.}\ }\textbf {\bibinfo {volume} {2005}}\bibinfo  {number} { (04)},\ \bibinfo {pages} {P04010}}\BibitemShut {NoStop}%
\bibitem [{\citenamefont {Schuch}\ \emph {et~al.}(2008)\citenamefont {Schuch}, \citenamefont {Wolf}, \citenamefont {Verstraete},\ and\ \citenamefont {Cirac}}]{Schuch_2008}%
  \BibitemOpen
\bibfield  {number} {  }\bibfield  {author} {\bibinfo {author} {\bibfnamefont {N.}~\bibnamefont {Schuch}}, \bibinfo {author} {\bibfnamefont {M.~M.}\ \bibnamefont {Wolf}}, \bibinfo {author} {\bibfnamefont {F.}~\bibnamefont {Verstraete}},\ and\ \bibinfo {author} {\bibfnamefont {J.~I.}\ \bibnamefont {Cirac}},\ }\bibfield  {title} {\bibinfo {title} {Entropy scaling and simulability by matrix product states},\ }\href {https://doi.org/10.1103/PhysRevLett.100.030504} {\bibfield  {journal} {\bibinfo  {journal} {Phys. Rev. Lett.}\ }\textbf {\bibinfo {volume} {100}},\ \bibinfo {pages} {030504} (\bibinfo {year} {2008})}\BibitemShut {NoStop}%
\bibitem [{\citenamefont {Jozsa}\ and\ \citenamefont {Miyake}(2008{\natexlab{a}})}]{Jozsa_2008}%
  \BibitemOpen
  \bibfield  {author} {\bibinfo {author} {\bibfnamefont {R.}~\bibnamefont {Jozsa}}\ and\ \bibinfo {author} {\bibfnamefont {A.}~\bibnamefont {Miyake}},\ }\bibfield  {title} {\bibinfo {title} {Matchgates and classical simulation of quantum circuits},\ }\href {https://doi.org/10.1098/rspa.2008.0189} {\bibfield  {journal} {\bibinfo  {journal} {Proc. R. Soc. A}\ }\textbf {\bibinfo {volume} {464}},\ \bibinfo {pages} {3089–3106} (\bibinfo {year} {2008}{\natexlab{a}})}\BibitemShut {NoStop}%
\bibitem [{\citenamefont {Wan}\ \emph {et~al.}(2023)\citenamefont {Wan}, \citenamefont {Huggins}, \citenamefont {Lee},\ and\ \citenamefont {Babbush}}]{wan_matchgate_2023}%
  \BibitemOpen
  \bibfield  {author} {\bibinfo {author} {\bibfnamefont {K.}~\bibnamefont {Wan}}, \bibinfo {author} {\bibfnamefont {W.~J.}\ \bibnamefont {Huggins}}, \bibinfo {author} {\bibfnamefont {J.}~\bibnamefont {Lee}},\ and\ \bibinfo {author} {\bibfnamefont {R.}~\bibnamefont {Babbush}},\ }\bibfield  {title} {\bibinfo {title} {Matchgate {shadows} for {fermionic} {quantum} {simulation}},\ }\href {https://doi.org/10.1007/s00220-023-04844-0} {\bibfield  {journal} {\bibinfo  {journal} {Commun. Math. Phys.}\ }\textbf {\bibinfo {volume} {404}},\ \bibinfo {pages} {629} (\bibinfo {year} {2023})}\BibitemShut {NoStop}%
\bibitem [{\citenamefont {Dias}\ and\ \citenamefont {Koenig}(2024)}]{Dias2024classicalsimulation}%
  \BibitemOpen
  \bibfield  {author} {\bibinfo {author} {\bibfnamefont {B.}~\bibnamefont {Dias}}\ and\ \bibinfo {author} {\bibfnamefont {R.}~\bibnamefont {Koenig}},\ }\bibfield  {title} {\bibinfo {title} {Classical simulation of non-{G}aussian fermionic circuits},\ }\href {https://doi.org/10.22331/q-2024-05-21-1350} {\bibfield  {journal} {\bibinfo  {journal} {{Quantum}}\ }\textbf {\bibinfo {volume} {8}},\ \bibinfo {pages} {1350} (\bibinfo {year} {2024})}\BibitemShut {NoStop}%
\bibitem [{\citenamefont {Zhao}(2023)}]{zhao2023learningopt}%
  \BibitemOpen
  \bibfield  {author} {\bibinfo {author} {\bibfnamefont {A.}~\bibnamefont {Zhao}},\ }\href@noop {} {\bibinfo {title} {Learning, optimizing, and simulating fermions with quantum computers}} (\bibinfo {year} {2023}),\ \Eprint {https://arxiv.org/abs/2312.10399} {arXiv:2312.10399 [quant-ph]} \BibitemShut {NoStop}%
\bibitem [{\citenamefont {Braccia}\ \emph {et~al.}(2025)\citenamefont {Braccia}, \citenamefont {Diaz}, \citenamefont {Larocca}, \citenamefont {Cerezo},\ and\ \citenamefont {García-Martín}}]{braccia2025optimalhaarrandomfermionic}%
  \BibitemOpen
  \bibfield  {author} {\bibinfo {author} {\bibfnamefont {P.}~\bibnamefont {Braccia}}, \bibinfo {author} {\bibfnamefont {N.~L.}\ \bibnamefont {Diaz}}, \bibinfo {author} {\bibfnamefont {M.}~\bibnamefont {Larocca}}, \bibinfo {author} {\bibfnamefont {M.}~\bibnamefont {Cerezo}},\ and\ \bibinfo {author} {\bibfnamefont {D.}~\bibnamefont {García-Martín}},\ }\href {https://arxiv.org/abs/2505.24212} {\bibinfo {title} {Optimal haar random fermionic linear optics circuits}} (\bibinfo {year} {2025}),\ \Eprint {https://arxiv.org/abs/2505.24212} {arXiv:2505.24212 [quant-ph]} \BibitemShut {NoStop}%
\bibitem [{\citenamefont {Sierant}\ \emph {et~al.}(2026{\natexlab{a}})\citenamefont {Sierant}, \citenamefont {Turkeshi},\ and\ \citenamefont {Tarabunga}}]{sierant2026theorymatchgatecommutant}%
  \BibitemOpen
  \bibfield  {author} {\bibinfo {author} {\bibfnamefont {P.}~\bibnamefont {Sierant}}, \bibinfo {author} {\bibfnamefont {X.}~\bibnamefont {Turkeshi}},\ and\ \bibinfo {author} {\bibfnamefont {P.~S.}\ \bibnamefont {Tarabunga}},\ }\href@noop {} {\bibinfo {title} {Theory of the matchgate commutant}} (\bibinfo {year} {2026}{\natexlab{a}}),\ \Eprint {https://arxiv.org/abs/2603.12392} {arXiv:2603.12392} \BibitemShut {NoStop}%
\bibitem [{\citenamefont {Braccia}\ \emph {et~al.}(2026)\citenamefont {Braccia}, \citenamefont {Diaz}, \citenamefont {Larocca}, \citenamefont {Cerezo},\ and\ \citenamefont {García-Martín}}]{braccia2026commutantfermionicgaussianunitaries}%
  \BibitemOpen
  \bibfield  {author} {\bibinfo {author} {\bibfnamefont {P.}~\bibnamefont {Braccia}}, \bibinfo {author} {\bibfnamefont {N.~L.}\ \bibnamefont {Diaz}}, \bibinfo {author} {\bibfnamefont {M.}~\bibnamefont {Larocca}}, \bibinfo {author} {\bibfnamefont {M.}~\bibnamefont {Cerezo}},\ and\ \bibinfo {author} {\bibfnamefont {D.}~\bibnamefont {García-Martín}},\ }\href@noop {} {\bibinfo {title} {{The commutant of fermionic Gaussian unitaries}}} (\bibinfo {year} {2026}),\ \Eprint {https://arxiv.org/abs/2603.19210} {arXiv:2603.19210} \BibitemShut {NoStop}%
\bibitem [{\citenamefont {West}\ \emph {et~al.}(2025)\citenamefont {West}, \citenamefont {Dowling}, \citenamefont {Southwell}, \citenamefont {Sevior}, \citenamefont {Usman}, \citenamefont {Modi},\ and\ \citenamefont {Quella}}]{West_2025}%
  \BibitemOpen
  \bibfield  {author} {\bibinfo {author} {\bibfnamefont {M.}~\bibnamefont {West}}, \bibinfo {author} {\bibfnamefont {N.}~\bibnamefont {Dowling}}, \bibinfo {author} {\bibfnamefont {A.}~\bibnamefont {Southwell}}, \bibinfo {author} {\bibfnamefont {M.}~\bibnamefont {Sevior}}, \bibinfo {author} {\bibfnamefont {M.}~\bibnamefont {Usman}}, \bibinfo {author} {\bibfnamefont {K.}~\bibnamefont {Modi}},\ and\ \bibinfo {author} {\bibfnamefont {T.}~\bibnamefont {Quella}},\ }\bibfield  {title} {\bibinfo {title} {A graph-theoretic approach to chaos and complexity in quantum systems},\ }\bibfield  {journal} {\bibinfo  {journal} {SciPost Physics Core}\ }\textbf {\bibinfo {volume} {8}},\ \href {https://doi.org/10.21468/scipostphyscore.8.4.081} {10.21468/scipostphyscore.8.4.081} (\bibinfo {year} {2025})\BibitemShut {NoStop}%
\bibitem [{\citenamefont {Collins}\ and\ \citenamefont {Sniady}(2006)}]{collins_integration_2006}%
  \BibitemOpen
  \bibfield  {author} {\bibinfo {author} {\bibfnamefont {B.}~\bibnamefont {Collins}}\ and\ \bibinfo {author} {\bibfnamefont {P.}~\bibnamefont {Sniady}},\ }\bibfield  {title} {\bibinfo {title} {Integration with {respect} to the {Haar} {measure} on {unitary}, {orthogonal} and {symplectic} {group}},\ }\href {https://doi.org/10.1007/s00220-006-1554-3} {\bibfield  {journal} {\bibinfo  {journal} {Commun. Math. Phys.}\ }\textbf {\bibinfo {volume} {264}},\ \bibinfo {pages} {773} (\bibinfo {year} {2006})}\BibitemShut {NoStop}%
\bibitem [{\citenamefont {Lastres}\ and\ \citenamefont {Moudgalya}(2026)}]{lastres2026geometryfreefermioncommutants}%
  \BibitemOpen
  \bibfield  {author} {\bibinfo {author} {\bibfnamefont {M.}~\bibnamefont {Lastres}}\ and\ \bibinfo {author} {\bibfnamefont {S.}~\bibnamefont {Moudgalya}},\ }\href@noop {} {\bibinfo {title} {Geometry of free fermion commutants}} (\bibinfo {year} {2026}),\ \Eprint {https://arxiv.org/abs/2604.05031} {arXiv:2604.05031 [quant-ph]} \BibitemShut {NoStop}%
\bibitem [{\citenamefont {Chapman}\ and\ \citenamefont {Flammia}(2025)}]{chapman2025fermioni}%
  \BibitemOpen
  \bibfield  {author} {\bibinfo {author} {\bibfnamefont {A.}~\bibnamefont {Chapman}}\ and\ \bibinfo {author} {\bibfnamefont {S.~T.}\ \bibnamefont {Flammia}},\ }\href@noop {} {\bibinfo {title} {Fermionic averaged circuit eigenvalue sampling}} (\bibinfo {year} {2025}),\ \Eprint {https://arxiv.org/abs/2504.01936} {arXiv:2504.01936} \BibitemShut {NoStop}%
\bibitem [{\citenamefont {Barvinok}(2008)}]{Barvinok2008Contingency}%
  \BibitemOpen
  \bibfield  {author} {\bibinfo {author} {\bibfnamefont {A.}~\bibnamefont {Barvinok}},\ }\bibfield  {title} {\bibinfo {title} {Enumerating contingency tables via random permanents},\ }\href {https://doi.org/10.1017/S0963548307008668} {\bibfield  {journal} {\bibinfo  {journal} {Combin. Probab. Comput.}\ }\textbf {\bibinfo {volume} {17}},\ \bibinfo {pages} {1} (\bibinfo {year} {2008})}\BibitemShut {NoStop}%
\bibitem [{\citenamefont {Zanardi}(2001)}]{Zanardi2001}%
  \BibitemOpen
  \bibfield  {author} {\bibinfo {author} {\bibfnamefont {P.}~\bibnamefont {Zanardi}},\ }\bibfield  {title} {\bibinfo {title} {Entanglement of quantum evolutions},\ }\href {https://doi.org/10.1103/PhysRevA.63.040304} {\bibfield  {journal} {\bibinfo  {journal} {Phys. Rev. A}\ }\textbf {\bibinfo {volume} {63}},\ \bibinfo {pages} {040304(R)} (\bibinfo {year} {2001})}\BibitemShut {NoStop}%
\bibitem [{\citenamefont {Rakovszky}\ \emph {et~al.}(2022)\citenamefont {Rakovszky}, \citenamefont {von Keyserlingk},\ and\ \citenamefont {Pollmann}}]{Rakovszky2022}%
  \BibitemOpen
  \bibfield  {author} {\bibinfo {author} {\bibfnamefont {T.}~\bibnamefont {Rakovszky}}, \bibinfo {author} {\bibfnamefont {C.~W.}\ \bibnamefont {von Keyserlingk}},\ and\ \bibinfo {author} {\bibfnamefont {F.}~\bibnamefont {Pollmann}},\ }\bibfield  {title} {\bibinfo {title} {Dissipation-assisted operator evolution method for capturing hydrodynamic transport},\ }\href {https://doi.org/10.1103/PhysRevB.105.075131} {\bibfield  {journal} {\bibinfo  {journal} {Phys. Rev. B}\ }\textbf {\bibinfo {volume} {105}},\ \bibinfo {pages} {075131} (\bibinfo {year} {2022})}\BibitemShut {NoStop}%
\bibitem [{\citenamefont {Begu\ifmmode \check{s}\else \v{s}\fi{}i\ifmmode~\acute{c}\else \'{c}\fi{}}\ and\ \citenamefont {Chan}(2025)}]{begusic2024realtime}%
  \BibitemOpen
  \bibfield  {author} {\bibinfo {author} {\bibfnamefont {T.}~\bibnamefont {Begu\ifmmode \check{s}\else \v{s}\fi{}i\ifmmode~\acute{c}\else \'{c}\fi{}}}\ and\ \bibinfo {author} {\bibfnamefont {G.~K.-L.}\ \bibnamefont {Chan}},\ }\bibfield  {title} {\bibinfo {title} {Real-time operator evolution in two and three dimensions via sparse pauli dynamics},\ }\href {https://doi.org/10.1103/PRXQuantum.6.020302} {\bibfield  {journal} {\bibinfo  {journal} {PRX Quantum}\ }\textbf {\bibinfo {volume} {6}},\ \bibinfo {pages} {020302} (\bibinfo {year} {2025})}\BibitemShut {NoStop}%
\bibitem [{\citenamefont {Schuster}\ \emph {et~al.}(2025)\citenamefont {Schuster}, \citenamefont {Yin}, \citenamefont {Gao},\ and\ \citenamefont {Yao}}]{schuster2024polynomialtime}%
  \BibitemOpen
  \bibfield  {author} {\bibinfo {author} {\bibfnamefont {T.}~\bibnamefont {Schuster}}, \bibinfo {author} {\bibfnamefont {C.}~\bibnamefont {Yin}}, \bibinfo {author} {\bibfnamefont {X.}~\bibnamefont {Gao}},\ and\ \bibinfo {author} {\bibfnamefont {N.~Y.}\ \bibnamefont {Yao}},\ }\bibfield  {title} {\bibinfo {title} {A polynomial-time classical algorithm for noisy quantum circuits},\ }\href {https://doi.org/10.1103/xct1-7kf2} {\bibfield  {journal} {\bibinfo  {journal} {Phys. Rev. X}\ }\textbf {\bibinfo {volume} {15}},\ \bibinfo {pages} {041018} (\bibinfo {year} {2025})}\BibitemShut {NoStop}%
\bibitem [{\citenamefont {Rudolph}\ \emph {et~al.}(2026)\citenamefont {Rudolph}, \citenamefont {Jones}, \citenamefont {Teng}, \citenamefont {Angrisani},\ and\ \citenamefont {Holmes}}]{rudolph2025paulipropag}%
  \BibitemOpen
  \bibfield  {author} {\bibinfo {author} {\bibfnamefont {M.~S.}\ \bibnamefont {Rudolph}}, \bibinfo {author} {\bibfnamefont {T.}~\bibnamefont {Jones}}, \bibinfo {author} {\bibfnamefont {Y.}~\bibnamefont {Teng}}, \bibinfo {author} {\bibfnamefont {A.}~\bibnamefont {Angrisani}},\ and\ \bibinfo {author} {\bibfnamefont {Z.}~\bibnamefont {Holmes}},\ }\bibfield  {title} {\bibinfo {title} {Pauli propagation: A computational framework for simulating quantum systems},\ }\href {https://doi.org/10.1103/6vd7-l9bn} {\bibfield  {journal} {\bibinfo  {journal} {PRX Quantum}\ }\textbf {\bibinfo {volume} {7}},\ \bibinfo {pages} {032001} (\bibinfo {year} {2026})}\BibitemShut {NoStop}%
\bibitem [{\citenamefont {Dowling}\ and\ \citenamefont {Pappalardi}(2026)}]{dowling2026page}%
  \BibitemOpen
  \bibfield  {author} {\bibinfo {author} {\bibfnamefont {N.}~\bibnamefont {Dowling}}\ and\ \bibinfo {author} {\bibfnamefont {S.}~\bibnamefont {Pappalardi}},\ }\href@noop {} {\bibinfo {title} {Page curve for local-operator entanglement from free probability}} (\bibinfo {year} {2026}),\ \Eprint {https://arxiv.org/abs/2605.02995} {arXiv:2605.02995} \BibitemShut {NoStop}%
\bibitem [{\citenamefont {Dowling}\ \emph {et~al.}(2026)\citenamefont {Dowling}, \citenamefont {Turkeshi}, \citenamefont {Nardis},\ and\ \citenamefont {Lami}}]{dowling2026noise}%
  \BibitemOpen
  \bibfield  {author} {\bibinfo {author} {\bibfnamefont {N.}~\bibnamefont {Dowling}}, \bibinfo {author} {\bibfnamefont {X.}~\bibnamefont {Turkeshi}}, \bibinfo {author} {\bibfnamefont {J.~D.}\ \bibnamefont {Nardis}},\ and\ \bibinfo {author} {\bibfnamefont {G.}~\bibnamefont {Lami}},\ }\href@noop {} {\bibinfo {title} {Noise-induced simulability transition from operator scrambling}} (\bibinfo {year} {2026}),\ \Eprint {https://arxiv.org/abs/2605.18943} {arXiv:2605.18943} \BibitemShut {NoStop}%
\bibitem [{\citenamefont {AI}\ and\ \citenamefont {{Collaborators}}(2025)}]{Abanin2025}%
  \BibitemOpen
  \bibfield  {author} {\bibinfo {author} {\bibfnamefont {G.~Q.}\ \bibnamefont {AI}}\ and\ \bibinfo {author} {\bibnamefont {{Collaborators}}},\ }\bibfield  {title} {\bibinfo {title} {Observation of constructive interference at the edge of quantum ergodicity},\ }\href {https://doi.org/10.1038/s41586-025-09526-6} {\bibfield  {journal} {\bibinfo  {journal} {Nature}\ }\textbf {\bibinfo {volume} {646}},\ \bibinfo {pages} {825} (\bibinfo {year} {2025})}\BibitemShut {NoStop}%
\bibitem [{\citenamefont {Nica}\ and\ \citenamefont {Speicher}(2006)}]{nica2006lectures}%
  \BibitemOpen
  \bibfield  {author} {\bibinfo {author} {\bibfnamefont {A.}~\bibnamefont {Nica}}\ and\ \bibinfo {author} {\bibfnamefont {R.}~\bibnamefont {Speicher}},\ }\href@noop {} {\emph {\bibinfo {title} {Lectures on the combinatorics of free probability}}},\ \bibinfo {series} {London Mathematical Society Lecture Note Series}, Vol.\ \bibinfo {volume} {335}\ (\bibinfo  {publisher} {Cambridge University Press},\ \bibinfo {year} {2006})\BibitemShut {NoStop}%
\bibitem [{\citenamefont {Fava}\ \emph {et~al.}(2025)\citenamefont {Fava}, \citenamefont {Kurchan},\ and\ \citenamefont {Pappalardi}}]{fava2023designsfreeprobability}%
  \BibitemOpen
  \bibfield  {author} {\bibinfo {author} {\bibfnamefont {M.}~\bibnamefont {Fava}}, \bibinfo {author} {\bibfnamefont {J.}~\bibnamefont {Kurchan}},\ and\ \bibinfo {author} {\bibfnamefont {S.}~\bibnamefont {Pappalardi}},\ }\bibfield  {title} {\bibinfo {title} {Designs via free probability},\ }\href {https://doi.org/10.1103/PhysRevX.15.011031} {\bibfield  {journal} {\bibinfo  {journal} {Phys. Rev. X}\ }\textbf {\bibinfo {volume} {15}},\ \bibinfo {pages} {011031} (\bibinfo {year} {2025})}\BibitemShut {NoStop}%
\bibitem [{\citenamefont {Dowling}\ \emph {et~al.}(2025{\natexlab{c}})\citenamefont {Dowling}, \citenamefont {Nardis}, \citenamefont {Heinrich}, \citenamefont {Turkeshi},\ and\ \citenamefont {Pappalardi}}]{dowling2025freeindep}%
  \BibitemOpen
  \bibfield  {author} {\bibinfo {author} {\bibfnamefont {N.}~\bibnamefont {Dowling}}, \bibinfo {author} {\bibfnamefont {J.~D.}\ \bibnamefont {Nardis}}, \bibinfo {author} {\bibfnamefont {M.}~\bibnamefont {Heinrich}}, \bibinfo {author} {\bibfnamefont {X.}~\bibnamefont {Turkeshi}},\ and\ \bibinfo {author} {\bibfnamefont {S.}~\bibnamefont {Pappalardi}},\ }\href@noop {} {\bibinfo {title} {Free independence and unitary design from random matrix product unitaries}} (\bibinfo {year} {2025}{\natexlab{c}}),\ \Eprint {https://arxiv.org/abs/2508.00051} {arXiv:2508.00051} \BibitemShut {NoStop}%
\bibitem [{\citenamefont {Fritzsch}\ and\ \citenamefont {Claeys}(2026)}]{Claeys2026}%
  \BibitemOpen
  \bibfield  {author} {\bibinfo {author} {\bibfnamefont {F.}~\bibnamefont {Fritzsch}}\ and\ \bibinfo {author} {\bibfnamefont {P.~W.}\ \bibnamefont {Claeys}},\ }\bibfield  {title} {\bibinfo {title} {Free probability in a minimal quantum circuit model},\ }\href {https://doi.org/10.1103/6mpz-p85s} {\bibfield  {journal} {\bibinfo  {journal} {Phys. Rev. X}\ }\textbf {\bibinfo {volume} {16}},\ \bibinfo {pages} {031027} (\bibinfo {year} {2026})}\BibitemShut {NoStop}%
\bibitem [{\citenamefont {Pappalardi}\ \emph {et~al.}(2025)\citenamefont {Pappalardi}, \citenamefont {Fritzsch},\ and\ \citenamefont {Prosen}}]{FritzschLattices2025}%
  \BibitemOpen
  \bibfield  {author} {\bibinfo {author} {\bibfnamefont {S.}~\bibnamefont {Pappalardi}}, \bibinfo {author} {\bibfnamefont {F.}~\bibnamefont {Fritzsch}},\ and\ \bibinfo {author} {\bibfnamefont {T.}~\bibnamefont {Prosen}},\ }\bibfield  {title} {\bibinfo {title} {Full eigenstate thermalization via free cumulants in quantum lattice systems},\ }\href {https://doi.org/10.1103/PhysRevLett.134.140404} {\bibfield  {journal} {\bibinfo  {journal} {Phys. Rev. Lett.}\ }\textbf {\bibinfo {volume} {134}},\ \bibinfo {pages} {140404} (\bibinfo {year} {2025})}\BibitemShut {NoStop}%
\bibitem [{\citenamefont {Jordan}\ and\ \citenamefont {Wigner}(1928)}]{Jordan1928-ko}%
  \BibitemOpen
  \bibfield  {author} {\bibinfo {author} {\bibfnamefont {P.}~\bibnamefont {Jordan}}\ and\ \bibinfo {author} {\bibfnamefont {E.}~\bibnamefont {Wigner}},\ }\bibfield  {title} {\bibinfo {title} {{\"Uber das Paulische \"Aquivalenzverbot}},\ }\href {https://doi.org/10.1007/bf01331938} {\bibfield  {journal} {\bibinfo  {journal} {Eur. Phys. J. A}\ }\textbf {\bibinfo {volume} {47}},\ \bibinfo {pages} {631} (\bibinfo {year} {1928})}\BibitemShut {NoStop}%
\bibitem [{\citenamefont {Jozsa}\ and\ \citenamefont {Miyake}(2008{\natexlab{b}})}]{Jozsa2008}%
  \BibitemOpen
  \bibfield  {author} {\bibinfo {author} {\bibfnamefont {R.}~\bibnamefont {Jozsa}}\ and\ \bibinfo {author} {\bibfnamefont {A.}~\bibnamefont {Miyake}},\ }\bibfield  {title} {\bibinfo {title} {Matchgates and classical simulation of quantum circuits},\ }\href {https://doi.org/10.1098/rspa.2008.0189} {\bibfield  {journal} {\bibinfo  {journal} {Proc. R. Soc. A}\ }\textbf {\bibinfo {volume} {464}},\ \bibinfo {pages} {3089} (\bibinfo {year} {2008}{\natexlab{b}})}\BibitemShut {NoStop}%
\bibitem [{\citenamefont {Collins}\ and\ \citenamefont {Matsumoto}(2009)}]{Collins_2009}%
  \BibitemOpen
  \bibfield  {author} {\bibinfo {author} {\bibfnamefont {B.}~\bibnamefont {Collins}}\ and\ \bibinfo {author} {\bibfnamefont {S.}~\bibnamefont {Matsumoto}},\ }\bibfield  {title} {\bibinfo {title} {{On some properties of orthogonal Weingarten functions}},\ }\href {https://doi.org/10.1063/1.3251304} {\bibfield  {journal} {\bibinfo  {journal} {J. Math. Phys.}\ }\textbf {\bibinfo {volume} {50}},\ \bibinfo {pages} {113516} (\bibinfo {year} {2009})}\BibitemShut {NoStop}%
\bibitem [{\citenamefont {Gross}\ \emph {et~al.}(2007)\citenamefont {Gross}, \citenamefont {Audenaert},\ and\ \citenamefont {Eisert}}]{Designs}%
  \BibitemOpen
  \bibfield  {author} {\bibinfo {author} {\bibfnamefont {D.}~\bibnamefont {Gross}}, \bibinfo {author} {\bibfnamefont {K.}~\bibnamefont {Audenaert}},\ and\ \bibinfo {author} {\bibfnamefont {J.}~\bibnamefont {Eisert}},\ }\bibfield  {title} {\bibinfo {title} {Evenly distributed unitaries: On the structure of unitary designs},\ }\href {https://doi.org/10.1063/1.2716992} {\bibfield  {journal} {\bibinfo  {journal} {J. Math. Phys.}\ }\textbf {\bibinfo {volume} {48}},\ \bibinfo {pages} {052104} (\bibinfo {year} {2007})}\BibitemShut {NoStop}%
\bibitem [{\citenamefont {Garc{\'\i}a-Mart{\'\i}n}\ \emph {et~al.}(2025)\citenamefont {Garc{\'\i}a-Mart{\'\i}n}, \citenamefont {Larocca},\ and\ \citenamefont {Cerezo}}]{Cerezo2025}%
  \BibitemOpen
  \bibfield  {author} {\bibinfo {author} {\bibfnamefont {D.}~\bibnamefont {Garc{\'\i}a-Mart{\'\i}n}}, \bibinfo {author} {\bibfnamefont {M.}~\bibnamefont {Larocca}},\ and\ \bibinfo {author} {\bibfnamefont {M.}~\bibnamefont {Cerezo}},\ }\bibfield  {title} {\bibinfo {title} {Quantum neural networks form gaussian processes},\ }\href {https://doi.org/10.1038/s41567-025-02883-z} {\bibfield  {journal} {\bibinfo  {journal} {Nature Physics}\ }\textbf {\bibinfo {volume} {21}},\ \bibinfo {pages} {1153} (\bibinfo {year} {2025})}\BibitemShut {NoStop}%
\bibitem [{\citenamefont {Jäger}\ \emph {et~al.}(2026)\citenamefont {Jäger}, \citenamefont {Braccia}, \citenamefont {Bermejo}, \citenamefont {Algaba}, \citenamefont {García-Martín},\ and\ \citenamefont {Cerezo}}]{jaeger2026provablescalablequantumgaussian}%
  \BibitemOpen
  \bibfield  {author} {\bibinfo {author} {\bibfnamefont {J.}~\bibnamefont {Jäger}}, \bibinfo {author} {\bibfnamefont {P.}~\bibnamefont {Braccia}}, \bibinfo {author} {\bibfnamefont {P.}~\bibnamefont {Bermejo}}, \bibinfo {author} {\bibfnamefont {M.~G.}\ \bibnamefont {Algaba}}, \bibinfo {author} {\bibfnamefont {D.}~\bibnamefont {García-Martín}},\ and\ \bibinfo {author} {\bibfnamefont {M.}~\bibnamefont {Cerezo}},\ }\href@noop {} {\bibinfo {title} {Provable and scalable quantum gaussian processes for quantum learning}} (\bibinfo {year} {2026}),\ \Eprint {https://arxiv.org/abs/2605.00099} {arXiv:2605.00099 [quant-ph]} \BibitemShut {NoStop}%
\bibitem [{\citenamefont {Miller}\ and\ \citenamefont {Harrison}(2011)}]{miller2011exactenumerationsamplingmatrices}%
  \BibitemOpen
  \bibfield  {author} {\bibinfo {author} {\bibfnamefont {J.~W.}\ \bibnamefont {Miller}}\ and\ \bibinfo {author} {\bibfnamefont {M.~T.}\ \bibnamefont {Harrison}},\ }\href {https://arxiv.org/abs/1104.0323} {\bibinfo {title} {Exact enumeration and sampling of matrices with specified margins}} (\bibinfo {year} {2011}),\ \Eprint {https://arxiv.org/abs/1104.0323} {arXiv:1104.0323} \BibitemShut {NoStop}%
\bibitem [{\citenamefont {Canfield}\ \emph {et~al.}(2008)\citenamefont {Canfield}, \citenamefont {Greenhill},\ and\ \citenamefont {McKay}}]{CANFIELD200832}%
  \BibitemOpen
  \bibfield  {author} {\bibinfo {author} {\bibfnamefont {E.~R.}\ \bibnamefont {Canfield}}, \bibinfo {author} {\bibfnamefont {C.}~\bibnamefont {Greenhill}},\ and\ \bibinfo {author} {\bibfnamefont {B.~D.}\ \bibnamefont {McKay}},\ }\bibfield  {title} {\bibinfo {title} {Asymptotic enumeration of dense 0-1 matrices with specified line sums},\ }\href {https://doi.org/https://doi.org/10.1016/j.jcta.2007.03.009} {\bibfield  {journal} {\bibinfo  {journal} {J. Comb. Th. A}\ }\textbf {\bibinfo {volume} {115}},\ \bibinfo {pages} {32} (\bibinfo {year} {2008})}\BibitemShut {NoStop}%
\bibitem [{\citenamefont {Alba}\ \emph {et~al.}(2019)\citenamefont {Alba}, \citenamefont {Dubail},\ and\ \citenamefont {Medenjak}}]{Alba2019}%
  \BibitemOpen
  \bibfield  {author} {\bibinfo {author} {\bibfnamefont {V.}~\bibnamefont {Alba}}, \bibinfo {author} {\bibfnamefont {J.}~\bibnamefont {Dubail}},\ and\ \bibinfo {author} {\bibfnamefont {M.}~\bibnamefont {Medenjak}},\ }\bibfield  {title} {\bibinfo {title} {Operator entanglement in interacting integrable quantum systems: The case of the rule 54 chain},\ }\href {https://doi.org/10.1103/PhysRevLett.122.250603} {\bibfield  {journal} {\bibinfo  {journal} {Phys. Rev. Lett.}\ }\textbf {\bibinfo {volume} {122}},\ \bibinfo {pages} {250603} (\bibinfo {year} {2019})}\BibitemShut {NoStop}%
\bibitem [{\citenamefont {Bertini}\ \emph {et~al.}(2020{\natexlab{a}})\citenamefont {Bertini}, \citenamefont {Kos},\ and\ \citenamefont {Prosen}}]{Kos2020II}%
  \BibitemOpen
  \bibfield  {author} {\bibinfo {author} {\bibfnamefont {B.}~\bibnamefont {Bertini}}, \bibinfo {author} {\bibfnamefont {P.}~\bibnamefont {Kos}},\ and\ \bibinfo {author} {\bibfnamefont {T.}~\bibnamefont {Prosen}},\ }\bibfield  {title} {\bibinfo {title} {{Operator Entanglement in Local Quantum Circuits II: Solitons in Chains of Qubits}},\ }\href {https://doi.org/10.21468/SciPostPhys.8.4.068} {\bibfield  {journal} {\bibinfo  {journal} {SciPost Phys.}\ }\textbf {\bibinfo {volume} {8}},\ \bibinfo {pages} {068} (\bibinfo {year} {2020}{\natexlab{a}})}\BibitemShut {NoStop}%
\bibitem [{\citenamefont {Alba}(2021)}]{Alba2021}%
  \BibitemOpen
  \bibfield  {author} {\bibinfo {author} {\bibfnamefont {V.}~\bibnamefont {Alba}},\ }\bibfield  {title} {\bibinfo {title} {Diffusion and operator entanglement spreading},\ }\href {https://doi.org/10.1103/PhysRevB.104.094410} {\bibfield  {journal} {\bibinfo  {journal} {Phys. Rev. B}\ }\textbf {\bibinfo {volume} {104}},\ \bibinfo {pages} {094410} (\bibinfo {year} {2021})}\BibitemShut {NoStop}%
\bibitem [{\citenamefont {Murciano}\ \emph {et~al.}(2024)\citenamefont {Murciano}, \citenamefont {Dubail},\ and\ \citenamefont {Calabrese}}]{Murciano_2024}%
  \BibitemOpen
  \bibfield  {author} {\bibinfo {author} {\bibfnamefont {S.}~\bibnamefont {Murciano}}, \bibinfo {author} {\bibfnamefont {J.}~\bibnamefont {Dubail}},\ and\ \bibinfo {author} {\bibfnamefont {P.}~\bibnamefont {Calabrese}},\ }\bibfield  {title} {\bibinfo {title} {More on symmetry resolved operator entanglement},\ }\href {https://doi.org/10.1088/1751-8121/ad30d1} {\bibfield  {journal} {\bibinfo  {journal} {J. Phys. A}\ }\textbf {\bibinfo {volume} {57}},\ \bibinfo {pages} {145002} (\bibinfo {year} {2024})}\BibitemShut {NoStop}%
\bibitem [{\citenamefont {Jonay}\ \emph {et~al.}(2018)\citenamefont {Jonay}, \citenamefont {Huse},\ and\ \citenamefont {Nahum}}]{Jonay2018}%
  \BibitemOpen
  \bibfield  {author} {\bibinfo {author} {\bibfnamefont {C.}~\bibnamefont {Jonay}}, \bibinfo {author} {\bibfnamefont {D.~A.}\ \bibnamefont {Huse}},\ and\ \bibinfo {author} {\bibfnamefont {A.}~\bibnamefont {Nahum}},\ }\href@noop {} {\bibinfo {title} {{Coarse-grained dynamics of operator and state entanglement}}} (\bibinfo {year} {2018}),\ \Eprint {https://arxiv.org/abs/1803.00089} {arXiv:1803.00089} \BibitemShut {NoStop}%
\bibitem [{\citenamefont {Bertini}\ \emph {et~al.}(2020{\natexlab{b}})\citenamefont {Bertini}, \citenamefont {Kos},\ and\ \citenamefont {Prosen}}]{Kos2020}%
  \BibitemOpen
  \bibfield  {author} {\bibinfo {author} {\bibfnamefont {B.}~\bibnamefont {Bertini}}, \bibinfo {author} {\bibfnamefont {P.}~\bibnamefont {Kos}},\ and\ \bibinfo {author} {\bibfnamefont {T.}~\bibnamefont {Prosen}},\ }\bibfield  {title} {\bibinfo {title} {{Operator entanglement in local quantum circuits I: Chaotic dual-unitary circuits}},\ }\href {https://doi.org/10.21468/SciPostPhys.8.4.067} {\bibfield  {journal} {\bibinfo  {journal} {SciPost Phys.}\ }\textbf {\bibinfo {volume} {8}},\ \bibinfo {pages} {067} (\bibinfo {year} {2020}{\natexlab{b}})}\BibitemShut {NoStop}%
\bibitem [{\citenamefont {Bertini}\ \emph {et~al.}(2026)\citenamefont {Bertini}, \citenamefont {Klobas}, \citenamefont {Kos},\ and\ \citenamefont {Malz}}]{Bertini_2026}%
  \BibitemOpen
  \bibfield  {author} {\bibinfo {author} {\bibfnamefont {B.}~\bibnamefont {Bertini}}, \bibinfo {author} {\bibfnamefont {K.}~\bibnamefont {Klobas}}, \bibinfo {author} {\bibfnamefont {P.}~\bibnamefont {Kos}},\ and\ \bibinfo {author} {\bibfnamefont {D.}~\bibnamefont {Malz}},\ }\bibfield  {title} {\bibinfo {title} {Random permutation circuits beyond qubits are quantum chaotic},\ }\href {https://doi.org/10.1103/ql9x-6kzj} {\bibfield  {journal} {\bibinfo  {journal} {Phys. Rev. B}\ }\textbf {\bibinfo {volume} {113}},\ \bibinfo {pages} {L100302} (\bibinfo {year} {2026})}\BibitemShut {NoStop}%
\bibitem [{\citenamefont {Jacoby}\ and\ \citenamefont {Gopalakrishnan}(2026)}]{Jacoby_2026}%
  \BibitemOpen
  \bibfield  {author} {\bibinfo {author} {\bibfnamefont {J.~A.}\ \bibnamefont {Jacoby}}\ and\ \bibinfo {author} {\bibfnamefont {S.}~\bibnamefont {Gopalakrishnan}},\ }\bibfield  {title} {\bibinfo {title} {Long-time limits of local operator entanglement in interacting integrable models},\ }\href {https://doi.org/10.21468/scipostphys.20.6.162} {\bibfield  {journal} {\bibinfo  {journal} {SciPost Phys.}\ }\textbf {\bibinfo {volume} {20}},\ \bibinfo {pages} {162} (\bibinfo {year} {2026})}\BibitemShut {NoStop}%
\bibitem [{\citenamefont {Prosen}\ and\ \citenamefont {Znidari\ifmmode~\check{c}\else \v{c}\fi{}}(2007)}]{Prosen2007}%
  \BibitemOpen
  \bibfield  {author} {\bibinfo {author} {\bibfnamefont {T.}~\bibnamefont {Prosen}}\ and\ \bibinfo {author} {\bibfnamefont {M.}~\bibnamefont {Znidari\ifmmode~\check{c}\else \v{c}\fi{}}},\ }\bibfield  {title} {\bibinfo {title} {Is the efficiency of classical simulations of quantum dynamics related to integrability?},\ }\href {https://doi.org/10.1103/PhysRevE.75.015202} {\bibfield  {journal} {\bibinfo  {journal} {Phys. Rev. E}\ }\textbf {\bibinfo {volume} {75}},\ \bibinfo {pages} {015202(R)} (\bibinfo {year} {2007})}\BibitemShut {NoStop}%
\bibitem [{\citenamefont {Hubig}\ \emph {et~al.}(2017)\citenamefont {Hubig}, \citenamefont {McCulloch},\ and\ \citenamefont {Schollwöck}}]{Hubig_2017}%
  \BibitemOpen
  \bibfield  {author} {\bibinfo {author} {\bibfnamefont {C.}~\bibnamefont {Hubig}}, \bibinfo {author} {\bibfnamefont {I.~P.}\ \bibnamefont {McCulloch}},\ and\ \bibinfo {author} {\bibfnamefont {U.}~\bibnamefont {Schollwöck}},\ }\bibfield  {title} {\bibinfo {title} {Generic construction of efficient matrix product operators},\ }\href {https://doi.org/10.1103/physrevb.95.035129} {\bibfield  {journal} {\bibinfo  {journal} {Phys. Rev. B}\ }\textbf {\bibinfo {volume} {95}},\ \bibinfo {pages} {035129} (\bibinfo {year} {2017})}\BibitemShut {NoStop}%
\bibitem [{\citenamefont {Bianchi}\ \emph {et~al.}(2021)\citenamefont {Bianchi}, \citenamefont {Hackl},\ and\ \citenamefont {Kieburg}}]{PhysRevB.103.L241118}%
  \BibitemOpen
  \bibfield  {author} {\bibinfo {author} {\bibfnamefont {E.}~\bibnamefont {Bianchi}}, \bibinfo {author} {\bibfnamefont {L.}~\bibnamefont {Hackl}},\ and\ \bibinfo {author} {\bibfnamefont {M.}~\bibnamefont {Kieburg}},\ }\bibfield  {title} {\bibinfo {title} {Page curve for fermionic gaussian states},\ }\href {https://doi.org/10.1103/PhysRevB.103.L241118} {\bibfield  {journal} {\bibinfo  {journal} {Phys. Rev. B}\ }\textbf {\bibinfo {volume} {103}},\ \bibinfo {pages} {L241118} (\bibinfo {year} {2021})}\BibitemShut {NoStop}%
\bibitem [{\citenamefont {Bianchi}\ \emph {et~al.}(2022)\citenamefont {Bianchi}, \citenamefont {Hackl}, \citenamefont {Kieburg}, \citenamefont {Rigol},\ and\ \citenamefont {Vidmar}}]{Bianchi_Hackl_Kieburg_Rigol_Vidmar_2022}%
  \BibitemOpen
  \bibfield  {author} {\bibinfo {author} {\bibfnamefont {E.}~\bibnamefont {Bianchi}}, \bibinfo {author} {\bibfnamefont {L.}~\bibnamefont {Hackl}}, \bibinfo {author} {\bibfnamefont {M.}~\bibnamefont {Kieburg}}, \bibinfo {author} {\bibfnamefont {M.}~\bibnamefont {Rigol}},\ and\ \bibinfo {author} {\bibfnamefont {L.}~\bibnamefont {Vidmar}},\ }\bibfield  {title} {\bibinfo {title} {Volume-law entanglement entropy of typical pure quantum states},\ }\href {https://doi.org/10.1103/PRXQuantum.3.030201} {\bibfield  {journal} {\bibinfo  {journal} {PRX Quantum}\ }\textbf {\bibinfo {volume} {3}},\ \bibinfo {pages} {030201} (\bibinfo {year} {2022})}\BibitemShut {NoStop}%
\bibitem [{\citenamefont {Bittel}\ \emph {et~al.}(2026)\citenamefont {Bittel}, \citenamefont {Eisert}, \citenamefont {Leone}, \citenamefont {Mele},\ and\ \citenamefont {Oliviero}}]{Bittel_2026}%
  \BibitemOpen
  \bibfield  {author} {\bibinfo {author} {\bibfnamefont {L.}~\bibnamefont {Bittel}}, \bibinfo {author} {\bibfnamefont {J.}~\bibnamefont {Eisert}}, \bibinfo {author} {\bibfnamefont {L.}~\bibnamefont {Leone}}, \bibinfo {author} {\bibfnamefont {A.~A.}\ \bibnamefont {Mele}},\ and\ \bibinfo {author} {\bibfnamefont {S.~F.}\ \bibnamefont {Oliviero}},\ }\bibfield  {title} {\bibinfo {title} {{A complete theory of the Clifford commutant}},\ }\href {https://doi.org/10.22331/q-2026-07-22-2171} {\bibfield  {journal} {\bibinfo  {journal} {Quantum}\ }\textbf {\bibinfo {volume} {10}},\ \bibinfo {pages} {2171} (\bibinfo {year} {2026})}\BibitemShut {NoStop}%
\bibitem [{\citenamefont {Shenker}\ and\ \citenamefont {Stanford}(2014)}]{Shenker_Stanford_2014}%
  \BibitemOpen
  \bibfield  {author} {\bibinfo {author} {\bibfnamefont {S.~H.}\ \bibnamefont {Shenker}}\ and\ \bibinfo {author} {\bibfnamefont {D.}~\bibnamefont {Stanford}},\ }\bibfield  {title} {\bibinfo {title} {Black holes and the butterfly effect},\ }\href {https://doi.org/10.1007/JHEP03(2014)067} {\bibfield  {journal} {\bibinfo  {journal} {J. High Energy Phys.}\ }\textbf {\bibinfo {volume} {2014}}\bibinfo  {number} { (3)},\ \bibinfo {pages} {67}}\BibitemShut {NoStop}%
\bibitem [{\citenamefont {Dowling}\ \emph {et~al.}(2023)\citenamefont {Dowling}, \citenamefont {Kos},\ and\ \citenamefont {Modi}}]{dowling2023scrambling}%
  \BibitemOpen
\bibfield  {number} {  }\bibfield  {author} {\bibinfo {author} {\bibfnamefont {N.}~\bibnamefont {Dowling}}, \bibinfo {author} {\bibfnamefont {P.}~\bibnamefont {Kos}},\ and\ \bibinfo {author} {\bibfnamefont {K.}~\bibnamefont {Modi}},\ }\bibfield  {title} {\bibinfo {title} {Scrambling {is} {necessary} but {not} {sufficient} for {chaos}},\ }\href {https://doi.org/10.1103/PhysRevLett.131.180403} {\bibfield  {journal} {\bibinfo  {journal} {Phys. Rev. Lett.}\ }\textbf {\bibinfo {volume} {131}},\ \bibinfo {pages} {180403} (\bibinfo {year} {2023})}\BibitemShut {NoStop}%
\bibitem [{\citenamefont {Dowling}(2025)}]{Dowling2025thesis}%
  \BibitemOpen
  \bibfield  {author} {\bibinfo {author} {\bibfnamefont {N.}~\bibnamefont {Dowling}},\ }\emph {\bibinfo {title} {{Multi-time structures from many-body dynamics}}},\ \href {https://doi.org/10.26180/29133272.v1} {Ph.D. thesis},\ \bibinfo  {school} {Monash University}, \bibinfo {address} {Melbourne, Australia} (\bibinfo {year} {2025})\BibitemShut {NoStop}%
\bibitem [{\citenamefont {Voiculescu}(1991)}]{Voiculescu1991}%
  \BibitemOpen
  \bibfield  {author} {\bibinfo {author} {\bibfnamefont {D.}~\bibnamefont {Voiculescu}},\ }\bibfield  {title} {\bibinfo {title} {Limit laws for random matrices and free products},\ }\href {https://doi.org/10.1007/bf01245072} {\bibfield  {journal} {\bibinfo  {journal} {Invent. Math.}\ }\textbf {\bibinfo {volume} {104}},\ \bibinfo {pages} {201–220} (\bibinfo {year} {1991})}\BibitemShut {NoStop}%
\bibitem [{\citenamefont {Vallini}\ and\ \citenamefont {Pappalardi}(2026)}]{Vallini2026longtimefreenessin}%
  \BibitemOpen
  \bibfield  {author} {\bibinfo {author} {\bibfnamefont {E.}~\bibnamefont {Vallini}}\ and\ \bibinfo {author} {\bibfnamefont {S.}~\bibnamefont {Pappalardi}},\ }\bibfield  {title} {\bibinfo {title} {Long-time {F}reeness in the {K}icked {T}op},\ }\href {https://doi.org/10.22331/q-2026-06-08-2129} {\bibfield  {journal} {\bibinfo  {journal} {{Quantum}}\ }\textbf {\bibinfo {volume} {10}},\ \bibinfo {pages} {2129} (\bibinfo {year} {2026})}\BibitemShut {NoStop}%
\bibitem [{\citenamefont {Popescu}\ \emph {et~al.}(2006)\citenamefont {Popescu}, \citenamefont {Short},\ and\ \citenamefont {Winter}}]{Popescu2006}%
  \BibitemOpen
  \bibfield  {author} {\bibinfo {author} {\bibfnamefont {S.}~\bibnamefont {Popescu}}, \bibinfo {author} {\bibfnamefont {A.~J.}\ \bibnamefont {Short}},\ and\ \bibinfo {author} {\bibfnamefont {A.}~\bibnamefont {Winter}},\ }\bibfield  {title} {\bibinfo {title} {Entanglement and the foundations of statistical mechanics},\ }\href {https://doi.org/10.1038/nphys444} {\bibfield  {journal} {\bibinfo  {journal} {Nat. Phys.}\ }\textbf {\bibinfo {volume} {2}},\ \bibinfo {pages} {754} (\bibinfo {year} {2006})}\BibitemShut {NoStop}%
\bibitem [{\citenamefont {Trigueros}\ \emph {et~al.}(2026)\citenamefont {Trigueros}, \citenamefont {Sun}, \citenamefont {Turkeshi}, \citenamefont {Sierant},\ and\ \citenamefont {Tarabunga}}]{trigueros2026unitarydesignsdopedmatchgate}%
  \BibitemOpen
  \bibfield  {author} {\bibinfo {author} {\bibfnamefont {F.~B.}\ \bibnamefont {Trigueros}}, \bibinfo {author} {\bibfnamefont {Z.-H.}\ \bibnamefont {Sun}}, \bibinfo {author} {\bibfnamefont {X.}~\bibnamefont {Turkeshi}}, \bibinfo {author} {\bibfnamefont {P.}~\bibnamefont {Sierant}},\ and\ \bibinfo {author} {\bibfnamefont {P.~S.}\ \bibnamefont {Tarabunga}},\ }\href {https://arxiv.org/abs/2606.23800} {\bibinfo {title} {Unitary designs from doped matchgate circuits}} (\bibinfo {year} {2026}),\ \Eprint {https://arxiv.org/abs/2606.23800} {arXiv:2606.23800 [quant-ph]} \BibitemShut {NoStop}%
\bibitem [{\citenamefont {Debertolis}(2025)}]{debertolis2025naturalsuper}%
  \BibitemOpen
  \bibfield  {author} {\bibinfo {author} {\bibfnamefont {M.}~\bibnamefont {Debertolis}},\ }\href@noop {} {\bibinfo {title} {Natural super-orbitals representation of many-body operators}} (\bibinfo {year} {2025}),\ \Eprint {https://arxiv.org/abs/2507.10690} {arXiv:2507.10690} \BibitemShut {NoStop}%
\bibitem [{\citenamefont {Sierant}\ \emph {et~al.}(2026{\natexlab{b}})\citenamefont {Sierant}, \citenamefont {Stornati},\ and\ \citenamefont {Turkeshi}}]{Sierant_2026}%
  \BibitemOpen
  \bibfield  {author} {\bibinfo {author} {\bibfnamefont {P.}~\bibnamefont {Sierant}}, \bibinfo {author} {\bibfnamefont {P.}~\bibnamefont {Stornati}},\ and\ \bibinfo {author} {\bibfnamefont {X.}~\bibnamefont {Turkeshi}},\ }\bibfield  {title} {\bibinfo {title} {Fermionic magic resources of quantum many-body systems},\ }\bibfield  {journal} {\bibinfo  {journal} {PRX Quantum}\ }\textbf {\bibinfo {volume} {7}},\ \href {https://doi.org/10.1103/3yx4-1j27} {10.1103/3yx4-1j27} (\bibinfo {year} {2026}{\natexlab{b}})\BibitemShut {NoStop}%
\bibitem [{\citenamefont {Turkeshi}\ \emph {et~al.}(2026)\citenamefont {Turkeshi}, \citenamefont {Sierant},\ and\ \citenamefont {Tarabunga}}]{turkeshi2026williamsonmajorizationtheoryfermionic}%
  \BibitemOpen
  \bibfield  {author} {\bibinfo {author} {\bibfnamefont {X.}~\bibnamefont {Turkeshi}}, \bibinfo {author} {\bibfnamefont {P.}~\bibnamefont {Sierant}},\ and\ \bibinfo {author} {\bibfnamefont {P.~S.}\ \bibnamefont {Tarabunga}},\ }\href@noop {} {\bibinfo {title} {Williamson majorization theory of fermionic non-gaussianity}} (\bibinfo {year} {2026}),\ \Eprint {https://arxiv.org/abs/2608.10140} {arXiv:2608.10140 [quant-ph]} \BibitemShut {NoStop}%
\bibitem [{\citenamefont {Eisert}(2021)}]{Eisert2021EntanglingPower}%
  \BibitemOpen
  \bibfield  {author} {\bibinfo {author} {\bibfnamefont {J.}~\bibnamefont {Eisert}},\ }\bibfield  {title} {\bibinfo {title} {Entangling power and quantum circuit complexity},\ }\href {https://doi.org/10.1103/PhysRevLett.127.020501} {\bibfield  {journal} {\bibinfo  {journal} {Phys. Rev. Lett.}\ }\textbf {\bibinfo {volume} {127}},\ \bibinfo {pages} {020501} (\bibinfo {year} {2021})},\ \Eprint {https://arxiv.org/abs/2104.03332} {arXiv:2104.03332 [quant-ph]} \BibitemShut {NoStop}%
\bibitem [{\citenamefont {Miller}\ \emph {et~al.}(2025)\citenamefont {Miller}, \citenamefont {Favre}, \citenamefont {Holmes}, \citenamefont {Özlem Salehi}, \citenamefont {Chakraborty}, \citenamefont {Nykänen}, \citenamefont {Zimborás}, \citenamefont {Glos},\ and\ \citenamefont {García-Pérez}}]{miller2025simulation}%
  \BibitemOpen
  \bibfield  {author} {\bibinfo {author} {\bibfnamefont {A.}~\bibnamefont {Miller}}, \bibinfo {author} {\bibfnamefont {J.}~\bibnamefont {Favre}}, \bibinfo {author} {\bibfnamefont {Z.}~\bibnamefont {Holmes}}, \bibinfo {author} {\bibnamefont {Özlem Salehi}}, \bibinfo {author} {\bibfnamefont {R.}~\bibnamefont {Chakraborty}}, \bibinfo {author} {\bibfnamefont {A.}~\bibnamefont {Nykänen}}, \bibinfo {author} {\bibfnamefont {Z.}~\bibnamefont {Zimborás}}, \bibinfo {author} {\bibfnamefont {A.}~\bibnamefont {Glos}},\ and\ \bibinfo {author} {\bibfnamefont {G.}~\bibnamefont {García-Pérez}},\ }\href@noop {} {\bibinfo {title} {Simulation of fermionic circuits using {Majorana} propagation}} (\bibinfo {year} {2025}),\ \Eprint {https://arxiv.org/abs/2503.18939} {arXiv:2503.18939} \BibitemShut {NoStop}%
\bibitem [{\citenamefont {D’Anna}\ \emph {et~al.}(2026)\citenamefont {D’Anna}, \citenamefont {Nys},\ and\ \citenamefont {Carrasquilla}}]{D_Anna_2026}%
  \BibitemOpen
  \bibfield  {author} {\bibinfo {author} {\bibfnamefont {M.}~\bibnamefont {D’Anna}}, \bibinfo {author} {\bibfnamefont {J.}~\bibnamefont {Nys}},\ and\ \bibinfo {author} {\bibfnamefont {J.}~\bibnamefont {Carrasquilla}},\ }\bibfield  {title} {\bibinfo {title} {Majorana string simulation of nonequilibrium dynamics in two-dimensional lattice fermion systems},\ }\bibfield  {journal} {\bibinfo  {journal} {Physical Review B}\ }\textbf {\bibinfo {volume} {114}},\ \href {https://doi.org/10.1103/f6mn-jc3r} {10.1103/f6mn-jc3r} (\bibinfo {year} {2026})\BibitemShut {NoStop}%
\bibitem [{\citenamefont {Nunez-Fernandez}\ \emph {et~al.}(2025)\citenamefont {Nunez-Fernandez}, \citenamefont {Debertolis},\ and\ \citenamefont {Florens}}]{2025resolvingspacetimestructuresquantum}%
  \BibitemOpen
  \bibfield  {author} {\bibinfo {author} {\bibfnamefont {Y.}~\bibnamefont {Nunez-Fernandez}}, \bibinfo {author} {\bibfnamefont {M.}~\bibnamefont {Debertolis}},\ and\ \bibinfo {author} {\bibfnamefont {S.}~\bibnamefont {Florens}},\ }\href@noop {} {\bibinfo {title} {Resolving space-time structures of quantum impurities with a numerically exact few-body algorithm}} (\bibinfo {year} {2025}),\ \Eprint {https://arxiv.org/abs/2503.13706} {arXiv:2503.13706} \BibitemShut {NoStop}%
\bibitem [{\citenamefont {Klich}(2014)}]{Klich_2014}%
  \BibitemOpen
  \bibfield  {author} {\bibinfo {author} {\bibfnamefont {I.}~\bibnamefont {Klich}},\ }\bibfield  {title} {\bibinfo {title} {A note on the full counting statistics of paired fermions},\ }\href {https://doi.org/10.1088/1742-5468/2014/11/p11006} {\bibfield  {journal} {\bibinfo  {journal} {Journal of Statistical Mechanics: Theory and Experiment}\ }\textbf {\bibinfo {volume} {2014}},\ \bibinfo {pages} {P11006} (\bibinfo {year} {2014})}\BibitemShut {NoStop}%
\bibitem [{\citenamefont {Chang}\ \emph {et~al.}(2026)\citenamefont {Chang}, \citenamefont {Krumtünger}, \citenamefont {Larocca},\ and\ \citenamefont {West}}]{chang2026classicalshadowssymmetricspaces}%
  \BibitemOpen
  \bibfield  {author} {\bibinfo {author} {\bibfnamefont {R.}~\bibnamefont {Chang}}, \bibinfo {author} {\bibfnamefont {M.}~\bibnamefont {Krumtünger}}, \bibinfo {author} {\bibfnamefont {M.}~\bibnamefont {Larocca}},\ and\ \bibinfo {author} {\bibfnamefont {M.}~\bibnamefont {West}},\ }\href@noop {} {\bibinfo {title} {Classical shadows over symmetric spaces}} (\bibinfo {year} {2026}),\ \Eprint {https://arxiv.org/abs/2605.05518} {arXiv:2605.05518 [quant-ph]} \BibitemShut {NoStop}%
\bibitem [{\citenamefont {S{\"u}nderhauf}\ \emph {et~al.}(2019)\citenamefont {S{\"u}nderhauf}, \citenamefont {Piroli}, \citenamefont {Qi}, \citenamefont {Schuch},\ and\ \citenamefont {Cirac}}]{Piroli2019}%
  \BibitemOpen
  \bibfield  {author} {\bibinfo {author} {\bibfnamefont {C.}~\bibnamefont {S{\"u}nderhauf}}, \bibinfo {author} {\bibfnamefont {L.}~\bibnamefont {Piroli}}, \bibinfo {author} {\bibfnamefont {X.-L.}\ \bibnamefont {Qi}}, \bibinfo {author} {\bibfnamefont {N.}~\bibnamefont {Schuch}},\ and\ \bibinfo {author} {\bibfnamefont {J.~I.}\ \bibnamefont {Cirac}},\ }\bibfield  {title} {\bibinfo {title} {Quantum chaos in the brownian syk model with large finite n : Otocs and tripartite information},\ }\href {https://doi.org/10.1007/JHEP11(2019)038} {\bibfield  {journal} {\bibinfo  {journal} {Journal of High Energy Physics}\ }\textbf {\bibinfo {volume} {2019}},\ \bibinfo {pages} {38} (\bibinfo {year} {2019})}\BibitemShut {NoStop}%
\bibitem [{\citenamefont {West}\ \emph {et~al.}(2026)\citenamefont {West}, \citenamefont {Cerezo},\ and\ \citenamefont {Larocca}}]{west2026strongmatchgatedesignsnearly}%
  \BibitemOpen
  \bibfield  {author} {\bibinfo {author} {\bibfnamefont {M.}~\bibnamefont {West}}, \bibinfo {author} {\bibfnamefont {M.}~\bibnamefont {Cerezo}},\ and\ \bibinfo {author} {\bibfnamefont {M.}~\bibnamefont {Larocca}},\ }\href@noop {} {\bibinfo {title} {Strong matchgate designs in nearly optimal depth}} (\bibinfo {year} {2026}),\ \Eprint {https://arxiv.org/abs/2609.26677} {arXiv:2609.26677 [quant-ph]} \BibitemShut {NoStop}%
\bibitem [{\citenamefont {Usoltcev}\ \emph {et~al.}(2026)\citenamefont {Usoltcev}, \citenamefont {Wille}, \citenamefont {Eisert},\ and\ \citenamefont {Altland}}]{AltlandMatchgates}%
  \BibitemOpen
  \bibfield  {author} {\bibinfo {author} {\bibfnamefont {M.}~\bibnamefont {Usoltcev}}, \bibinfo {author} {\bibfnamefont {C.}~\bibnamefont {Wille}}, \bibinfo {author} {\bibfnamefont {J.}~\bibnamefont {Eisert}},\ and\ \bibinfo {author} {\bibfnamefont {A.}~\bibnamefont {Altland}},\ }\href@noop {} {\bibinfo {title} {Continuum field theory of matchgate tensor network ensembles}} (\bibinfo {year} {2026}),\ \Eprint {https://arxiv.org/abs/arXiv:2603.06202} {arXiv:arXiv:2603.06202} \BibitemShut {NoStop}%
\bibitem [{\citenamefont {Bayona-Pena}\ \emph {et~al.}(2026)\citenamefont {Bayona-Pena}, \citenamefont {Gramellini},\ and\ \citenamefont {Piroli}}]{bayonapena2026operato}%
  \BibitemOpen
  \bibfield  {author} {\bibinfo {author} {\bibfnamefont {P.}~\bibnamefont {Bayona-Pena}}, \bibinfo {author} {\bibfnamefont {E.}~\bibnamefont {Gramellini}},\ and\ \bibinfo {author} {\bibfnamefont {L.}~\bibnamefont {Piroli}},\ }\href@noop {} {\bibinfo {title} {Operator-space entanglement entropy in random fermionic gaussian dynamics}} (\bibinfo {year} {2026}),\ \Eprint {https://arxiv.org/abs/2610.21723} {arXiv:2610.21723 [quant-ph]} \BibitemShut {NoStop}%
\end{thebibliography}
